\documentclass[11pt,leqno,textwidth=10cm]{amsart}
\usepackage{amssymb}
\allowdisplaybreaks
\usepackage{mathrsfs}
\usepackage{amsmath, amssymb, mathtools, nicefrac}
\usepackage{csquotes}
\usepackage{abstract}
\usepackage[hidelinks]{hyperref}
\usepackage{marginnote}
\usepackage{xcolor}
\usepackage{graphicx}
\usepackage{orcidlink}

\usepackage{etoolbox}
\makeatletter
\patchcmd{\@maketitle}{\newpage}{}{}{} 
\makeatother
\numberwithin{equation}{section}

\theoremstyle{definition}
\newtheorem{definition}{Definition}[section]

\newtheorem{remark}[definition]{Remark}
\theoremstyle{plain}
\newtheorem{theorem}[definition]{Theorem}
\newtheorem{lemma}[definition]{Lemma}
\newtheorem{corollary}[definition]{Corollary}
\newtheorem{prop}[definition]{Proposition}
\newtheorem{assumption}[definition]{Assumption}

\newcommand{\A}{\mathbf{A}}
\newcommand{\B}{\mathbf{B}}

\newcommand{\E}{\mathfrak{E}}
\newcommand{\f}{\change{f}}
\newcommand{\g}{\overline{g}}
\newcommand{\G}{\underline{G}}

\newcommand{\I}{\mathbb{I}}
\renewcommand{\j}{\jmath}

\renewcommand{\L}{\mathcal{L}}
\newcommand{\M}{\overline{M}}

\newcommand{\N}{\mathbb{N}}

\newcommand{\nabbar}{\overline{\nabla}}

\renewcommand{\P}{\mathcal{P}}

\newcommand{\R}{\mathbb{R}}

\newcommand{\Ric}{\text{\normalfont{Ric}}}
\newcommand{\Riem}{\text{\normalfont{Riem}}}
\newcommand{\supp}{\mathrm{supp}}
\renewcommand{\S}{\mathbb{S}}

\newcommand{\vol}[1]{{\text{\normalfont{vol}}}_{#1}}
\newcommand{\X}{\bm{X}}
\newcommand{\Z}{\mathbb{Z}}
\renewcommand{\epsilon}{\varepsilon}

\newcommand{\phibar}{{\phi}_{FLRW}}
\newcommand{\del}{\partial}
\newcommand{\Lap}{\Delta}
\renewcommand{\div}{\text{\normalfont{div}}}
\newcommand{\curl}{\text{\normalfont{curl}}}

\newcommand{\Gamhat}{\hat{\Gamma}}
\newcommand{\Gambar}{\overline{\Gamma}}
\newcommand{\nabhat}{\hat{\nabla}}
\newcommand{\nabsak}{\underline{\nabla}}
\newcommand{\Gamsak}{\underline{\Gamma}}
\newcommand{\Ltilde}{\tilde{\L}}
\newcommand{\sak}[1]{\underline{#1}}

\newcommand{\numberthis}{\addtocounter{equation}{1}\tag{\theequation}}

\renewcommand{\theequation}{\arabic{section}.\arabic{equation}}

\usepackage{bm}
\newcommand{\RE}{\mathcal{E}}
\newcommand{\RB}{\mathcal{B}}
\newcommand{\epsilonLC}{\bm{\epsilon}}

\newcommand{\change}[1]{#1}
\newcommand{\changediss}[1]{#1}
\newcommand{\delete}[1]{}
\usepackage{cancel}
\newcommand{\deletemath}[1]{}

\title[On the past maximal development of near-FLRW data for the ESFV system]{On the past maximal development of near-FLRW data for the Einstein scalar-field Vlasov system}
\author{David Fajman, Liam Urban}
\usepackage[foot]{amsaddr}
\address{
\begin{tabular}[h]{l@{\extracolsep{8em}}l} 
David Fajman  & Liam Urban \\
Faculty of Physics & Faculty of Mathematics\\ 
University of Vienna & University of Vienna \\
Boltzmanngasse 5 & Oskar-Morgenstern-Platz 1 \\
1090 Vienna, Austria & 1090 Vienna, Austria\\
david.fajman@ univie.ac.at & liam.urban@ univie.ac.at \\
\orcidlink{0000-0003-3034-6232}\ 0000-0003-3034-6232 & \orcidlink{0000-0001-9185-9627}\ 0000-0001-9185-9627
\end{tabular}
}

\begin{document}

\maketitle

\begin{abstract}
We show that the maximal globally hyperbolic development of near-FLRW initial data for the Einstein scalar-field Vlasov system exhibits stable Big Bang formation in the collapsing direction. The solutions exhibit stable Kretschmann scalar blow-up, causing the spacetime to become causally geodesically past incomplete, and are asymptotically velocity term dominated. This is the first stability result for the Einstein equations in the collapsing spacetime direction in presence of Vlasov matter that does not rely on any symmetry assumptions. \change{Furthermore, the Vlasov distribution remains close to that of the FLRW solution as a function on the co-mass shell, and so does its momentum support if one assumes it to be close to that of the FLRW distribution initially. On the other hand, the leading order terms in components of the Vlasov energy-momentum tensor exbihit an offset in asymptotic order controlled by the perturbation size, and when viewed on the mass shell, the distribution asymptotically concentrates in certain preferred velocity directions. }To ensure that this behaviour is sufficiently mitigated by the scalar field, we crucially exploit a scaling hierarchy between horizontal and vertical derivatives in the \change{commuted Vlasov }equation. 
\end{abstract}

\section{Introduction}

\subsection{Setting and Results}

The goal of this paper is to study the past behaviour of the Einstein equations in presence of both scalar field and massive ($m=1$) or massless ($m=0$) Vlasov matter near Friedman-Lema\^\i tre-Robertson-Walker (FLRW) solutions. On a \change{$(3+1)$}--dimensional spacetime $(\M,\g)$, these equations, which we refer to as the Einstein scalar-field Vlasov (ESFV) system, are given as follows:
\begin{subequations}\label{eq:EVSF}
\begin{align}
\label{eq:EVSF1}\Ric[\g]_{\mu\nu}-\frac12R[\g]\g_{\mu\nu}=&\,8\pi \left(T^{SF}_{\mu\nu}[\g,\phi]+T^{Vl}_{\mu\nu}[\g,\f]\right)\\
\label{eq:SFT}T^{SF}_{\mu\nu}[\g,\phi]=&\,\nabbar_{\mu}\phi\nabbar_{\nu}\phi-\frac12\g_{\mu\nu}\nabbar^\alpha\phi\nabbar_\alpha\phi\\
\label{eq:VT}T^{Vl}_{\mu\nu}[\g,\f]=&\,\int_{P_{(t,x)}}p_\mu p_\nu\f(t,x,p)\vol{P_{(t,x)}}\\
\label{eq:EVSF2}\square_{\g}\phi=&\,0\,\\
\label{eq:EVSF3}\mathcal{X}\f=&\,0
\end{align}
\end{subequations}
Here, $\mathcal{X}$ refers to the geodesic spray with respect to $\g$ on the \change{co-mass shell 
\[P=\bigsqcup_{(t,x)\in\M}P_{(t,x)},\ \text{where }P_{(t,x)}=\left\{p\in T^\ast_{(t,x)}\M\, \vert\, p_\mu p^\mu=-m^2,\ p^0>0\right\}\,.
\]}
FLRW solutions take the form
\begin{subequations}\label{eq:FLRW}
\begin{equation}
(I\times M,-dt^2+a(t)^2\gamma)\,,
\end{equation}
where $(M,\gamma)$ is a closed $3$-manifold of constant sectional curvature $\kappa\in\R$ and the scale factor $a\in C^\infty(I)$ satisfies the Friedman equation \eqref{eq:Friedman} with $a(0)=0$. For $C\in \R^+$ and ${\mathcal{F}}\in C^\infty_c(\R_0^+,\R_0^+)$, the scalar field $\phi_{FLRW}\in C^\infty(\M)$ and the Vlasov distribution function $\f_{FLRW}\in C^\infty(P)$ are given by
\begin{gather}
\del_t\phi_{FLRW}=Ca^{-3},\ \nabla\phi_{FLRW}=0\ \text{and}\\
\f_{FLRW}(t,x,p)=\change{{\mathcal{F}}((\gamma^{-1})^{ij}p_ip_j)={\mathcal{F}}(a(t)^{2}\lvert p\rvert_{g_{FLRW}}^2)}\,.
\end{gather}
\end{subequations}
We remark that the scale factor behaves like $t^\frac13$ toward the past singularity at $t=0$ (see Lemma \ref{lem:scale-factor}). The main result, stated in full in Theorem \ref{thm:main}, can be summarized as follows.

\begin{theorem}[Past stability of FLRW solutions to the ESFV system, short version]\label{thm:main-intro} Let $M$ be a closed Riemannian $3$-manifold that admits a metric $\gamma$ of constant sectional curvature. Let $\left(M,\mathring{g},\mathring{k},\mathring{\pi},\mathring{\psi},\mathring{f}\right)$
be CMC initial data in the sense discussed in Section \ref{subsubsec:initial-data} to the ESFV system \eqref{eq:EVSF} that is sufficiently close to \change{FLRW initial data with isotropic matter given by a non-trivial scalar field and Vlasov initial data $\mathring{f}_{FLRW}$ with compact momentum support. }In the case of massless particles, additionally assume that \change{$\mathring{f}_{FLRW}$ and $\mathring{f}$ vanish }in an open neighbourhood of the zero section of $P$. Then, the following statements hold.\\
\begin{enumerate}
\item The past maximal globally hyperbolic development $(\M,\g,\phi,\f)$  within the Einstein scalar-field Vlasov system is foliated by CMC hypersurfaces $M_s=t^{-1}(s),\ s>0,\ $ and exhibits stable Kretschmann scalar blow-up of order $t^{-4}$ toward the Big Bang singularity at $t=0$.\\
\item \change{The solution variables $(g,k,\phi,\f)$ exhibit asymptotically velocity term dominated behaviour towards the Big Bang singularity. In particular, in an appropriate function space, $\f(t,\cdot,\cdot)$ converges to a limiting distribution $f_{\normalfont Bang}$ on the cotangent bundle of the Big Bang hypersurface. This distribution is close to the initial distribution and vanishes for sufficiently large and, in the massless case, sufficiently small momenta.}
\end{enumerate}
\end{theorem}

\change{If one additionally assumes that the momentum support of $\mathring{f}$ is close to that of the FLRW solution, then so is that of $f_{\normalfont Bang}$, see Corollary \ref{cor:main-thm-support}. We note that, when viewing the Vlasov distribution on the mass shell rather than the co-mass shell, the particle momenta concentrate in directions determined by the leading order asymptotics of the shear, as is discussed in Remark \ref{rem:save-my-skin}.\\
 
Additionally, for $\kappa>0$, Theorem \ref{thm:main-intro} }also implies stable Big Crunch formation toward the future: The respective FLRW solutions exist on $\M=(0,T)\times M$ for some $T\in\R^+$ and are preserved under the time reflection $t\in(0,T)\leadsto T-t$ (see Lemma \ref{lem:FLRW} and \ref{lem:scale-factor}). Thus, the entire analysis toward the Big Bang singularity can be re-applied approaching the Big Crunch singularity in the case of positive spatial curvature. \\

This is, to our knowledge, the first asymptotic stability result in the contracting direction involving Vlasov matter that does not require any symmetry assumptions. Furthermore, since Vlasov matter asymptotically matches the equation of state of radiation (see Section \ref{subsubsec:intro-vlasov}), this extends the regime in which the presence of a scalar field stabilizes Big Bang formation compared to the recent result \cite{BO24EESF} (see the discussion in Section \ref{subsubsec:euler-sf}). 

\subsection{Background}\label{subsec:bg}

\subsubsection{Singularity formation and quiescent cosmology}\label{subsubsec:quiesc}

Since the singularity theorems by Penrose \cite{Pen65} and Hawking \cite{Hawk67}, understanding the nature of singularity formation has been one of the core pursuits of mathematical relativity. While these theorems describe under which conditions spacetimes become causally incomplete, they provide no information about the type of singularity. In cosmological settings, the two following conjectures concerning past singularity formation have garnered particular attention:\\

Firstly, the \textit{Strong Cosmic Censorship (SCC) conjecture} states that the maximal development of generic initial data becomes inextendible when it becomes geodesically incomplete. In view of results in the homogeneous setting by Chruściel and Rendall \cite{CR95} and Ringström \cite{Ring09} as well in Gowdy symmetry by Ringström \cite{Ring09Annals}, the cosmological version of this conjecture usually specifies this to be $C^2$-inextendibility caused by blow-up of the Kretschmann scalar. Since directly computing the asymptotic behaviour of the Kretschmann scalar given general initial data is often difficult if not impossible, evidence toward this conjecture is commonly provided by showing that the Kretschmann scalar blow-up of a particular explicit solution is stable within solutions to the Einstein equations. We will discuss the available results in this direction for FLRW and Kasner spacetimes in Section \ref{subsubsec:quiesc-context}.\\

Secondly, the \textit{BKL conjecture} (originally formulated in the work of Belinskiǐ, Khalatnikov and Lifshits \cite{BKL70}) states that cosmological solutions to the Einstein equations are asymptotically velocity term dominated (AVTD) toward the past, i.e., the maximal development exhibits the same asymptotics as the formal solution of the truncated system obtained from the Einstein equations by dropping terms containing spatial derivatives. This behaviour should then cause the metric evolution to become highly oscillatory and chaotic, often referred to as \enquote{Mixmaster} behaviour (see \cite{Misner69}). Originally, it was conjectured in \cite{BKL70} that the choice of matter model does not influence the AVTD asymptotics. However, Belinskiǐ and Khalatnikov later observed in \cite{BK73} that scalar field and stiff fluid matter appear to have an exceptional role: These matter models exhibit \textit{quiescent Big Bang formation}, i.e.,~ while the AVTD picture still holds true, this choice of matter dampens any potential oscillations and generates stable past asymptotics. In particular, this would mean that the SCC conjecture holds for such solutions, since the Kretschmann scalar blows up for FLRW solutions with non-trivial scalar field. For a recent comprehensive result on which types of initial data to the Einstein nonlinear scalar field system generate quiescent singularities, see \cite{GPR23}.

\subsubsection{Vlasov matter}\label{subsubsec:intro-vlasov}

The universe in its current state can, at large scales, reasonably be modelled as a collisionless self-gravitating gas, where stars are viewed as particles with rest mass $m$ and are distributed according to a distribution function $f$ (see \cite[\change{Chapters }1.2 and 4.1]{BT08}). Since one needs to require that all particles of the gas travel along timelike geodesics, $f$ must satisfy the Vlasov equation \eqref{eq:EVSF3}. The asymptotic behaviour of massive Vlasov matter on FLRW spacetimes in particular corresponds very well to the dominant matter components expected in standard cosmological models in the early and late universe (see, for example, \cite[Sections 5.5 and 6.4]{Ryd16}): Toward the past, where the universe is believed to have been radiation dominated, the density of massive and massless Vlasov matter behaves like $a^{-4}$ and the equation of state approaches that of a radiation fluid ($\mathfrak{p}^{Vl}\approx \frac13\rho^{Vl}$; this holds with equality in the massless case). Toward the future, where one expects matter domination, the pressure converges to $0$, approximating the equation of state of dust, and the Vlasov density is of order $a^{-3}$. Thus, even though a fully physical model of gas toward the Big Bang may need to include collisions, Vlasov matter is already a realistic matter type to consider from the perspective of FLRW cosmology. 

While we will discuss some relevant results regarding past existence and asymptotic behaviour in Section \ref{subsec:intro-context}, we also refer to \cite{Andr11} for a more general overview of results regarding the Einstein-Vlasov system.

\subsection{Relationship with previous work}\label{subsec:intro-context}

This paper both expands the known regime of quiescent Big Bang formation caused by scalar fields and, to our knowledge, is the first result on past stability in presence of Vlasov matter without symmetry. Thus, we will proceed as follows: First, we compare our approach with the existing literature on Big Bang formation for the Einstein scalar-field system. Then, we contrast them with results for Einstein scalar-field fluid systems, since modelling the universe as a fluid is a common approach that is closely related to the Vlasov approach (see Section \ref{subsubsec:intro-vlasov}). Finally, we compare our results to those for Einstein (scalar-field) Vlasov matter.\changediss{\phantom{m}\\}

\subsubsection{Quiescent singularity formation for scalar field matter}\label{subsubsec:quiesc-context}


Rodnianski and Speck proved the first result \cite{Rodnianski2018} regarding past stability of the Einstein scalar-field system without symmetries, showing that Big Bang formation in FLRW spacetimes with spatial geometry $\mathbb{T}^3$ is linearly stable and asymptotically velocity term dominated, i.e., the spacetime and matter variables asymptotically behave like the solution to the velocity dominated ODE. This was extended to nonlinear stability in \cite{Rodnianski2014}, and then in turn to spherical and closed hyperbolic spatial hypersurfaces in \cite{Speck2018}, respectively \cite{FU23}. In the case of flat spatial geometry, the regime of past asymptotic stability has been extended to moderately anisotropic generalized Kasner spacetimes\footnote{\textit{Generalized Kasner spacetimes} are solutions to the Einstein scalar-field system of the form
\[\g_{Kas}=-dt^2+\sum_{i=1}^dt^{2p_i}dx^i\otimes dx^i, \quad\phi_{Kas}=A\log(t), \quad\text{where } \sum_{i=1}^dp_i=1,\quad \sum_{i=1}^dp_i^2=1-A^2,\quad A\in[0,1].\]
Note that, except for the FLRW case $p_1=\dots=p_d$, these are spatially homogeneous but not isotropic. Solutions in the vacuum case ($A=0$) are simply referred to as Kasner spacetimes.} in \cite{RodSp22} and then to the all subcritical\footnote{A generalized Kasner spacetime is called subcritical if the Kasner exponents $(p_i)_{i=1,\dots, d}$ satisfy $p_i+p_j-p_k<1$ for all $i,j,k=1,\dots,d$.} ones in \cite{RodSpFou20}. In all of these cases, the Kretschmann scalar exhibits stable blow-up at order $t^{-4}$ approaching the Big Bang ($t\downarrow 0$); the solutions are AVTD. Furthermore, Big Bang stability near FLRW spacetimes can be localized, see \cite{BeyOl21}. These recent developments indicate that spatial geometry has little bearing on the past asymptotic behaviour of non-vacuum solutions to the Einstein scalar-field equations. In fact, in their recent work, \change{Oude Groeniger}, Petersen and Ringström \cite{GPR23} were able to establish a very general criterion for quiescent Big Bang formation in the non-linear Einstein scalar-field system, essentially only requiring the mean curvature to be large enough compared to the initial data and the eigenvalues of the initial Weingarten map are to be sufficiently far from one another (see \cite[Theorem 12]{GPR23}). In terms of Big Bang stability, the latter may be dropped near a wide class of cosmological model solutions, including FLRW spacetimes (see \cite[Theorem 49]{GPR23}). \\

This independence of spatial geometry is reflected in our result insofar as that we allow space in the reference FLRW spacetime to be of arbitrary constant sectional curvature. To be more precise, the energy formalism we use for the scalar field and spacetime variables extends the approach in \cite{FU23} to FLRW spacetimes with spatial geometry $\mathbb{T}^3$ and $\S^3$, and then couples this with an energy mechanism for Vlasov matter. However, restricting ourselves to isotropic reference spacetimes is essential when Vlasov matter is present. Else, both the tracefree part of $T^{Vl}_{ij}$ and the energy flux no longer vanish in the reference spacetime, and one would likely lose control of the shear and related constrained quantities. Thus, while it might be possible to integrate Vlasov stability analysis into the frameworks of \cite{RodSpFou20} or \cite{GPR23}, which both included anisotropic spacetimes, there seems to be no loss of generality in restricting ourselves to a methodology that only considers the case of an isotropic reference. In particular, we note that one likely cannot use the same approach as in \cite{BeyOl21} to establish localized stability in the scalar-field Vlasov setting, as this seems bound to run into similar issues as radiation fluids do in \cite{BO24EESF} (see Section \ref{subsubsec:euler-sf}). \changediss{\phantom{m}\\}

\subsubsection{Past stability in Einstein (scalar-field) Euler models}\label{subsubsec:euler-sf}

Given how robust scalar field matter appears to be when considering past asymptotics of cosmological spacetimes, it is a natural next step to analyze whether it can dampen the oscillatory behaviour one would otherwise expect from matter models with AVTD behaviour, as discussed in Section \ref{subsubsec:quiesc}.\\

For example, the dynamics of spatially homogeneous solutions to the Einstein-Euler equations with $c_s^2\in(0,1)$ are well understood to asymptote toward vacuum solutions, often of Kasner type, that exhibit anisotropic curvature blow-up (see, for example, \cite{WE97, Rin01}). Additionally, Beyer and Le Floch analysed homogeneous fluids on fixed $(1+d)\,$ -- dimensional Kasner backgrounds with linear equation of state $\mathfrak{p}=c_s^2\rho$ in \cite{BeyLeF17}. In this setting, the fluid velocity diverges toward the Big Bang in the supercritical case $c_s^2<\frac{1}{d}$, is time-independent in the critical case $c_s^2=\frac{1}{d}$, and vanishes toward the Big Bang in the subcritical case $1\geq c_s^2>\frac{1}{d}$ (see \cite[Theorem 3.3]{BeyLeF17}). In \cite[Theorems 4.2-4.3]{BeyLeF17}, they showed that subcritical and certain critical solutions to the Einstein-Euler system in Gowdy symmetry are velocity dominated, while this can even break down on the Kasner background in the supercritical case (see \cite[Section 4.3]{BeyLeF17}).\\

Given the AVTD nature of Big Bang formation in the Einstein scalar-field system, one could thus hope that coupling with subcritical fluids preserves quiescence. Indeed, Beyer and Oliynyk were able to prove in \cite{BO24EESF} that subcritical FLRW solutions to the Einstein scalar-field Euler system are past stable (see \cite[Theorems 3.1 and 3.6]{BO24EESF}). However, the critical case cannot be controlled by the Fuchsian approach in \cite{BO24EESF} since leading terms in the evolution of the fluid velocity become indefinite (see \cite[Remark 9.2]{BO24EESF}), which may cause the fluid to degenerate. However, the results of \cite{BeyLeF17} indicate that the critical regime may still exhibit stability under stronger assumptions.\\

Since the critical case corresponds to radiation fluids in spacetime dimension $1+d=4$, and thus to the asymptotic behaviour of Vlasov matter, our results essentially extend the stability regime of \cite{BO24EESF} to the critical regime for $m=0$ and slightly beyond for $m=1$. That FLRW solutions to \eqref{eq:EVSF} fall into a critical regime is reflected in the fact that \change{only the Vlasov distribution itself remains close to the leading term of its velocity term dominated counterpart, while the natural spatial metric on slices of the CMC foliation cannot be controlled sharply. This causes components of the Vlasov energy momentum tensor to diverge, including the macroscopic pressure and density that one would like to compare to the fluid model. Nevertheless, the solution variables that determine the leading order asymptotics of the spacetime fully retain their AVTD behaviour (see Remark \ref{rem:avtd}). }Further, that spatial velocity directions cannot be ignored toward the Big Bang is also strongly indicated by even VTD radiation fluids on Kasner backgrounds exhibiting this feature, and that stability of their initial profile can only be ensured under additional symmetry assumptions on the initial data (see \cite[Theorem 4.3]{BeyLeF17}). 

\subsubsection{Vlasov matter in cosmological spacetimes}

The Einstein-Vlasov system has been studied extensively, especially in the context of future nonlinear stability, while little is known in the contracting regime.\\

In \change{$2+1$ } gravity, the first author has analyzed the past and future maximal development of cosmological solutions to the Einstein-Vlasov system (see \cite{Faj16, Faj17, Faj17m0, Faj18}). In \change{$3+1$ } dimensions, the system is known to be nonlinearly stable near Minkowski space (see \cite{Tay17, LiTay20, FJS21, BFJSTh21}) and near Milne spacetimes (see \cite{AndFaj20,BaFaj20}). Moreover, future stability is known for the Einstein nonlinear scalar-field massive-Vlasov system with accelerated expansion (see \cite{Rin13}). Most of the tools used therein were extended to the massless setting in \cite[Paper B]{Sve12}. The works \cite{Rin13} by Ringström and \cite[Paper B]{Sve12} by Svedberg are of particular relevance to this paper, since the local well-posedness results therein can be carried over to our setting. We discuss this in Section \ref{subsec:lwp}. 

Regarding the past direction, it was shown in \cite{DaRe16} that Kretschmann blow-up occurs for any surface-symmetric non-vacuum cosmological solution to the Einstein-Vlasov system, but no statement on the nature of blow-up is made. The first stability result for Vlasov was due to Rein (see \cite{Rein96}) in the same setting; the approach was extended to the scalar-field Vlasov setting by Tegankong, Noutchegueme and Rendall in \cite{TNR04, TR06}. In these works, the authors studied solutions of the form
\[-e^{-2\mu(a,r)}da^2+e^{2\lambda(a,r)}dr^2+a^2\sigma_{\varepsilon}\,\]
and obtained a lower bound on the Kretschmann scalar by $a^{-6}$, ensuring that the Strong Cosmic Censorship conjecture holds in this symmetry class. However, when comparing our results to \cite{Rein96}, it should be noted that near-FLRW spacetimes are not covered by the stability result \cite[Theorem 5.3]{Rein96}, as discussed in the subsequent Remark. Indeed, in these coordinates, FLRW spacetimes correspond to the choices
\[\lambda(a,r)=\text{ln}(a),\quad e^{-2\mu(a)}=\frac{8\pi}{3}a^2\rho(a)\,.\]
\change{Renormalised Vlasov solutions $f^\sharp$ on the mass shell depend on the expansion rescaled \change{velocity }variables $w$ and $F$}. In \cite{Rein96}, the maximum of these quantities on the support of \change{$f^\sharp$ }at the initial hypersurface $a=1$ is denoted by $w_0$ and $F_0$, respectively. This notation in hand, one computes that FLRW spacetimes satisfy
\begin{align*}
10\pi^2w_0F_0\sqrt{1+w_0^2+F_0}\left\|e^{2\mu(1)}\right\|_{C^0}\left\|\change{f^\sharp_{FLRW}}(1)\right\|_{C^0}\geq\,&5\pi\int_{-w_0}^{w_0}\int_0^{F_0} \sqrt{1+w^2+F}\mathcal{F}(w^2+F)\,dFdw\\
=&\,5\pi e^{2\mu(1)}\rho(1)=\frac{15}8\,.
\end{align*}
The requirement in \cite[Theorems 4.1 and 5.3]{Rein96} that the quantity on the left-hand side is bounded from above by $\frac12$ thus excludes near-FLRW spacetimes from the analysis.
There are no such restrictions in the scalar-field Vlasov case studied by Tegankong and Rendall in \cite{TR06} so that the aforementioned lower bound on the Kretschmann scalar still applies to FLRW solutions. This bound is sufficient to show that the Strong Cosmic Censorship conjecture is satisfied for surface symmetric cosmological solutions to \eqref{eq:EVSF}. However, it does not necessarily reflect the leading order asymptotic behaviour of the Kretschmann scalar. Indeed, in presence of a scalar field, the Kretschmann scalar of the FLRW solution blows up at order $a^{-12}$ and for the pure Vlasov case at order $a^{-8}$. In \eqref{eq:blow-up-rates}, we show that this blow-up is stable, i.e., near-FLRW initial data also leads to blow-up at order $a^{-12}$ and the rescaled Kretschmann scalar remains close to its FLRW counterpart. The gap between our work and the lower bounds in \cite{Rein96,TR06} appears since non-negative terms which are dominant near FLRW solutions are dropped in their computations.  Moreover, in \cite[Theorem 3.5]{TR06}, it is shown that the expansion normalized eigenvalues of the second fundamental form converge under rather restrictive assumptions on the behaviour on $f^\sharp$ along the entire evolution.\footnote{In \cite{TR06}, the proof of Theorem 3.5 uses inequality (3.3) from Proposition 3.2. This proposition assumes that the support of the expansion normalised Vlasov distribution function converges in the $p^1$-direction at the rate $t^\alpha$ for some $\alpha>0$. As the authors remark after the proof of Proposition 3.2, it is not clear whether such an estimate holds for the past maximal development if $f\neq 0$ or how the result could be utilized for a potential bootstrap argument.} In \eqref{eq:asymp-K}, we show that this is possible for generic near-FLRW data. Thus, our results are not only the first to establish stable past asymptotics without symmetry assumptions, but also refine previous results \cite{Rein96,TR06} in surface symmetry when considering near FLRW spacetimes. \change{It is also worth noting that \cite{DaRe16} does not make a statement on the asymptotics of the Vlasov distribution itself. In \cite{Rein96,TR06}, the momentum support of the Vlasov distribution is contained in arbitrarily small, uniformly shrinking balls around the zero section. By contrast, we show in \eqref{eq:vtd-f} that, in a weighted $C^{10}$-space, the Vlasov distribution converges as a function on the co-mass shell to a footprint that is close to the initial distribution. However, one must exhibit some caution in comparing these results, as all of the analysis mentioned above is performed while considering the Vlasov distribution function on the mass shell, while it only converges toward the Big Bang in this work as a function on the co-mass shell. Indeed, when working on the mass shell, particle velocities concentrate in negative eigendirections of the expansion normalised shear $K_{Bang}^{\parallel}$, as is discussed in Remark \ref{rem:save-my-skin}. We leave it open to future work whether any $K_{Bang}^{\parallel}$ close to $0$ can be realised given suitable near-FLRW initial data, but remark that, in his recent result \cite{Li24}, Li establishes an isomorphism between near-Kasner initial data and data on the Big Bang hypersurface for the linear Einstein scalar-field system. Since the evolution of the shear is dominated by the scalar field, and the spacetime evolution is dominated by the linearized system, similar statements may also hold in the nonlinear ESFV setting. Thus, we conjecture that this concentration of velocities is generic since, if no negative eigendirection exists, $K_{Bang}^{\parallel}$ must vanish identically.}\\

\subsection{Proof outline}

We prove stability via a bootstrap argument on supremum norms of the spacetime metric components, the scalar field and components of the Vlasov energy-momentum-tensor (see \eqref{eq:BsC}) as well as on low order supremum norms of the Vlasov distribution function (see \eqref{eq:BsVlasovhor}). The core argument to improve these assumptions relies on a series of energy estimates (see Definition \ref{def:energies}), where the first energy hierarchy adapts the arguments of \cite{FU23} to the current setting, and the second hierarchy provides sufficiently strong bounds on Vlasov energies.\\

We will first discuss the structural features of the Vlasov equation. Then, we discuss how the argument of \cite{FU23} is adapted to include Vlasov matter. We refer to \cite[Section 1.4.1]{FU23} for a more detailed discussion of the bootstrap mechanism and the energy hierarchy for the spacetime and the scalar field variables.\changediss{\phantom{m}\\}

\subsubsection{Features of Vlasov matter toward the Big Bang}\label{subsubsec:vlasov-feature}\phantom{m}\\

\change{By viewing the Vlasov distribution as a function on the co-mass shell, the Vlasov equation (see \eqref{eq:vlasov-resc}) admits a particularly simple form in CMCTC gauge. In particular, momentum characteristics remain uniformly bounded throughout the analysis if the momentum support is bounded initially. However, this structure becomes significantly more complicated once one commutes the equation with covariant derivatives adapted to this foliation. Firstly, one generates terms of order $t^{-1}$ containing the shear which, a priori, prevent us from obtaining integrable estimates. This essentially means that shear terms reappear that one would also have at order zero in the mass shell formulation, see \eqref{eq:Vlasov-sharp}. Secondly, even if one were to use a different commutation method, \eqref{eq:vlasov-resc} still contains high order lapse terms that can only be estimated at the cost of scaling and are thus effectively borderline on an energy level. To obtain closed estimates, we leverage the internal structure of the commuted Vlasov equation in, to our knowledge, a novel manner: Toward the past, terms containing momentum derivatives (\enquote{vertical derivatives}) of $f$ may diverge faster than those with derivatives orthogonal to the momentum (\enquote{horizontal derivatives}). However, such terms either occur nonlinearly and thus inherit smallness from lower order a priori estimates (see Lemma \ref{lem:AP}), or are linear high order shear and lapse terms arising due to the reference distribution. In particular, when commuting the Vlasov equation with, for example, $K$ horizontal and $L-K$ vertical derivatives, those terms with more than $K$ horizontal derivatives remain integrable in time toward the Big Bang even when the Vlasov distribution function itself can only be controlled imprecisely (e.g., using the bootstrap assumption). This allows us to obtain improved estimates by iterating over increasing $K$, given one takes care regarding the coupling with lapse and shear terms as discussed below.}\\

We remark that the horizontal derivative as it occurs in the Vlasov equation involves Christoffel symbols of the evolved metric and is thus not necessarily globally defined. While low order supremum bounds (see Lemma \ref{lem:APVlasov}) can essentially be proven locally, energy estimates can only be obtained globally. However, using the blockdiagonal structure of suitably rescaled Sasaki metrics on the \change{cotangent bundle}, along with the associated covariant derivatives,  one can still contract over the subspace of vertical derivatives and its complement (see Remark \ref{rem:contractions-well-def}), which allows one to globally define energies that only measure a fixed number of horizontal derivatives of $f$. This is crucial as it allows us to exploit this scaling hierarchy within our global energy formalism.\changediss{\phantom{m}\\}

\subsubsection{Integrating Vlasov matter into the scalar field stability mechanism}\label{subsubsec:integrate-Vlasov-3+1}
\phantom{m}\\

\textbf{Initial data and the bootstrap argument.} We may take initial data to be CMC without loss of generality (see Remark \ref{rem:cmc}), and are thus able to work in CMC gauge with zero shift. Furthermore, we assume (see Assumption \ref{ass:init}) that the solution variables are \change{$\epsilon^2$-}close to their FLRW counterparts in Sobolev spaces $H^\ell$ with $\ell\leq 19$, with their norms collected in $\mathcal{H}$ and $\mathcal{H}_{top}$, as well as in $C^{m}$ with $m\leq 16$, collected in $\mathcal{C}$. These spaces are defined with respect to the contraction normalized spatial metric $G$ (see \eqref{eq:rescalingGK}) or, for the Vlasov distribution function, the rescaled Sasaki metric $\G_0$ (see \eqref{eq:def-G0}).\\

\change{On the bootstrap interval, we first assume 
\[\mathcal{C}\lesssim\epsilon a^{-c\sigma}\]
for $c>0$ and some sufficiently small $\sigma\gg\epsilon^\frac18$, see \eqref{eq:BsC}. This small divergent factor is necessary due to the spatial metric in particular degenerating toward the Big Bang singularity. Our goal is to improve this bound to 
\[\mathcal{C}\lesssim \epsilon^\frac74a^{-c^\prime\epsilon^\frac18}\]
for some $c^\prime>0$ by establishing improved energy bounds. }We then extend these to $\mathcal{H}$ with conditional coercivity estimates and then to $\mathcal{C}$ using Sobolev embeddings with respect to $\gamma$. Regarding Vlasov matter, Lemma \ref{lem:density-control} provides control on the energy-momentum components based on specific Vlasov energies, which allows us to control the Vlasov matter quantities in $\mathcal{C}$ directly upon Sobolev embedding with respect to $\gamma$. While we may still obtain strong Sobolev bounds for the Vlasov distribution function (see \eqref{eq:sob-norm-f-main}) using similar conditional coercivity bounds with respect to the Sasaki metric (see Lemma \ref{lem:vlasov-rearrange}), this is not required for the bootstrap argument. Additionally, we assume a weaker bound on low order derivatives of $f-f_{FLRW}$ that are not exclusively vertical (see \eqref{eq:BsVlasovhor}). This is only needed to obtain a priori bounds on the Vlasov distribution function which then also improve upon this assumption.\\

That the bootstrap assumptions hold on an initial interval in the chosen gauge is ensured by the local well-posedness result in Lemma \ref{lem:local-wp-CMC}. To establish this, we take the local existence, continuation criteria and Cauchy stability results in harmonic gauge from \cite{Rin13} in the massive and \cite[Paper B]{Sve12} in the massless case. From there, we argue as in \cite{Shao11, Rodnianski2014} to extend the result from harmonic gauge to CMC gauge, where the continuation criteria need to be adapted accordingly.\\

\textbf{A priori bounds.} To deal with nonlinear terms in the evolution equations, we need improved bounds on the solution variables at low order. In the FLRW picture, the Vlasov matter quantities are dominated by the scalar field. More precisely, the energy density and pressure of Vlasov matter behave like $a^{-4}$ while their scalar field counterparts behave like $a^{-6}$. This is preserved by the bootstrap assumptions up to a factor of $a^{\pm \sigma}$. Meanwhile, the respective energy fluxes $T_{0i}^{SF}$ and $T_{0i}^{Vl}$ vanish for the FLRW spacetime, and their perturbed counterparts can both only be bounded by $\epsilon a^{-3-c\sigma}$. Since, however, the energy flux only occurs in the momentum constraint equation, which is not used to establish a priori bounds, low order improvements for shear, Bel-Robinson, metric and scalar field variables are then obtained as in \cite[Section 4]{FU23} (see Lemma \ref{lem:AP0} and Lemma \ref{lem:AP}), where the Vlasov terms are always asymptotically negligible.\\

As for the Vlasov distribution itself, the distribution function at order zero and the evolution of its support can be dealt with by applying the method of characteristics to the rescaled Vlasov equation \eqref{eq:vlasov-resc}. The commuted equations at order $L$ cause borderline error terms of the same order. As outlined in Section \ref{subsubsec:vlasov-feature}, these borderline terms can only add vertical derivatives and are nonlinear, so that the bounds in Lemma \ref{lem:AP0} and Lemma \ref{lem:AP} can be applied. Thus, one still obtains a priori estimates for $L$ vertical derivatives of $f-f_{FLRW}$ along similar lines as at order zero, using the low order Vlasov bootstrap assumption on the horizontal derivatives. This can be repeated over increasing order of horizontal derivatives, where all borderline terms contain fewer horizontal derivatives and hence are already improved. In the case where all derivatives are horizontal, the bootstrap assumption is not required, yielding a stronger improvement than previously obtained for the mixed terms. Successively replacing the bootstrap assumption with this improved bound then yields sufficiently strong a priori bounds on the distribution function that, in particular, improve \eqref{eq:BsVlasovhor}.\\

\textbf{Energy estimates.} Regarding Vlasov energies, the transport equation for the distribution function $f$ yields a mechanism for energy conservation using integration by parts (see Lemma \ref{lem:vlasov-energy-mech}). We then use the Sasaki-commuted Vlasov equations (see Lemma \ref{lem:vlasov-comm}) to obtain high order energy estimates. At these higher orders, we again need to deal with borderline error terms: Considering Vlasov terms with $K$ horizontal and $L-K$ vertical derivatives, differentiating in time and applying the Vlasov equation, \change{commuting $\del_t$ with horizontal Sasaki derivatives }leads to error terms of the form 
\begin{equation}\label{eq:vlasov-beyond-bl}
a^{-3}\ast\nabla\Sigma\ast \lvert\nabsak^{L-K}_{vert}\nabsak_{hor}^{K-1}(f-f_{FLRW})\rvert_{\G_0}^2\,.
\end{equation}
At first glance, these diverge too strongly for the argument to close, since $\nabla \Sigma$ can only be bounded by $\epsilon a^{-c\sqrt{\epsilon}}$ using the a priori bound \eqref{eq:APmidSigma}. However, as for the a priori estimates, such terms always contain less horizontal derivatives than the terms they originate from. Moreover, they do not occur in the purely vertical case. Thus, we can obtain a suitable energy estimate for the case of only having vertical derivatives  (see Lemma \ref{lem:vlasov-vertical}), up to a term of the form
\[\int_t^{t_0}a(s)^{-1-c\sigma-\omega}\cdot a(s)^\omega\E^{(L)}_{1,1}(f,s)\,ds\,\]
for some fixed $\omega>0$. Similarly, we obtain a series of hierarchized energy estimates for mixed and purely horizontal Vlasov energies (see Lemmas \ref{lem:vlasov-interim} and \ref{lem:vlasov-hor}), where energies with $K$ horizontal derivatives are scaled by $a^{(K+1)\omega}$. Said scaling means that borderline terms such as \eqref{eq:vlasov-beyond-bl} are scaled more strongly than the energies they are controlled by, and thus do not enter the energy estimates as borderline terms. This can be combined to a total Vlasov energy estimate on the sum of all Vlasov energies at order $L$ with their respective scaling (see Lemma \ref{lem:vlasov-total-scaled}).\\
Additionally, we note that the Sasaki metric contains first order derivatives of the spatial metric in its components, and the Vlasov equation also contains first order derivatives. Hence, it is not obvious that our commutation method does not lose regularity. In fact, to ensure that high order metric errors arising from commutator errors are kept under control, we use scaled curvature energy estimates (see \eqref{eq:en-est-Ric-top} and \eqref{eq:en-est-Ric-toptop}).\\

To ultimately improve the bootstrap assumptions on Vlasov matter quantities, these scaled energy estimates are insufficient. To complete the argument, we derive a complementary, unscaled series of estimates with the same hierarchical approach. Regarding terms of the type described in \eqref{eq:vlasov-beyond-bl}, the fact that $\nabla\Sigma$ is small (albeit divergent, see \eqref{eq:APmidSigma}) ensures that no loss of $\epsilon$-control occurs as soon as we allow for horizontal derivatives. Otherwise, these energy estimates are proven in much the same way as their scaled counterparts. Since they are slightly more delicate, we always focus on proving the latter type in Section \ref{sec:vlasov} and sketch the relevant adaptations for the former scaled estimates.\\

Energy estimates for spacetime quantities and scalar field matter follow along similar lines as in \cite{FU23}, since Vlasov energies enter the respective estimates with time-integrable weights, as we discuss in Section \ref{subsec:energy-old}. In particular, since the sectional curvature is negligible compared to both matter contributions in the Friedman equation \eqref{eq:Friedman} and we do not make use of the specific value of sectional curvature otherwise, extending the argument of \cite{FU23} to $\kappa\in\R$ also only leads to error terms of low order (see, in particular, Lemma \ref{lem:en-error-cancellation}). When considering how Vlasov matter couples into the remaining field equations, the rescaled momentum constraint \eqref{eq:REEqMom} is the only term that is not evidently asymptotically weaker than scalar field matter. Here, the leading scalar field term is $Ca^{-3}\nabla\phi$. Note that scalar field energies scale $\nabla\phi$ by $a^4$. Accounting for this, Vlasov energies are always enter spacetime and scalar field energy estimates with time integrable weights that are strictly smaller than those for scalar field energies. \change{In anticipation of the total energy discussed next paragraph, we also derive energy estimates scaled by powers of $a^\frac{\omega}4$ for the shear, for Bel-Robinson variables and curvature terms, which follow by similar arguments as their unscaled counterparts.}\\

\textbf{Bootstrap improvement and the main theorem.} \change{Structurally, our strategy to obtain strong energy bounds is to construct estimates of the following schematic form
\begin{align*}
\numberthis\label{eq:caricature}\mathcal{F}(t)\lesssim&\,\mathcal{F}(t_0)+\int_t^{t_0}\left(\epsilon^\frac18 s^{-1}+\langle\text{time-integrable terms}\rangle\right)\mathcal{F}(s)\,ds+\langle\text{already improved terms}\rangle
\end{align*}
As the first step in this procedure, we slightly oversimplify the total energy, see Definition \ref{def:total}, into the caricature
\[\E_{total}=\E_{total,SF}+\left(\epsilon^\frac14+a^\frac{\omega}4\right)\E_{total,metric}+\E_{total,Vl}\,.\]
Here, the contribution
\[\E_{quiesc}=\E_{total,SF}+\epsilon^\frac14\E_{total,metric}\]
corresponds to the main terms of the total energy for the Einstein scalar-field system as in \cite{FU23}, and $\E_{total,Vl}$ is a hierarchical combination of Vlasov energies as outlined above, where any individual term is scaled by at least $a^{\omega}$. This leads to the leading order scalar field and metric terms entering from Vlasov energy estimates to be bounded by
\begin{equation}\label{eq:Vlasov-en-sc}
\int_t^{t_0}a(s)^{-3-c\sqrt{\epsilon}+\frac{\omega}8}\cdot \left(\E_{total,SF}(s)+a(s)^\frac{\omega}4\E_{total,metric}(s)\right)\,ds\,
\end{equation}
up to multiplicative constant. In other words, if we choose $\omega\gg\sigma>0$, this is compatible with \eqref{eq:caricature} applied to $\E_{total}$ since the prefactor in \eqref{eq:Vlasov-en-sc} is integrable in time, see Lemma \ref{lem:scale-factor}. On the other hand, Vlasov energies at worst couple into the estimates of the remaing energies as
\begin{equation}\label{eq:Vlasov-insert-intro}
\int_t^{t_0}a(s)^{-2-c\sqrt{\epsilon}-20\,\omega} \E_{total,Vl}(s)\,ds\,.
\end{equation}
Consequently, this prefactor can be kept time-integrable if one chooses $\omega$ to be sufficiently small, while still significantly larger than $\sigma$. In summary, this yields an estimate of the form \eqref{eq:caricature} for $\E_{total}$ and consequently to the strong bound $\E_{total}\lesssim \epsilon^4a^{-c\epsilon^\frac18}$, see Proposition \ref{prop:en-imp}.\\

While this bound is sufficient to improve solely the $\epsilon$-prefactor in $\mathcal{C}$, it would only show that geometric and Vlasov matter energies diverge at rates no worse than $a^{-c\omega}$ for some $c>0$, which is insufficient to improve the boostrap assumption since $\omega$ must be sufficiently larger than $\sigma$. However, this bound is still strong enough to ensure that \eqref{eq:Vlasov-insert-intro} can be bounded by $\epsilon^4$ up to multiplicative constant. Thus, this term is already sufficiently improved, and we can essentially reapply the same schematic argument to $\E_{quiesc}$ as to the total energy in \cite{FU23} to obtain an estimate of the form \eqref{eq:caricature}. This leads to strong bounds on all components of the spacetime metric, see Proposition \ref{prop:en-imp-quiesc} and Corollary \ref{cor:en-imp-interim}. }To obtain strong Vlasov energy bounds (see Proposition \ref{prop:vlasov-improvement}), we then iterate the unscaled Vlasov energy estimates over horizontal and total derivative orders, \change{using }that the other energy quantities are already \change{improved to obtain a series of inequalities of the schematic form \eqref{eq:caricature}}.\\

The energy bounds then yield an improvement of the bootstrap assumption \eqref{eq:BsC} as outlined at the start of this subsection, closing the bootstrap argument. Most parts of the main theorem, including geodesic incompleteness and curvature blow-up stability, now follow from these improved bounds just as in \cite[Theorem 15.1]{Rodnianski2014} or \cite[Theorem 8.2]{FU23}. \change{In particular, one can perform a similar argument for the Vlasov distribution itself, though one loses one derivative order in doing so since one needs to directly estimate spatial and momentum derivatives of the distribution that appear in the Vlasov equation.} 


\subsection{Structure of the paper}

In Section \ref{sec:prelim}, we introduce notation and the decomposed Einstein equations as well as some useful formulas. Section \ref{sec:norms} collects the necessary solution norms to phrase the initial data and bootstrap assumptions, and we establish a local well-posedness result therein along with continuaton criteria for solutions to the ESFV system. In Section \ref{sec:ap}, we prove a priori bounds and estimates for the solution variables based on the bootstrap assumptions, including low order supremum norm bounds and some near-norm equivalencies. Energy estimates for spacetime and scalar field variables are collected in Section \ref{sec:energy-new}, and the energy estimates for Vlasov matter in Section \ref{sec:vlasov}. All of these energy estimates are combined in Section \ref{sec:improve}, which is then used to prove Big Bang stability in Section \ref{sec:main-thm}.\\

\noindent \textbf{Acknowledgements.} The authors thank Michael Eichmair for his feedback on early versions of this manuscript\change{, as well as the anonymous referee of \cite{U24} for their comments that also helped to improve this work.}
\section{Preliminaries}\label{sec:prelim}

\subsection{Notation}\label{subsec:notation}

\subsubsection{Foliation, CMC gauge and spacetime objects}

In this paper, we study the Einstein equations in CMC gauge with zero shift unless stated otherwise. Our \change{$(3+1)$}-dimensional spacetime $(\M,\g)$ is thus foliated by constant mean curvature Cauchy hypersurfaces $M_{s}=t^{-1}(\{s\})$ for a time function $t$. These hypersurfaces are diffeomorphic to the closed orientable 3-manifold $M$, and thus we will often replace $M_s$ by $\{s\}\times M$ or simply $M$ without further comment. With respect to this foliation, the spacetime metric is written as
\[\g=-n^2dt^2+g_{ab}dx^adx^b\,.\]
We write ${{\xi(t)}^{a_1\dots a_r}}_{b_1\dots b_r}$ for the $\g$-orthogonal projection of any $M_t$-tangent tensor ${\xi^{\alpha_1\dots\alpha_r}}_{\beta_1\dots\beta_s}$ onto the hypersurface $M_t$. We usually suppress the time dependency in this notation.\\

\noindent The $(0,2)$-tensor
\[k(X,Y)=-\g(\nabla_X\del_0,Y)\]
 denotes the second fundamental form with respect to this foliation, where $\del_0=n^{-1}\del_t$ is the future directed timelike unit normal. For the smooth scale factor $a$ (see Lemma \ref{lem:FLRW}) and when working in CMC gauge, the mean curvature of $M_t$ is given by
\begin{equation}\label{eq:CMC}
\tau(t):=-3\frac{\dot{a}(t)}{a(t)}\,.
\end{equation}
Defining the dual of the Weyl tensor $W\equiv W[\g]$ by
\[W^\ast_{\alpha\beta\gamma\delta}=\frac12\epsilonLC[\g]_{\alpha\beta\mu\nu}{W^{\mu\nu}}_{\gamma\delta},\]
the Bel-Robinson variables are $M_t$-tangent (0,2)-tensors defined by
\[E_{ac}:=W_{a0c0}, B_{ac}:=W^\ast_{a0c0}\,.\]
Finally, we write
\[J_{\beta\gamma\delta}:=\nabbar^\alpha W_{\alpha\beta\gamma\delta},\ J^\ast_{\beta\gamma\delta}=\nabbar^{\alpha}W^\ast_{\alpha\beta\gamma\delta}\,.\]
The Einstein equations imply
\[\frac{1}{8\pi}J_{i0j}=\nabbar_0T_{ij}-\nabbar_iT_{j0},\quad \frac{1}{8\pi}J_{i0j}^\ast={{\epsilonLC[g]}^{kl}}_i\nabbar_kT_{jl}\,.\]
Since one has $T=T^{SF}+T^{Vl}$, we can analogously decompose these expressions into scalar field and Vlasov components, which we denote by $(J^{SF})^{(\ast)}_{i0j}$ and $(J^{Vl})^{(\ast)}_{i0j}$ respectively.\changediss{\phantom{m}\\}

\subsubsection{Indices}

We apply the following index conventions:

\begin{itemize}
\item Lowercase Latin indices $a,b,\dots,i,j,\dots$ run from $1$ to $3$. In the context of spatial derivatives or components of spatial tensors, these are always with respect to an arbitrary local coordinate system $(x^i)_{i=1,2,3}$ on $M$.
\item Lowercase Greek indices $\alpha,\beta,\dots,\mu,\nu,\dots$ are spacetime indices and run from $0$ to $3$. Unless stated otherwise, the index $0$ denotes components relative to $\del_0$, and indices $1,2,3$ denote spatial indices as discussed above.
\item Uppercase Latin letters $I,J,\dots$ index objects on the \change{cotangent bundle $T^\ast M$ }(or, respectively, the mass shell $P$, see Section \ref{subsec:sasaki}) and run from $1$ to $6$. Here, indices $1$ to $3$ refer to spatial coordinates and horizontal directions, while indices $4$ to $6$ refer to momentum coordinates and vertical directions (also see Section \ref{subsec:sasaki}). \change{We note that we will always use coordinates $(p_i)_{i=1,2,3}$ on the cotangent planes canonically induced from a local coordinate system $(x^i)_{i=1,2,3}$, see also Remark \ref{rem:mass-shell-lift}.}
\end{itemize}

Indices are raised and lowered with respect to $\g$ and $g$. When raising, respectively lowering, with respect to the rescaled metric $G$, we denote this by $\mathfrak{T}^\sharp$, respectively $\mathfrak{T}^\flat$. We will not raise or lower indices with respect to \change{cotangent bundle metrics }to avoid potential confusion with schematic notation (see Section \ref{subsubsec:metric-not}).

\subsubsection{Metric conventions and notation}\label{subsubsec:metric-not}

For any Semi-Riemannian  metric $\mathfrak{g}$, its Levi-Civita connection is denote by $\nabla[\mathfrak{g}]$, and the Riemannian curvature tensor is defined by
\[\nabla[\mathfrak{g}]_\alpha\nabla[\mathfrak{g}]_\beta Z_\gamma-\nabla[\mathfrak{g}]_{\beta} \nabla[\mathfrak{g}]_{\alpha}Z_\gamma={\Riem[\mathfrak{g}]_{\alpha\beta\gamma}}^\delta Z_\delta\,.\]
The Levi-Civita tensor is written as $\epsilonLC[\mathfrak{g}]$. If $\mathfrak{g}$ is Riemannian, $\langle\,\cdot,\cdot\,\rangle_{\mathfrak{g}}$ and $\lvert\,\cdot\,\rvert_{\mathfrak{g}}$ denote the inner product and norm respectively. Further, the respective Japanese bracket is defined by $\langle \,\cdot\,\rangle_\mathfrak{g}=\sqrt{1+\lvert\,\cdot\,\rvert_\mathfrak{g}^2}$. The volume form and element of $\mathfrak{g}$ are written as $\vol{\mathfrak{g}}$ and $\mu_{\mathfrak{g}}=\sqrt{\det\mathfrak{g}}$. Further, we define the following operations on $M_t$-tangent $(0,2)$-tensors $A,\tilde{A}$ and vector fields $v$ over $M_t$, where all indices are raised and lowered with respect to $\mathfrak{g}$:
\begin{align*}
(A\odot_{\mathfrak{g}}\tilde{A})_{ij}&=A_{ik}{\tilde{A}^k}_{\ j}\\
(A\wedge_{\mathfrak{g}} \tilde{A})_i&={\epsilonLC[\mathfrak{g}]_i}^{jp}{A_j}^q\tilde{A}_{qp}\\
(v\wedge_{\mathfrak{g}} A)_{ab}&={\epsilonLC[\mathfrak{g}]_a}^{cd}v_cA_{db}+{\epsilonLC[\mathfrak{g}]_b}^{cd}v_cA_{ad}\\
(A\times_\mathfrak{g} \tilde{A})_{ij}&={\epsilonLC[\mathfrak{g}]_{i}}^{ab}{\epsilonLC[\mathfrak{g}]_j}^{pq}A_{ap}\tilde{A}_{bq}+\frac13\langle A,\tilde{A}\rangle_{\mathfrak{g}}\mathfrak{g}_{ij}-\frac13(\text{tr}_\mathfrak{g}A)(\mathfrak{tr}_\mathfrak{g}\tilde{A})g_{ij}\\
(\curl A)_{ij}=(\curl_g A)_{ij}&=\frac12\left[{\epsilonLC[\mathfrak{g}]_i}^{cd}\nabla_dA_{cj}+{\epsilonLC[\mathfrak{g}]_j}^{cd}\nabla_dA_{ci}\right]\\
(\div_\mathfrak{g} A)_i&=\nabla^bA_{ib}
\end{align*}

To simplify notation, we will write $\ast_{\mathfrak{g}}$ to denote any contraction between two tensors with respect to $\mathfrak{g}$, up to an implicit multiplicative constant. For $m\in\N$, we write
\[\mathfrak{T}^{\ast m}=\underbrace{\mathfrak{T}\ast\dots\ast\mathfrak{T}}_{m-\text{times}}\,.\]
If no metric is specified, $\ast$ contracts indices using the expansion normalized spatial metric $G$. Multiple covariant derivatives of a tensor $\mathfrak{T}$ with respect to (projected components of a) Levi-Civita connection $\nabla[\mathfrak{g}]$ are denoted by $\nabla[\mathfrak{g}]^I\mathfrak{T}$. These superscripts indicate the number of derivatives instead of an index if and only if they are uppercase Latin letters or formulas.\\

We note that the Weyl tensor vanishes in three dimensions, so that
\[\Riem[\mathfrak{g}]_{ijkl}=\Ric[\mathfrak{g}]_{ik}\mathfrak{g}_{jl}-\Ric[\mathfrak{g}]_{il}\mathfrak{g}_{jk}+\Ric[\mathfrak{g}]_{jl}\mathfrak{g}_{ik}-\Ric[\mathfrak{g}]_{jk}\mathfrak{g}_{il}-\frac12R[\mathfrak{g}]\left(\mathfrak{g}_{ik}\mathfrak{g}_{jl}-\mathfrak{g}_{il}\mathfrak{g}_{jk}\right)\,.\]
For three-dimensional Riemannian metrics $\mathfrak{g}$, $\Riem[\mathfrak{g}]$ and $\Ric[\mathfrak{g}]$ are equivalent in our schematic notation and will be treated as such throughout.\changediss{\phantom{m}\\}

\subsubsection{Constants}

For scalar functions $\zeta_1,\zeta_2$ with $\zeta_2\geq 0$, we write $\zeta_1\lesssim\zeta_2$ if and only if there exists a constant $K>0$ such that $\max\left(\zeta_1,0\right)\leq K\zeta_2$, and $\zeta_1\simeq\zeta_2$ if $\zeta_1,\zeta_2\geq 0$, $\zeta_1\lesssim\zeta_2$ and $\zeta_2\lesssim\zeta_1$. Implicit constants are allowed to depend on data from the FLRW reference solution as well as the initial time $t_0$.

\subsection{The background solution and the Friedman equation}

We collect properties of the FLRW reference solution:

\begin{lemma}[FLRW solutions to the ESFV system]\label{lem:FLRW}
Consider an FLRW spacetime $(\M,\g_{FLRW})$ of the form
\[\change{(\M=(0,T)\times M,\ \g_{FLRW}=-dt^2+a(t)^2\gamma)}\]
for $T\in \R^+\cup\{\infty\}$, a closed orientable 3-manifold $(M,\gamma)$  of constant sectional curvature $\kappa\in\R$ and $a\in C^\infty(0,\infty)$ . Further, consider isotropic scalar field matter and Vlasov matter, i.e.,  
\begin{subequations}
\begin{gather}
\label{eq:FLRW-wave}
\del_t\phibar=C a(t)^{-3},\quad \nabla\phibar=0, \\
\label{eq:FLRW-Vlasov}
\change{{f}_{FLRW}(t,x,p_i)={\mathcal{F}}\left((\gamma^{-1})^{ij}p_ip_j\right),\quad {\mathcal{F}}\in C_c^\infty(\R^+_0,\R^+_0)}
\end{gather}
and
\begin{align}
{\rho}_{FLRW}^{Vl}(t)=&\,4\pi\int_0^\infty R^2\sqrt{m^2a(t)^2+R^2} \mathcal{F}(R^2)dR\,,\\
{\mathfrak{p}}_{FLRW}^{Vl}(t)=&\,\frac{4\pi}3\int_0^\infty \frac{R^4}{\sqrt{m^2a(t)^2+R^2}}\mathcal{F}(R^2)dR\,.
\end{align}

\end{subequations}
\noindent Then $(\M,{\g}_{FLRW},\phibar,\f_{FLRW})$ is a solution to the ESFV system \eqref{eq:EVSF} if and only if
\begin{subequations}\label{eq:Friedman-both}
\begin{equation}\label{eq:Friedman}
\dot{a}^2=\frac{4\pi}3C^2 a^{-4}+\frac{8\pi}3a^{-2}{\rho}_{FLRW}^{Vl}-\kappa\,.
\end{equation}
Further, one has
\begin{align}
\ddot{a}=&\,-\frac{8\pi}3C^2a^{-5}-\frac{4\pi}3a^{-3}\left({\rho}_{FLRW}^{Vl}+3{\mathfrak{p}}_{FLRW}^{Vl}\right)
\end{align}
\end{subequations}
\end{lemma}
\begin{proof}
This is a standard computation, see \cite[p.~345]{ONeill83} .
\end{proof}
Given the Friedman equations \eqref{eq:Friedman-both}, the mean curvature (see \eqref{eq:CMC}) satisfies the following identities
\begin{subequations}
\begin{align}
\tau^2=&\,12\pi C^2a^{-6}+24\pi a^{-4}{\rho}^{Vl}_{FLRW}-9\kappa a^{-2}\\
\del_t\tau=&\,12\pi C^2a^{-6}+12\pi a^{-4}\left({\rho}_{FLRW}^{Vl}+{\mathfrak{p}}_{FLRW}^{Vl}\right)-3\kappa a^{-2}\label{eq:delt-tau}
\end{align}
\end{subequations}
Since the scalar field remains the leading term in \eqref{eq:Friedman} approaching $a=0$, the asymptotic behaviour of the scale factor toward the Big Bang singularity is not affected by Vlasov matter or the sectional curvature. More precisely, it satisfies the following properties. 

\begin{lemma}[Scale factor analysis]\label{lem:scale-factor}Let $a$ be a solution \eqref{eq:Friedman} for $C>0$ with $a(0)=0$ that is positive somewhere. Then, there exists some maximal time of existence $0<T\leq\infty$ with $a\in C([0,T),\R_0^+)\cap C^\omega((0,T),\R^+)$. For $\kappa\leq 0$, $a$ is strictly increasing and one has $T=\infty$. Else, $T>0$ is the unique time such that $\dot{a}(\frac{T}2)=0$. Then, $a$ is strictly increasing on $[0,\frac{T}2]$ and $t\in(-\frac{T}2,\frac{T}2)\mapsto a\left(t+\frac{T}2\right)$ is even.\\

As $t\downarrow 0$, one has $a(t)\simeq t^{\frac13}$. Further, letting $t_0\in\R^+$, with $t_0\leq\frac{T}2$ if $\kappa>0$, the following bounds hold for any $t\in(0,t_0)$, any $q>0$ and constants $c,K>0$ that are independent of $t,t_0,p$ and $q$.
\begin{align}\label{eq:diff-ineq-Friedman}
\sqrt{\frac{4\pi}3}C a(t)^{-2}\leq&\, \dot{a}(t)+\kappa\\
\label{eq:a-integrals}
\int_t^{t_0}a(s)^{-3+p}\,ds\leq&\,
\begin{cases}
\frac1pa(t_0)^{p} & p>0 \\
\frac1{\lvert p\rvert}a(t)^{-\lvert p\rvert} & p<0
\end{cases}\\
\label{eq:log-est}
\int_t^{t_0}a(s)^{-3}\,ds\leq&\,\frac{K}qa(t)^{-cq}
\end{align}
\end{lemma}
\begin{proof}
For the first statement, we note that the functions
\[a\mapsto \int_0^\infty R^2\sqrt{m^2a^2+R^2}\mathcal{F}(R^2)dR,\ a\mapsto \int_0^\infty \frac{R^4}{\sqrt{m^2a^2+R^2}}\mathcal{F}(R^2)dR\]
are analytic on $(0,\infty)$ since $\mathcal{F}$ is compactly supported. Thus, \eqref{eq:Friedman} is of the form $\dot{a}=G(a)$, where $G$ is analytic on $(0,\infty)$. By the same argument as in the proof of \cite[Lemma 2.2]{FU22}, the first statement holds for $0<T\leq\infty$ with $T=\infty$ for $\kappa\leq 0$. The case of $\kappa>0$ follows as in the proof of \cite[Lemma 3.5]{Speck2018}. \eqref{eq:a-integrals}-\eqref{eq:log-est} are proven as in \cite[Lemma 2.4]{FU23} for $\kappa\leq 0$ and \cite[Lemma 3.7]{Speck2018} for $\kappa>0$. \eqref{eq:diff-ineq-Friedman} follows immediately from \eqref{eq:Friedman}.
\end{proof}
Below, in the case where $\kappa>0$, we will always tacitly assume $t_0\leq\frac{T}2$, i.e., we take our initial data during the era in which the universe does not contract. By Lemma \ref{lem:cauchy-stab}, we may do so without loss of generality. By reflection at $\frac{T}2$, all of our statements regarding Big Bang stability in this case also apply identically to stability of the Big Crunch as $t\uparrow T$.

\subsection{Initial data}\label{subsubsec:initial-data}

Herein, we briefly describe the form initial data for the ESFV system \eqref{eq:EVSF} takes within this paper: Consider data
\[(M,\mathring{g},\mathring{k},\mathring{\pi},\mathring{\psi},\mathring{f}),\]
where $\mathring{g}$ is a positive definite symmetric $(0,2)$-tensor on \change{$M$}, $\mathring{k}$ a $(1,1)$-tensor, $\mathring{\pi}$ is a $(0,1)$-tensor, and $\mathring{\psi}$ and $\mathring{f}$ are scalar functions on $M$ and \change{$T^\ast M$ }respectively. Moreover, since $\mathring{f}$ represents a particle distribution function, we will assume $\mathring{f}\geq 0$.\\
 For constant mean curvature (CMC) initial data, we additionally require that
\[\text{tr}_{\mathring{g}}\mathring{k}=-3\change{\frac{\dot{a}(t_0)}{a(t_0)}}\,.\]
Additionally, the following constraint equations must hold: For $x\in M$, we require
\begin{subequations}
\begin{align}
\left\{\text{R}[\mathring{g}]+1-\left({\mathring{k}^{a}}_{\ b}{\mathring{k}^b}_{\ a}\right)\right\}(x)=&\,8\pi\left[\lvert\mathring{\psi}(x)\rvert^2+\lvert\mathring{\pi}\rvert_{\mathring{g}_x}^2\right]+16\pi\int_{P_x}p^0\mathring{f}(x,p)\,\vol{P_x}\ \label{eq:init-Hamilton}\,\text{ and}\\
\{\div_{\mathring{g}}\mathring{k}\}_i(x)=&\,-8\pi\mathring{\psi}(x)\mathring{\pi}_i(x)-8\pi\int_{P_x}p_i\mathring{f}(x,p)\vol{P_x}\label{eq:init-momentum}\,.
\end{align}
\end{subequations}
For the Vlasov quantities in these equations to be well-defined, we will assume $\mathring{f}$ to have compact momentum support, i.e.,
$\{p\,\vert\,(x,p)\in \supp\mathring{f}\}$ is a compact set. In particular, this implies that $(x,p)\mapsto\lvert p\rvert_{\gamma_x}$ is uniformly bounded on $\supp\mathring{f}$. 

In the following, we understand initial data to be embedded into a spacetime $(\M,\g)$ along $\iota$ such that $(\M,\g,\phi,\overline{f})$ solves the ESFV system with
\[\iota(M)=M_{t_0},\ \iota^\ast\g=\mathring{g},\,\iota^\ast k=\mathring{k},\,\iota^\ast\pi=\mathring{\pi},\,\iota^\ast \del_0\phi=\mathring{\psi_0}\ \text{and}\ (\change{T^\ast\iota})^\ast\f(t_0,\cdot,\cdot)=\mathring{f},\]
where \change{$T^\ast \iota$ }denotes the \change{cotangent }map of $\iota$. The maximal globally hyperbolic development of such data is unique, as shown in \cite{FB52, CBGer69}. We thus assume said solution to be globally hyperbolic throughout. \change{FLRW initial data is, surpressing the embedding, given by}
\[\left(\change{M},a(t_0)^2\gamma,-\dot{a}(t_0)a(t_0)\gamma,0,Ca(t_0)^{-3},\change{(x,p)\mapsto \mathcal{F}(\lvert p\vert_\gamma^2)}\right)\,.\]

\subsection{Expansion normalised variables}

\begin{definition}[Expansion normalised variables for Big Bang stability]\label{def:rescaled}
We will use the following expansion normalised variables:
\begin{subequations}
\begin{gather}
\label{eq:rescalingGK}
G_{ij}=a^{-2}g_{ij},\quad (G^{-1})^{ij}=a^2g^{ij},\quad \Sigma_{ij}=a\hat{k}_{ij}\\
\label{eq:rescalingL} N=n-1\\
\label{eq:rescalingBR}
\RE_{ij}=a^4\cdot E_{ij},\quad \RB_{ij}=a^4\cdot B_{ij}\\
\label{eq:rescalingMatter}
\Psi=a^3\del_0\phi-C
\end{gather}
\end{subequations}
Further, on the \change{co-mass shell}
\change{\begin{equation}\label{eq:mass-shell}
P=\bigsqcup_{(t,x)\in\M}P_{(t,x)},\ \text{where }P_{(t,x)}=\{p\in T^\ast_{(t,x)}\M\, \vert\, p_\mu p^\mu=-m^2,\ p_0<0\}\,
\end{equation}}
with $m=0$ in the massless and $m=1$ in the massive case, we define the rescaled momentum variables
\change{\begin{equation}
v_i=p_i,\quad v^0=a p^0=\sqrt{m^2a^2+\lvert p\rvert_G^2},\quad v^0_\gamma=\sqrt{m^2a^2+\lvert v\rvert_\gamma^2}\,,
\end{equation}}
We refer to Remark \ref{rem:mass-shell-lift} regarding the pointwise norms $\lvert v\rvert_G$ and $\lvert v\rvert_\gamma$. 
\delete{The Vlasov distribution function in these rescaled variables is written as}
\deletemath{\begin{equation}
f(t,x,v)=\f(t,x,(a(t)^{-2}v^i))=\f(t,x,p^i)\,.
\end{equation}}
\delete{Additionally, we introduce notation for the rescaled version of the reference distribution:}
\deletemath{\begin{equation}
f_{FLRW}(t,x^i,v^j):=\mathcal{F}(\lvert v\rvert_\gamma^2)
\end{equation}}
For any $t\in(0,t_0]$, we define the following momentum support bounds:
\begin{subequations}
\begin{align}
\P(t):=&\,\sup_{x\in M_t}\sup_{v\in\supp f(t,x,\,\cdot\,)}\lvert v\rvert_G\,,\\
\P^0(t):=&\,\sup_{x\in M_t}\sup_{v\in\supp f(t,x,\,\cdot\,)} v^0\,.
\end{align}
\end{subequations}
$\P_\gamma$, respectively $\P^0_\gamma$, are defined replacing $\lvert v\rvert_G$ with $\lvert v\rvert_\gamma$, respectively $v^0$ with $v^0_\gamma$. Finally, we write the rescaled components of the Vlasov energy-momentum tensor as follows:
\begin{align*}
\rho^{Vl}=&\,a^4T_{00}^{Vl}[\g,\f]=\int_{\R^3} \change{v^0f(t,x,v)\mu_{G}^{-1}}dv\\
{(S^{Vl})^i}_j=&\,a^4{(T^{Vl})^i}_j[\g,\f]=\int_{\R^3}\change{\frac{v_j(G^{-1})^{il}v_l}{v^0}f(t,x,v)\mu_{G}^{-1}}dv\\
\mathfrak{p}^{Vl}=&\,\frac13 \text{tr}_GS^{Vl}=\frac13\int_{\R^3}\change{\frac{\lvert v\rvert_G^2}{v^0}f(t,x,v)\mu_{G}^{-1}}dv\\
{(S^{Vl,\parallel})^i}_j=&\,{(S^{Vl})^i}_{j}-\mathfrak{p}^{Vl}\I^i_{j}\\
\j^{Vl}_l=&\,a^3T_{0l}^{Vl}[\g,\f]=\int_{\R^3} \change{v_l f(t,x,v)\mu_{G}^{-1}}dv
\end{align*}
\end{definition}
Note that, in the massless case $m=0$, one has $\rho^{Vl}=\frac13\mathfrak{p}^{Vl}$. Furthermore, by reflection symmetry, one has $\change{\j_{FLRW}^{Vl}}=0$ and thus
\begin{equation}\label{eq:j-Vl}
\change{\j^{Vl}_l}=a^{-3}\int_{\change{T^\ast_xM}} \change{v_l (f-f_{FLRW})(t,x,v)\mu_{G}^{-1}}dv\,.
\end{equation}
Finally, we observe that
\begin{align*}
{(S^{Vl})^i}_j=&\,\int_{\change{T^\ast_xM}}\change{\frac{v_j(G^{-1})^{il}v_l}{v^0}(f-f_{FLRW})(t,x,v)\mu_{G}^{-1}}dv+\int_{\change{T^\ast_xM}}\change{\frac{v_j(G^{-1})^{il}v_l}{v^0}f_{FLRW}(t,x,v)\mu_{G}^{-1}}dv\,.
\end{align*}
Regarding the second term, notice that if $i\neq j$, the term must vanish by reflection symmetry of $f_{FLRW}$. Furthermore, $f_{FLRW}$ is preserved under permutations of the coordinates $(v^1,v^2,v^3)$ by isotropy. This implies 
\[\int_{\change{T^\ast_xM}}\change{\frac{(G^{-1})^{il}v_jv_l}{v^0}f_{FLRW}(t,x,v)\mu_{G}^{-1}}dv=\int_{\change{T^\ast_xM}}\frac{\lvert v\rvert_{G}^2}{3v^0}\mathcal{F}(\lvert v\rvert_\gamma^2)(t,x,v)\change{\mu_G^{-1}}dv\cdot\I^i_j\]
and consequently
\[\int_{\change{T^\ast_xM}}\change{\frac{(G^{-1})^{il}v_jv_l}{v^0}f_{FLRW}(t,x,v)\mu_G^{-1}dv}-\mathfrak{p}^{Vl}\I^i_j=\int_{\change{T^\ast_xM}}\frac{\lvert v\rvert_G^2}{3\change{v^0}}(f_{FLRW}-f)(t,x,v)\mu_Gdv\cdot \I^i_j\,.\]
Hence, we can write
\begin{equation}\label{eq:S-Vl-parallel}
{(S^{Vl,\parallel})^i}_j=\int_{\change{T^\ast_xM}}\change{\frac{(G^{-1})^{il}v_jv_l}{v^0}(f-f_{FLRW})(t,x,v)\mu_G^{-1}}dv-\int_{\change{T^\ast_xM}}\frac{\lvert v\rvert_G^2}{3\change{v^0}}(f-f_{FLRW})(t,x,v)\change{\mu_G^{-1}}dv\cdot \I^i_j\,.
\end{equation}

\subsection{Sasaki metrics}\label{subsec:sasaki}

In this section, we review properties of the Sasaki metric and the resulting geometry on the cotangent bundle that will underpin our analysis of the Vlasov equation. For a more detailed overview, we refer to \cite[Section 2.4]{A-CGaSar22}.\\

The rescaled metric $G$ on $M_t$ induces the Sasaki metric on \change{$T^\ast M_t$ }by
\change{\[\G=G_{ij}dx^idx^j+(G^{-1})^{ij}Dv_iDv_j\,,\]
where $Dv_i=dv_i-\Gamma[G]^j_{ik}v_jdx^k$. }In the following, $\nabsak$ denotes the covariant derivative with respect to $\G$. We will express this connection using the local frame
\change{\begin{equation}\label{eq:AB}
\A_i=\del_{x^i}+v_j\Gamma[G]^j_{ik}\del_{v_k},\quad \B^i=\del_{v_i}\,.
\end{equation}
Note that one has $Dv_i(\A_j)=0$ for $i,j=1,2,3$. We will denote this frame by $\{S_{\sak{I}}\}_{\sak{I}=1,\dots,6}$ with, for $i=1,2,3$, $S^{\sak{I}}=\A_i$ for $\sak{I}=i$ and $S_{\sak{I}}=\B^{i}$ for $\sak{I}=i+3$. Coefficients with respect to this frame will consequently be labeled by underlined, capitalized Latin indices. For reasons explained below, we will refer to $\{\A_i\}_{i=1,2,3}$ as the horizontal frame components, and $\{\B^i\}_{i=1,2,3}$ as the vertical frame components. For the sake of notational convenience, we will simply write $\sak{i}$ and $\sak{i+3}$ to distinguish between horizontal and vertical elements of the frame $S_{\sak{I}}$. For any tensor $\mathcal{T}$ on $T^\ast M_t$, we further simplify notation as follows:
\begin{equation}
\nabsak_{i}\mathcal{T}=\nabsak_{\sak{i}}\mathcal{T}=\nabsak_{\A_i}\mathcal{T},\quad \nabsak^{i+3}\mathcal{T}=\nabsak_{\sak{i+3}}\mathcal{T}=\nabsak_{\B^i}\mathcal{T}\,,
\end{equation}
where $i=1,2,3$. With respect to the frame $\{S_{\sak{i}}\}$, the connection coeffients of $\G$ are as follows (see \cite[Appendix D]{AndFaj20} for the tangent bundle analogue).}
\begin{subequations}\label{eq:conn-coeff}
\change{\begin{align}
{\Gamsak^{\sak{k}}}_{\sak{i}\,\sak{j}}=&\,\Gamma[G]^{{k}}_{{i}{j}}&\quad {\Gamsak^{\sak{k+3}}}_{\sak{i}\,\sak{j}}=&\,\frac12v_l{\Riem[G]^{l}}_{kij}\\
{\Gamsak^{\underline{k}}}_{(\sak{i+3})\,\sak{j}}=&\,\frac12v_l{\Riem[G]^{\sharp lik}}_j&\quad {\Gamsak^{\sak{k}}}_{\sak{i}\,(\sak{j+3})}=&\,\frac12v_{l}{\Riem[G]^{\sharp ljk}}_i\\
{\Gamsak^{\sak{k+3}}}_{\sak{(i+3)}\,\sak{j}}=&\,-\Gamma[G]_{{jk}}^i& \quad {\Gamsak^{\sak{k+3}}}_{\sak{i}\,\sak{(j+3)}}=&\,{\Gamsak^{k}}_{\sak{(i+3)}\,\sak{j+3}}={\Gamsak^{\sak{k+3}}}_{\sak{(i+3)}\,\sak{(j+3)}}=0
\end{align}}
\end{subequations}
\noindent We also introduce the rescaled Sasaki metric as follows
\change{\begin{equation}\label{eq:def-G0}
\G_0(\A_i,\A_j)=G_{ij},\quad \G_0(\A_i,\B^j)=0,\quad \G_0(\B^i,\B^j)=(v^0)^{-2}(G^{-1})^{ij}
\end{equation}}
as well as the following analogues with respect to $\gamma$.
\change{\begin{align*}
\hat{D}v_i=&\,dv_i-\Gamhat^j_{ik}v_jdx^k\\
\underline{\gamma}=&\,\gamma_{ij}dx^idx^j+(\gamma^{-1})^{ij}\hat{D}v_i\hat{D}v_j\\
\underline{\gamma}_0=&\,\gamma_{ij}dx^idx^j+(m^2+\lvert v\rvert_\gamma^2)^{-1}(\gamma^{-1})^{ij}\hat{D}v_i\hat{D}v_j
\end{align*}}

\begin{remark}[Notation for horizontal and vertical covariant derivatives]\label{rem:contractions-well-def}  A technical difficulty arises from the fact that the horizontal frame $\{\A_i\}_{i=1,2,3}$ is not globally defined. On the other hand, it naturally arises in the Vlasov equation \eqref{eq:vlasov-resc}, and our hierarchy will require us to treat horizontal and vertical covariant derivatives differently. One can, however, split vector fields on \change{$T^\ast M_t$ }into their horizontal and vertical components globally. \\

Let $\pi_{\change{T^\ast M}}:\change{T^\ast M}\rightarrow M$ denote the canonical projection and $T\pi_{\change{T^\ast M}}: T\change{T^\ast M}\rightarrow \change{T^\ast M}$ its tangent map. Recall that $\ker{T\pi_{\change{T^\ast M}}}$ defines the vertical \change{cotangent bundle}. Let $\mathfrak{V}(\change{T^\ast M})$ be the space of sections of $T\change{T^\ast M}$ over $\change{T^\ast M}$ which are valued in $\ker{T\pi_{\change{T^\ast M}}}$. Then, $\mathfrak{V}(\change{T^\ast M})$ is spanned by \change{$\{\B^1,\B^2,\B^3\}$}. Its orthogonal complement $\mathfrak{V}(\change{T^\ast M})^\perp$ admits the local basis $\{\A_1,\A_2,\A_3\}$. We now split vector fields using the $\G$-orthogonal projection
\[\mathbb{P}: \mathfrak{X}(T\change{T^\ast M})\longrightarrow \mathfrak{V}(\change{T^\ast M}),\ \mathbb{P}W=\sum_{i=1,3}\change{\G(\B^i,W)\B^i}\,,\]
onto vertical vector fields. Covariant derivatives can now also be split into horizontal and vertical components by taking
\[\nabsak_{vert}: \mathfrak{X}(\change{T^\ast M})\times \mathfrak{X}(T\change{T^\ast M})\longrightarrow \mathfrak{X}(T\change{T^\ast M}),\ (\nabsak_{vert})_{W}Z=\nabsak_{\mathbb{P}(W)}Z\]
and analogously for $\nabsak_{hor}$ replacing $\mathbb{P}$ with $\mathbb{P}^\perp=1-\mathbb{P}$. In particular, extending this to tensors, contractions along these projected covariant derivatives reduce to expressions of the following type for scalar functions $\xi_1,\xi_2$ on $\change{T^\ast M}$:
\change{\begin{align*}
&\,\langle \nabsak_{vert}^{L-K}\nabsak_{hor}^K\xi_1,\nabsak_{vert}^{L-K}\nabsak_{hor}^K\xi_1\rangle_{\G_0}\\
=&\,\left[(v^0)^{2(L-K)}G_{i_1j_1}\dots G_{i_{L-K}j_{L-K}}(G^{-1})^{k_1l_1}\dots(G^{-1})^{k_Ml_M}\cdot\right.\\
&\,\left.\cdot\nabsak^{i_1+3}\dots\nabsak^{i_{L-K}+3}\nabsak_{k_1}\dots\nabsak_{k_M}\xi_1\,\nabsak_{j_1+3}\dots\nabsak_{j_{L-K}+3}\nabsak_{l_1}\dots\nabsak_{l_M}\xi_2\right]
\end{align*}}
The term $\langle\bm{X}\nabsak_{vert}^{L-M}\nabsak_{hor}\xi_1,\nabsak_{vert}^{L-M}\nabsak_{hor}\xi_2\rangle_{\G_0}$ (see \eqref{eq:def-X}) decomposes similarly.
By construction, such expressions are defined globally on $\change{T^\ast M}$.
For the sake of simplicity, we will often view $\nabsak_{vert}$ as an operator defined on $\mathfrak{V}(\change{T^\ast M})\times \mathfrak{X}(T\change{T^\ast M})$, i.e., as only applying vertical covariant derivatives, and analogously view $\nabsak_{hor}$. 
\end{remark}

\begin{remark}[Lifting to the mass shell]\label{rem:mass-shell-lift}
In the following, we will need $\nabsak$ (and the derived connections $\nabsak_{vert}$ and $\nabsak_{hor}$ as well as $\hat{\nabsak}$) to act on quantities that are not a priori defined on the mass shell and thus, technically, need to be lifted to it. We will surpress such lifts throughout the paper: For example, when considering any $M_t$-tangent tensor $\mathfrak{T}$, $\nabsak\mathfrak{T}$ will refer to the covariant Sasaki derivative on the pullback of $\mathfrak{T}$ along the canonical projection $\pi_t:P_t\simeq \change{T^\ast M}_t\rightarrow M_t$. In particular, one has
\change{\[\nabsak_{{i}}\mathfrak{T}=\nabla_i\mathfrak{T},\quad \nabsak^{i+3}\mathfrak{T}=0\,.\]}
Conversely, we lift the momentum variables $v$ to the \change{covector field} $\tilde{v}$ on the \change{co-}mass shell at time $t$ by
\change{\[\tilde{v}_i=0,\quad \tilde{v}_{i+3}=v_i\,,\]
and simply write $\nabsak v_i:=\nabsak \tilde{v}_{i+3}$. }We collect various derivatives of momentum variables in Lemma \ref{lem:mom-der}. We write \change{$v^{\sharp i}=(\G\vert_{vert}^{-1})^{ij}v_j=(G^{-1})^{ij}v_j$ and define $\nabsak v^\sharp$ analogously}, as well as $\lvert v\rvert_G$ and $\lvert v\rvert_\gamma$.
\end{remark}


\subsection{$3+1$-decomposition of the Einstein equations}

\subsubsection{The rescaled Vlasov equation}

The Vlasov equation \eqref{eq:EVSF3} reads
\change{\begin{align*}
\mathcal{X}f=&\,\g^{\alpha\beta}p_\alpha\del_\beta f+\g^{\lambda\mu}\Gambar^\nu_{i\lambda}p_\mu p_\nu\del_{p_i} f\\
=&\,n^{-1}p^0\left(\del_t f+\frac{g^{ij}p_i}{p^0}n\A_j f-p^0\nabla_in{\B}^i f\right)=0\,.
\end{align*}
Above, we used that connection coefficient terms containing precisely one zero index precisely cancel in these coordinates since $p^0=-p_0$ and consequently
\[p_0p_j g^{jl}\Gambar^0_{il}+p^0p_j\Gambar^j_{i0}=p_j\left(-p_0 g^{jl}k_{il}-p^0k^{j}_{\ i}\right)=0.\]
Rearranging and using the rescaled variables from Definition \ref{def:rescaled}, the Vlasov equation becomes
\begin{subequations}
\begin{align*}\numberthis\label{eq:vlasov-resc}
\del_tf&\,=-a^{-1}(N+1)\frac{v^{\sharp j}}{v^0}\A_jf+a^{-1}v^0\nabla_jN\B^j f=:\X f
\end{align*}
We extend $\X$ to tensors $\mathcal{T}$ on the mass shell as follows:
\begin{equation}\label{eq:def-X}
\X\mathcal{T}=-a^{-1}(N+1)\frac{v^{\sharp j}}{v^0}\nabsak_j\mathcal{T}+a^{-1}v^0\nabla_{j}N\B^j\mathcal{T}
\end{equation}
\end{subequations}}

\subsubsection{Constraint equations and equations for the rescaled spacetime and scalar field variables}\label{subsubsec:REEq}

\begin{subequations}
We collect the evolution equations for the shear, the metric and derived variables:

\begin{align*}
\del_tG_{ij}=&-2(N+1)a^{-3}\Sigma_{ij}+2N\frac{\dot{a}}{a}G_{ij} \numberthis\label{eq:REEqG}\\
\del_t(G^{-1})^{ij}=&\,2(N+1)a^{-3}(\Sigma^\sharp)^{ij}-2N\frac{\dot{a}}{a}(G^{-1})^{ij}\numberthis\label{eq:REEqG-1}\\
\del_t{(\Sigma^\sharp)^i}_j=&\,\tau N{(\Sigma^\sharp)^i}_j-a\nabla^{\sharp i}\nabla_jN+(N+1)a\left[{\left(\Ric[G]^\sharp\right)^i}_j-2\kappa\I^i_j\right]\numberthis\label{eq:REEqSigmaSharp}\\
&\,-8\pi (N+1)a\nabla^{\sharp i}\phi\nabla_j\phi+N\left[4\pi C^2a^{-3}+4\pi a^{-1}\left({\rho}_{FLRW}^{Vl}+\mathfrak{p}^{Vl}_{FLRW}\right)-\kappa a\right]\I^i_j\\
&\,-8\pi a^{-1}{((S^{Vl})^\parallel)^{i}}_j-4\pi a^{-1}\left[\left(\rho^{Vl}-\rho^{Vl}_{FLRW}\right)-\left(\mathfrak{p}^{Vl}-\mathfrak{p}^{Vl}_{FLRW}\right)\right]\I^i_j\\
\del_t\Gamma_{ij}^k[G]=&\,-(N+1)a^{-3}\left[\nabla_i{(\Sigma^\sharp)^k}_j+\nabla_j{(\Sigma^\sharp)^k}_i-\nabla^{\sharp k}\Sigma_{ij}\right]\numberthis\label{eq:REEqChr}\\
&\,-a^{-3}\left[\nabla_iN{(\Sigma^\sharp)^k}_j+\nabla_jN{(\Sigma^\sharp)^k}_i-\nabla^{\sharp k}N\Sigma_{ij}\right]\\
&\,+\frac{\dot{a}}a\left[\nabla_i N\cdot\I^k_j+\nabla_jN\cdot\I^k_i-\nabla^{\sharp k}N\cdot G_{ij}\right]\\
\numberthis\label{eq:REEqRic}\del_t\Ric[G]_{ab}=&\,a^{-3}(N+1)(\Lap_G\Sigma_{ab}-\nabla^{\sharp d}\nabla_a\Sigma_{db}-\nabla^{\sharp d}\nabla_b\Sigma_{da})\\
&\,+a^{-3}\nabla^{\sharp d}N(2\nabla_d\Sigma_{ab}-\nabla_a\Sigma_{db}-\nabla_b\Sigma_{da})\\
&\,-a^{-3}\left(\nabla_a N(\div_G\Sigma)_b+\nabla_b(\div_G\Sigma)_a\right)
+\Lap_GN(a^{-3}\Sigma_{ab}+\frac{\tau}3G_{ab})\\
&\,-a^{-3}\Bigr(\nabla^{\sharp d}\nabla_a N\cdot \Sigma_{db}+\nabla^{\sharp d}\nabla_b N\cdot\Sigma_{da}\Bigr)+\frac\tau3\nabla_a\nabla_bN
\end{align*}

\noindent The rescaled lapse equations\footnote{Recall that, for $\kappa>0$, we have assumed to be in the contracting regime without loss of generality (see Lemma \ref{lem:scale-factor} and the subsequent remark). Thus, recalling \eqref{eq:delt-tau} \[\del_t\tau=12\pi C^2a^{-4}-3\kappa+12\pi a^{-2}\left({\rho}_{FLRW}^{Vl}+{\mathfrak{p}}_{FLRW}^{Vl}\right)\geq 0\] holds and the lapse equations \eqref{eq:REEqLapse1} and \eqref{eq:REEqLapse2} are always solvable.} read

\begin{align*}
\numberthis\label{eq:REEqLapse1}\Lap_GN=&\,\left[12\pi C^2a^{-4}-3\kappa+12\pi a^{-2}\left({\rho}_{FLRW}^{Vl}+{\mathfrak{p}}_{FLRW}^{Vl}\right)\right]N\\
&\,+(N+1)a^{-4}\left(\lvert\Sigma\rvert_G^2+16\pi C\Psi+8\pi\Psi^2\right)\\
&\,+4\pi a^{-2}(N+1)\left[\left(\rho^{Vl}-{\rho}_{FLRW}^{Vl}\right)+\left(\mathfrak{p}^{Vl}-{\mathfrak{p}}_{FLRW}^{Vl}\right)\right]\,,\\
\numberthis\label{eq:REEqLapse2}\Lap_GN=&\,\left[12\pi C^2a^{-4}-3\kappa+12\pi a^{-2}\left({\rho}_{FLRW}^{Vl}+{\mathfrak{p}}_{FLRW}^{Vl}\right)\right]N+(N+1)(R[G]-6\kappa)\\
&\,+12\pi a^{-2}(N+1)\left[\left(\mathfrak{p}^{Vl}-{\mathfrak{p}}_{FLRW}^{Vl}\right)-\left(\rho^{Vl}-{\rho}_{FLRW}^{Vl}\right)\right]\,.
\end{align*}
We arranged the Vlasov terms so that the evolved density and pressure appear as in the solution norms used later on (see \eqref{eq:def-C}), so that we can expect these terms to be small, at least initially.\\
The constraint equations read

\begin{align*}
R[G]-6\kappa-a^{-4}\langle\Sigma,\Sigma\rangle_G=&\,8\pi\left[a^{-4}\Psi^2+2Ca^{-4}\Psi+\lvert\nabla\phi\rvert_G^2\right]+16\pi a^{-2}\left[\rho^{Vl}-{\rho}_{FLRW}^{Vl}\right]\,,\numberthis\label{eq:REEqHam}\\
(\div_G\Sigma)_{l}=&\,-8\pi\nabla_l\phi(\Psi+C)+8\pi\j^{Vl}\numberthis\label{eq:REEqMom}\,,
\end{align*}
with the respective Bel-Robinson analogues
\begin{align*}
\numberthis\label{eq:REEqConstrE}
\RE_{ij}=&\,a^4\Ric[G]_{ij}+\frac29\tau^2 a^6G_{ij}+\frac{\tau}3a^3\Sigma_{ij}-(\Sigma\odot_G\Sigma)_{ij}-4\pi a^4T_{ij}-\frac{4\pi}3(3T_{00}-\text{tr}_{\g}T)a^6G_{ij}\\
=&\,a^4\left(\Ric[G]_{ij}-2\kappa G_{ij}\right)+\frac{\tau}3a^3\Sigma_{ij}-(\Sigma\odot_G\Sigma)_{ij}-4\pi a^4\nabla_i\phi\nabla_j\phi-4\pi a^2(S^{Vl})^{\parallel,\flat}_{ij}\\
&\,-\left[\frac{4\pi}3a^4\lvert\nabla\phi\rvert_G^2+\frac{8\pi}3\Psi^2+\frac{16\pi}3C\Psi\right]G_{ij}-\frac{16\pi}3\left(\rho^{Vl}-{\rho}_{FLRW}^{Vl}\right)\cdot a^2G_{ij}\,,\\
\RB_{ij}=&\,-a^3\curl_g\Sigma_{ij}\numberthis\label{eq:REEqConstrB}\,.
\end{align*}
The evolution equations for the Bel-Robinson variables are as follows:
\begin{align*}
\numberthis\label{eq:REEqE}\del_t\RE_{ij}=&\,\left(3-N\right)\frac{\dot{a}}a\RE_{ij}-a^{-1}(\nabla N\wedge_G\RB)_{ij}+(N+1)a^{-1}\curl_G\RB_{ij}\\
&\,-(N+1)\left[\frac52a^{-3}(\RE\times_G\Sigma)_{ij}+\frac23a^{-3}\langle\RE,\Sigma\rangle_GG_{ij}\right]\\
&\,+4\pi(N+1) a^{-3}(\Psi+C)^2\Sigma_{ij}-4\pi(N+1) \dot{a}a^3\nabla_i\phi\nabla_j\phi+4\pi a\nabla_{(i}N\nabla_{j)}\phi(\Psi+C)\\
&\,-4\pi a(N+1)\left[\nabla_i\Psi\nabla_j\phi+\nabla_j\Psi\nabla_i\phi+(\Sigma^\sharp)^l_{(i}\nabla_{j)}\phi\nabla_l\phi-(\Psi+C)\nabla_i\nabla_j\phi\right]\\
&\,+(N+1)\left[\frac{2\pi}3a^6\del_0\left(a^{-6}(\Psi+C)^2+a^{-2}\lvert\nabla\phi\rvert_G^2\right)+4\pi\frac{\dot{a}}a(\Psi+C)^2\right]G_{ij}\\
&\,+a^4\frac{N+1}2\left[J_{i0j}^{Vl}+J_{i0j}^{Vl}\right]\\[0.5em]
\numberthis\label{eq:REEqB}\del_t\RB_{ij}=&\,\frac{\dot{a}}a\left(3-N\right)\RB_{ij}+a^{-1}(\nabla N\wedge_G\RE)_{ij}-a^{-1}(N+1)\curl_G\RE_{ij}\\
&\,-(N+1)\left[\frac52a^{-3}(\RB\times_G\Sigma)_{ij}+\frac23a^{-3}\langle\RB,\Sigma\rangle_GG_{ij}\right]\,\\
&\,-4\pi(N+1)\epsilonLC[G]_{lmj}\left(a^3\nabla^{\sharp l}\nabla_{j}\phi\nabla^{\sharp m}\phi+a^{-1}{(\Sigma^\sharp)^l}_i\nabla^{\sharp m}\phi(\Psi+C)\right)\\
&\,+a^4\frac{N+1}2\left[(J^{Vl})^{\ast}_{i0j}+(J^{Vl})^{\ast}_{j0i}\right]
\end{align*}
\noindent To compute the new matter terms in the Bel Robinson evolution equations, we first insert \eqref{eq:vlasov-resc} and then integrate by parts for any vertical derivatives that occur to obtain the following:
\change{\begin{align*}\numberthis\label{eq:JVlpar}
\frac{N+1}{2\cdot 8\pi}\left(J_{i0j}^{Vl}+J_{j0i}^{Vl}\right)=&\,\frac{N+1}{8\pi}\left(\nabbar_0T^{Vl}_{ij}-\nabbar_{(i} T^{Vl}_{j)0}\right)\\
=&\,-2(N+1)\,\frac{\dot{a}}a \cdot a^{-2}(S^{Vl})^{\parallel,\flat}_{ij}-(N+1)a^{-3}\nabla_{(i}\j^{Vl}_{j)}
\\
&\,+a^{-3}\int_{\R^3}\left[\frac{v_iv_jv^{\sharp k}}{(v^0)^2}(N+1)\nabsak_kf+\B^k\left({v_iv_j}\right)\nabla_{k}Nf \right]\mu_G^{-1}dv\\
&\,-(N+1)a^{-5}\int_{\R^3}\Sigma_{kl}\frac{v_iv_jv^{\sharp k}v^{\sharp l}}{2(v^0)^3}f\,\mu_G^{-1}dv\\
&\,-\frac{\dot{a}}{a}\cdot a^{-2}(N+1)\int_{\R^3}\frac{v_iv_j}{v^0}\cdot\frac{m^2a^2}{(v^0)^2}(f-f_{FLRW})\,\mu_G^{-1}dv+\langle\text{pure trace terms}\rangle
\end{align*}
Note that, in the final line, we have already expanded $f$ by $(f-f_{FLRW})+f_{FLRW}$ and absorbed the latter into the collection of pure trace terms, by the same argument as for $S^{Vl,\parallel}$ in \eqref{eq:S-Vl-parallel}. Furthermore, we have
\begin{align*}
\numberthis\label{eq:JVlast}\frac{N+1}{2\cdot 8\pi}\left({(J^{Vl})}_{i0j}^{\ast}+{(J^{Vl})}_{j0i}^{\ast}\right)=&\,a^{-1}{\epsilon[G]^{\sharp kl}}_{(i}\nabbar_kT^{Vl}_{j)l}\\
=&\,a^{-4}{\epsilon[G]^{\sharp kl}}_{(i}\int_{\R^3}\frac{v_{j)}v_l}{v^0}\nabsak_{k}f\,\mu_G^{-1} dv-a^{-5}{\epsilon[G]^{\sharp kl}}_{(i}\Sigma_{j)k}\int_{\R^3}v_lf\,\mu_G^{-1}dv\,
\end{align*}}
Finally, the rescaled wave equation reads
\begin{equation}\label{eq:REEqWave}
\del_t\Psi=a(N+1)\Lap\phi+a\langle\nabla N,\nabla\phi\rangle_G-3\frac{\dot{a}}aN(\Psi+C)\,.
\end{equation}
\end{subequations}


%
%
%

%
%

\subsection{Miscellaneous formulas}

To finish this preliminary section, we collect some useful commutator and derivative formulas which can all be proven by straightforward computation:

\begin{lemma}[Commuting derivatives and integrals] Let $\zeta: \M\rightarrow\R$ and $\xi:P\rightarrow \R$ be sufficiently regular. Then, the following statements hold:
\begin{align}
\del_t\int_M\zeta\vol{G}=&\,\int_M \left(\del_t\zeta-N\tau\zeta\right)\vol{G}\label{eq:delt-intM}\\
\change{\del_t\int_M\int_{T^\ast_{\cdot}M}\xi\vol{\G}=&\,\change{\int_M\int_{T^\ast_{\cdot}M}\del_t\xi\vol{\G}}}\label{eq:delt-intTM}\\
\left[\nabla_{i_1}\dots\nabla_{i_L}\int_{\change{T^\ast_{\cdot}M}}\xi(t,\cdot,v) \vol{\G\vert_{\change{T^\ast_{\cdot}M}}}\right](x)=&\,\int_{\change{T^\ast_{\cdot}M}}\nabsak_{i_1}\dots\nabsak_{i_L}\xi(t,x,v)\vol{\G\vert_{T^\ast_xM}}\label{eq:nabla-intTM}
\end{align}
\end{lemma}

\begin{lemma}[Covariant derivatives of momentum functions]\label{lem:mom-der} Recalling the notation from Remark \ref{rem:mass-shell-lift}, the following formulas hold:
\change{\begin{subequations}
\begin{gather}
\nabsak_j v_k=\nabsak_j v^0=\nabsak_j\langle v\rangle_G=\nabsak_j\left(\frac{v_k}{v^0}\right)=0\label{eq:hor-mom-zero}\\
\nabsak_jv^0_\gamma=\left[\Gamma_{jk}^l-\Gamhat_{jk}^l\right](\gamma^{-1})^{kr}\frac{v_rv_l}{v^0_\gamma}\\
\nabsak^{j+3} v^0 = \frac{v^{\sharp j}}{v^0},\quad \nabsak^{j+3}\langle v\rangle_G=\frac{v^{\sharp j}}{\langle v\rangle_G},\quad \nabsak^{j+3} \left(\frac{v_k}{v^0}\right)=\frac1{v^0}\I^j_k-\frac{v_kv^{\sharp j}}{(v^0)^3}\label{eq:mom-mom}
\end{gather}
\end{subequations}}
\end{lemma}

\begin{lemma}[Frame commutators]\label{lem:commutators-zero} Let $\xi:P\rightarrow\R$ be sufficiently regular. Then, one has:
\change{\begin{subequations}
\begin{align*}
[\A_i,\A_j]\xi=&\,v^l{\Riem[G]^k}_{lij}\B_k\xi\numberthis\label{eq:AA-comm}\\
[\A_i,\B^j]\xi=&\,-\Gamma[G]^j_{il}\B^l\xi\numberthis\label{eq:AB-comm}\\
[\A_i,v_j\B^j]\xi=&\,0\numberthis\label{eq:AvB-comm}\\
[\del_t,\nabsak_i]\xi=&\,v_k\left(\del_t\Gamma[G]^k_{ij}\right)\B^j\xi \numberthis\label{eq:delt-nabhor-comm}\\
[\del_t,\nabsak^{j+3}]\xi=&\,0\label{eq:vert-comm}\numberthis\\
[\bm{X},\A_i]\xi=&\,\,-a^{-1}(N+1)\left[\frac{v^{\sharp j}v_l}{v^0}{\Riem[G]^l}_{mji}\B^m\xi-\frac{v_j}{v^0}\Gamma[G]_{ik}^j\B^k\xi\right]\numberthis\label{eq:XA-comm}\\
&\,+a^{-1}\frac{v_j}{v^0}\nabla_iN\A_j\xi-a^{-1}v^0\left(\nabla_i\nabla_{j}N\right)\B^j\xi\\
\numberthis\label{eq:XB-comm}[\bm{X},\B^i]\xi=&\,a^{-1}\left(\frac1{v^0}(G^{-1})^{ij}-\frac{v^{\sharp i}v^{\sharp j}}{(v^0)^3}\right)(N+1)\A_j\xi\\
&\,+a^{-1}\frac{v^{\sharp j}}{v^0}\Gamma[G]_{jl}^i\B^l\xi-a^{-1}\frac{v^{\sharp i}}{v^0}\left(\nabla_jN\right)\B^j\xi\\
\end{align*}
\end{subequations}}
\end{lemma}
\begin{proof}We take \eqref{eq:AA-comm}-\eqref{eq:AB-comm} from \change{\cite[(38)]{A-CGaSar22}}. \eqref{eq:AvB-comm}-\eqref{eq:vert-comm} are immediate. \eqref{eq:XA-comm} and \eqref{eq:XB-comm} follow with \eqref{eq:AA-comm}-\eqref{eq:AvB-comm} and the product rule. In particular, for \eqref{eq:XA-comm}, one computes
\change{\begin{align*}
&\,v^0\left(\nabla_jN\right)\B^j\A_i\xi-\A_i\left(v^0\nabla_{j}N\B^j\xi\right)\\
=&\,v^0\left[\left(\nabla_jN\right)\cdot\Gamma[G]_{ik}^j\B^k\xi-\left(\del_i\nabla_jN\right)\B^j\xi\right]\\
=&\,-v^0\left(\nabla_i\nabla_jN\right)\B^j\xi\,,
\end{align*}}
\end{proof}

\section{Norm setup, local well-posedness and bootstrap assumptions}\label{sec:norms}

\subsection{Norms and energies}

In this section, we introduce the necessary norms to state our initial data and bootstrap assumptions, as well the energies we will utilise in our bootstrap improvement mechanism.

\begin{definition}[Norms and function spaces] Let $\mathfrak{T}$ be a $M_t$-tangent tensor, let $\xi$ be a function on $P$, let $z\in M$ and let $\mu\in\R_0^+$ as well as $K\in\N$. We define the following semi-norms:
\begin{subequations}
\begin{align}
\|\mathfrak{T}\|_{\dot{C}^K_G(M_t)}:=&\,\sup_{x\in M_t}\left(\lvert\nabsak^K\mathfrak{T}\rvert_G\right)_x\\
\|\xi\|_{\dot{C}^{K}_{\mu,\G_0}(\change{T^\ast M_t})}:=&\,\sup_{(x,v)\in \change{T^\ast M_t}}\langle v\rangle_G^{\mu}\left\lvert \nabsak^K\xi(t,x,v)\right\rvert_{\G_0}\\
\|\mathfrak{T}\|_{\dot{H}^K_G(M_t)}:=&\,\left(\int_{M_t}\lvert \nabla^K\mathfrak{T}\rvert_G^2\,\vol{G}\right)^\frac12\\
\|\xi\|_{\dot{H}^K_{\mu,\G_0,0}(\change{T^\ast M_t})}:=&\,\left(\int_{\change{T^\ast M_t}}\langle v\rangle_G^{2\mu}\frac{1+\lvert v\rvert_G^2}{\lvert v\rvert_G^2}\lvert\nabsak^K\xi(t,x,v)\rvert_{\G_0}^2\,\vol{\G}\right)^\frac12\\
\|\xi\|_{\dot{H}^{K}_{\mu,\G_0,1}(\change{T^\ast M_t})}:=&\,\left(\int_{\change{T^\ast M_t}}\langle v\rangle_G^{2\mu}\lvert\nabsak^K\xi(t,x,v)\rvert_{\G_0}^2\vol{\G}\right)^\frac12\\
\|\xi(t,z,\cdot)\|_{\dot{H}^K_{\mu,\G_0\vert_{vert},0}(\change{T^\ast_xM_t})}:=&\,\left(\int_{\change{T^\ast_xM_t}}\langle v\rangle_G^{2\mu}\frac{1+\lvert v\rvert_G^2}{\lvert v\rvert_G^2}\lvert\nabsak_{vert}^K\xi(t,z,v)\rvert_{\G_0\vert_{vert}}^2\,\vol{\G\vert_{vert}}\right)^\frac12\\
\|\xi(t,z,\cdot)\|_{\dot{H}^{K}_{\mu,\G_0\vert_{vert},1}(\change{T^\ast_xM_t})}:=&\,\left(\int_{\change{T^\ast_xM_t}}\langle v\rangle_G^{2\mu}\lvert\nabsak_{vert}^K\xi(t,z,v)\rvert_{\G_0\vert_{vert}}^2\vol{\G\vert_{vert}}\right)^\frac12
\end{align}
The corresponding norms when replacing $G$ with $\gamma$ or $\G_{(0)}$ with $\underline{\gamma}_{(0)}$ are defined analogously, as well as the respective $H^K$- and $C^K$-norms. The corresponding function spaces are defined by completion with the respective norm. Finally, we write $\|\cdot\|_{\dot{H}^K_{\mu,\G_0}}:=\|\xi\|_{\dot{H}^K_{\mu,\G_0,1}}$ (and similar for the other Sasaki Sobolev norms) and define
\begin{equation}
\|\xi\|_{L^\infty_x{H}^{K}_{\mu,\G_0\vert_{vert},m}(\change{T^\ast M_t})}=\sup_{x\in M}\|\xi(t,x,\cdot)\|_{\dot{H}^{K}_{\mu,\G_0\vert_{vert},m}(\change{T^\ast_xM_t})}\,.
\end{equation}
\end{subequations}
\end{definition}

The following combined solution norms will be used to efficiently encode the initial data and bootstrap assumptions:

\begin{definition}[Solution norms]\label{def:sol-norm} We define the following norms to measure the size of near-FLRW solutions:
\begin{subequations}
\begin{align*}
\numberthis\label{eq:def-H}\mathcal{H}(t)=&\,\|\Psi\|_{{H^{18}_G}(M_t)}+\|\nabla\phi\|_{{H^{17}_G}(M_t)}+a^2\|\nabla\phi\|_{{\dot{H}^{18}_G}(M_t)}\\
&\,+\|\Sigma\|_{H^{18}_G(M_t)}+\|\RE\|_{H^{18}_G(M_t)}+\|\RB\|_{H^{18}_G(M_t)}\\
&\,+\|G-\gamma\|_{H^{18}_G(M_t)}+\|\Ric[G]-2\kappa G\|_{H^{16}_G(M_t)}+a^{-2}\|N\|_{H^{16}_G(M_t)}\\
&\,+\|f-f_{FLRW}\|_{H^{18}_{1,\G_0}(\change{T^\ast M_t})}\\
\numberthis\label{eq:def-H-top}\mathcal{H}_{top}(t)=&\,a^2\|\Psi\|_{{\dot{H}^{19}_G}(M_t)}+a^4\|\nabla\phi\|_{\dot{H}^{19}_G(M_t)}\\
&\,+a^2\|\Sigma\|_{\dot{H}^{19}_G(M_t)}+a^2\|\Ric[G]-2\kappa\|_{\dot{H}^{17}_G(M_t)}+a^2\|f-f_{FLRW}\|_{\dot{H}^{19}_{1,\G_0}(\change{T^\ast M_t})}\,
\end{align*}
\begin{align*}
\numberthis\label{eq:def-C}\mathcal{C}(t)=&\,\|\Psi\|_{{C^{16}_G}(M_t)}+\|\nabla\phi\|_{C^{15}_G(M_t)}+\|\Sigma\|_{C^{16}_G(M_t)}+\|\RE\|_{C^{16}_G(M_t)}+\|\RB\|_{C^{16}_G(M_t)}\\
&\,+\|G-\gamma\|_{C^{16}_G(M_t)}
+\|\Ric[G]-2\kappa G\|_{C^{14}_G(M_t)}+a^{-2}\|N\|_{C^{14}_G(M_t)}\\
&\,+\|\rho^{Vl}-{\rho}_{FLRW}^{Vl}\|_{C^{16}_G(M_t)}+\|\mathfrak{p}^{Vl}-\mathfrak{p}_{FLRW}^{Vl}\|_{C^{16}_G(M_t)}+\|\j^{Vl}\|_{C^{16}_G(M_t)}+\|S^{Vl,\parallel}\|_{C^{16}_G(M_t)}
\end{align*}
\end{subequations}
\end{definition}
We note that we do not include the highest order lapse norm that we could control in $\mathcal{H}$ and $\mathcal{C}$: As one can see from Section \ref{sec:energy-lapse}, we can in fact control lapse energies up to order $21$, but at high orders, these scaled energy estimates become progressively weaker. Conversely, for the bootstrap assumption, we only need to assume low order lapse bounds to control nonlinear error terms, and in particular will need to use that the lapse converges to $1$ near the Big Bang hypersurface. Thus, there is no benefit in including higher order lapse norms in $\mathcal{C}$ for the bootstrap assumption. In fact, including the lapse in $\mathcal{H}$ is redundant due to the elliptic estimates in Section \ref{sec:energy-lapse}, and we only include it to for the sake of convenience.
%
To improve the bootstrap assumptions, we will make use of the following energies:

\begin{definition}[Energies]\label{def:energies}
Let $K,L\in\N,$ and, where $K$ occurs, and $0<K\leq L$. We define:
\begin{subequations}
\begin{align*}
\numberthis\E^{(L)}(\phi,\cdot)=&\,(-1)^L\int_M\Psi\Lap^L\Psi-a^4\phi\Lap^{L+1}\phi\,\vol{G}
\label{eq:energydef-ibp}\\
\numberthis\E^{(L)}(W,\cdot)=&\,(-1)^l\int_M \langle\RE,\Lap^L \RE\rangle_G +\langle\RB,\Lap^L\RB\rangle_G\,\vol{G}\\
\numberthis\E^{(L)}(\Sigma,\cdot)=&\,(-1)^L\int_M\langle\Sigma,\Lap^L\Sigma\rangle_G\,\vol{G}\\
\numberthis\E^{(L)}(\Ric,\cdot)=&\,(-1)^L\int_M\left\langle \Ric[G]-2\kappa G,\Lap^L\left(\Ric[G]-2\kappa G\right)\right\rangle_G\,\vol{G}\\
\numberthis\E^{(L)}(N,\cdot)=&\,(-1)^L\int_M\langle N,\Lap^LN\rangle_G\,\vol{G}\\
\E^{(L)}_{\mu,0}(f,\cdot)=&\,\int_{\change{T^\ast M}}\langle v\rangle^{2\mu}_G\lvert \nabsak^L_{vert}(f-f_{FLRW})\rvert^2_{\G_0}\,\vol{\G}\numberthis\label{def:en-vlasov-vert}\\
\E^{(L)}_{\mu,K}(f,\cdot)=&\,\int_{\change{T^\ast M}} \langle v\rangle^{2\mu}_G\lvert\nabsak_{vert}^{L-K}\nabsak_{hor}^{K}f\rvert_{\G_0}^2\,\vol{\G}\numberthis\label{def:en-vlasov-hor}\\
\numberthis\E^{(\leq L)}_{\mu,\leq K}(f,\cdot)=&\sum_{R=1}^L\sum_{S=0}^{\min\{K,L\}}\E^{(R)}_{\mu,S}(f,\cdot)
\end{align*}
\end{subequations}
\end{definition}

%
\subsection{Initial data assumptions}

Using the solution norms introduced above, we can compactly state our initial data assumption as follows:

\begin{assumption}[Initial data assumptions]\label{ass:init} Let $(M,\mathring{g},\mathring{k},\mathring{\pi},\mathring{\psi},\mathring{f})$ be CMC initial data to the ESFV system as introduced in Section \ref{subsubsec:initial-data}. In particular, $\mathring{f}\geq 0$ is assumed to have compact momentum support, and in the case of massless Vlasov matter, we additionally assume that $f$ vanishes on an open neighbourhood of the zero section of $\change{T^\ast M}$. \\
For some sufficiently small $\epsilon>0$ and some $t_0>0$, we assume 
\begin{equation}\label{eq:init-ass}
\mathcal{H}(t_0)+\mathcal{H}_{top}(t_0)+\mathcal{C}(t_0)\leq\epsilon^2
\end{equation}
as well as
\begin{equation}\label{eq:init-ass-C-vlasov}
\|f-f_{FLRW}\|_{C^{11}_{1,\G_0}(\change{T^\ast M_{t_0})}}\leq \epsilon^2\,.
\end{equation}
\end{assumption}

\subsection{Local well-posedness}\label{subsec:lwp}

For all that follows, we need to establish that solutions to the ESFV system \eqref{eq:EVSF} locally exist in CMC gauge with zero shift, that all solution norms and energies are sufficiently regular in time and that we can extend solutions as long as we can control the size of our variables sufficiently well. To this end, we first collect local well-posedness results in harmonic gauge from \cite{Rin13,Sve12}, and then sketch how their proofs can be combined with results for the Einstein stiff-fluid system in CMC gauge to obtain the result for the ESFV system. In particular, we also collect a Cauchy stability result in harmonic gauge to argue that restricting the analysis to CMC data can be done without loss of generality (see Remark \ref{rem:cmc}).

\begin{lemma}[Local well-posedness for the ESFV system in harmonic gauge]\label{lem:local-wp-harmonic}
Let $l\in\N,\,l\geq 3$ and $\mu\in\R,\,\mu\geq 3$. Further, let $m=0,1$ denote the Vlasov mass, and let $(\mathring{g},\mathring{k},\mathring{\pi},\mathring{\psi},\mathring{f})$ be initial data to the ESFV system \eqref{eq:EVSF} as in Section \ref{subsubsec:initial-data} with
\begin{subequations}
\begin{equation}
\|\mathring{g}\|_{H^{l+1}_\gamma(M)}+\|\mathring{k}\|_{H^{l}_\gamma(M)}+\|\mathring{\pi}\|_{H^l_\gamma(M)}+\|\mathring{\psi}\|_{H^l_\gamma(M)}<\infty
\end{equation}
as well as
\begin{equation}
\mathring{f}\in H^{l}_{\mu,\underline{\gamma}_0,m}(\change{T^\ast M_{t_0}}),
\end{equation}
\end{subequations}
Then, a local solution $(\M,\g,\nabla\phi,\del_0\phi,\f)$ in harmonic gauge with Vlasov-mass $m$ exists in the following sense: For some $h>0$ and writing $J=(t_0-h,t_0+h)$, a spacetime manifold $\M$ can be foliated by spacelike Cauchy hypersurfaces $(M_t)_{t\in J}$, each of which is endowed with a collection of coordinates $(t\equiv x_U^0,x_U^1,x_U^2,x_U^3)$, where $x_U^i$ are spatial coordinates on an open subset $U\subseteq M$ and $M$ is covered by finitely many such coordinate neighbourhoods. In these coordinates, the harmonic gauge constraint
\[\g^{\rho\sigma}\left(\Gamma[\g]^\nu_{\rho\sigma}-\Gamma[\g_{FLRW}]^\nu_{\rho\sigma}\right)=0\]
is satisfied, one has
\[(\g_{FLRW})_{ij}=a(t)^2\gamma_{ij}\,,\]
and $(\M,\g,\nabla\phi,\del_0\phi,\f)$ solves \eqref{eq:EVSF}. Furthermore, the spacetime metric components and matter variables enjoy the following regularity in these coordinates:
\begin{align*}
\g_{\rho\sigma}\in&\,C^{l-1}(J\times M)\cap C^0(J,H_\gamma^{l+1}(M))\\
\del_t\g_{\rho\sigma}\in&\,C^{l-2}(J\times M)\cap C^0(J,H_\gamma^{l}(M))\\
\del_\rho\phi\in&\,C^{l-2}(J\times M)\cap C^0(J,H_\gamma^{l}(M))\\
\left((t,x)\mapsto \|f(t,x,\cdot)\|^2_{H^1_{\mu,\underline{\gamma}_0\vert_{vert},m}(\change{T^\ast_xM})}\right)\in&\,C_{dt^2+\gamma}^{l-1}(J\times M)
\end{align*}
For $l\geq 4$, one additionally has
\begin{equation}\label{eq:Vlasov-cont}
f\in C^{l-4}_{dt^2+\underline{\gamma}}(J\times \change{T^\ast M})\,.
\end{equation}
\begin{subequations}
If $(\mathfrak{t},t_0]$ is the past maximal interval of existence, then one has
\begin{equation}\label{eq:cont-crit-metric-sf-harm}
\lim_{t\downarrow \mathfrak{t}}\left(\|\g_{\rho\sigma}\|_{C^2_\gamma(M_t)}+\|\del_t\g_{\rho\sigma}\|_{C^1_\gamma(M_t)}+\|\del_\mu\phi\|_{C^1(M_t)}\right)=\infty\,,
\end{equation}
\begin{equation}\label{eq:cont-crit-Vlasov-harm}
\lim_{t\downarrow \mathfrak{t}}\sup_{x\in M}\left[\int_{\change{T^\ast_xM}} \langle p\rangle_{\gamma}^{2\mu}\cdot \left(\frac{\langle p\rangle_\gamma}{p^0_\gamma}\right)^{2l}\left(\left\lvert \f(t,x,p)\right\rvert^2+(p^0_\gamma)^2\left\lvert\B \f(t,x,p)\right\rvert_\gamma^2\right)\,\change{\mu_\gamma^{-1}}dp\right]=\infty\,.
\end{equation}
\end{subequations}
or $\mathfrak{t}=-\infty$.
\end{lemma}
\begin{proof}[Proof-Outline.] As argued, for example, in \cite[Section 22.2.3]{Rin13}, the Einstein equations and the wave equation for the scalar field can be modified in harmonic gauge such that one can equivalently solve a quasilinear system in $\g_{\mu\nu}$ and $\phi$, coupled with the transport equation for $\f$. For such systems, \cite[Proposition 19.76 and Lemma 19.83]{Rin13} lead to the stated result for $m=1$, and respectively \cite[Paper B, Proposition 8.48 and Lemma 9.1]{Sve12} for $m=0$, since one easily checks that all operators and nonlinearities remain admissible when no scalar field potential is present (see \cite[Section 19.1]{Rin13}, respectively \cite[Paper B, Section 8.1]{Sve12}). \change{We also note that, while the results in \cite{Rin13,Sve12} are proven while viewing the Vlasov distribution as a function on the mass shell, rather than the co-mass shell, we can induce initial data $\mathring{f}^\sharp:T^\ast M\rightarrow\R$ on the mass shell by
\[\mathring{f}^\sharp(x,q^i)=\mathring{f}\left(x,(\mathring{g}^{-1})^{ij}p_j\right)\]
which, due to the regularity requirements for $\mathring{g}$, enjoys the same Sobolev regularity as $\mathring{f}$, but now with respect to the corresponding weighted Sasaki metric induced on the tangent bundle. Then working with canonical coordinates $q^\mu$ on the mass shell, one can apply the aforementioned results and transform the corresponding solution $f^\sharp$ back via
\[f(t,x,p_i)=f^\sharp(t,x,\g_{i\mu}q^\mu)\,.\]
Again, the regularity of the metric components that the regularity of $f^\sharp$ transfers to $f$ in the corresponding function spaces with respect to the cotangent bundle and the co-mass shell.}
\end{proof}
We note that the non-negativity of the Vlasov distribution function is ensured by the Vlasov equation, and can be argued as in Lemma \ref{lem:vlasov-nonneg}. Furthermore, uniqueness follows from uniqueness of the maximal globally hyperbolic development (see Section \ref{subsubsec:initial-data}).\\
To prove that we can take initial data to be CMC without loss of generality, we need the following result:

\begin{lemma}[Cauchy stability of near-FLRW solutions to the ESFV system]\label{lem:cauchy-stab}
Let $l\in\N,\ l\geq 3$ and $\mu\in\R,\ \mu\geq 3$, and let $\delta>0$ be sufficiently small. Further, consider initial data $(M,\mathring{g},\mathring{k},\mathring{\pi},\mathring{\psi},\mathring{f})$ to the ESFV system with Vlasov mass $m=0,1$ (again as in Section \ref{subsubsec:initial-data}) such that one has
\begin{subequations}
\begin{equation}
\|\mathring{g}-\gamma\|_{H^{l+1}_\gamma(M)}+\|\mathring{k}+\gamma\|_{H^{l}_\gamma(M)}+\|\mathring{\pi}\|_{H^l_\gamma(M)}+\|\mathring{\psi}-C\|_{H^l_\gamma(M)}<\delta
\end{equation}
and
\begin{equation}
\change{\|(\mathring{f}-f_{FLRW})(t_0,\cdot,\cdot)\|_{H^{l}_{\mu,\underline{\gamma}_0,m}(\change{T^\ast M})}<\delta\,.}
\end{equation}
\end{subequations}
hold. In the massless Vlasov case, we additionally assume that $\mathring{f}$ vanishes in an open neighbourhood of $\{(x,p)\in T^\ast M \,\vert\, p=0\}$. Then the solution in the sense of Lemma \ref{lem:local-wp-harmonic} satisfies the following bound for some $K>0$ that is independent of $\delta>0$ for any $t\in[t_1,t_0], t_1>\mathfrak{t}$:
\begin{align*}\numberthis\label{eq:cauchy-stab}
\|g-a^2\gamma\|_{H_\gamma^{l+1}(M_t)}+\|k-\frac{\tau}3 a^2\gamma\|_{H_\gamma^l(M_t)}+\|\nabla \phi\|_{H_\gamma^l(M_t)}&\\+\|\del_0\phi-Ca^{-3}\|_{H_\gamma^l(M_t)}
+\|\f-\f_{FLRW}\|_{H^l_{\mu,\underline{\gamma}_0,m}(\change{T^\ast M_t})}&\,\leq K\delta
\end{align*}
\end{lemma}
\begin{proof}[Proof-Outline.]
For the massive case, this can be proven as in \cite[Corollary 24.10]{Rin13}, which in turn follows along similar lines as \cite[Theorem 15.10]{Ring09}. Essentially, the argument therein reduces to using the respective result for the metric components and matter variables in harmonic gauge from \cite[Corollary 20.7]{Rin13} and patching these local results to a spatially global result. The aforementioned statements are, at first glance, only applicable for smooth initial data, but this clearly extends to non-smooth data by a standard approximation argument. \\
This mostly applies analogously to $m=0$ -- the only issue that arises is that \eqref{eq:cont-crit-Vlasov-harm} is regularity dependent, thus solutions at high regularity might only be extendible in lower regularity. However, if the distribution function vanishes on an open neighbourhood of the zero section throughout the evolution, any choice of $l\geq 3$ in \eqref{eq:cont-crit-Vlasov-harm} is equivalent to any other (see the proof of \cite[Paper B, Lemma 9.2]{Sve12}). That this is the case can be shown as in \cite[Paper B, Lemma 8.25]{Sve12}, analogous to the proof of \eqref{eq:APMomMassless}.
\end{proof}

\begin{remark}[Existence of a CMC hypersurface]\label{rem:cmc}
Before moving on to our main well-posedness result in CMC gauge with zero shift, we illustrate why the above results imply that, as in \cite{Rodnianski2014,Speck2018,FU23}, it is is not a true restriction to take initial data to be CMC : \\

If the initial data is close to FLRW data as in \eqref{eq:init-ass}, but not CMC, Lemma \ref{lem:local-wp-harmonic} yields a local-in-time solution by evolving in harmonic gauge, and by Lemma \ref{lem:cauchy-stab}, we can restrict ourselves to a time interval such that we remain close to the FLRW data throughout the local time evolution. An implicit function theorem argument as in \cite[Section 2.5]{FajKr20} now implies that, within this local solution, we can find a Cauchy hypersurface $\Sigma^\prime$ close to the initial hypersurface $M_{t_0}$ that is CMC. The adaptation of the argument in \cite[Section 2.5]{FajKr20} can be done as sketched in \cite[Remark 8.1]{FU23}: The only noteworthy change compared to the latter is in computing that $(\g_{FLRW},t_0)$ remains a regular point of the mean curvature map $H$ under the addition of Vlasov matter and regardless of sectional curvature. One computes using \eqref{eq:Friedman}:
\begin{equation*}
a(t_0)^2dH_{(\g_{FLRW},t_0)}(0,w)=\left[\frac13\Lap_\gamma+\kappa-4\pi C^2a(t_0)^{-4}-4\pi a(t_0)^{-2}\left(\rho_{FLRW}^{Vl}+\mathfrak{p}_{FLRW}^{Vl}\right)\right]w
\end{equation*}
Since we assume $f_{FLRW}\geq 0$ and $C>0$, and since $\Lap_\gamma$ has no positive eigenvalues, this clearly is a regular point for $\kappa\leq 0$. For $\kappa>0$, $\kappa<4\pi C^2a(t_0)^{-4}$ can be ensured for small enough $t_0>0$, which Lemma \ref{lem:cauchy-stab} allows us to take without loss of generality.\\

Having obtained a CMC hypersurface, we need to ensure that Lemma \ref{lem:cauchy-stab} is sufficient to ensure Assumption \ref{ass:init} for suitably small $\delta>0$. This follows from Sobolev embedding for all terms except those associated with Vlasov matter: The Vlasov equation ensures that \change{${f}$ }is still nonnegative (as in Lemma \ref{lem:vlasov-nonneg}), and that it has compact momentum support as well as, in the massless case, support bounded away from the zero section (as in Lemma \ref{lem:APMom}). In particular, for $\delta<1$, any momentum weights can be dropped or added up to a constant depending on \change{$f_{FLRW}$}, and thus \change{$\|{f}-{f}_{FLRW}\|_{H^{18}_{1,\underline{\gamma}_0}(T^\ast M_t)}\lesssim \delta$ }also holds. \change{Using a co-mass shell analogue of \cite[Lemma 15.29]{Rin13} or estimates similar to the energy bounds in Lemma \ref{lem:density-control}}, this ensures the equivalent bound on all Vlasov matter terms in $\mathcal{C}$ after potentially updating constants. Finally, \eqref{eq:init-ass-C-vlasov} can be ensured since \cite[Lemma 15.29]{Rin13} \change{applied analogously to the co-mass shell }and the standard Sobolev embedding on $\R^3$ imply
\change{\[\|{f}-{f}_{FLRW}\|_{C^{11}_{1,\underline{\gamma}_0}(\change{T^\ast M_t})}\lesssim \sup_{x\in M}\|{f}-{f}_{FLRW}\|_{L^\infty_xH^{13}_{1,\underline{\gamma}_0\vert_{vert}}(\change{T^\ast_xM_t})}\lesssim \|{f}-{f}_{FLRW}\|_{H^{13}_{1,\underline{\gamma}_0}(\change{T^\ast M_t})}\,.\]}
Again, we note that, on the momentum support of \change{${f}-{f}_{FLRW}$ }on $[t_1,t_0]\times M$, $\langle v\rangle_\gamma \simeq 1$ holds.
\end{remark}

\change{\begin{remark}[Smallness of $t_0$]\label{rem:close-to-bb} To establish a strong total energy bound in Proposition \ref{prop:en-imp}, we will need to assume that $a(t)$ is sufficiently small for $t\in(t_{Boot},t_0]$, which simply means that $t_0$ needs to be sufficiently small since $a$ is strictly increasing on  $(0,T/2)$, see Lemma \ref{lem:scale-factor}. Lemma \ref{lem:cauchy-stab} ensures that, if we choose the initial perturbation to be small enough, we can use Cauchy stability to ensure that the solution remains sufficiently close to the FLRW solution in any fixed time interval $[t_1,t_0]$. We can then run the argument in Remark \ref{rem:cmc} to find a nearby CMC hypersurface, and can thus overall assume $t_0>0$ to be sufficiently small without loss of generality. 
\end{remark}

Having now prepared initial data to be CMC as well as sufficiently close to the Big Bang }if needed, we can state our main local well-posedness result:

\begin{lemma}[Local well-posedness in CMC gauge with zero shift]\label{lem:local-wp-CMC}
Let $(M,\mathring{g},\mathring{k},\mathring{\pi},\mathring{\psi},\mathring{f})$ be CMC initial data as in Lemma \ref{lem:local-wp-harmonic} and let $\mu\geq 3$. Additionally, assume that 
\[\inf_{x\in M}\mathring{\psi}(x)^2-\lvert\mathring{\pi}\rvert_{\mathring{g}_x}^2>0\,\,.\]
Then, this data launches a unique solution $(\M,g,k,\nabla\phi,\del_0\phi,\overline{f})$ of the CMC-transported ESFV-system (see \eqref{eq:EVSF} with CMC condition \eqref{eq:CMC}) toward the past, foliated by CMC hypersurfaces $(M_t)_{t\in J}$ with $J=(t_0-h,t_0]$ for some $h>0$. One has:
\begin{align*}
g\in&\,C_{dt^2+\gamma}^{l-1}(J\times M)\cap C^0(J,H_\gamma^{l+1}(M))\\
k\in&\,C_{dt^2+\gamma}^{l-2}(J\times M)\cap C^0(J,H_\gamma^{l}(M))\\
\nabla\phi\in&\,C_{dt^2+\gamma}^{l-2}(J\times M)\cap C^0(J,H_\gamma^{l}(M))\\
\del_t\phi\in&\,C_{dt^2+\gamma}^{l-2}(J\times M)\cap C^0(J,H_\gamma^{l}(M))\\
n\in&\,C_{dt^2+\gamma}^{l}(J\times M)\cap C^0(J,H_\gamma^{l+2}(M))\\
f\in&\,C^{l-4}_{dt^2+\underline{\gamma}}(J\times TM)\cap C^{l-1}(J,L^\infty_xH_{\mu,\underline{\gamma}_0\vert_{vert},m}^1(\change{T^\ast_{(\cdot)}M}))
\end{align*}
For the past-maximal interval of existence $(\mathfrak{t},t_0]$, one has $\mathfrak{t}=0$ or one of the following blow-up criteria are satisfied:\\
There exists a sequence $(t_m,x_m)$ with $t_m\downarrow\mathfrak{t}$ such that
\begin{enumerate}
\item such that the smallest eigenvalue of $g(t_m,x_m)$ converges to $0$, 
\item or such that $n(t_m,x_m)$ converges to $0$, 
\item or such that $\left(\lvert\del_0\phi\rvert^2+\lvert\nabla\phi\rvert_g^2\right)(t_m,x_m)$ converges to $0$,
\item or one of the following maps is unbounded:
\begin{subequations}
\begin{align}
s\in(\mathfrak{t},t_0]\mapsto&\,\|g\|_{C^2_\gamma(M_s)}+\|k\|_{C^1_\gamma(M_s)}+\|n\|_{C^2_\gamma(M_s)}+\|\del_t\phi\|_{C^1_\gamma(M_s)}+\|\nabla\phi\|_{C^1_\gamma(M_s)}\,,\label{eq:blowup-crit-geom}\\
s\in(\mathfrak{t},t_0]\mapsto&\, \|\rho^{Vl}\|_{C^1_\gamma(M_s)}+\|\j^{Vl}\|_{C^1_\gamma(M_s)}+\|S^{Vl}\|_{C^1_\gamma(M_s)}\,,\label{eq:blowup-crit-Vlasov-quant}\\ 
s\in(\mathfrak{t},t_0]\mapsto&\,\sup_{x\in M_s}\left[\int_{\change{T^\ast_xM_s}} \langle p\rangle_{\gamma}^{2\mu}\cdot \left(\frac{\langle p\rangle_\gamma}{p^0_\gamma}\right)^{2l}\left(\left\lvert \f(s,x,p)\right\rvert^2+(p^0_\gamma)^2\left\lvert \change{\B f}(s,x,p)\right\rvert_\gamma^2\right)\,\change{\mu_\gamma^{-1}} dp\right]\label{eq:blowup-crit-Vlasov-harm} 
\end{align}
\end{subequations}
\end{enumerate}
\end{lemma}
\begin{proof}[Proof-Outline.] In short, this can be proven along the same lines as in the proof of \cite[Theorem 14.1]{Rodnianski2014} for the stiff fluid, with the analysis of Vlasov matter following from \cite{Rin13,Sve12} as above. The continuation criterion \eqref{eq:blowup-crit-geom} then extends from the scalar field case, \eqref{eq:blowup-crit-Vlasov-quant} is the direct analogue of \eqref{eq:blowup-crit-geom} for the Vlasov matter quantities that occur in the elliptic-hyperbolic system for the metric variables, and \eqref{eq:blowup-crit-Vlasov-harm} is needed to ensure boundedness of coefficients arising from Vlasov energy estimates as in the proof in harmonic gauge.\\

To be more precise, the proof of \cite[Theorem 14.1]{Rodnianski2014} relies on showing that a modified system of equations leads to an elliptic-hyperbolic system in which the second fundamental form satisfies a wave equation, where its evolution equation is considered as an additional constraint quantity. The only properties of stiff fluid matter that are used in the proof are that the stiff fluid equation implies that the associated energy-momentum tensor being divergence-free, that it satisfies the strong energy condition and that it generates standard energy estimates. The former two are also satisfied for our matter model (where we again note that one has $f\geq 0$, and thus the strong energy condition is ensured). Since scalar field matter constitutes a subcase of stiff fluid matter, the energy estimates for stiff fluids apply to the scalar field as well. Thus, energy estimates for the geometry and scalar field can be straightforwardly extended from the stiff fluid setting.\\

The necessary energy estimates for the Vlasov distribution function can be established as in \cite{Rin13,Sve12}, where they are proven for a wide class of relativistic transport operators of the form
\change{\[L\xi=\del_t\xi+\frac{{z^i}}{q^0}\del_{x^i}\xi+\frac{G^i}{q^0}\del_{q^i}\xi\,,\]}
where \change{$z^i\equiv z^i[u]$ }and $G^i\equiv G^i[u]$ are sufficiently regular coefficient functions depending on the quantities $u$ that solve the coupled system of wave equations, which correspond to the spacetime metric components. \change{Again, note that these statements are proven for the mass shell, so we need to consider the Vlasov equation on the mass shell, i.e., apply this to the equation
\begin{equation}\label{eq:Vlasov-sharp}
Lf^\sharp=\del_t\xi+\frac{q^i}{q^0}\,n\,\del_{x^i}f^\sharp+\left(-\frac{q^iq^l}{q^0}\,n\,\Gamma[g]^j_{il}-q^0\,g^{ij}\,\nabla_in+2n\,k^{j}_{\ i}\,q^i\right)\del_{q^j}\xi=0\,.
\end{equation}
Nevertheless, }we can establish analogous estimates for these components using the elliptic-hyperbolic system from \cite[Theorem 14.1]{Rodnianski2014}, and the transport operator for \change{$f^\sharp$ }is actually simpler than the one that needs to be considered in harmonic gauge, energy estimates for the Vlasov distribution also extend. Finally, the components of the Vlasov energy momentum tensor can be bounded up to constant by the energies for the distribution function as in Lemma \ref{lem:density-control}.\\

Thus, we can prove local well-posedness of the modified elliptic-hyperbolic system metric and scalar field coupled to the Vlasov equation by a constructing a similar convergent sequence as in \cite[Proposition 19.76]{Rin13} and \cite[Paper B, Proposition 8.48]{Sve12}. That this then also solves the Einstein equations -- i.e., that the constraints are propagated -- can again be proven as in the stiff fluid setting since we only use properties of the energy momentum tensor shared by both matter models.\\

The blow-up criteria follow by combining what arises from the elliptic-hyperbolic system (extending from \cite[Theorem 14.1]{Rodnianski2014}) with those arising from the transport equation from Lemma \ref{lem:local-wp-harmonic}. For the latter, \eqref{eq:cont-crit-metric-sf-harm} is covered by the requirements \eqref{eq:blowup-crit-geom} from the elliptic-hyperbolic system, and thus only \eqref{eq:cont-crit-Vlasov-harm} needs to be added.
\end{proof}

\begin{remark}[On regularity of norms and energies]
In addition to the initial data assumptions (see Assumption \ref{ass:init}), we can assume, without loss of generality by a standard approximation argument, that the initial data is sufficiently regular. More precisely, using Lemma \ref{lem:local-wp-CMC}, we will assume tacitly $l\in\N$ to be sufficiently large such that all norms and energies are continuously differentiable in time. We note that, in the massless case, the blow-up criterion \eqref{eq:cont-crit-Vlasov-harm} is dependent on how regular we choose our initial data to be a priori. However, due to control on the support of $f$ (see Lemma \ref{lem:APMom}), it will turn out that this additional weight can be ignored when checking this criterion (see Proof of Theorem \ref{thm:main}).
\end{remark}

\subsection{Bootstrap assumptions}

By Lemma \ref{lem:local-wp-CMC}, initial data as in Assumption \ref{ass:init} generates a local-in-time solution toward the past such that $\mathcal{C}(t)$ and $\|f-f_{FLRW}\|_{C^{11}_{1,\G_0}(\change{T^\ast M_t})}$ are continuous in $t$. Thus, the following bootstrap assumption can be satisfied:

\begin{assumption}[Bootstrap assumptions]\label{ass:bootstrap}
Let $\sigma=\epsilon^\frac{1}{16}$, fix $K_0>0$ and $c_0>0$ such that $c_0\sigma<1$. We assume that, for the bootstrap time $t_{Boot}\in[0,t_0)$, the following holds for any $t\in(t_{Boot},t_0]$:
\begin{subequations}
\begin{equation}\label{eq:BsC}
\mathcal{C}(t)\leq K_0\epsilon a^{-c_0\sigma}\,.
\end{equation}
Additionally, we assume
\begin{equation}\label{eq:BsVlasovhor}
\|\nabsak_{hor}(f-f_{FLRW})\|_{C^{10}_{1,\G_0}(\change{T^\ast M_t})}\leq K_0{\epsilon}^\frac14a^{-c\sigma}
\end{equation}
\end{subequations}
\end{assumption}

Our goal is to show that the bootstrap assumptions imply 
\begin{equation*}
\mathcal{C}\leq K_1\epsilon a^{-c_1\epsilon^\frac18},\quad \|\nabsak_{hor}(f-f_{FLRW})\|_{C^{10}_{1,\G_0}(\change{T^\ast M_t)}}\leq K_1\sqrt{\epsilon} a^{-c\sqrt{\epsilon}}
\end{equation*}
for suitable constants $c_1,K_1>0$ and any $t\in(t_{Boot},t_0]$, which is a strict improvement if $\epsilon>0$ is chosen sufficiently small.

\begin{remark}[Lapse bounds]\label{rem:BsC}
We note that, in particular, \eqref{eq:BsC} implies
\begin{equation}\label{eq:BsN}
\|N\|_{C^{14}_G(M_t)}\lesssim \epsilon a(t)^{2-c\sigma}\,,
\end{equation}
and we will in fact only be able to improve the convergence rate to $\epsilon a^{2-c\epsilon^\frac18}$ by Theorem \ref{thm:main}. Thus, the lapse converges less strongly than in the pure scalar field case, and the bootstrap assumption is also weaker (see \cite[(3.18h)]{FU23}). This is due to Vlasov matter being the asymptotically strongest term in \eqref{eq:REEqLapse2}. However, it does not lead to substantial changes in the argument since the fact that the lapse converges at a rate stronger than, for example, $a(t)$ for sufficiently small $\sigma>0$ is enough to make otherwise borderline terms containing low orders of $N$ converge.
\end{remark}

%

\section{A priori estimates}\label{sec:ap}

\subsection{Strong low order bounds on spacetime metric and scalar field variables}

In this section, we collect the necessary improved low order estimates for all variables that are not the Vlasov distribution function or associated matter quantities. Since Vlasov matter is asympotically negligible in the evolution of metric variables and the shear as well as in the Hamiltonian and Bel-Robinson constraint equations, and does not occur in the wave equation at all, all of these bounds follow almost directly from \cite[Sections 4.1 and 4.2]{FU23}, and we only sketch how one controls the occuring Vlasov terms with the bootstrap assumption.

\begin{lemma}[Strong $C^0_G$ bounds]\label{lem:AP0}For $t\in(t_{Boot},t_0]$, the following estimates hold:
\begin{subequations}
\begin{align}
\|\Psi\|_{C^0_G(M_t)}\lesssim&\,\epsilon\label{eq:APPsi}\\
\|\Sigma\|_{C^0_G(M_t)}+\|\Sigma^\sharp\|_{C^0_\gamma(M_t)}\lesssim&\,\epsilon\label{eq:APSigma}\\
\|\RE\|_{C^0_G(M_t)}\lesssim&\,\epsilon\label{eq:APE}
\end{align}
\end{subequations}
\end{lemma}
\begin{proof}
\underline{\eqref{eq:APPsi}} is proven identically to the scalar field case (see \cite[(4.4a)]{FU23}), replacing the lapse bootstrap assumption therein with \eqref{eq:BsN}. \\
Regarding  \underline{\eqref{eq:APSigma}}, one has
\begin{equation}\label{eq:APSigma-step}
-\del_t\lvert\Sigma\rvert_G^2 = 2{\left(\del_t\Sigma^\sharp\right)^i}_j{\left(\Sigma^\sharp\right)^j}_i\lesssim \left\lvert (\del_t\Sigma^\sharp)^\parallel\right\rvert_G\lvert\Sigma\rvert_G\,,
\end{equation}
where we recall that $(\del_t\Sigma^\sharp)^\parallel$ is the tracefree, symmetric part of $\del_t\Sigma^\sharp$. Estimating all non-Vlasov terms as in the proof of \cite[(4.2b)]{FU23} using \eqref{eq:BsC}, we have
\[\left\lvert \del_t\Sigma^\sharp\right\rvert_G\lesssim \epsilon a^{-1-c\sigma}(1+\lvert \Sigma\rvert_G)+a^{-1}\left\lvert S^{Vl,\parallel}\right\rvert_G\,.\]
That this bound is weaker than its analogue in the pure scalar field case is, again, down to \eqref{eq:BsN}. The Vlasov term can also be bounded by $\lesssim \epsilon a^{-1-c\sigma}$ using the bootstrap assumption \eqref{eq:BsC}. Integrating \eqref{eq:APSigma-step} and inserting this bound along with \eqref{eq:init-ass} and \eqref{eq:a-integrals} then implies
\[\lvert\Sigma\rvert_G^2(t)\lesssim \lvert \Sigma\rvert_G^2(t_0)+\int_t^{t_0}\epsilon a(s)^{-1-c\sigma}\,ds\lesssim \epsilon.\]
The bound for $\lvert \Sigma^\sharp \rvert_{C^0_\gamma}$ is proven identically since the bootstrap assumption on $G-\gamma$ ensures that all quantities in $\mathcal{C}$ satisfy the bootstrap assumptions with $C_G$ replaced by $C_\gamma$, up to updating constants.\\
\underline{\eqref{eq:APE}} then follows directly from \eqref{eq:REEqConstrE}: Due to $\langle G,\RE\rangle_G=0$, the second line in \eqref{eq:REEqConstrE} can be ignored, and using \eqref{eq:APSigma} for all shear terms and \eqref{eq:BsC} for all other terms, including Vlasov matter, one has $\langle\RE,\RE\rangle_G\lesssim \epsilon\lvert \RE\rvert_G$ and thus the statement.
\end{proof}


\begin{lemma}[Strong low-order $C_G$-bounds for spacetime and scalar field variables]\label{lem:AP}For $t\in(t_{Boot},t_0]$, one has:
\begin{subequations}
\begin{align}
\|\Psi\|_{C^{13}_G(M_t)}\lesssim&\,\epsilon a(t)^{-c\sqrt{\epsilon}}\,\label{eq:APmidPsi}\\
\|\Sigma\|_{C^{12}_G(M_t)}\lesssim&\,\epsilon a(t)^{-c\sqrt{\epsilon}}\,\label{eq:APmidSigma}\\
\|G^{\pm 1}-\gamma^{\pm 1}\|_{C^{12}_G(M_t)}\lesssim&\sqrt{\epsilon}a(t)^{-c\sqrt{\epsilon}}\,\label{eq:APmidG}\\
\|\nabla\phi\|_{C^{12}_G(M_t)}\lesssim&\,\sqrt{\epsilon}a(t)^{-c\sqrt{\epsilon}}\label{eq:APmidphi}\\
\|\Ric[G]+\frac29G\|_{C^{10}_G(M_t)}\lesssim&\,\sqrt{\epsilon}a(t)^{-c\sqrt{\epsilon}}\label{eq:APmidRic}\\
\|\RE\|_{C^{12}_G(M_t)}\lesssim&\,\epsilon a(t)^{-c\sqrt{\epsilon}}\label{eq:APmidE}\\
\|\RB\|_{C^{11}_G(M_t)}\lesssim&\,\epsilon a(t)^{2-c\sqrt{\epsilon}}\label{eq:APmidB}
\end{align}
\end{subequations}
\end{lemma}
\begin{proof}
Regarding \eqref{eq:APmidSigma}, \eqref{eq:BsC} implies
\[\left\lvert \nabla^J\del_t\Sigma^\sharp\right\rvert_G\lesssim \epsilon a^{-1-c\sigma},\]
where the Vlasov matter terms and the lapse bound \eqref{eq:BsN} lead to the leading terms. Since the integral of the right hand side over $[0,t_0]$ is still bounded by $\epsilon$ up to constant due to \eqref{eq:a-integrals}, \eqref{eq:APmidSigma} is then proven as in \cite[Lemma 4.3]{FU23}. \\
\eqref{eq:APmidPsi}, \eqref{eq:APmidG}-\eqref{eq:APmidRic} and \eqref{eq:APmidB} are also proven almost entirely identically to their analogues in \cite[Lemma 4.3]{FU23}, the only differences arising again from the weaker bootstrap bound \eqref{eq:BsN} on the lapse that is, as with the shear, sufficient to prove that terms containing $N$ are asymptotically negligible. Finally, \eqref{eq:APmidB} is proven by applying $\nabla^J$ to \eqref{eq:REEqConstrE} and then arguing as in the $C^0_G$-case, using \eqref{eq:APmidSigma} for the shear and \eqref{eq:BsC} for the remaining terms.
\end{proof}

Finally, we collect that the volume element $\mu_G$ is uniformly bounded:

\begin{lemma}
For $\mu_G=\sqrt{\det G}$ and $t\in(t_{Boot},t_0]$, one has
\begin{equation}\label{eq:APmuG}
\|\mu_G-\mu_\gamma\|_{C^0(M_t)}\lesssim \epsilon\,.
\end{equation}
\end{lemma}
\begin{proof}
This follows immediately from $\del_t\mu_G=-N\tau \mu_G$, \eqref{eq:BsN}, the initial data assumption and the Gronwall lemma.
\end{proof}

\subsection{A priori observations for Vlasov matter}

\begin{remark}[On the charateristic method]\label{rem:characteristic-method}
The following a priori estimates involving the Vlasov distribution function will rely on the fact that initial value problems of the form
\[\del_t\xi+b(t,x,v)^i\del_{x^i}\xi(t,x,v)+\change{B_i}(t,x,v)_{i}\del_{v_i}\xi(t,x,v)+c(t,x,v,\xi)=0,\ \xi(0,x,v)=\xi_0(x,v)\]
can be locally solved using the method of characteristics, i.e.: There exist functions $X:I\rightarrow U\subseteq \R^3$, $V:I\rightarrow\R^3$, $Z:I\rightarrow \R$ with
\[Z(s)=\xi(s,X(s),V(s))\]
that satisfy the following equations:
\begin{align*}
\dot{X}^i(s)=&\,b^i(s,X(s),V(s))\\
\change{\dot{V}_i(s)=}&\,\change{B_{i}(s,X(s),V(s))}\\
\dot{Z}(s)=&\,-c(s,X(s),V(s),Z(s))
\end{align*}
This is a consequence of \cite[p. 107, Theorem 2]{Evans98}. In particular, we will apply this to equations that are shifts or derivatives of the rescaled Vlasov equation \eqref{eq:vlasov-resc}. Since, on the bootstrap interval, the solution to the Vlasov equation exists and is unique, it can be expressed in this fashion in the entire bootstrap interval after choosing inital data on some hypersurface $M_t$ for $t\in(t_{Boot},t_0]$. In particular, solutions to \eqref{eq:vlasov-resc} remain constant along integral curves $(X,V)$ of the following characteristic system:
\begin{subequations}\label{eq:char}
\begin{align*}
\frac{dX^i}{dt}=&\,a^{-1}(N+1)\frac{(G^{-1})^{ij}V_j}{\sqrt{m^2a^2+\lvert V\rvert_G^2}}\numberthis\label{eq:charX}\\
\change{\frac{dV_i}{dt}=}&\change{\,a^{-1}(N+1)(G^{-1})^{jk}\Gamma[G]^l_{ik}\frac{V_jV_l}{\sqrt{m^2a^2+\lvert V\rvert_{G}^2}}-a^{-1}\sqrt{m^2a^2+\lvert V\rvert_G^2}\nabla_{i}N}\numberthis\label{eq:charV}\\
\end{align*}
\end{subequations}
We introduce the following notation:
\[\mathcal{V}=\lvert V\rvert_G,\quad V^0=\sqrt{m^2a^2+\mathcal{V}^2},\quad V^0_\gamma=\sqrt{m^2a^2+\lvert V\rvert_\gamma^2}\]
\end{remark}

\begin{lemma}[Non-negativity of the Vlasov distribution function]\label{lem:vlasov-nonneg}
For any $t\in(t_{Boot},t_0]$, $f(t,\cdot,\cdot)$ is nonnegative, as well as $\rho^{Vl}(t,\cdot)$ and $\mathfrak{p}^{Vl}(t,\cdot)$.
\end{lemma}
\begin{proof}
Recall that $s\mapsto f(s,X(s),V(s))$ remains constant for characteristics $(X,V)$ satisfying \eqref{eq:char} and that $\mathring{f}$ is nonnegative. Thus, for any $(x,v)\in \change{T^\ast M_t}$, there exist $(x_{init},v_{init})\in \change{T^\ast M}$ such that one has
\[0\leq \mathring{f}(x_{init},v_{init})=f(t_0,x_{init},v_{init})=f(t,x,v)\,.\] 
Since $v^0$ is also non-negative, so are then $\rho^{Vl}(t,\cdot)$ and $\mathfrak{p}^{Vl}(t,\cdot)$.
\end{proof}

\begin{lemma}[Bounds for the momentum support]\label{lem:APMom} There exists a constant $K>0$ that is independent of $\epsilon$ such that the following bounds hold:
\begin{subequations}
\begin{align}
\P(t)\leq \P^0(t)\leq &\,\change{Ka(t)^{-c\sqrt{\epsilon}}}\label{eq:APMom}\\
\change{\P_\gamma(t)\leq\P_\gamma^0(t)}\leq &\,\change{K}\label{eq:APMomgamma}
\end{align}
\end{subequations}
Furthermore, for $m=0$, one additionally has that the momentum support of $f$ is bounded away from zero as follows:
\begin{align}\label{eq:APMomMassless}
\sup\left\{\lvert v\rvert_G^{-2}\ \vert\ x\in M_t,\ v\in \supp f(t,x,\cdot)\right\}&\,\change{\leq K a(t)^{-c\sqrt{\epsilon}}}\\
\change{\sup\left\{\lvert v\rvert_\gamma^{-2}\ \vert\ x\in M_t,\ v\in \supp f(t,x,\cdot)\right\}}&\,\change{\leq K}\label{eq:APMomMasslessgamma}
\end{align}
\end{lemma}
\begin{proof}
Noting $\mathcal{V}\leq V^0$, we get the following from the characteristic equation \eqref{eq:charV}: 
\begin{align*}
\numberthis\label{eq:charV0}-\frac{d(V^0)^2}{dt}&\,\lesssim\change{-2\frac{\dot{a}}a\cdot m^2a^2+\lvert\del_tG^{-1}\rvert_G\mathcal{V}^2+a^{-1}\lvert N+1\rvert\lvert \Gamma[G]\rvert_G\frac{\mathcal{V}^3}{V^0}+a^{-1} V^0\mathcal{V}^2\lvert \nabla N\rvert_G}\\
&\,\change{\lesssim \left(a^{-3}\|\Sigma\|_{C^0_G}+a^{-3}\|N\|_{C^0_G}+a^{-1}(\|G-\gamma\|_{C^1_G}+\|\gamma\|_{C^1_G})+a^{-1}\|N\|_{\dot{C}^1_G}\right)(V^0)^2}
\end{align*}
Using \eqref{eq:BsC} for the metric and lapse as well as \eqref{eq:APSigma}, this can be estimated as follows:
\begin{align*}
-\frac{d(V^0)^2}{dt}\lesssim&\,\left(\epsilon a^{-3}+a^{-1-c\sigma}\right)(V^0)^2
\end{align*}
Updating $c>0$, the Gronwall lemma then implies 
\begin{equation}\label{eq:APcharV}
(V^0(t))^2\leq (V^0(t_0))^2\cdot a(t)^{-2c\sqrt{\epsilon}}\,.
\end{equation}
For any $x\in M_t$ and any $v\in \supp f(t,x,\cdot)$, $(x,v)$ must lie on a unique characteristic $(X,V)$ emanating from $(X(t_0),V(t_0))\in \supp f(t_0,\cdot,\cdot)$. Since $\mathring{f}$ has compact momentum support, \eqref{eq:APcharV} implies
\[V^0(t)\lesssim \P(t_0)a(t)^{-c\sqrt{\epsilon}}\]
for $V(t)=v\in\supp f(t,x,\cdot)$, and thus \eqref{eq:APMom} after taking the supremum in $v$. \change{\eqref{eq:APMomgamma} is proven identically, noting that since $\del_t\gamma=0$, one has
\[-\frac{d(V^0_\gamma)^2}{dt}\lesssim a^{-1-c\sigma}(V^0_\gamma)^2\]
and thus $V^0_\gamma$ remains uniformly bounded. 
\eqref{eq:APMomMassless}-\eqref{eq:APMomMasslessgamma} also follow similarly}, using that $\mathring{f}$ vanishes on an open neighbourhood of the zero section.

\end{proof}

To obtain strong low order bounds on the Vlasov distribution function, we will need to compare it to the reference distribution function $f_{FLRW}$ while using the Sasaki connection with respect to $G$. To control the resulting error terms, we first collect the following bounds on the reference distribution:

\begin{lemma}[Derivatives of the reference distribution function]\label{lem:hor-deriv-ref}
For $0<K<L\leq 19$, one has
\begin{subequations}
\begin{align*}\numberthis\label{eq:hor-deriv-ref-mixed}
\|\nabsak_{vert}^{L-K}\nabsak_{hor}^{K}f_{FLRW}\|_{L^2_{1,\G_0}(\change{T^\ast M_t})}\lesssim a(t)^{-c\sqrt{\epsilon}}&\,\Big(\underbrace{\sqrt{\epsilon}\|\Ric[G]-2\kappa G\|_{H^{K-1}_G(M_t)}}_{\text{if }K>1}\\
&\,+\|\Gamma-\Gamhat\|_{H^{K-1}_G(M_t)}+\change{\|G^{-1}-\gamma^{- 1}\|_{H^{K-1}_G(M_t)}}\Big)
\end{align*}
Additionally, for $L>0$, the following holds:
\begin{align*}\numberthis\label{eq:hor-deriv-ref}
\|\nabsak_{hor}^{L}f_{FLRW}\|_{L^2_{1,\G_0}(\change{T^\ast M_t})}\lesssim a(t)^{-c\sqrt{\epsilon}}&\,\Big(\underbrace{\sqrt{\epsilon}\|\Ric[G]-2\kappa G\|_{H^{L-2}_G(M_t)}}_{\text{if }L>2}\\
&\,+\|\Gamma-\Gamhat\|_{H^{L-1}_G(M_t)}+\change{\|G^{-1}-\gamma^{-1}\|_{H^{L-1}_G(M_t)}}\Big)
\end{align*}
Finally, let $0<K<L\leq 12$. Then, one has
\begin{align}\numberthis\label{eq:hor-deriv-ref-C}
\|\nabsak_{vert}^{L-K}\nabsak_{hor}^{K}f_{FLRW}\|_{C^0_{1,\G_0}(\change{T^\ast M_t})}\lesssim&\,\sqrt{\epsilon}a(t)^{-c\sqrt{\epsilon}}\,,\\
\|\nabsak_{vert}^{L-(K-1)}\nabsak_{hor}^{K-1}\bm{X}f_{FLRW}\|_{C^0_{1,\G_0}(\change{T^\ast M_t})}\lesssim&\,\change{\epsilon a(t)^{-1-c\sigma}}\,.\numberthis\label{eq:hor-deriv-Xref-C}
\end{align}
\end{subequations}
\end{lemma}
Recall that, by Remark \ref{rem:contractions-well-def}, all of the objects above are well-defined.
\begin{proof}
Before going into the proof, we note that vertical derivatives of $f_{FLRW}$ are a linear combination of terms of the form
\[v^{\ast_{\gamma}m_1}\cdot \mathcal{F}^{(m_2)}(\lvert v\rvert_\gamma^2)\]
with $m_1,m_2\in \N$. Since we have $\mathcal{F}\in C_0^\infty(\R_0^+)$, the distribution term is uniformly bounded. Further, on the support of $f_{FLRW}$, $v$ is uniformly bounded with respect to $\underline{\gamma}\vert_{vert}$. Thus,
\[(v^0)^L\lvert\nabsak_{vert}^L f_{FLRW}(t,x,v)\rvert_{\G}= \lvert \nabsak^L_{vert} f_{FLRW}(t,x,v)\rvert_{\G}\lesssim a(t)^{-c\sqrt{\epsilon}}\]
holds by \eqref{eq:APMom} and \eqref{eq:APmidG} for any $L>0$, as well as
\[\|\nabsak_{vert}^L f_{FLRW}\|_{L^2_{1,\G_0}}\lesssim a^{-c\sqrt{\epsilon}}\]
(also using \eqref{eq:APMomMassless} for $m=0$).\\
We compute:
\change{\begin{align*}
\numberthis\label{eq:ref-hor-formula}\A_i f_{FLRW}(\cdot,\cdot,v)=&\,\mathcal{F}^\prime(\lvert v\rvert_{\gamma}^2)\left[\del_i(\gamma^{-1})^{jk}v_jv_k+2\Gamma_{ij}^rv_kv_r(\gamma^{-1})^{jk}\right]\\
=&\,\mathcal{F}^\prime(\lvert v\rvert_{\gamma}^2)\left[-\Gamhat^j_{il}(\gamma^{-1})^{lk}v_jv_k-\Gamhat^{k}_{il}(\gamma^{-1})^{jl}v_jv_k+2\Gamma_{ij}^rv_kv_r\right]\\
=&\,2\mathcal{F}^\prime(\lvert v\rvert_{\gamma}^2)\,(\gamma^{-1})^{jk}\,\left[\Gamma_{ij}^l-\Gamhat_{ij}^l\right]\,v_kv_l\\
=&\,2\mathcal{F}^\prime(\lvert v\rvert_{\gamma}^2)\,\left[(G^{-1})^{jk}-(G^{-1}-\gamma^{-1})^{jk}\right]\,\left[\Gamma_{ij}^l-\Gamhat_{ij}^l\right]\,v_kv_l
\end{align*}
Thus, we can estimate
\begin{align*}
\lvert \nabsak_{hor}f_{FLRW}(\cdot,\cdot,v)\rvert_{\G_0}\lesssim&\,\lvert\Gamma-\Gamhat\rvert_G\left(1+\lvert G^{-1}-\gamma^{-1}\rvert_G\right)\lvert v\rvert_G^2\chi_{\supp \mathcal{F}}(\lvert v\rvert_\gamma^2)\,.
\end{align*}}
This proves \eqref{eq:hor-deriv-ref-C} at order $1$ follows by integrating over $(\change{T^\ast M},\underline{G}_0)$ and using that $\mathcal{F}$ has compact support. \eqref{eq:hor-deriv-ref} for $L=1$ follows from the same estimates by inserting \eqref{eq:APmidG}.\\ 
For $L=2$, notice that
\change{\begin{equation}\label{eq:nabsak2fref}
\nabsak_i\nabsak_jf_{FLRW}=\A_i\A_jf_{FLRW}-\Gamma_{ij}^k\A_kf_{FLRW}-\frac12v_l{{\Riem[G]^l}_{kij}}\B^{k}f_{FLRW}
\end{equation}}
The first two terms \change{in }\eqref{eq:nabsak2fref} can be dealt with as before, while for the final term, we compute
\change{\begin{align*}
-\frac12v_l{\Riem[G]^l}_{kij}\B^{k}f_{FLRW}(\cdot,\cdot,v)=&\,-v_l{\Riem[G]^l}_{kij}(\gamma^{-1})^{mk}v_m\cdot \mathcal{F}^\prime(\lvert v\rvert_\gamma^2)\\
=&\,\left[{\Riem[G]^l}_{kij}(G^{-1}-\gamma^{-1})^{mk}v_lv_m-\underbrace{v_kv_l{\Riem[G]^{\sharp lk}}_{ij}}_{=0}\right]\cdot \mathcal{F}^\prime(\lvert v\rvert_\gamma^2)\,.
\end{align*}}
Thus, this can be estimated in $L^2_{1,\G_0}$ by \change{$a^{-c\sqrt{\epsilon}}\|G^{-1}-\gamma^{-1}\|_{L^2_G}^2$} using that $\mathcal{F}$ is compactly supported as well as \eqref{eq:APmidG} and \eqref{eq:APmidRic}. This yields \eqref{eq:hor-deriv-ref} for $L=2$. For $L\geq 3$, we repeat this procedure using the product rule, incurring curvature error terms. \eqref{eq:hor-deriv-ref-mixed} is proven similarly, and so is \eqref{eq:hor-deriv-ref-C} using the a priori estimates \eqref{eq:APmidRic}, \eqref{eq:APmidG} and \eqref{eq:APMom}.

Regarding \eqref{eq:hor-deriv-Xref-C}, we split up $\bm{X}f_{FLRW}$ into the horizontal and vertical derivative terms and prove that both parts satisfy the claimed bound: The terms arising from the horizontal components of $\bm{X}f$ can be bounded directly using \eqref{eq:hor-deriv-ref-C} since $\A_i\xi=\nabsak_i\xi$. For the remaining vertical components, we may encounter terms where only vertical derivatives fall on $f_{FLRW}$. However, since we are only interested in a uniform pointwise bound, we don't need to ensure that all Christoffel terms that arise from the Sasaki connection coefficients reduce to difference tensor terms, and note that \eqref{eq:APmidG} also implies the crude bound
\[\lvert \del^{\leq 11}\Gamma\rvert_G\lesssim a^{-c\sqrt{\epsilon}}\]
in a coordinate neighbourhood on which $\Gamhat$ remains bounded. The bound then follows for the vertical parts of $\bm{X}f$ using \change{\eqref{eq:BsC} }and that $M$ is compact.
\end{proof}

\begin{lemma}[\change{$C^0_{1,\cdot}$-estimates }for the Vlasov distribution function]\label{lem:APVlasov0}For $t\in(t_{Boot},t_0]$, one has:
\begin{subequations}
\change{\begin{align}
\|f-f_{FLRW}\|_{C^0_{1,\underline{\gamma}_0}(T^\ast M_t)}}\lesssim&\,\epsilon\label{eq:APVlasov0gamma}\\
\|f-f_{FLRW}\|_{C^0_{1,\G_0}(T^\ast M_t)}\lesssim&\,\epsilon a(t)^{-c\sqrt{\epsilon}}\label{eq:APVlasov0}
\end{align}
\end{subequations}
\end{lemma}

\begin{proof}
We start out by rewriting the rescaled Vlasov equation \eqref{eq:vlasov-resc} as an inhomogeneous transport equation for $f-f_{FLRW}$:
\change{\begin{align*}
\del_t\left(f-f_{FLRW}\right)=&\,\bm{X}(f-f_{FLRW})-a^{-1}\frac{v^{\sharp j}}{v^0}(N+1)\A_jf_{FLRW}+a^{-1}v^0\,\nabla_{j}N\,\B^j f_{FLRW}
\end{align*}}
Let $(X,V,Z)$ be the corresponding characteristics satisfying \eqref{eq:char} emanating from $(x,v)\in \change{T^\ast M_t}$ with $Z(t)=(f-f_{FLRW})(t,x,v)$. Then, we have
\change{\begin{align*}
\frac{dZ}{dt}=&\,a^{-1}\frac{V_l}{V^0}(G^{-1})^{lj}(N+1)\A_jf_{FLRW}-a^{-1}V^0\,\nabla_jN\,\B^j f_{FLRW}
\end{align*}}
Note that, on the support of $(f-f_{FLRW})(t,x,\cdot)$, $V^0$ can be controlled by $a^{-c\sqrt{\epsilon}}$ using \eqref{eq:APmidG} and \eqref{eq:APMom}. We obtain in total:
\change{\begin{equation}\label{eq:APchar}
\left\lvert\frac{dZ}{dt}\right\rvert\lesssim {\epsilon} a^{-1-c\sigma}
\end{equation}}
Using the initial data assumption \eqref{eq:init-ass-C-vlasov} and again \eqref{eq:APMom} and \eqref{eq:APmidG} to control $v^0$ on the support of $f-f_{FLRW}$, this implies
\change{\begin{align*}
\lvert \langle v\rangle_\gamma(f-f_{FLRW})(t,x,v)\rvert\leq &\, \langle v\rangle_\gamma\lvert (f-f_{FLRW})(t,X(t),V(t))-(f-f_{FLRW})(t_0,X(t_0),V(t_0))\rvert\\
&\,+\langle v\rangle_\gamma(f-f_{FLRW})(t_0,X(t_0),V(t_0))\rvert\\
\leq&\,\langle v\rangle_\gamma\left(\chi_{\supp (f-f_{FLRW})(t,x,\cdot)}(v)\right) \lvert Z(t)-Z(t_0)\rvert+\|f-f_{FLRW}\|_{C^0_{1,\underline{\gamma}_0}(T^\ast M_{t_0})}\\
\lesssim&\,\change{\left[\int_t^{t_0}\epsilon a(s)^{-1-c\sigma}\,ds+\epsilon^2\right]}
\end{align*}
\eqref{eq:APVlasov0gamma} now follows by \eqref{eq:a-integrals}. \eqref{eq:APVlasov0} can be obtained by the same argument or from \eqref{eq:APVlasov0gamma} using \eqref{eq:APmidG}.}
\end{proof}

\begin{lemma}[Strong $C^l_{1,\G_0}$-norm estimates for the Vlasov distribution function]\label{lem:APVlasov} Let $t\in(t_{Boot},t_0]$. Then, the following holds:
\begin{align}
\label{eq:APVlasov}\|f-f_{FLRW}\|_{C^{11}_{1,\G_0}(\change{T^\ast M_t})}\lesssim \sqrt{\epsilon}a(t)^{-c\sqrt{\epsilon}}
\end{align}
\end{lemma}
Note that, for $\epsilon>0$ suitably small, \eqref{eq:APVlasov} is a bootstrap improvement for \eqref{eq:BsVlasovhor}. 
\begin{proof}
The following proof will require us to choose a spatial coordinate frame to express horizontal frame derivatives, so we must first establish suitable coordinate systems to cover the entire evolution: Since $M$ is compact, there exists some sufficiently small $r>0$ such that an open ball $B_r(x)$ of radius $r>0$ around any $x\in M$ is a coordinate neighbourhood, and there exist $x_1,\dots,x_m\in M$ such that $\{B_{\frac{r}2}(x_i)\}_{i=1,\dots,m}$ covers $M$. Furthermore, for \eqref{eq:charX}, $\lvert \frac{d}{dt}X\rvert_G\lesssim a^{-1}$ holds, which is integrable on $[0,t_0]$ (see \eqref{eq:a-integrals}), since $\lvert N+1\rvert\lesssim 1$ holds by the bootstrap assumption \eqref{eq:BsN}. \eqref{eq:charX} will describe the spatial characteristics for all systems considered in this proof. In particular, there exists a finite partition of $(t_{Boot},t_0]$ by intervals $(t_{j},t_{j-1}]$ such that $\int_{t_{j+1}}^{t_j}a(s)^{-1}\,ds<\frac{r}2$ holds. Thus, for any $y\in B_{\frac{r}2}(x_i)$, the characteristic $X$ emanating from $y$ at time $t_j$ is contained in the coordinate neighbourhood $B_r(x_i)$ on $(t_{j},t_{j-1}]$. Thus, we can prove \eqref{eq:APVlasov} for any $t\in(t_1,t_0]$ using local coordinates by having found a choice of coordinates that remains valid throughout the evolution, and since we can cover $M$ in a finite amount of such coordinate neighbourhoods. Finally, the bounds then extend to $t_1$ by continuity, and thus we can prove the statement iteratively on the entire bootstrap interval by proving the bound on $(t_{j},t_{j-1}]$ from data emanating from $M_{t_{j-1}}$.\\

Having now justified that we can essentially perform this proof in local coordinates, we move on to the proof itself: First, we commute \eqref{eq:vlasov-resc} with \change{$\B^i$ }to obtain
\change{\begin{align*}
\del_t\left(v^0_\gamma\B^i(f-f_{FLRW})\right)=&\,v_\gamma^0\B^i\del_tf+\frac{\del_tv_\gamma^0}{v_\gamma^0}\,v_\gamma^0\B^i(f-f_{FLRW})\\
=&\,\X\left(v_\gamma^0\B^i(f-f_{FLRW})\right)+m^2\frac{\dot{a}a}{(v_\gamma^0)^2}\,v_\gamma^0\B^i(f-f_{FLRW})\\
&\,+v_\gamma^0\B^i\X f_{FLRW}+(\mathcal{K}_{1,0})^i
\end{align*}}
with
\change{\begin{align*}
(\mathcal{K}_{1,0})^i(\cdot,\cdot,v)=&\,v_\gamma^0\,[\B^i,\X](f-f_{FLRW})(\cdot,\cdot,v)+\frac{\X v^0_\gamma}{v^0_\gamma}\,\B^i(f-f_{FLRW})
\end{align*}}
(see \eqref{eq:XB-comm}). Moving to the corresponding characteristic system \change{$(X,V,Z^i)$}, $X$ and $V$ again satisfy \eqref{eq:char} emanating from $(x,v)\in \change{T^\ast M_t}$, while \change{$Z^i(t)=v^0_\gamma\B^i(f-f_{FLRW})(t,X(t),V(t))$ }satisfies 
\change{\[\frac{dZ^i}{dt}=m^2\frac{\dot{a}a}{(V^0_\gamma(t))^2}\,Z^i+(\mathcal{K}_{1,0})^i(t,X(t),V(t))+V_\gamma^0(t)\,\B^i\X f_{FLRW}(t,X(t),V(t))\,.\]
The horizontal term derivative term in \eqref{eq:XB-comm} can be bounded with respect to $\gamma$ by $\epsilon^\frac14 a^{-c\sigma}$ using \eqref{eq:BsVlasovhor} and \eqref{eq:BsC}, while the remaining terms can be bounded by $a^{-1-c\sigma}\lvert Z\rvert_\gamma$ using \eqref{eq:BsC}. Further and using \eqref{eq:APmidG} to switch metrics where needed, one computes
\[\lvert \X V^0_\gamma\rvert\lesssim a^{-1-c\sqrt{\epsilon}}(\lvert \Gamma-\Gamhat\rvert_G+\lvert \nabla N\rvert_G)V^0_\gamma\,,\]
so, with \eqref{eq:BsC}, the remaining terms can be bounded by $\epsilon a^{-1-c\sigma}\lvert Z\rvert_\gamma$. In summary, we obtain}
\begin{align*}
\lvert\mathcal{K}_{1,0}(\cdot,X,V)\rvert_\gamma\lesssim&\,{\epsilon}^\frac14a^{-1-c\sigma}+\change{a^{-1-c\sigma}}\lvert Z\rvert_\gamma\,.
\end{align*}
Due to \eqref{eq:hor-deriv-Xref-C} and \eqref{eq:APmidG} and since $f_{FLRW}$ \change{has }compact momentum support with respect to $\gamma$, we also have
\begin{align*}
V_\gamma^0\lvert\B\X f_{FLRW}\rvert_{\gamma}\lesssim&\,\change{\epsilon a^{-1-c\sigma}}\,.
\end{align*}
Thus, dropping the remaining term since it has favourable sign, one altogether has
\change{\begin{equation*}
-\frac{1}{2\lvert Z\rvert_\gamma}\frac{d}{dt}\left[\lvert Z\rvert_\gamma^2\right]\lesssim a^{-1-c\sigma}\lvert Z\rvert_\gamma+\epsilon a^{-1-c\sigma}+{\epsilon}^\frac14a^{-1-c\sigma}\,.
\end{equation*}}
\change{After }integrating and applying the Gronwall lemma as well as the scale factor integral bounds from Lemma \ref{lem:scale-factor} and the initial data \change{assumption, yields}
\change{\begin{equation}\label{eq:APvertgamma}
\lvert \nabsak_{vert}(f-f_{FLRW})(t,X(t),V(t))\rvert_{\underline{\gamma}_0}\lesssim \epsilon^\frac14\,.
\end{equation}}
Repeating the same procedure for $\A_i$, we have
\begin{equation*}
\del_t\A_i(f-f_{FLRW})=\X\A_i(f-f_{FLRW})+\A_i\X f_{FLRW}+(\mathcal{K}_{1,1})_i
\end{equation*}
with
\begin{align*}
(\mathcal{K}_{1,1})_i=&\,[\A_i,\X](f-f_{FLRW})+[\del_t,\A_i](f-f_{FLRW})\,,
\end{align*}
where \eqref{eq:XA-comm} and \eqref{eq:delt-nabhor-comm} imply, using \eqref{eq:BsC} and \eqref{eq:APSigma}:
\begin{align*}
\lvert [\X,\A_{(\cdot)}](f-f_{FLRW})\rvert_\gamma\lesssim&\,a^{-1-c\sigma}\lvert \A_{(\cdot)}(f-f_{FLRW})\rvert_\gamma +\epsilon a^{-3-c\sqrt{\epsilon}}\lvert\nabsak_{vert}(f-f_{FLRW})\rvert_{\underline{\gamma}_0}\\
\lvert[\del_t,\A_{(\cdot)}](f-f_{FLRW})\rvert_\gamma\lesssim&\,\epsilon a^{-3-c\sqrt{\epsilon}}\lvert \nabsak_{vert}\left(f-f_{FLRW}\right)\rvert_{\underline{\gamma}_0}
\end{align*}
Using \eqref{eq:APvertgamma}, it follows that
\begin{equation*}
\lvert\mathcal{K}_{1,1}\rvert_\gamma\lesssim a^{-1-c\sigma}\lvert \A_{(\cdot)}(f-f_{FLRW})\rvert_\gamma+ \epsilon^\frac54 a^{-3-c\sqrt{\epsilon}}
\end{equation*}
With \eqref{eq:hor-deriv-Xref-C}, one also computes 
\[\lvert \A_{(\cdot)}\X f_{FLRW}\rvert_\gamma \lesssim \epsilon a^{-1-c\sigma}+\epsilon a^{-3-c\sqrt{\epsilon}}\,,\]
and now obtains that
\begin{equation}\label{eq:APhorgamma}
\lvert \A_{(\cdot)}(f-f_{FLRW})(t,X(t),V(t))\rvert_\gamma\lesssim \sqrt{\epsilon}a^{-c\sqrt{\epsilon}}
\end{equation}
holds. Finally, we can redo the proof of \eqref{eq:APvertgamma}, replacing \eqref{eq:BsVlasovhor} with the improved bound \eqref{eq:APhorgamma}, which improves the bound in \eqref{eq:APvertgamma} to $\sqrt{\epsilon}a^{-c\sqrt{\epsilon}}$ after updating constants. Altogether, one now concludes as in the proof of Lemma \ref{lem:APVlasov0} that
\[\|\nabsak(f-f_{FLRW})\|_{{C}^0_{1,\underline{\gamma}_0}(\change{T^\ast U_t})}\lesssim \sqrt{\epsilon}a(t)^{-c\sqrt{\epsilon}}\]
holds for suitably chosen coordinate neighbourhoods as outlined at the start of the proof, and thus the same bounds on \change{$T^\ast M_t$}. With \eqref{eq:APmidG}, this also implies
\begin{align*}
\|f-f_{FLRW}\|_{\dot{C}^1_{1,\G_0}(\change{T^\ast M_t})}\lesssim&\,\sqrt{\epsilon}a(t)^{-c\sqrt{\epsilon}}\,.
\end{align*}
We now sketch how to extend this procedure to higher orders inductively, since the line of argumentation is almost identical: Assume the improved bound has been proven up to $C^{{\change{L}}-1}_{1,\G_0}$ for ${\change{L}}\leq 11$. Commuting \eqref{eq:vlasov-resc} with \change{$(v^0_\gamma)^L\B^{i_1}\dots\B^{i_{L}}$  }implies the following transport system:
\change{\[(\del_t-\bm{X})\B^{i_1}\dots\B^{i^{{\change{L}}}}(f-f_{FLRW})=(\mathcal{K}_{{\change{L}},0})^{i_1\dots i_{\change{L}}}+(v^0_\gamma)^L\B^{i_1}\dots\B^{i_{\change{L}}}\X f_{FLRW}\]
with
\begin{align*}
(\mathcal{K}_{{\change{L}},0})^{i_1\dots i_{\change{L}}}=&\,[(v^0_\gamma)^L\B^{i_1}\dots\B^{i^{{\change{L}}}},\bm{X}](f-f_{FLRW})\\
=&\,(v^0_\gamma)^L\B^{i_1}\dots\B^{i_{{\change{L}}-1}}[\B^{i_{{\change{L}}}},\bm{X}](f-f_{FLRW})+\dots+[\B^{i_1},\X]\B^{i_2}\dots\B_{i^{{\change{L}}}}(f-f_{FLRW})\\
&\,+(\X (v^0_\gamma)^L)\B^{i_1}\dots\B^{i_L}(f-f_{FLRW}) 
\end{align*}}
Once, again the characteristics $(X,V)$ are as in \eqref{eq:char}, while one has
\[\frac{dZ^{i_1\dots i_{\change{L}}}}{dt}=(\mathcal{K}_{{\change{L}},0})^{i_1\dots i_{\change{L}}}+\B^{i_1}\dots\B^{i_{\change{L}}}\X f_{FLRW}\]
with
\[(V^0_\gamma)^{{\change{L}}}\lvert \mathcal{K}_{{\change{L}},0}\rvert_\gamma\lesssim \change{a^{-1-c\sigma}(V^0_\gamma)^{\change{L}}\lvert Z\rvert_\gamma}+\epsilon^\frac14a^{-1-c\sigma}+\change{\epsilon a^{-1-c\sigma}\|f-f_{FLRW}\|_{C_{1,\G_0}^{{\change{L}}-1}}}\]
and
\[(V^0_\gamma)^{{\change{L}}}\left\lvert\B^{\change{L}}\X f_{FLRW}\right\rvert_\gamma\lesssim \change{\epsilon a^{-1-c\sigma}}\]
This yields
\[\left(V^0_\gamma(t)\right)^{\change{L}}\left\lvert\B^{\change{L}}(f-f_{FLRW})(X(t),V(t))\right\rvert_{{\gamma}}\lesssim\epsilon^\frac14 a(t)^{-c\sqrt{\epsilon}}\]
as with ${\change{L}}=0$.\\
From here, we can repeat this procedure for $\change{(v^0_\gamma)^{L-R}}\B^{\change{L-R}}\A^{\change{R}}$ by iterating over increasing $\change{R}$ for $\change{R}=1,\dots,{\change{L}}$, where we will have proven
\[\left(V^0_\gamma(t)\right)^{\change{L-\tilde{R}}}\left\lvert\B^{\change{L-\tilde{R}}}\A^{\change{\tilde{R}}}(f-f_{FLRW})(t,X(t),V(t))\right\rvert_{{\gamma}}\lesssim{\epsilon}^\frac14a(t)^{-c\sqrt{\epsilon}}\]
for $0\leq\change{\tilde{R}}<\change{R}$ in the previous steps. Additionally, the induction hypothesis allows us to use \eqref{eq:APVlasov} up to order ${\change{L}}-1$. We use all of these bounds to control Vlasov terms of lower order and of order ${\change{L}}$ with less horizontal derivaitves in $\mathcal{K}_{{\change{L}},\change{R}}$, along with\eqref{eq:APSigma},\eqref{eq:APmidSigma} and \eqref{eq:BsC} to control shear, metric and lapse as well as \eqref{eq:APMom} and \eqref{eq:APmidG} to control superfluous momentum terms. Additionally, for $\change{R<L}$, we use the bootstrap assumption \eqref{eq:BsVlasovhor} to bound the term
\change{\begin{equation}\label{eq:vlasov-AP-problem}
-a^{-1}(N+1)\left(\frac{1}{v^0}\I^j_{i_K}-\frac{v^\flat_{i_K}v^j}{(v^0)^3}\right)\B^{i_1}\dots\B^{i_{K-1}}\B^{i_{K+1}}\dots\B^{i_{\change{L-R}}}\A_{i_{\change{L-R+1}}}\dots\A_{i_{\change{L}}}\A_j(f-f_{FLRW})\,,
\end{equation}}
which occurs when applying the product rule to the horizontal derivative term in $\X$, by ${\epsilon}^\frac14a^{-1-c\sigma}$. In total \change{and now considering ${Z^{i_1\dots\i_{L-R}}}_{i_{L-R+1}\dots i_L}$ in analogy to above}, this then implies
\change{\begin{equation}\label{eq:K-ell-r}
\lvert \mathcal{K}_{\change{L-R,R}}\rvert_\gamma\lesssim\begin{cases} a^{-1-c\sigma}\lvert Z\rvert_\gamma+\epsilon^\frac14a^{-1-c\sigma}+\epsilon^\frac54 a^{-3-c\sqrt{\epsilon}} & \change{R<L}\\
a^{-1-c\sigma}\lvert Z\rvert_\gamma+\epsilon^\frac54 a^{-3-c\sqrt{\epsilon}} & \change{R=L} \end{cases}\,.
\end{equation}}
We again use Lemma \ref{lem:hor-deriv-ref} to control the reference terms. In total, \change{at order $\change{L}$}, we use strong low order bounds on 
\[\|\Ric[G]-2\kappa G\|_{C_G^{{\change{L}}-1}(M_t)},\ \|\Gamma-\Gamhat\|_{C^{{\change{L}}-1}_G(M_t)}, \text{ and } \|\Sigma\|_{C^{{\change{L}}}_G(M_t)}\,\]
(along with using the bootstrap assumption \eqref{eq:BsN} for $\|N\|_{C^{{\change{L}}+1}_G(M_t)}$). Thus, this procedure can only be performed until ${\change{L}}=11$. Altogether, one obtains
\begin{align*}
\change{\sum_{R=0}^{L}}(V^0_\gamma(t))^{\change{L-R}}\left\lvert \B^{\change{L-R}}\A^{\change{R}}(f-f_{FLRW})(t,X(t),V(t))\right\rvert_{\gamma}\lesssim&\,\epsilon^\frac14 a(t)^{-c\sqrt{\epsilon}}\,,
\end{align*}
and, since \eqref{eq:vlasov-AP-problem} does not occur for \change{$R=L$}, also
\[\left\lvert \A^{\change{L}}(f-f_{FLRW})(t,X(t),V(t))\right\rvert_\gamma\lesssim\sqrt{\epsilon}a(t)^{-c\sqrt{\epsilon}}\]
just like for $\change{L=R}=1$. Now, we iterate over decreasing \change{$R$}: Assume that one has proven, for any $\change{\tilde{R}>R}\geq 0$,
\[\left(V^0_\gamma(t)\right)^{\change{L-\tilde{R}}}\left\lvert\B^{{\change{L}}-\change{\tilde{R}}}\A^{\change{\tilde{R}}}(f-f_{FLRW})(X(t),V(t))\right\rvert_{{\gamma}}\lesssim\sqrt{\epsilon}a(t)^{-c\sqrt{\epsilon}}\,.\]
Then, one can bound \eqref{eq:vlasov-AP-problem} by $\sqrt{\epsilon}a^{-1-c\sqrt{\epsilon}}$ instead of applying \eqref{eq:BsVlasovhor}, improving the error term bound \eqref{eq:K-ell-r} as before and thus yielding the same bound as above for $\change{\tilde{R}=R}$. In total, we obtain
\begin{equation*}
\sum_{r=0}^{{\change{L}}}(V^0_\gamma(t))^{{\change{L}-R}}\lvert \B^{{\change{L}-R}}\A^{\change{R}}(f-f_{FLRW})(X(t),V(t))\rvert_{\gamma}\lesssim\sqrt{\epsilon}a(t)^{-c\sqrt{\epsilon}}\,.
\end{equation*}
Finally, note that \eqref{eq:APmidRic} and \eqref{eq:APmidG} \change{imply}
\[\lvert\nabsak_{vert}^{\change{L-R}}\nabsak_{hor}^{\change{R}}(f-f_{FLRW})\rvert_{\G_0}\lesssim \lvert\B^{\change{L-R}}\A^{\change{R}}(f-f_{FLRW})\rvert_{\G_0}+a^{-c\sqrt{\epsilon}}\lvert \nabsak^{{\change{L}}-1}(f-f_{FLRW})\rvert_{\G_0}\lesssim \sqrt{\epsilon}a^{-c\sqrt{\epsilon}},\]
and subsequently \eqref{eq:APVlasov} at order ${\change{L}}$. This closes the induction over ${\change{L}}\leq 11$, proving \eqref{eq:APVlasov} 
\end{proof}

%

Additionally, the bounds on the Vlasov distribution function admit bounds on its energy-momentum tensor:

\begin{corollary}[Strong low order bounds on Vlasov matter quantities]\label{lem:APdensity}
The following bounds hold:
\begin{equation}\label{eq:APdensity}
\|\rho^{Vl}-\rho^{Vl}_{FLRW}\|_{C^{11}_G(M_t)}+\|\mathfrak{p}^{Vl}-\mathfrak{p}^{Vl}_{FLRW}\|_{C^{11}_G(M_t)}+\|\change{\j^{Vl}}\|_{C^{11}_G(M_t)}+\|S^{Vl,\parallel}\|_{C^{11}_G(M_t)}\lesssim\sqrt{\epsilon}a(t)^{-c\sqrt{\epsilon}}
\end{equation}
\end{corollary}
\begin{proof}
Starting with $\rho^{Vl}-\rho^{Vl}_{FLRW}$, we compute the following:
\begin{align*}
\left(\rho^{Vl}-{\rho}_{FLRW}^{Vl}\right)(t,x)
=&\,\int_{\R^3} \left[f(t,x,v)-\frac{v_\gamma^0}{v^0}\change{\frac{\mu_G}{\mu_\gamma}}f_{FLRW}(t,x,v)\right]v^0\change{\mu_G^{-1}}dv
\end{align*}
Note that one has 
\[\left\lvert v^0-v^0_\gamma\right\rvert=\frac{(v^0)^2-(v^0_\gamma)^2}{v^0+v^0_\gamma}\leq\frac{\change{\left\lvert (G^{-1}-\gamma^{-1})^{ij}v_iv_j\right\rvert}}{v^0+v^0_\gamma}\lesssim \change{\lvert G^{-1}-\gamma^{-1}\rvert_G}\cdot v^0.\]
Furthermore, by expanding the determinant and using \eqref{eq:APmuG} and $\lvert \mu_\gamma\rvert\lesssim 1$, one obtains with \eqref{eq:APmidG} that
\[\change{\lvert \mu_G^{-1}-\mu_\gamma^{-1}\rvert\lesssim \lvert \det G^{-1}-\det\gamma^{-1}\rvert\lesssim a^{-c\sqrt{\epsilon}}\lvert G^{-1}-\gamma^{-1}\rvert_G}\,.\]
Since $f_{FLRW}$ is also uniformly bounded, one now obtains:
\begin{align}\label{eq:APdensity-step}
\left\lvert\rho^{Vl}-{\rho}_{FLRW}^{Vl}\right\rvert\lesssim &\,\int_{\supp (f-f_{FLRW})}\left[\lvert (f-f_{FLRW})(\cdot,\cdot,v)\rvert+a^{-c\sqrt{\epsilon}}\change{\lvert G^{-1}-\gamma^{-1}\rvert_G}\right] v^0\,\change{\mu_G^{-1}}dv
\end{align}
By \eqref{eq:APMom}, the integral over any polynomial of $v^0$ is bounded by $a^{-c\sqrt{\epsilon}}$, and thus the bound follows at order $0$ by applying the uniform estimates \eqref{eq:APVlasov0} and \eqref{eq:APmidG}. For higher orders, note that $\nabla{\rho}_{FLRW}^{Vl}=0$. Thus, for any of the matter quantities, we pull the covariant derivatives $\nabla^L$ past the integral according to \eqref{eq:nabla-intTM} and can directly bound the integrand uniformly using \eqref{eq:APVlasov} and \eqref{eq:APmidG}, proving \eqref{eq:APdensity} for the energy density. The bounds for the remaining quantities are proven identically.
\end{proof}

Finally, we collect the following lemma that will be needed to bound the evolution of weights in our energy approach:

\begin{lemma}[Evolution bound on momentum variables]
\begin{equation}\label{eq:deltv0}
-\frac{\del_tv^0}{v^0}\lesssim \epsilon a^{-3}\,,\quad -\frac{\del_t\langle v\rangle_G}{\langle v\rangle_G}\lesssim \epsilon a^{-3}
\end{equation}
\end{lemma}
\begin{proof}
\change{One }computes the following:
\change{\begin{equation*}
\del_tv^0=\frac1{2v^0}\left[2m^2a\dot{a}+(\del_tG^{-1})^{ij}v_iv_j\right]=(N+1)a^{-3}\Sigma^{\sharp ij}\frac{v_iv_j}{v^0}-N\frac{\dot{a}}a\frac{\lvert v\rvert_G^2}{v^0}+\frac{\dot{a}}a\cdot \frac{m^2a^2}{v^0}
\end{equation*}}
The statement then follows from \eqref{eq:APSigma} and \eqref{eq:BsC}, where the final summand has favourable sign and is thus dropped. The estimate for $\langle v\rangle_G$ is obtained identically.
\end{proof}

\subsection{Semi-norm equivalences and commutators for the Vlasov distribution}

In this section, we collect estimates on covariant commutator errors for functions on the mass shell, that in particular will imply that the Vlasov energies \eqref{def:en-vlasov-vert}-\eqref{def:en-vlasov-hor} are coercive up to error terms dependent on lower order Vlasov energies as well as metric norms and energies. This will be sufficient to deduce control of the Vlasov matter quantities in $\mathcal{H}$ and $\mathcal{C}$ from improved energy bounds in Section \ref{sec:improve}.

\begin{lemma}[Commutator formulas for the Sasaki metric]\label{lem:sasaki-com-gen} Let $\xi$ be a sufficiently regular function on $\displaystyle\cup_{t\in(t_{Boot},t_0], x\in M}P_{(t,x)}$ and let $M_1,M_2\in \N$. Then, understanding that terms do not occur if they would involve negative orders of derivatives, the following schematic expressions hold:
\begin{subequations}
\begin{align*}
\numberthis\label{eq:com-sasaki-gen-hor}[\nabsak_{vert}^{M_1}\nabsak_{hor}^{M_2},\nabsak_{hor}]\xi=&\,v\ast\Ric[G]\ast\nabsak_{vert}^{M_1-1}\nabsak_{hor}^{M_2}\xi+v\ast\Ric[G]\ast\nabsak_{vert}^{M_1+1}\nabsak_{hor}^{M_2-1}\xi\\
&\,+\underbrace{v\ast\nabla^{M_2}(\Ric[G]-2\kappa G)\ast\nabsak_{vert}^{M_1}\xi}_{\text{if }M_2>0}+\Xi_{M_1,M_2,hor}\xi\\
\numberthis\label{eq:com-sasaki-gen-t}[\del_t,\nabsak_{vert}^{M_1}\nabsak_{hor}^{M_2}]\xi=&\,v\ast\nabla^{M_2-2}\del_t\Ric[G]\ast\nabsak_{vert}^{M_1+1}\xi
+v\ast\del_t\Gamma[G]\ast\nabsak_{vert}^{M_1+1}\nabsak_{hor}^{M_2-1}\xi\\
&\,+\Xi_{M_1,M_2,t}\xi
\end{align*}
\end{subequations}
If, on the momentum support of $\xi(t,\cdot,\cdot)$,  $v^0$ is uniformly bounded by $a(t)^{-c\sqrt{\epsilon}}$ up to constant, the lower order error terms can be bounded as follows:
\begin{subequations}\label{eq:com-sasaki-gen-err}
\begin{align*}
\numberthis\label{eq:com-sasaki-gen-hor-err}\|\langle v\rangle_G&(v^0)^{M_1}\Xi_{M_1,M_2,hor}\xi\|_{L^2_{1,\G}(\change{T^\ast M_t})}\lesssim a(t)^{-c\sqrt{\epsilon}}\|\xi\|_{H_{1,\G_0}^{M_1+M_2-1}(\change{T^\ast M_t})}\\
&\,+a(t)^{-c\sqrt{\epsilon}}\|\Ric[G]-2\kappa G\|_{H^{M_2-1}_G(M_t)}\|\xi\|_{C^{M_1+1}_{1,\G_0}(\change{T^\ast M_t})}\\
\numberthis\label{eq:com-sasaki-gen-t-err}\|\langle v\rangle_G&(v^0)^{M_1}\Xi_{M_1,M_2,t}\xi\|_{L^2_{1,\G}(\change{T^\ast M_t})}\lesssim\epsilon a(t)^{-3-c\sqrt{\epsilon}}\|\xi\|_{H^{M_1+M_2-1}_{1,\G_0}(\change{T^\ast M_t})}\\
&\,+a(t)^{-c\sqrt{\epsilon}}\left(\|\del_t\Ric[G]\|_{H^{M_2-2}_G(M_t)}+\|\del_t\Gamma[G]\|_{H^{M_2-2}_G(M_t)}\right)\|\xi\|_{C^{M_1+1}_{1,\G_0}(\change{T^\ast M_t})}
\end{align*}
\end{subequations}
\end{lemma}
\noindent By Lemma \ref{lem:APMom}, \eqref{eq:com-sasaki-gen-err} applies for $\xi=f$.
\begin{proof}
These formulas follow from straightforward computations using the commutator formulas in Lemma \ref{lem:commutators-zero} and the connection coefficient expressions from \eqref{eq:conn-coeff}. For the error bounds, we recall the bounds \eqref{eq:APMom} on the momentum support, and estimate lower order curvature and metric errors with \eqref{eq:APmidRic} and \eqref{eq:APmidG} to obtain \eqref{eq:com-sasaki-gen-hor-err}. For \eqref{eq:com-sasaki-gen-t}, recall \eqref{eq:REEqRic} and \eqref{eq:REEqChr} and use \eqref{eq:APmidSigma} as well as the bootstrap assumption to bound lower order terms in the rescaled shear and lapse respectively.
\end{proof}

\begin{lemma}[Rearrangement of derivatives in intermediate Vlasov energies]\label{lem:vlasov-rearrange}
\begin{subequations}
Let $0<K<L$ and let $\mathcal{D}_{L,K}$ denote a combination of $L-K$ vertical and $K$ horizontal covariant derivatives with respect to the Sasaki metric $\G$. Then, the following holds:
\begin{align*}
\numberthis\label{eq:vlasov-rearrange-with-ref}&\left\lvert\|\mathcal{D}_{L,K}(f-f_{FLRW})\|^2_{L^2_{1,\G_0}(\change{T^\ast M_t})}-\|\nabsak_{vert}^{L-K}\nabsak_{hor}^K(f-f_{FLRW})\|^2_{L^2_{1,\G_0}(\change{T^\ast M_t})}\right\rvert\\
\lesssim&\,a(t)^{-c\sqrt{\epsilon}}\left(\E^{(\leq L-1)}_{1,\leq K}(f,t)+\epsilon{\E^{(\leq K-1)}(\Ric,t)}+\|\Gamma-\Gamhat\|_{H^{K-1}_G(M_t)}^2+\|\change{G^{\pm 1}-\gamma^{\pm 1}}\|_{H^{K-1}_G(M_t)}^2\right)
\end{align*}
Similarly, one has for $1\leq K<L$,
\begin{align*}
\numberthis\label{eq:vlasov-rearrange-no-ref-K}\left\lvert\|\mathcal{D}_{L,K}f\|_{L^2_{1,\G_0}(\change{T^\ast M_t})}^2-\E^{(L)}_{L,K}(f,t)\right\rvert\lesssim&\,a(t)^{-c\sqrt{\epsilon}}\E^{(\leq L-1)}_{1,\leq K}(f,t)+\epsilon a(t)^{-c\sqrt{\epsilon}}\E^{(\leq K-1)}(\Ric,t)\\
&\,+a(t)^{-c\sqrt{\epsilon}}\left(\|\Gamma-\Gamhat\|_{L^2_G(M_t)}^2+\|\change{G^{\pm 1}-\gamma^{\pm 1}}\|_{H^1_G(M_t)}^2\right)\,.
\end{align*}
If one additionally \change{takes }$K\leq 10$ in any of the above bounds, all terms but the Vlasov energy quantities can be dropped. 
\end{subequations}
\end{lemma}
\begin{proof}
Reviewing the commutator formulas from Lemma \ref{lem:commutators-zero}, one schematically has for $K>0$:
\begin{align*}
\mathcal{D}_{L,K}\xi-\nabsak_{vert}^{L-K}\nabsak_{hor}^K\xi=&\,\sum_{I=0}^{K-1}\left[v\ast \nabla^I\Riem[G]\ast\nabsak^{L-I-1}\xi+\nabla^{I}\Riem[G]\ast\nabsak^{L-I}\xi\right]\\
&\,+\langle \text{lower order nonlinear terms containing }v,\Riem[G],\xi\rangle\,
\end{align*}
where we recall $\Riem[G]=\Ric[G]\ast G$. Thus, the product rule and the a priori estimates \eqref{eq:APmidRic}, \eqref{eq:APMom} and \eqref{eq:APVlasov} yield the following:
\begin{align*}
&\,\left\lvert\|\mathcal{D}_{L,K}(f-f_{FLRW})\|^2_{L^2_{1,\G_0}(\change{T^\ast M})}-\|\nabsak_{vert}^{L-K}\nabsak_{hor}^K(f-f_{FLRW})\|^2_{L^2_{1,\G_0}(\change{T^\ast M})}\right\rvert\\
\lesssim&\,a^{-c\sqrt{\epsilon}}\|f-f_{FLRW}\|^2_{H^{L-1}_{1,\G_0}(\change{T^\ast M})}+\epsilon a^{-c\sqrt{\epsilon}}\E^{(\leq K-1)}(\Ric,\cdot)
\end{align*}
Furthermore, after applying the triangle inequality to $\|\nabsak_{vert}^{L-K}\nabsak_{hor}^K(f-f_{FLRW})\|_{L^2_{1,\G_0}(\change{T^\ast M})}$, we can estimate the reference terms with \eqref{eq:hor-deriv-ref-mixed}. \eqref{eq:vlasov-rearrange-with-ref} then follows by iterating over iteratively replacing mixed terms on the right hand side until only Vlasov energy terms remain.\\
For \eqref{eq:vlasov-rearrange-no-ref-K}, we largely proceed as above replacing $f-f_{FLRW}$ with $f$. However, whenever we encounter a term with only one horizontal derivative by commuting, we replace $f$ with $(f-f_{FLRW})+f_{FLRW}$. For the former, we simply use \eqref{eq:vlasov-rearrange-with-ref}. Regarding the latter, the analogous bounds hold as for $\mathcal{D}_{L,1}f_{FLRW}$ as for $\nabsak_{vert}^{L-1}\nabsak_{hor}f_{FLRW}$ in \eqref{eq:hor-deriv-ref-mixed}, except for incurring metric terms up to second order since taking vertical derivatives causes metric terms to occur at order zero that need to be differentiated: For example, \change{one has}
\change{\[\nabsak_{hor}\nabsak_{vert}f_{FLRW}=\nabsak_{hor}(\mathcal{F}^\prime(\lvert v\rvert_\gamma^2)\ast v\ast\gamma)=v\ast\left(\nabsak_{hor}\mathcal{F}^\prime(\lvert v\rvert_\gamma^2)\ast \gamma^{-1}+\mathcal{F}^\prime(\lvert v\rvert_\gamma^2)\ast\nabla(G^{-1}-\gamma^{-1})\right)\,.\]}
Note that, since we only need to apply \eqref{eq:hor-deriv-ref-mixed} when only one horizontal derivative occurs, we do not incur higher order metric errors in \eqref{eq:vlasov-rearrange-no-ref-K}. 
\end{proof}

Since Lemma \ref{lem:vlasov-rearrange} allows us to exchange arrange covariant derivatives in arbitrary order up to curvature errors, this will imply that Vlasov energies are coercive up to lower order metric errors, along with bounds on the reference distribution function to control the purely vertical case. We will use this in Corollary \ref{cor:norm-imp-vlasov} to improve bounds on Sobolev norms containing $f-f_{FLRW}$. To then improve the bootstrap assumptions, we need to show that these sufficiently control the energy momentum tensor, which is ensured by the following lemma:

\begin{lemma}[Energy bounds for Vlasov energy-momentum tensor components]\label{lem:density-control}
For $\,L\in\Z_+$ and $t\in(t_{Boot},t_0]$, the following bounds hold:
\begin{subequations}
\begin{align}
\|\rho^{Vl}-{\rho}_{FLRW}^{Vl}\|_{L^2_G(M_t)}+\|\mathfrak{p}^{Vl}-{\mathfrak{p}}_{FLRW}^{Vl}\|_{L^2_G(M_t)}\lesssim&\,a(t)^{-c\sqrt{\epsilon}}\left(\sqrt{\E^{(0)}_{1,0}(f,t)}+\change{\|G^{-1}-\gamma^{-1}\|_{L^2_G(M_t)}}\right)\label{eq:density-control-L2}\\
\|\rho^{Vl}-{\rho}_{FLRW}^{Vl}\|_{\dot{H}^L_G(M_t)}+\|\mathfrak{p}^{Vl}-{\mathfrak{p}}_{FLRW}^{Vl}\|_{\dot{H}^L_G(M_t)}\lesssim&\,a(t)^{-c\sqrt{\epsilon}}\sqrt{\E^{(L)}_{1,L}(f,t)}\label{eq:density-control}\\
\|\j^{Vl}\|_{\dot{H}^L_G(M_t)}\lesssim&\, a(t)^{-c\sqrt{\epsilon}}\sqrt{\E^{(L)}_{1,L}(f,t)}\label{eq:flux-control}\\
\|(S^{Vl})^{\parallel}\|_{\dot{H}^L_G(M_t)}\lesssim&\,a(t)^{-c\sqrt{\epsilon}}\sqrt{\E^{(L)}_{1,L}(f,t)}\label{eq:stress-control}
\end{align}
\end{subequations}
\end{lemma}
\begin{proof}
Note that we can apply Jensen's inequality on continuous functions $\xi$ on $P$ that have compact support after renormalising, i.e., for $\lvert \supp\xi\rvert_G=\sup_{v\in \supp\xi(t,x,\cdot)}\lvert v\rvert_G$,
\begin{align*}
\int_M\left(\int_{\R^3}v^0\xi\change{\mu_G^{-1}}dv\right)^2\vol{G}=&\int_M\change{\mu_G^{-2}}\lvert \supp\xi\rvert_G^2\left(\lvert \supp\xi\rvert_G^{-1}\int_{\supp\xi}v^0\xi\,dv\right)^2\,\vol{G}\\
\leq&\,\int_{\change{T^\ast M}}\change{\mu_G^{-1}}\lvert\supp\xi\rvert_G (v^0)^2\xi^2\,\vol{\G}\,,
\end{align*}
and that
\[\change{\mu_G^{-1}}\lvert\supp (f-f_{FLRW})\rvert_G\lesssim (1+\P+\|G-\gamma\|_{C^0_G(M)})^3\lesssim a^{-c\sqrt{\epsilon}}\]
since $\mu_G$ is uniformly bounded due to \eqref{eq:APmuG} and using \eqref{eq:APMom} and \eqref{eq:APmidG}. Thus, recalling \eqref{eq:APdensity-step}, we compute for the energy density:
\begin{align*}
\|\rho^{Vl}-{\rho}_{FLRW}^{Vl}\|^2_{L^2_G(M_t)}\lesssim&\,\int_{M_t}\left(\int_{\supp (f-f_{FLRW})(t,x,\cdot)}\left[\lvert f-f_{FLRW}\rvert+\change{\lvert G^{-1}-\gamma^{-1}\rvert_G}\right] v^0\,\change{\mu_G^{-1}}dv\right)^2\,\vol{G}\\
\lesssim&\,\change{a(t)}^{-c\sqrt{\epsilon}}\int_{\change{T^\ast M_t}}(v^0)^2\left[\lvert f-f_{FLRW}\rvert^2+\change{\lvert G^{-1}-\gamma^{-1}\rvert_G}^2\right]\vol{\G}\\
\lesssim&\,\change{a(t)}^{-c\sqrt{\epsilon}}\left(\E^{(0)}_{1,0}(f,\change{t})+\change{a(t)}^{-c\sqrt{\epsilon}}\change{\|G^{-1}-\gamma^{-1}\|^2_{L^2_G(\change{M_t})}}\right)
\end{align*}
The computations for the pressure are fully analogous, proving \eqref{eq:density-control-L2} after updating constants. By \eqref{eq:j-Vl} and \eqref{eq:S-Vl-parallel}, we can express $\j^{Vl}$ and $(S^{Vl,\parallel})^i_{\ j}$ as weighted integrals over $f-f_{FLRW}$, avoiding metric error terms. Thus, \eqref{eq:flux-control} and \eqref{eq:stress-control} at $L=0$ follow directly by applying the Jensen inequality as above. Higher orders are proven fully analogously, observing as in the proof of Lemma \ref{lem:APdensity} that $\nabla\rho^{Vl}_{FLRW}=\nabla\mathfrak{p}^{Vl}_{FLRW}=0$, as well as the Vlasov energy definition \eqref{def:en-vlasov-hor} -- this allows us directly apply the Jensen inequality to all matter variables without having to expand the integrand.
\end{proof}

\section{Estimates for spacetime and scalar field variables}\label{sec:energy-new}

\change{All of the estimates in this and the following section hold for $t\in(t_{Boot},t_0]$. Furthermore, we let $\omega\gg \sigma>0$ be a constant that is suitably small: For example, if $\sigma$ is chosen small enough that $c\epsilon^\frac18<c\sigma<\frac1{1000}$ holds throughout the proof, $\omega=\frac1{100}$ is an admissible choice, but this choice is far from optimal. In general, $\omega$ only needs to be small compared to the order of derivatives used in our estimates, but large compared to $\sigma$.}

\subsection{Lapse energy estimates}\label{sec:energy-lapse}

In this section, we collect estimates for the lapse that follow from \eqref{eq:REEqLapse1}-\eqref{eq:REEqLapse2}. In particular, in contrast to the estimates in \cite[Section 5]{FU23}, we will also prove estimates for odd order lapse energies, since these will be necessary to control high order lapse terms. 

\begin{definition}[Elliptic operators]\label{def:ell-op} 
\begin{subequations}
\begin{align*}
\L \zeta=&\,a^4\Lap \zeta-h\cdot \zeta,& h=&-3\kappa a^4+12\pi C^2+12\pi a^2\left({\rho}_{FLRW}^{Vl}+{\mathfrak{p}}_{FLRW}^{Vl}\right)+H\numberthis\\
&&H=&\,\langle\Sigma,\Sigma\rangle_G+8\pi\Psi^2+16\pi C\Psi+4\pi a^2\left(\rho^{Vl}-{\rho}_{FLRW}^{Vl}+3\mathfrak{p}^{Vl}-3{\mathfrak{p}}_{FLRW}^{Vl}\right)\\
\Ltilde \zeta=&\,a^4\Lap\zeta-\tilde{h}\cdot\zeta,&\quad \tilde{h}=&-3\kappa a^4+12\pi C^2+12\pi a^2({\mathfrak{p}}_{FLRW}^{Vl}+{\rho}_{FLRW}^{Vl})+\tilde{H}\numberthis\\
&&\tilde{H}=&\,a^4\left[R[G]+\frac23-8\pi\lvert\nabla\phi\rvert^2_G\right]+12\pi a^2\left[(\mathfrak{p}^{Vl}-{\mathfrak{p}}_{FLRW}^{Vl})-(\rho^{Vl}-{\rho}_{FLRW}^{Vl})\right]
\end{align*}
\end{subequations}
\end{definition}
\noindent With these operators, the rescaled lapse equations \eqref{eq:REEqLapse1}-\eqref{eq:REEqLapse2} then read
\begin{equation}\label{eq:lapse-ell}
\L N=\,H \text{ and } \Ltilde N=\,\tilde{H}\,,
\end{equation}
and admit the following estimates:


\begin{lemma}[Elliptic estimates for $\L$ and $\Ltilde$]\label{lem:ell-lapse}
For scalar functions $\zeta,Z$ on $M_t$ with $\L \zeta=Z$ or $\Ltilde \zeta=Z$, one has
\begin{equation}\label{eq:ell-L-est}
a^4\|\Lap \zeta\|_{L^2_G(M)}+a^2\|\nabla \zeta\|_{L^2_G(M)}+\|\zeta\|_{L^2_G(M)}\lesssim \|Z\|_{L^2_G(M)}
\end{equation}
\end{lemma}
\begin{proof}
The only difference in the proof in the pure scalar-field case is that Vlasov terms occur in $h$ and $\tilde{h}$. However, we observe that ${\rho}_{FLRW}^{Vl}+{\mathfrak{p}}_{FLRW}^{Vl}$ is nonnegative and only depends on $t$, while one has $\lvert H\rvert \lesssim \epsilon$ by the \eqref{eq:APSigma}, \eqref{eq:APPsi} and \eqref{eq:BsC}, as well as $\lvert \tilde{H}\rvert\lesssim \epsilon a^{2-c\sigma}$ by \eqref{eq:BsC}. Hence, one has
\begin{equation}\label{eq:ell-est-cond}
h,\tilde{h}\geq 6\pi C^2\text{ and }\lvert\nabla h\rvert_G,\lvert\nabla\tilde{h}\rvert_G\lesssim \epsilon
\end{equation}
 for small enough $\epsilon>0$, which were the only properties needed to prove the above lemma in the scalar field case. 
\end{proof}

\begin{lemma}[Odd order elliptic estimates for $\L$ and $\Ltilde$]\label{lem:ell-lapse-odd}
Consider scalar functions $\zeta,Z$ with $\L\zeta=Z$ or $\Ltilde\zeta=Z$. Then, one has
\begin{equation}\label{eq:ell-lapse-odd}
a^4\|\nabla\Lap\zeta\|_{L^2_G(M)}+a^2\|\Lap\zeta\|_{L^2_G(M)}+\|\nabla\zeta\|_{L^2_G(M)}^2\lesssim \|Z\|_{\dot{H}^1_G(M)}+\sqrt{\epsilon}a^{-c\sqrt{\epsilon}}\|Z\|_{L^2_G(M)}
\end{equation}
\end{lemma}
\begin{proof} We only prove the statement for $\L$, the proof is identical for $\Ltilde$. First, note that
\[\nabla Z=\nabla(\L\zeta)=a^4\nabla\Lap\zeta-h\nabla\zeta+\zeta\nabla h.\]
Contracting $\nabla_i\L\zeta$ with $-\nabla^{\sharp i}\zeta$ and integrating, we thus obtain
\begin{align*}
\int_Ma^4\lvert\Lap\zeta\rvert_G^2+h\lvert\nabla\zeta\rvert^2-\zeta\langle\nabla h,\nabla\zeta\rangle_G\,\vol{G}=-\int_M\langle\nabla Z,\nabla\zeta\rangle_G\,\vol{G}\,.
\end{align*}
Using integration by parts to rewrite the third term on the left, applying \eqref{eq:ell-est-cond} and rearranging then yields
\begin{align*}
\int_M a^4\lvert \nabla^2\zeta\rvert_G^2+\lvert\nabla\zeta\rvert_G^2\,\vol{G}\lesssim&\,\int_M \lvert\nabla Z\rvert_G^2+\|\Lap h\|_{L^\infty_G}\lvert\zeta \rvert_G^2\,\vol{G}\,.
\end{align*}
Using the a priori estimates \eqref{eq:APmidSigma} and \eqref{eq:APmidPsi} as well as \eqref{eq:BsC} for the Vlasov terms, along with Lemma \ref{lem:ell-lapse}, this then becomes
\begin{align*}
\int_M a^4\lvert \nabla^2\zeta\rvert_G+\lvert\nabla\zeta\rvert_G^2\,\vol{G}\lesssim&\,\|Z\|_{\dot{H}^1}^2+\sqrt{\epsilon}a^{-c\sqrt{\epsilon}}\|Z\|_{L^2_G}^2\,.
\end{align*}
Now contracting $\nabla_i\L\zeta$ with $a^4\nabla^{\sharp i}\Lap\zeta$ and integrating and using integration by parts, one has
\begin{align*}
&\,\int_Ma^8\lvert\nabla\Lap\zeta\rvert_G^2+ha^4\lvert\Lap\zeta\rvert^2+a^4\langle \nabla h,\nabla\zeta\rangle\Lap\zeta+a^4\zeta\langle\nabla h,\nabla\Lap\zeta\rangle_G\,\vol{G}\\
=&\,\int_Ma^8\lvert\nabla\Lap\zeta\rvert_G^2+ha^4\lvert\Lap\zeta\rvert^2-a^4\Lap h\cdot \zeta\Lap\zeta\,\vol{G}\\
=&\,\int_Ma^4\langle\nabla Z,\nabla\Lap\zeta\rangle_G\,\vol{G}\,.
\end{align*}
Thus, by \eqref{eq:ell-est-cond}, the following holds for some $K>0$:
\begin{align*}
&\,\int_Ma^8\lvert\nabla\Lap\zeta\rvert_G^2+6\pi C^2a^4\lvert\Lap\zeta\rvert^2\,\vol{G}\\
\leq&\,\frac12\|Z\|_{\dot{H}^1_G(M)}^2+\int_M\frac12a^8\lvert\nabla\Lap\zeta\rvert_G^2+ K\epsilon\left(\lvert\zeta\rvert^2+a^4\lvert\Lap\zeta\rvert^2\right)\,\vol{G}
\end{align*}
The statement then follows for small enough $\epsilon$ by rearranging and combining both bounds.
\end{proof}

\begin{lemma}[Lapse energy estimate]\label{lem:lapse-en-est}
The following estimates hold:
\begin{subequations}
\begin{align*}\numberthis\label{eq:lapse-en-est0}
&\,a^8\E^{(2)}(N,\cdot)+a^4\E^{(1)}(N,\cdot)+\E^{(0)}(N,\cdot)\\
\lesssim&\,\epsilon^2\E^{(0)}(\Sigma,\cdot)+\E^{(0)}(\phi,\cdot)+a^4\E^{(0)}_{1,0}(f,\cdot)+a^{4-c\sqrt{\epsilon}}\change{\|G^{-1}-\gamma^{-1}\|_{L^2_G(M)}^2}\,,\\
\numberthis\label{eq:lapse-en-est-tilde0}
&\,a^8\E^{(2)}(N,\cdot)+a^4\E^{(1)}(N,\cdot)+\E^{(0)}(N,\cdot)\\
\lesssim&\,a^8\E^{(0)}(\Ric,\cdot)+{\epsilon}a^{8-c\sqrt{\epsilon}}\|\nabla\phi\|_{L^2_G(M)}^2+a^{4-c\sqrt{\epsilon}}\E^{(0)}_{1,0}(f,\cdot)+a^{4-c\sigma}\change{\|G^{-1}-\gamma^{-1}\|_{L^2_G(M)}^2}
\end{align*}
For any $L\in\{1,\dots,18\}$, one additionally has:
\begin{align*}\numberthis\label{eq:lapse-en-est}
&\,a^8\E^{(L+2)}(N,\cdot)+a^4\E^{(L+1)}(N,\cdot)+\E^{(L)}(N,\cdot)\\
\lesssim&\,\epsilon^2\E^{(L)}(\Sigma,\cdot)+\E^{(L)}(\phi,\cdot)+\epsilon^2a^{-c\sqrt{\epsilon}}\left[\E^{(\leq \max\{0,L-2\})}(\Sigma,\cdot)+\E^{(\leq \max\{0,L-2\})}(\phi,\cdot)\right]\\
&\,+\underbrace{\left(\epsilon^4a^{-c\sqrt{\epsilon}}+\epsilon^2a^{2-c\sigma}\right)\E^{(\leq L-3)}(\Ric,\cdot)}_{\text{not present for }L\leq 2}+a^{4-c\sqrt{\epsilon}}\E^{(L)}_{1,L}(f,\cdot)\\
\numberthis\label{eq:laspse-en-est-tilde}
&\,a^8\E^{(L+2)}(N,\cdot)+a^4\E^{(L+1)}(N,\cdot)+\E^{(L)}(N,\cdot)\\
\lesssim&\,a^8\E^{(\leq L)}(\Ric,\cdot)+{\epsilon}a^{8-c\sqrt{\epsilon}}\|\nabla\phi\|_{H^{L}_G(M)}^2+a^{4-c\sqrt{\epsilon}}\E^{(L)}_{1,L}(f,\cdot)
\end{align*}
\end{subequations}
\end{lemma}
\begin{proof}
This proof now proceeds as in the pure scalar field setting (see \cite[Corollary 5.5, 5.8]{FU23}): We commute \eqref{eq:lapse-ell} with $\Lap^l$ and $\nabla\Lap^l$ for $l=\lfloor\frac{L}2\rfloor=0,\dots,9$ and apply the elliptic estimates from Lemma \ref{lem:ell-lapse} and \ref{lem:ell-lapse-odd} respectively. For the estimates with respect to $\L$, the right hand sides at order $L$ are of the form
\[\sum_{I=1}^{L-1}\nabla^IH\ast\nabla^{L-I}N+(N+1)\begin{cases}\Lap^\frac{L}2H & L\text{ even}\\
\nabla\Lap^{\lfloor\frac{L}2\rfloor}H& L\text{ odd}\end{cases}\,.\]
Again, the only changes occur when considering Vlasov terms: At high order, we apply Lemma \ref{lem:density-control} to estimate the density and pressure terms by Vlasov energies up to curvature errors. Conversely, low order Vlasov terms are bounded using \eqref{eq:APdensity}. Regarding the remaining quantities, we note that \cite[Lemma 4.5]{FU23} fully extends to our setting, allowing us to move between Sobolev norms and energies up to curvature errors at high orders.
\end{proof}

\subsection{Energy and norm estimates for the scalar field and remaining spacetime variables}\label{subsec:energy-old}

In this section, we collect individual energy and Sobolev norm estimates for the shear as well as scalar field and metric variables. Compared to the estimates in \cite[Section 6]{FU23} and their respective proofs, there are two differences we need to account for: Firstly, we now consider arbitrary (constant) spatial sectional curvature instead of $\kappa=-\frac19$, which leads to a slightly more refined argument being necessary to prove Lemma \ref{lem:en-est-SF}. Secondly, Vlasov matter now enters the respective equations and the above lapse estimates. \change{Beyond needing to estimate Vlasov matter terms where they occur, the fact that Vlasov matter enters into the elliptic lapse estimate \eqref{eq:lapse-en-est} leads to a linear coupling between scalar field and Vlasov matter. To deal with this, we need to introduce an additional time-scaling hierarchy in the total energy along with a hierarchy based on $\epsilon$-scaling, see Definition \eqref{def:total}. To this end, we additionally derive time-scaled estimates along similar lines as their unscaled analogues, and often distribute both types of weights in anticipation of the total energies in Definition \eqref{def:total}. The proofs of the latter otherwise }will follow by the same arguments as in the pure scalar field case, so we will often only focus on how these changes are accounted for in the respective proofs. 

\begin{lemma}[Energy estimates for the scalar field]\label{lem:en-est-SF} For $L\in 2\Z_+,\,L\leq 18$, the following estimate holds:
\begin{subequations}
\begin{align*}
\numberthis\label{eq:en-est-SF}&\,\E^{(L)}(\phi,t)+\int_t^{t_0}\dot{a}(s)a(s)^{3}\E^{(L+1)}(N,s)+\frac{\dot{a}(s)}{a(s)}\E^{(L)}(N,s)\,ds\\
\lesssim&\,\E^{(L)}(\phi,t_0)+\int_t^{t_0}\left({\epsilon}a(s)^{-3}+a(s)^{-1-c\sqrt{\epsilon}}\right)\E^{(L)}(\phi,s)\,ds\\
&\,+\int_t^{t_0}\left\{\change{\epsilon^\frac34a(s)^{-3}\cdot \epsilon^\frac14\E^{(L)}(\Sigma,s)+\epsilon a(s)^{-3}\cdot\epsilon^\frac12\E^{(L-2)}(\Ric,s)}\right.\\
&\,\phantom{+\int_t^{t_0}}+a(s)^{-2-c\sqrt{\epsilon}-(L+1)\omega}\cdot a(s)^{(L+1)\omega}\E^{(\leq L)}_{1,\leq L}(f,s)+\change{a(s)^{-2-c\sqrt{\epsilon}-\frac{\omega}2}\cdot a(s)^{\frac{\omega}2}\|G^{-1}-\gamma^{-1}\|_{L^2_G(M_s)}^2}\\
&\,\phantom{+\int_t^{t_0}}\left.+\sqrt{\epsilon}a(s)^{-3-c\sqrt{\epsilon}}\E^{(\leq L-2)}(\phi,s)+\change{\epsilon^\frac34 a(s)^{-3-c\sqrt{\epsilon}}\cdot \epsilon^\frac14}\E^{(\leq L-2)}(\Sigma,s)\right\}\,ds\\
&\underbrace{+\int_t^{t_0}\change{\epsilon a(s)^{-3-c\sqrt{\epsilon}}\cdot\epsilon^\frac12}\E^{(\leq L-4)}(\Ric,s)\,ds}_{\text{if }L\geq4}
\end{align*}
Similarly, one has
\begin{align*}
\numberthis\label{eq:en-est-SF-top}&\,a(t)^4\E^{(L+1)}(\phi,t)+\int_t^{t_0}\left\{\dot{a}(s)a(s)^7\E^{(L+2)}(N,s)+\dot{a}(s)a(s)^3\E^{(L+1)}(N,s)\right\}\,ds\\
\lesssim&\,a(t_0)^4\E^{(L+1)}(\phi,t_0)+\int_t^{t_0}\left(\epsilon a(s)^{-3}+a(s)^{-1-c\sqrt{\epsilon}}\right)\cdot a(s)^4\E^{(L+1)}(\phi,s)\,ds\\
&\,+\int_t^{t_0}\left\{\change{\epsilon^\frac12 a(s)^{-3}\cdot \epsilon^\frac12}a(s)^4\E^{(L+1)}(\Sigma,s)+\left(\epsilon a(s)^{-3}+a(s)^{-1-c\sqrt{\epsilon}}\right)\E^{(L)}(\phi,s)\right.\\
&\,\phantom{+\int_t^{t_0}}+\change{\epsilon^\frac34 a(s)^{-3}\cdot\epsilon^\frac12}\E^{(L)}(\Sigma,s)+\change{\epsilon^\frac14 a(s)^{-1-c\sqrt{\epsilon}}\cdot \epsilon^\frac34}a(s)^4\E^{(L-1)}(\Ric,s)\\
&\,\phantom{+\int_t^{t_0}}+\change{\left(\epsilon^\frac52 a(s)^{-3}+\epsilon^\frac12 a(s)^{-1-c\sqrt{\epsilon}}\right)\cdot\epsilon^\frac12}\E^{(L-2)}(\Ric,s)\\
&\,\phantom{+\int_t^{t_0}}+a(s)^{-2-c\sqrt{\epsilon}-(L+2)\omega}\cdot a(s)^{4+(L+2)\omega}\E^{(L+1)}_{1,L+1}(f,s)+a(s)^{2-c\sqrt{\epsilon}-(L+1)\omega}\cdot a(s)^{(L+1)\omega}\E^{(\leq L)}_{1,\leq L}(f,s)\,\\
&\,\phantom{+\int_t^{t_0}}\left.+\change{a(s)^{2-c\sqrt{\epsilon}-\frac{\omega}2}\cdot \change{a(s)^\frac{\omega}2}\|G^{-1}-\gamma^{-1}\|_{L^2_G(M_s)}^2}\right\}\,ds
\end{align*}
For $L=0$, analogous estimates hold with all energies of order $L-1$ or lower dropped from the respective right hand sides.
\end{subequations}
\end{lemma}
\begin{proof} We only sketch how to adapt the proof of \eqref{eq:en-est-SF} from that of \cite[Lemma 6.2]{FU23} -- analogous adaptations can then be applied to the proof of \cite[Lemma 6.4]{FU23} to obtain \eqref{eq:en-est-SF-top} as well as the low order analogues.\\
After taking the time derivative, inserting \eqref{eq:REEqWave} and integrating by parts, the following schematic equality holds: 
\begin{subequations}
\begin{align}
-\del_t\E^{(L)}(\phi,\cdot)\lesssim&\,\int_M \left\{6C\frac{\dot{a}}a\Lap^\frac{L}2N\cdot\Lap^\frac{L}2\Psi+3a\langle\nabla\Lap^\frac{L}2N,\nabla\phi\rangle_G\cdot\Lap^\frac{L}2\Psi\right.\label{eq:en-ineq-SF1}\\
&\,\phantom{\int_M}\left.-2Ca\langle \nabla\Lap^\frac{L}2N,\nabla\Lap^\frac{L}2\phi\rangle_G\right\}-4\dot{a}a^3\lvert\nabla\Lap^\frac{L}2\phi\rvert_G^2\,\vol{G}+\langle\text{error terms}\rangle\label{eq:en-ineq-SF2}
\end{align}
\end{subequations}
Regarding the first term in \eqref{eq:en-ineq-SF1}, we insert the following zero using \eqref{eq:REEqLapse1} and \eqref{eq:Friedman}:
\begin{align*}
0=&\,\int_M\div_G\left(\nabla\Lap^\frac{L}2N\cdot\Lap^\frac{L}2N\right)\,\vol{G}\\
=&\,\int_M -\frac3{8\pi}\dot{a}a^3\lvert\nabla\Lap^{\frac{L}2}N\rvert_G^2-\frac{3}{8\pi}\left(-3\kappa\dot{a}a^3+12\pi C^2\frac{\dot{a}}a+12\pi a^{-2}\left({\rho}_{FLRW}^{Vl}-{\mathfrak{p}}_{FLRW}^{Vl}\right)\right)\lvert\Lap^{\frac{L}2}N\rvert^2\\
&\quad -6 C\frac{\dot{a}}a\Lap^{\frac{L}2}N\cdot \Lap^{\frac{L}2}\Psi+\frac92\pi a^{-2}\Lap^\frac{L}2\left[\left(\rho^{Vl}-{\rho}_{FLRW}^{Vl}\right)-\left(\mathfrak{p}^{Vl}-{\mathfrak{p}}^{Vl}_{FLRW}\right)\right]\change{\Lap^\frac{L}2N}\\
&\quad+\langle\text{borderline and junk terms}\rangle\,\vol{G}
\end{align*}
Regarding the first line, the first and third term can be pulled to the left hand side, while the second and fourth can be bounded by $\lesssim a^{-2}\E^{(L)}(N,\cdot)$. The first term in the second line precisely cancels the first term in \eqref{eq:en-ineq-SF1}, and the \change{term containing Vlasov matter }can be bounded by
\change{\begin{equation}\label{eq:en-est-SF-Vlasovterm}
a^{-2-c\sqrt{\epsilon}}\sqrt{\E^{(\leq L)}_{1,\leq L}(f,\cdot)}\sqrt{\E^{(L)}(N,\cdot)}\,
\end{equation}}
by Lemma \ref{lem:density-control}. Note that, for the same argument at $L=0$, this also incurs a metric term. \\
The second term in \eqref{eq:en-ineq-SF1} and the explicit term in \eqref{eq:en-ineq-SF2} can be treated similarly, so we focus on the latter, denoting it by $(\ast)$: Applying \eqref{eq:diff-ineq-Friedman}, we obtain the following:
\begin{align*}
\lvert(\ast)\rvert\leq&\,\int_M2\sqrt{\frac{3}{4\pi}}\left(\dot{a}a^3+\kappa a^3\right)\lvert\nabla\Lap^\frac{L}2N\rvert_G\lvert\nabla\Lap^\frac{L}2\phi\rvert_G\,\vol{G}\\
\leq&\,\int_M \left(4\dot{a}a^3\lvert\nabla\Lap^{\frac{L}2}\phi\rvert_G^2+\frac{3}{16\pi}\dot{a}a^3\lvert\nabla\Lap^\frac{L}2N\rvert_G^2\right)\,\vol{G}+\lvert \kappa\rvert a^{-1}\left(\E^{(L)}(\phi,\cdot)+a^4\E^{(L+1)}(N,\cdot)\right)
\end{align*}
The first term cancels with the second term in \eqref{eq:en-ineq-SF2}, the second can be pulled to the left and absorbed into already present lapse energy terms after adapting constants. \\
Combining these observations and dealing with borderline and junk terms as in the pure scalar field setting, one obtains the following differential estimate:
\begin{subequations}
\begin{align*}
&-\del_t\E^{(L)}(\phi,\cdot)+\dot{a}a^{3}\E^{(L+1)}(N,\cdot)+\frac{\dot{a}}a\E^{(L)}(N,\cdot)\\
\numberthis\label{eq:diff-ineq-SF1}\lesssim&\,\left(\epsilon a^{-3}+a^{-1-c{\sigma}}\right)\E^{(L)}(\phi,\cdot)+\left(\epsilon a^{-3}+\change{a^{-2-c\sqrt{\epsilon}}}\right)\left(a^4\E^{(L+1)}(N,\cdot)+\E^{(L)}(N,\cdot)\right)\\
\numberthis\label{eq:diff-ineq-SF2}&\,+\epsilon a^{-3}\E^{(L)}(\Sigma,\cdot)+\epsilon^\frac32 a^{-3}\E^{(L-2)}(\Ric,\cdot)+
{\epsilon}a^{-3-c\sqrt{\epsilon}}\E^{(\leq L-2)}(\phi,\cdot)\\
\numberthis\label{eq:diff-ineq-SF3}&\,+\epsilon a^{-3-c\sqrt{\epsilon}}\E^{(\leq L-2)}(\Sigma,\cdot)+\left[\epsilon a^{-3-c\sqrt{\epsilon}}+\sqrt{\epsilon}a^{-1-c\sqrt{\epsilon}}\right]\E^{(\leq L-2)}(N,\cdot)\\
\numberthis\label{eq:diff-ineq-SF4}&\,+a^{-2-c\sqrt{\epsilon}}\E^{(\leq L)}_{1,\leq L}(f,\cdot)+\underbrace{\epsilon^\frac32a^{-3-c\sqrt{\epsilon}}\E^{(\leq L-4)}(\Ric,\cdot)}_{\text{not present for }L=2}
\end{align*}
\end{subequations}
The statement then follows by applying Lemma \ref{lem:lapse-en-est} to the lapse energies in \eqref{eq:diff-ineq-SF1} and \eqref{eq:diff-ineq-SF3} and integrating. We note that, for $L=0$, one not only incurs metric norms by applying elliptic lapse estimates at low orders, but also when estimating \eqref{eq:en-est-SF-Vlasovterm}.
\end{proof}

We will need an additional norm bound to obtain optimal control for the scalar field:

\begin{lemma}[Norm bound for $\nabla\phi$]\label{lem:norm-est-nablaphi}
\begin{equation*}
\|\nabla\phi\|_{H^l_G(M)}\lesssim(1+\epsilon a^{-c\sqrt{\epsilon}})\|\Sigma\|_{H^{l+1}_G(\change{M})}+\sqrt{\epsilon} a^{-c\sqrt{\epsilon}}\|\Psi\|_{H^l_G(M)}+(1+\epsilon a^{-c\sqrt{\epsilon}})\|\j^{Vl}\|_{H^l_G(M)}
\end{equation*}
\end{lemma}
\begin{proof}
The momentum constraint \eqref{eq:REEqMom} leads to the following expression after rearranging and using \eqref{eq:APPsi} to ensure $\Psi+C>\frac{C}2$:
\begin{align*}
\lvert \nabla^l\nabla\phi\rvert_G=&\,\left\lvert \nabla^l\left(\frac{\frac1{8\pi}\div_G\Sigma-\j^{Vl}}{\Psi+C}\right)\right\rvert_G\\
\lesssim&\,\lvert \nabla^{l+1}\Sigma\rvert_G+\lvert\nabla^l\j^{Vl}\rvert_G+\sum_{I_1+I_2=l-1}\left(\lvert \nabla^{I_1+1}\Sigma\rvert_G+\lvert\nabla^{I_1}\j^{Vl}\rvert_G\right)\lvert \nabla^{I_2+1}\Psi\rvert_G
\end{align*}
The statement then follows by applying the Young inequality with \eqref{eq:APmidSigma}, \eqref{eq:APdensity} and \eqref{eq:APPsi}.
\end{proof}

Before moving on the energy estimates for Bel-Robinson variables, we prepare the necessary bounds on the Vlasov matter components that occur in the Bel-Robinson evolution equations \eqref{eq:REEqE} and \eqref{eq:REEqB}:

\begin{lemma}[Bounds for Bel-Robinson Vlasov matter components]\label{lem:bound-J-Vl} Recalling \eqref{eq:JVlpar} and \eqref{eq:JVlast}, writing \[(J^{Vl})^{(\ast)}_{i0j,symm}=\frac12\left[(J^{Vl})^{(\ast)}_{i0j}+(J^{Vl})^{(\ast)}_{j0i}\right]\]
and using $J^{Vl,\parallel}_{(\cdot)0(\cdot),symm}$ to denote the same expression as in \eqref{eq:JVlpar} where pure trace terms are dropped, one has
\begin{align*}
\|a^4(N+1)J^{Vl,\parallel}_{(\cdot)0(\cdot),symm}\|_{L^2_G(M)}&\,\lesssim a^{-1-c\sqrt{\epsilon}-\frac{\omega}2}\sqrt{a^\omega\E^{(0)}_{1,0}(f,\cdot)}+a^{-1-c\sqrt{\epsilon}-\omega}\sqrt{a^{4+2\omega}\E^{(1)}_{1,\leq 1}(f,\cdot)}
\\
&\,+a^{-1-c\sqrt{\epsilon}}\sqrt{a^4\E^{(1)}(N,\cdot)}+a^{1-c\sqrt{\epsilon}}\sqrt{\E^{(0)}(N,\cdot)}+a^{-1-c\sqrt{\epsilon}}\sqrt{\E^{(0)}(\Sigma,\cdot)}\\
&\end{align*}
as well as, for $L\in2\Z_+, L\leq 18$,
\begin{align*}
&\,\|a^4\Lap^\frac{L}2\left((N+1)J^{Vl,\parallel}_{(\cdot)0(\cdot),symm}\right)\|_{L^2_G(M)}\\
\lesssim&\,a^{-1-c\sqrt{\epsilon}-\frac{L+2}2\omega}\sqrt{a^{4+(L+2)\omega}\E^{(L+1)}_{1,L+1}(f,\cdot)}+a^{-1-c\sqrt{\epsilon}-\frac{L+1}2\omega}\sqrt{a^{(L+1)\omega}\E^{(\leq L)}_{1,\leq L}(f,\cdot)}\\
&\,+a^{-1-c\sqrt{\epsilon}}\sqrt{a^4\E^{(L+1)}(N,\cdot)}+a^{1-c\sqrt{\epsilon}}\sqrt{\E^{(L)}(N,\cdot)}\\
&\,+\sqrt{\epsilon}a^{1-c\sqrt{\epsilon}}\sqrt{\E^{(\leq L-2)}(N,\cdot)}+a^{-1-c\sqrt{\epsilon}}\sqrt{\E^{(\leq L)}(\Sigma,\cdot)}\\
&\,+\sqrt{\epsilon}a^{-1-c\sigma}\sqrt{a^4\E^{(L-1)}(\Ric,\cdot)}+\sqrt{\epsilon} a^{-1-c\sqrt{\epsilon}}\sqrt{\E^{(\leq L-2)}(\Ric,\cdot)}
\end{align*}
and, for $L\in2\N, L\leq 18$,
\begin{align*}
&\,\|a^4\Lap^{\frac{L}2}\left((N+1)J^{Vl,\ast}_{(\cdot)0(\cdot),symm}\right)\|_{L^2_G}\\
\lesssim&\,\change{a^{-2-c\sqrt{\epsilon}-\frac{(L+2)}2\omega}\sqrt{a^{4+(L+2)\omega}\E^{(L+1)}_{1,L+1}(f,\cdot)}+\epsilon a^{-1-c\sqrt{\epsilon}-\frac{L+1}2\omega}\sqrt{a^{(L+1)\omega}\E^{(\leq L)}_{1,L}(f,\cdot)}}\\
&\,+a^{-1-c\sqrt{\epsilon}}\sqrt{\E^{(\leq L)}(\Sigma,\cdot)}+\underbrace{\epsilon a^{-1-c\sqrt{\epsilon}}\sqrt{\E^{(\leq L-2)}(\Ric,\cdot)}}_{\text{if }L\neq 0}
\end{align*}
\begin{proof}
We can directly use Lemma \ref{lem:density-control} to control $(S^{Vl})^{\parallel}$ and $\nabla\j^{Vl}$. \change{Additionally, we note that the final term in \eqref{eq:JVlpar} can be estimated identically to $(S^{Vl})^{\parallel}$ after observing that $m^2a^2\leq (v^0)^2$. }For the remaining terms, recall that we can apply \eqref{eq:nabla-intTM} to pull the covariant derivative past the respective momentum integrals. By Lemma \ref{lem:mom-der}, any of the prefactors depending only on the momentum are uniformly bounded by $v^0$ up to constant and vanish when hit by a horizontal derivative. Keeping this in mind, we apply the product rule: When estimating high order Vlasov terms, we apply the Jensen inequality as in the proof of Lemma \ref{lem:density-control} to connect the respectice quantities to our Vlasov energies, and apply the low order bounds \eqref{eq:APmidSigma} and \eqref{eq:BsC} for shear and lapse respectively. On the other hand, if most derivatives hit lapse or shear terms, we estimate the respective distribution function terms with \eqref{eq:APVlasov} by $a^{-c\sqrt{\epsilon}}$, and simply apply \eqref{eq:APMom} to deal with the remaining integrals over functions of $v$. Finally, recall \cite[Lemma 4.5]{FU23} to bound Sobolev norms in $\Sigma$ and $N$ to their respective energy terms up to curvature errors.
\end{proof}
\end{lemma}

\begin{corollary}[Energy estimates for Bel-Robinson variables]\label{lem:en-est-BR} The following \change{holds }
for $L\in 2\N,\,2\leq L\leq 18$:
\begin{subequations}\label{eq:en-est-BR-all}
\change{\begin{align*}
\numberthis\label{eq:en-est-BR-scaled}&\,a(t)^{\frac{\omega}4}\E^{(L)}(W,t)+\int_t^{t_0}\int_M a(s)^{\frac{\omega}4}\left[8\pi C^2a(s)^{-3}(N+1)\langle\Lap^\frac{L}2\Sigma,\Lap^\frac{L}2\RE\rangle_G+6\frac{\dot{a}(s)}{a(s)}\lvert\Lap^\frac{L}2\RE\rvert_G^2\right]\,\vol{G}\,ds\\
&\,\lesssim a(t_0)^{\frac{\omega}4}\E^{(L)}(W,t_0)+\int_t^{t_0} \left(\epsilon a(s)^{-3}+a(s)^{-3+\frac{\omega}8}+a(s)^{-2-c\sqrt{\epsilon}}\right)\,a(s)^\frac{\omega}4\E^{(L)}(W,s)\,ds\\
&\,+\int_t^{t_0}\left\{ \left(a(s)^{-3+\frac{\omega}8}+a(s)^{-1-c\sqrt{\epsilon}}\right)\,\left(a(s)^4\E^{(L+1)}(\phi,s)+\E^{(L)}(\phi,s)\right)\right.\\
&\,\phantom{\int_t^{t_0}}+\left(\epsilon a(s)^{-3}+a(s)^{-1-c\sqrt{\epsilon}}\right)\,a(s)^{\frac{\omega}4} \E^{(L)}(\Sigma,s)\\
&\,\phantom{\int_t^{t_0}}+\epsilon^\frac14a(s)^{-3+\frac{\omega}8}\,\cdot \epsilon^\frac34a(s)^4\E^{(L-1)}(\Ric,s)+\left(\epsilon^\frac{7}2 a(s)^{-3+\frac{\omega}8}+\epsilon^\frac12 a(s)^{-1-c\sqrt{\epsilon}+\frac{\omega}8}\right)\,\cdot\epsilon^\frac12\E^{(\leq L-2)}(\Ric,s)\\
&\,\phantom{\int_t^{t_0}}+a(s)^{-2-c\sqrt{\epsilon}-(L+2)\omega}\,\cdot a(s)^{4+(L+2)\omega}\E^{(L+1)}_{1,L+1}(f,s)+a(s)^{-1-c\sqrt{\epsilon}-(L+1)\omega}\cdot a(s)^{(L+1)\omega}\E^{(\leq L)}_{1,\leq L}(f,s)\\
&\,\phantom{\int_t^{t_0}}+a(s)^{3-c\sqrt{\epsilon}-\frac{\omega}4}\cdot a(s)^{\frac{\omega}2}\|G^{-1}-\gamma^{-1}\|^2_{L^2_G(M_s)}+\left(\epsilon^2 a(s)^{-3-c\sqrt{\epsilon}+\frac{\omega}8}+a(s)^{-1-c\sqrt{\epsilon}+\frac{\omega}8}\right)\E^{(\leq L-2)}(\phi,s)\\
&\,\phantom{\int_t^{t_0}}+\epsilon a(s)^{-3-c\sqrt{\epsilon}}\cdot a(s)^\frac{\omega}4\left(\E^{(\leq L-2)}(\Sigma,s)+\E^{(\leq L-2)}(W,s)\right)\\
&\,\phantom{\int_t^{t_0}}+\underbrace{\epsilon^\frac12 a(s)^{-3-c\sqrt{\epsilon}+\frac{\omega}8}\cdot\epsilon^{\frac{1}{2}}\E^{(\leq L-4)}(\Ric,s)}_{\text{not present for }L=2}\Bigr\}\,ds\,.
\end{align*}}
\begin{align*}
\numberthis\label{eq:en-est-BR}&\,\E^{(L)}(W,t)+\int_t^{t_0}\int_M\left[8\pi C^2a(s)^{-3}(N+1)\langle\Lap^\frac{L}2\Sigma,\Lap^\frac{L}2\RE\rangle_G+6\frac{\dot{a}(s)}{a(s)}\lvert\Lap^\frac{L}2\RE\rvert_G^2\right]\,\vol{G}\,ds\\
&\,\lesssim\E^{(L)}(W,t_0)+\int_t^{t_0}\left(\epsilon^\frac18 a(s)^{-3}+a(s)^{-2-c\sqrt{\epsilon}}\right)\E^{(L)}(W,s)\,ds\\
&\,+\int_t^{t_0}\left\{\epsilon^{-\frac18}a(s)^{-3}\cdot a(s)^4\E^{(L+1)}(\phi,s)+\left(\epsilon^\frac18 a(s)^{-3}+a(s)^{-1}\right)\E^{(L)}(\phi,s)\right.\\
&\,\phantom{\int_t^{t_0}}+\left(\epsilon a(s)^{-3}+a\change{(s)}^{-1-c\sqrt{\epsilon}}\right)\E^{(L)}(\Sigma,s)\\
&\,\phantom{\int_t^{t_0}}+\change{\epsilon^\frac{1}8a(s)^{-3}\cdot \epsilon^\frac34}a(s)^4\E^{(L-1)}(\Ric,s)+\change {\left(\epsilon^\frac{27}8 a(s)^{-3}+\epsilon^\frac12a(s)^{-1-c\sqrt{\epsilon}}\right)\cdot\epsilon^\frac12}\E^{(\leq L-2)}(\Ric,s)\\
&\,\phantom{\int_t^{t_0}}+a(s)^{-2-c\sqrt{\epsilon}-(L+2)\omega}\cdot a(s)^{4+(L+2)\omega}\E^{(L+1)}_{1,L+1}(f,s)+a(s)^{-1-c\sqrt{\epsilon}-(L+1)\omega}\cdot a(s)^{(L+1)\omega}\E^{(\leq L)}_{1,\leq L}(f,s)\\
&\,\phantom{\int_t^{t_0}}+\change{a(s)^{3-c\sqrt{\epsilon}-\frac{\omega}4}\cdot a(s)^{\frac{\omega}2}\|G^{-1}-\gamma^{-1}\|_{L^2_G(M_s)}^2}+\left(\epsilon^\frac{15}8 a(s)^{-3-c\sqrt{\epsilon}}+a(s)^{-1-c\sqrt{\epsilon}}\right)\E^{(\leq L-2)}(\phi,s)\\
&\,\left.\phantom{\int_t^{t_0}}+\epsilon a(s)^{-3-c\sqrt{\epsilon}}\left(\E^{(\leq L-2)}(\Sigma,s)+\E^{(\leq L-2)}(W,s)\right)+\underbrace{\change{\epsilon^\frac12 a(s)^{-3-c\sqrt{\epsilon}}\cdot \epsilon^\frac12}\E^{(\leq L-4)}(\Ric,s)}_{\text{not present for }L=2}\right\}\,ds\,.
\end{align*}
\end{subequations}

For $L=0$, the \change{analogues of \eqref{eq:en-est-BR-all} hold }where all energies of order below $L$ are dropped from the right hand \change{sides}.
\end{corollary}
\begin{proof}
The proof proceeds exactly as in the pure scalar field case (see \cite[Lemma 6.6]{FU23}), where we use Lemma \ref{lem:bound-J-Vl} to bound the Vlasov matter terms, and apply Lemma \ref{lem:lapse-en-est} to bound lapse energies.
 \end{proof}

\begin{lemma}[Integral energy estimates for the second fundamental form]\label{lem:en-est-Sigma}
For $L\in 2\Z_+, L\leq 18$, the following holds:
\begin{subequations}\label{eq:en-est-Sigma-all}
\change{\begin{align*}
\numberthis\label{eq:en-est-Sigma-scaled}&\,a(t)^\frac{\omega}4\E^{(L)}(\Sigma,t)+2\int_t^{t_0}\int_M a(s)^\frac{\omega}4\left[a(s)^{-3}(N+1)\langle\Lap^\frac{L}2\RE,\Lap^{\frac{L}2}\Sigma\rangle_G+\frac{\dot{a}(s)}{a(s)}(N+1)\lvert\Lap^{\frac{L}2}\Sigma\rvert_G^2\right]\,\vol{G}\,ds\\
&\,\lesssim a(t_0)^\frac{\omega}4\E^{(L)}(\Sigma,t_0)+\int_t^{t_0}\left(\epsilon a(s)^{-3}+a(s)^{-3+\frac{\omega}8}+a(s)^{-1-c\sqrt{\epsilon}}\right) a(s)^{\frac{\omega}4}\E^{(L)}(\Sigma,s)\,ds\\
&\,+\int_t^{t_0}\Bigr\{a(s)^{-3+\frac{\omega}8}\E^{(L)}(\phi,s)+\epsilon^\frac32a(s)^{3-c\sigma+\frac{\omega}4}\cdot\epsilon^\frac12\E^{(L-1)}(\Ric,s)+\epsilon^\frac32 a(s)^{-3+\frac{\omega}8}\cdot\epsilon^\frac12\E^{(L-2)}(\Ric,s)\\
&\,\phantom{\int_t^{t_0}}+\epsilon^2a(s)^{-3-c\sqrt{\epsilon}}\cdot a(s)^\frac{\omega}4\E^{(\leq L-2)}(\Sigma,s)+\epsilon^2a(s)^{-3-c\sqrt{\epsilon}+\frac{\omega}8}\E^{(\leq L-2)}(\phi,s)\\
&\,\phantom{\int_t^{t_0}}+a(s)^{-1-c\sqrt{\epsilon}-(L+1)\omega}\cdot \left(a(s)^{(L+1)\omega}\E^{(\leq L)}_{1,\leq L}(f,s)+a(s)^{\frac{\omega}2}\|G^{-1}-\gamma^{-1}\|_{L^2_G(M_s)}^2\right)\\
&\,\phantom{\int_t^{t_0}}+\underbrace{\epsilon^\frac32a(s)^{-3-c\sqrt{\epsilon}+\frac{\omega}8}\cdot\epsilon^\frac12\E^{(\leq L-4)}(\Ric,s)}_{\text{not present for }L=2}\Bigr\}\,ds
\end{align*}}
\begin{align*}
\numberthis\label{eq:en-est-Sigma}&\,\E^{(L)}(\Sigma,t)+2\int_t^{t_0}\int_M\left[a(s)^{-3}(N+1)\langle\Lap^\frac{L}2\RE,\Lap^{\frac{L}2}\Sigma\rangle_G+\frac{\dot{a}(s)}{a(s)}(N+1)\lvert\Lap^{\frac{L}2}\Sigma\rvert_G^2\right]\,\vol{G}\,ds\\
&\,\lesssim\E^{(L)}(\Sigma,t_0)+\int_t^{t_0}\left(\epsilon^\frac18a(s)^{-3}+a(s)^{-1-c\sqrt{\epsilon}}\right)\E^{(L)}(\Sigma,s)\,ds\\
&\,+\int_t^{t_0}\Bigr\{\epsilon^{-\frac18}a(s)^{-3}\E^{(L)}(\phi,s)+\change{\epsilon^\frac54 a(s)^{3-c\sigma}\cdot \epsilon^\frac34}\E^{(L-1)}(\Ric,s)+\epsilon^2a(s)^{-3}\E^{(L-2)}(\Ric,s)\\
&\,\phantom{\int_t^{t_0}}+\epsilon^2a(s)^{-3-c\sqrt{\epsilon}}\E^{(\leq L-2)}(\Sigma,s)+\left(\epsilon^\frac{15}8a(s)^{-3-c\sqrt{\epsilon}}+\epsilon a(s)^{-1-c\sqrt{\epsilon}}\right)\E^{(\leq L-2)}(\phi,s)\\
&\,\phantom{\int_t^{t_0}}+a(s)^{-1-c\sqrt{\epsilon}-(L+1)\omega}\cdot \left(a(s)^{(L+1)\omega}\E^{(\leq L)}_{1,\leq L}(f,s)+a(s)^{(L+1)\omega}\change{\|G^{-1}-\gamma^{-1}\|_{L^2_G(M_s)}^2}\right)\\
&\,\phantom{\int_t^{t_0}}+\underbrace{\change{\epsilon^\frac32 a(s)^{-3-c\sqrt{\epsilon}}\cdot\epsilon^\frac12}\E^{(\leq L-4)}(\Ric,s)}_{\text{not present for }L=2}\Bigr\}\,ds
\end{align*}
\end{subequations}
For $L=0$, the same estimate holds with all energies of order $L-1$ or lower dropped from the upper bound.
\end{lemma}
\begin{proof}
Note that the only term that is not pure trace that the presence of Vlasov matter adds to the evolution equation is $S^{Vl,\parallel}$. Hence, commuting \eqref{eq:REEqSigmaSharp} with $\Lap^{\frac{L}2}$ and carrying over error term notation from \cite[(11.14)]{FU23}, we have the following:
\begin{subequations}
\begin{align}
-\del_t\E^{(L)}(\Sigma,\cdot)=&\,\int_M \left\{2a\langle\nabla^2\Lap^{\frac{L}2} N,\Lap^{\frac{L}2}\Sigma\rangle_G-2a(N+1)\langle\Lap^{\frac{L}2}\Ric[G],\Lap^{\frac{L}2}\Sigma\rangle_G\right.\label{eq:en-est-Sigma1}\\
&\,\qquad+16\pi a^{-1}\langle\Lap^\frac{L}2 S^{Vl,\parallel},\Lap^\frac{L}2\Sigma\rangle_G\label{eq:en-est-Sigma2}\\
&\,\qquad+(\del_tG^{-1})\ast G^{-1}\ast\Lap^{\frac{L}2}\Sigma\ast\Lap^{\frac{L}2}\Sigma-3N\frac{\dot{a}}a\lvert\Lap^{\frac{L}2}\Sigma\rvert_G^2\label{eq:en-est-Sigma3}\\
&\,\qquad\left.-2\langle\mathfrak{S}_{L,Border},\Lap^{\frac{L}2}\Sigma\rangle_G-2\langle\mathfrak{S}_{L,Junk}^\parallel,\Lap^{\frac{L}2}\Sigma\rangle_G\right\}\,\vol{G}\label{eq:en-est-Sigma4}
\end{align}
\end{subequations}
Regarding \eqref{eq:en-est-Sigma1}, we apply the lapse energy estimate \eqref{eq:lapse-en-est} (resp., for $L=0$, \eqref{eq:lapse-en-est0}) for the first term and have the following:
\begin{align*}
&\,a\|\nabla^2\Lap^\frac{L}2N\|_{L^2_G}\sqrt{\E^{(L)}(\Sigma,\cdot)}\lesssim a^{-3}\sqrt{a^8\E^{(L+2)}(N,\cdot)+\E^{(L)}(N,\cdot)}\sqrt{\E^{(L)}(\Sigma,\cdot)}\\
\lesssim&\, \left(\epsilon^\frac18 a^{-3}+a^{-1-c\sqrt{\epsilon}}\right)\E^{(L)}(\Sigma,\cdot)+\epsilon^{-\frac18}a^{-3}\E^{(L)}(\phi,\cdot)+a^{-1-c\sqrt{\epsilon}-(L+1)\omega}\change{\cdot}\,a^{(L+1)\omega}\E^{(L)}_{1,L}(f,\cdot)\\
&\,+\begin{cases}
\change{a^{-1-c\sqrt{\epsilon}-\frac{\omega}2}\cdot a^\frac{\omega}2\|G^{-1}-\gamma^{-1}\|_{L^2_G(M)}^2} & L=0\\
\epsilon^\frac{15}8a^{-3-c\sqrt{\epsilon}}\left(\E^{(\leq L-2)}(\phi,\cdot)+\E^{(\leq L-2)}(\Sigma,\cdot)\right)+\epsilon^\frac{15}8 a^{-3-c\sqrt{\epsilon}}\E^{(\leq L-3)}(\Ric,\cdot) & L\geq 2
\end{cases}
\end{align*}
For the second term in \eqref{eq:en-est-Sigma1}, we insert the Bel-Robinson constraint \eqref{eq:REEqConstrE}. High order lapse terms are treated similarly as above, and the only Vlasov term that is tracefree and thus isn't cancelled by the inner product is
\[-2(N+1)a^{-3}\cdot 4\pi a^2\langle\Lap^\frac{L}2 S^{Vl,\parallel},\Lap^\frac{L}2\Sigma\rangle_G\,.\]
This term is bounded in absolute value by $a^{-1-c\sqrt{\epsilon}}\sqrt{\E^{(L)}_{1,L}(f,\cdot)}\sqrt{\E^{(L)}(\Sigma,\cdot)}$ up to constant with Lemma \ref{lem:density-control}, as is the term in \eqref{eq:en-est-Sigma2}. \eqref{eq:en-est-Sigma3}-\eqref{eq:en-est-Sigma4} only contribute error terms that can be dealt with as in the scalar field setting, accounting for \eqref{eq:BsN}. 
\end{proof}

\begin{lemma}[Elliptic energy estimate for the second fundamental form]\label{lem:en-est-Sigma-top} For any $L\in 2\Z_+$, $2\leq L\leq 18$, we have
\begin{subequations}
\begin{align*}
\numberthis\label{eq:en-est-Sigma-top}a^4\E^{(L+1)}(\Sigma,\cdot)&\,\lesssim\left(a^{4-c\sqrt{\epsilon}}+\epsilon a^{2-c\sqrt{\epsilon}}\right)\E^{(L)}(\Sigma,\cdot)+\E^{(L)}(\phi,\cdot)+\E^{(L)}(W,\cdot)\\
&\,\change{+\epsilon^\frac54\cdot \epsilon^\frac34 a^4\E^{(L-1)}(\Ric,\cdot)+a^{4-c\sqrt{\epsilon}-(L+1)\omega}\cdot a^{(L+1)\omega}\E^{(L)}_{1,L}(f,\cdot)}\\
&\,+\epsilon a^{2-c\sqrt{\epsilon}}\E^{(\leq L-2)}(\phi,\cdot)+a^{2-c\sqrt{\epsilon}}\E^{(\leq L-2)}(\Sigma,\cdot)+\change{\epsilon^\frac12 a^{2-c\sqrt{\epsilon}}\cdot \epsilon^\frac12\E^{(\leq L-2)}(\Ric,\cdot)}\,.
\end{align*}
For $L=0$, one analogously has
\begin{equation}\label{eq:en-est-Sigma-1}
a^4\E^{(1)}(\Sigma,\cdot)\lesssim\,\left(a^{4-c\sqrt{\epsilon}}+\epsilon a^{2-c\sqrt{\epsilon}}\right)\E^{(0)}(\Sigma,\cdot)+\E^{(0)}(\phi,\cdot)+\E^{(0)}(W,\cdot)+a^4\E^{(0)}_{1,0}(f,\cdot)\,.
\end{equation}
\end{subequations}
\end{lemma}
\begin{proof}
Recall from the proof of \cite[Lemma 6.10]{FU23} that, since $\Lap^\frac{L}2\Sigma$ is tracefree, one has
\begin{equation*}
\E^{(L+1)}(\Sigma,\cdot)\lesssim a^{-c\sqrt{\epsilon}}\E^{(L)}(\Sigma,\cdot)+\int_M\lvert \div_G\Lap^{\frac{L}2}\Sigma\rvert_G^2+\lvert\curl\Lap^\frac{L}2\Sigma\rvert_G^2\,\vol{G}
\end{equation*}
Inserting commuted analogues of \eqref{eq:REEqMom} and \eqref{eq:REEqConstrB}, we can argue as in the scalar field setting, where the commuted analogue of \eqref{eq:REEqMom} now additionally contains the term
\[8\pi\Lap^\frac{L}2\j^{Vl},\]
where we again apply Lemma \ref{lem:density-control}, leading to the respective Vlasov energy terms. 
\end{proof}

We close out this section by collecting all necessary bounds for metric variables:

\begin{lemma}[Sobolev norm estimates for metric variables]\label{lem:norm-est-G} For any $l\in\N,\,l\leq 18$, we have the following:
\begin{subequations}
\begin{align*}\numberthis\label{eq:metric-int-est-scaled}
\change{a(t)^\frac{\omega}2}\|\change{G^{\pm 1}-\gamma^{\pm 1}}&\|_{H^l_G(M_t)}^2\lesssim\epsilon^4+\int_t^{t_0}\left(a(s)^{-3+\change{\frac{\omega}8}}+a(s)^{-1-c\sigma}\right)\cdot \change{a(s)^\frac{\omega}2}\|\change{G^{\pm 1}-\gamma^{\pm 1}}\|_{H^l_G(M_s)}^2\,ds\\
&\,+\int_t^{t_0}a(s)^{-3+\change{\frac{\omega}8-c\sqrt{\epsilon}}}\left(a(s)^\frac{\omega}4\E^{(\leq l)}(\Sigma,s)+\E^{(\leq l)}(\phi,s)+\underbrace{\change{\epsilon\E^{(\leq l-2)}(\Ric,s)}}_{\text{if }l\geq 3}\right)\,ds\\
&\,+\int_t^{t_0}a(s)^{-1-c\sqrt{\epsilon}-(l+1)\omega}\cdot a(s)^{(l+1)\omega}\E^{(\leq l)}_{1,\leq l}(f,s)\,ds\\
\numberthis\label{eq:norm-est-G}
\|\change{G^{\pm 1}-\gamma^{\pm 1}}&\|_{H^l_G(M_t)}^2\lesssim a^{-c\epsilon^\frac18}\left(\epsilon^4+\epsilon^{-\frac14}\sup_{s\in(t,t_0)}\left(\|N\|_{H^{l}_G(M_s)}^2+\|\Sigma\|_{H^{l}_G(M_s)}^2\right)\right)
\end{align*}
\end{subequations}
For $l>0$, one additionally has
\begin{subequations}
\begin{align*}\numberthis\label{eq:chr-int-est-scaled}
a(t)^{\beta+\change{\frac{\omega}2}}&\,\|\Gamma-\Gamhat\|_{H^{l-1}_G(M_t)}^2\lesssim \epsilon^4+\int_t^{t_0}\left(a(s)^{-3+\change{\frac{\omega}8}-c\sqrt{\epsilon}}+a(s)^{-1-c\sigma}\right)\cdot a(s)^{\beta+\frac{\omega}2}\|\Gamma-\Gamhat\|_{H^{l-1}_G(M_s)}^2\,ds\\
&\,+\int_t^{t_0}a(s)^{-3+\change{\frac{\omega}8}}\cdot a(s)^{\beta}\left(a(s)^\frac{\omega}4\E^{(\leq l)}(\Sigma,s)+\E^{(\leq l)}(\phi,s)+\underbrace{\change{\epsilon\E^{(\leq l-2)}(\Ric,s)}}_{\text{if }l\geq 3}\right)\,ds\\
&\,+\int_t^{t_0}a(s)^{-1-c\sqrt{\epsilon}-\change{(l+1)\omega}}\cdot a(s)^{\beta+(l+1)\omega}\E^{(\leq l)}_{1,\leq l}(f,s)\,ds\,\text{ and}\\
\numberthis\label{eq:chr-norm-est}&\,\|\Gamma-\Gamhat\|_{H^{l-1}_G(M_t)}^2\lesssim a(t)^{-c\epsilon^\frac18}\left(\epsilon^4+\epsilon^{-\frac14}\sup_{s\in(t,t_0)}\left(\|N\|_{H^{l}_G(M_s)}^2+\|\Sigma\|_{H^{l}_G(M_s)}^2\right)\right)
\end{align*}
\end{subequations}
as well as the following:
\begin{subequations}
\change{\begin{align*}
\numberthis\label{eq:en-est-Ric-scaled}a(t)^\frac{\omega}2&\,\E^{(L-2)}(\Ric,t)\lesssim a(t_0)^\frac{\omega}{4}\E^{(L-2)}(\Ric,t_0)\\
&\,+\int_t^{t_0}\left(\epsilon^\frac18a(s)^{-3}+a(s)^{-3+\frac{\omega}8}+a(s)^{-1-c\sqrt{\epsilon}}\right)\change{a(s)^\frac{\omega}2}\E^{(L-2)}(\Ric,s)\,ds\\
&\,+\int_t^{t_0}\Bigr\{a(s)^{-3+\frac{\omega}8}\left(\E^{(L)}(\phi,s)+a(s)^\frac{\omega}4\E^{(L)}(\Sigma,s)\right)\\
&\,\quad+a(s)^{-1-c\sqrt{\epsilon}-(L+1)\omega}\cdot a(s)^{(L+1)\omega}\E^{(\leq L)}_{1,\leq L}(f,s)\\
&\,\quad+a(s)^{-3+\frac{\omega}8-c\sqrt{\epsilon}}\left(\E^{(\leq L-2)}(\phi,s)+a(s)^{\frac{\omega}4}\E^{(\leq L-2)}(\Sigma,s)\right)\\
&\,\quad+\epsilon^\frac78a(s)^{-3-c\sqrt{\epsilon}}\cdot a(s)^{\frac{\omega}2}\E^{(\leq L-4)}(\Ric,s)\Bigr\}\,ds\,,\\
\numberthis\label{eq:en-est-Ric-top-scaled}
a(t)^{4+{\omega}}&\,\E^{(L-1)}(\Ric,t)\lesssim\epsilon^4+\int_t^{t_0}\left(\epsilon^\frac18a(s)^{-3}+a(s)^{-3+\frac{\omega}4}\right)\cdot a(s)^{4+{\omega}}\E^{(L-1)}(\Ric,s)\,ds\\
+\int_t^{t_0}\Bigr\{&\,a(s)^{-3+\frac{\omega}4}\cdot a(s)^{4+\frac{\omega}2}\E^{(L+1)}(\Sigma,s)+\left(a(s)^{-1-c\sqrt{\epsilon}}+\epsilon^2a(s)^{-3+\frac{\omega}2}\right)a(s)^\frac{\omega}4\E^{(L)}(\Sigma,s)\\
&\,\quad+\left(a(s)^{-3+\frac{3\omega}4}+a(s)^{-1-c\sqrt{\epsilon}}\right)\E^{(L)}(\phi,s)\\
&\,\quad+\left(a(s)^{-1-c\sqrt{\epsilon}}+\epsilon^2 a(s)^{-3-c\sqrt{\epsilon}+\frac{\omega}2}\right)\left(\E^{(\leq L-2)}(\phi,s)+a(s)^\frac{\omega}4\E^{(\leq L-2)}(\Sigma,s)\right)\\
&\,\quad+a(s)^{-1-c\sigma-(L+1)\omega}\cdot a(s)^{(L+1)\omega}\E^{(\leq L)}_{1,\leq L}(f,s)\\
&\,\quad\left.+\epsilon a(s)^{-3+\frac{\omega}4}\cdot\,a(s)^{\frac{\omega}2}\E^{(L-2)}(\Ric,s)+\epsilon a(s)^{-3-c\sqrt{\epsilon}+\frac{\omega}4}\,\cdot a(s)^{\frac{\omega}2}\E^{(\leq L-4)}(\Ric,s)\right\}\,ds
\end{align*}}
\begin{align*}
\numberthis\label{eq:en-est-Ric}\E^{(L-2)}(\Ric,t)\lesssim&\,\E^{(L-2)}(\Ric,t_0)+\int_t^{t_0}\left(\epsilon^\frac18a(s)^{-3}+a(s)^{-1-c\sqrt{\epsilon}}\right)\E^{(L-2)}(\Ric,s)\,ds\\
+\int_t^{t_0}\Bigr\{&\,\epsilon^{-\frac18}a(s)^{-3}\left(\E^{(L)}(\phi,s)+\E^{(L)}(\Sigma,s)\right)+a(s)^{-1-c\sqrt{\epsilon}-(L+1)\omega}\cdot a(s)^{(L+1)\omega}\E^{(\leq L)}_{1,\leq L}(f,s)\\
&\,+\epsilon^{-\frac18}a(s)^{-3-c\sqrt{\epsilon}}\left(\E^{(\leq L-2)}(\phi,s)+\E^{(\leq L-2)}(\Sigma,s)\right)\\
&\,+\epsilon^\frac78a(s)^{-3-c\sqrt{\epsilon}}\E^{(\leq L-4)}(\Ric,s)\Bigr\}\,ds\,,\\
\numberthis\label{eq:en-est-Ric-top}
a(t)^{4}&\,\E^{(L-1)}(\Ric,t)\lesssim\epsilon^4+\int_t^{t_0}\change{\left(\epsilon^\frac18a(s)^{-3}+a(s)^{-1-c\sqrt{\epsilon}}\right)}\cdot a(s)^4\E^{(L-1)}(\Ric,s)\,ds\\
+&\,\int_t^{t_0}\left\{\epsilon^{-\frac18}a(s)^{-3}\cdot a(s)^4\E^{(L+1)}(\Sigma,s)+\change{\left(\epsilon^\frac{15}8a(s)^{-3}+a(s)^{-1-c\sqrt{\epsilon}}\right)}\E^{(L)}(\Sigma,s)\right.\\
&\,+\change{\left(\epsilon^{-\frac18}a(s)^{-3}+a(s)^{-1-c\sqrt{\epsilon}}\right)}\E^{(L)}(\phi,s)\\
&\,+\change{\left(\epsilon^\frac{15}8 a(s)^{-3-c\sqrt{\epsilon}}+a(s)^{-1-c\sqrt{\epsilon}}\right)}\left(\E^{(\leq L-2)}(\phi,s)+\E^{(\leq L-2)}(\Sigma,s)\right)\\
&\,+a(s)^{-1-c\sigma-(L+1)\omega}\cdot a(s)^{(L+1)\omega}\E^{(\leq L)}_{1,\leq L}(f,s)\\
&\,\left.+\epsilon a(s)^{-3}\E^{(L-2)}(\Ric,s)+\epsilon a(s)^{-3-c\sqrt{\epsilon}}\E^{(\leq L-4)}(\Ric,s)\right\}\,ds
\end{align*}
\change{For $L=2$, the previous two estimates extend if all energies with negative order are dropped, while one additionally has the following:}
\change{\begin{align*}
\numberthis\label{eq:en-est-Ric0-scaled}a(t)^\frac{\omega}2\E^{(0)}&\,(\Ric,t)\lesssim a(t_0)^\frac{\omega}2\E^{(0)}(\Ric,t_0)+\int_t^{t_0}\left(a(s)^{-3+\frac{\omega}8}+a(s)^{-1-c\sqrt{\epsilon}}\right)\E^{(0)}(\Ric,s)\,ds\\
&\,+\int_t^{t_0}a(s)^{-3+\frac{\omega}8}\left(\E^{(0)}(\phi,s)+a(s)^\frac{\omega}4\E^{(0)}(\Sigma,s)\right)\\
&\,\qquad+a(s)^{-1-c\sqrt{\epsilon}-{\omega}}\left(a(s)^{\omega}\E^{(0)}_{1,0}(f,s)+a(s)^\frac{\omega}2\|G^{-1}-\gamma^{-1}\|_{L^2_G(M_s)}^2\right)\,ds
\end{align*}}
\begin{align*}
\numberthis\label{eq:en-est-Ric0}\E^{(0)}&\,(\Ric,t)\lesssim\E^{(0)}(\Ric,t_0)+\int_t^{t_0}\left(\epsilon^\frac18a(s)^{-3}+a(s)^{-1-c\sqrt{\epsilon}}\right)\E^{(0)}(\Ric,s)\,ds\\
&\,+\int_t^{t_0}\epsilon^{-\frac18}a(s)^{-3}\left(\E^{(0)}(\phi,s)+\E^{(0)}(\Sigma,s)\right)\,ds\\
&\,+\int_t^{t_0}\change{a(s)^{-1-c\sqrt{\epsilon}-\omega}\left(a(s)^\omega\E^{(0)}_{1,0}(f,s)+a(s)^{\frac{\omega}2}\|G^{-1}-\gamma^{-1}\|_{L^2_G(M_s)}^2\right)}\,ds
\end{align*}
Finally, \change{one }has the following bound:
\begin{align*}\numberthis\label{eq:en-est-Ric-toptop}
a^8\E^{(l)}(\Ric,\cdot)\lesssim&\,\E^{(l)}(W,\cdot)+\E^{(l)}(\Sigma,\cdot)+\epsilon a^{-c\sqrt{\epsilon}}\E^{(\leq l-2)}(\Sigma,\cdot)+\change{\E^{(l)}(\phi,\cdot)+\epsilon a^{-c\sqrt{\epsilon}}\E^{(\leq l-2)}(\phi,\cdot)}
\\
&\,+a^{4-c\sqrt{\epsilon}}\E^{(\leq l)}_{1,\leq l}(f,\cdot)\underbrace{+a^{4-c\sqrt{\epsilon}}\|\change{G^{-1}-\gamma^{-1}}\|_{L^2_G}^2}_{\text{if }l=0}+\change{\underbrace{\epsilon\E^{(l-2)}(\Ric,\cdot)+\epsilon a^{-c\sqrt{\epsilon}}\E^{(\leq l-4)}(\Ric,\cdot)}_{\text{if }l\geq 4}}
\end{align*}
\end{subequations}
\end{lemma}
\begin{proof} Since \eqref{eq:REEqG}, \eqref{eq:REEqChr} and \eqref{eq:REEqRic} do not contain Vlasov matter variables, the proofs for \eqref{eq:norm-est-G}, \eqref{eq:chr-norm-est} and \eqref{eq:en-est-Ric}-\eqref{eq:en-est-Ric0} are largely unchanged to those of \cite[Lemma 6.11 - 6.14]{FU23}, where we apply Lemma \ref{lem:lapse-en-est} to estimate lapse energy terms. Regarding the scaled estimates, we consider \eqref{eq:metric-int-est-scaled} for even $l\geq 2$; the other cases as well as \eqref{eq:chr-int-est-scaled} are proven similarly: Commuting \eqref{eq:REEqG} with $\Lap^\frac{l}2$, we have the following:
\begin{align*}
&\,-\del_t\left(\change{a(t)^\frac{\omega}2}\|\change{G^{-1}-\gamma^{-1}}\|_{H^l_G}^2\right)\\
\lesssim&\,a^{-3-c\sqrt{\epsilon}+\change{\frac{\omega}2}}\left(\sqrt{\E^{(\leq l)}(\Sigma,\cdot)}+\sqrt{\E^{(\leq l)}(N,\cdot)}+\sqrt{\E^{(\leq l-2)}(\Ric,\cdot)}\right)\|\change{G^{-1}-\gamma^{-1}}\|_{H^l_G}
\end{align*}
We apply the Young inequality to, for example, the shear term as follows:
\[a^{-3-c\sqrt{\epsilon}+\change{\frac{\omega}2}}\sqrt{\E^{(\leq l)}(\Sigma,\cdot)}\|\change{G^{-1}-\gamma^{-1}}\|_{H^l_G}\lesssim a^{-3+\change{\frac{\omega}8}}\cdot \change{a^\frac{\omega}2}\|\change{G^{-1}-\gamma^{-1}}\|_{H^l_G}^2 +a^{-3+\change{\frac{\omega}8-c\sqrt{\epsilon}}}\,\cdot a^\frac{\omega}4\E^{(\leq l)}(\Sigma,\cdot)\]
The statement then follows by inserting Lemma \ref{lem:lapse-en-est}, applying the Young inequality similarly to the other terms and integrating. \change{From comparing \eqref{eq:REEqG-1} to \eqref{eq:REEqG}, one quickly observes that the analogous bound for the inverse metric follows identically.}\\
Finally, \eqref{eq:en-est-Ric-toptop} is a direct consequence of applying $\Lap^\frac{l}2$ (resp. $\nabla\Lap^\frac{l-1}2$) to the constraint equation \eqref{eq:REEqConstrE}, rearranging and applying the a priori estimates \eqref{eq:APSigma}, \eqref{eq:APmidSigma}, \eqref{eq:APmidphi}, \eqref{eq:APPsi} and \eqref{eq:APmidPsi}. In particular, regarding the Vlasov terms, one uses \eqref{eq:stress-control} and \eqref{eq:density-control-L2} to connect their norms to our energy quantities. \change{For the remaining terms, we use estimates as in \cite[Lemma 4.5]{FU23}, along with the previously mentioned a priori estimates, to convert Sobolev norms into energies up to curvature errors.}
\end{proof}

\section{Vlasov energy estimates}\label{sec:vlasov}

In this section, we collect the necessary energy estimates to obtain the bootstrap improvements in the following section. In particular, for each derivative order and each number of horizontal derivatives in the energy, we prove both a time-scaled estimate, which are all combined in Lemma \ref{lem:vlasov-total-scaled} to obtain bootstrap improvements for the spacetime and scalar field variables, as well as non-scaled estimates that lose smallness, which we use to improve the bootstrap assumptions on Vlasov matter. Regarding the time-scaled estimates, we scale by $a^\omega$ and powers there of to indicate scaling necessary within the hierarchy of vertical and horizontal derivatives at a given order (foreshadowing \eqref{eq:def-total-vlasov}), and by $a^\beta$ to indicate additional scaling necessary for the total energy at odd derivative orders (foreshadowing the final line of \eqref{eq:def-total-en}).
\subsection{Preliminaries}

First, we collect some tools that will be used throughout:

\begin{remark}[Integrating $M_t$-tangent tensors over the tangent space]
We recall that by Lemma \ref{lem:APMom}, positive powers of $\lvert v\rvert_G$ can be bounded by $a^{-c\sqrt{\epsilon}}$, and thus integrals over powers of $\lvert v\rvert_G$ and $v^0$ on the support of $f$, $f_{FLRW}$ or $f-f_{FLRW}$ are bounded by $a^{-c\change{\sqrt{\epsilon}}}$. In particular, for any $M_t$-tangent tensor $\mathfrak{T}$ and any $\alpha>-3$, this implies
\[\int_{\change{T^\ast M_t}}\lvert \mathfrak{T}\rvert_G^2\cdot \langle v\rangle_G^{2\alpha}\,\vol{\G}=\int_{M_t}\lvert \mathfrak{T}\rvert_G^2\int_{\change{T^\ast_xM_t}}\left(\langle v\rangle_G^{2\alpha}\change{\mu_G^{-1}}dv\right)\vol{G}\lesssim \|\mathfrak{T}\|_{L^2_G(M_t)}^2\cdot a^{-c\sqrt{\epsilon}}\,.\]
We use this in the sequel without further comment, usually when needing to estimate high order terms of $G-\gamma,\,\Sigma,\,N$ or scalar field quantities.
\end{remark}

Next, we collect the core estumate that ensures the transport equation almost conserves Vlasov energies:

\begin{lemma}[Energy mechanism for the Vlasov equation]\label{lem:vlasov-energy-mech} Let $\mu\geq 0$ and $0\leq L\leq 19$. Then, the following holds:
\begin{align*}
\int_{\change{T^\ast M_t}}\langle v\rangle^{2\mu}\langle\X\nabsak_{vert}^{L}(f-f_{FLRW}),\nabsak_{vert}^{L}(f-f_{FLRW})\rangle_{\G_0}\vol{\G}\lesssim&\,\epsilon a^{1-c\sigma}\E^{(L)}_{\mu,0}(f,\cdot)
\end{align*}
Similarly, if one additionally has $1\leq K\leq L$, we have
\begin{align*}
\int_{\change{T^\ast M_t}}\langle v\rangle^{2\mu}\langle\X\nabsak_{vert}^{L-K}\nabsak_{hor}^Kf,\nabsak_{vert}^{L-K}\nabsak_{hor}^Kf\rangle_{\G_0}\vol{\G}\lesssim&\,\change{\epsilon a^{1-c\sigma}\E^{(L)}_{\mu,K}(f,\cdot)}\,.
\end{align*}
\end{lemma}
\begin{proof}
Both estimates are proved identically, so we only write out the proof for the second. Writing out $\bm{X}$ (see \eqref{eq:def-X}), the left hand side takes the following form:
\begin{align*}
\change{\int_{\change{T^\ast M_t}}\langle v\rangle_G^{2\mu}(v^0)^{2(L-K)}\left\langle\left[-a^{-1}\frac{v^{\sharp j}}{v^0}(N+1)\nabsak_j+a^{-1}v^0\nabla_jN\,\B^j\right]\nabsak_{vert}^{L-K}\nabsak_{hor}^Kf,\nabsak_{vert}^{L-K}\nabsak_{hor}^Kf\right\rangle_{\G}\,\vol{\G}}
\end{align*}
Regarding the first term, we integrate by parts with respect to $\nabsak_j$, noting that the horizontal covariant derivative of any momentum weight vanishes. For the remaining \change{term}, we use Fubini's theorem to replace the integral by one over $\int_M\int_{\change{T^\ast_xM_t}}$ and then integrate by parts in $\B$. This leads to the following after rearranging and applying \eqref{eq:mom-mom}:
\change{\begin{align*}
\int_{T^\ast M_t}&\,\langle v\rangle_G^{2\mu}\langle\X\nabsak_{vert}^{L}(f-f_{FLRW}),\nabsak_{vert}^{L}(f-f_{FLRW})\rangle_{\G_0}\vol{\G}\\
=\int_{T^\ast M_t}&\,a^{-1}\langle v\rangle^{2\mu}_G\,\left(v^{\sharp j}\nabla_jN\right)\,v^0\left(2\mu\langle v\rangle_G^{-2}+2(L-K)(v^0)^{-2}\right)\lvert \nabsak_{vert}^{L-K}\nabsak_{hor}^Kf\rvert_{\G_0}^2\,\vol{\G}
\end{align*}
The statement then follows using \eqref{eq:BsN}.}
\end{proof}

To be able to apply the above identity, we need to commute $\del_t$ past the respective covariant Sasaki derivatives, apply \eqref{eq:vlasov-resc}, and then commute $\bm{X}$ back past them. We have already collected error term formulas for the former in Lemma \ref{lem:sasaki-com-gen}, and now collect schematic formulas for the latter:

\begin{lemma}[The commuted rescaled Vlasov equation]\label{lem:vlasov-comm} The following commuted analogues of \eqref{eq:vlasov-resc} hold: For $1\leq L\leq 19$, one has
\begin{subequations}
\change{\begin{align*}
\numberthis\label{eq:vlasov-L0}\nabsak_{vert}^L\del_tf= &\,\X \nabsak_{vert}^L(f-f_{FLRW})+a^{-1}(N+1)\frac{(v^0)^2+v\ast v}{(v^0)^3}\ast\nabsak_{vert}^{L-1}\nabsak_{hor}(f-f_{FLRW})\\
&\,+a^{-1}(N+1)\frac{v\ast v}{v^0}\ast\Ric[G]\ast\nabsak_{vert}^{L-1}\nabsak_{hor}(f-f_{FLRW})\\
&\,+a^{-1}\frac{v}{v^0}\ast\nabla N\ast \nabsak_{vert}^{L}(f-f_{FLRW})\\
&\,+\mathfrak{F}_{L,0}+\nabsak_{vert}^L\X f_{FLRW}
\end{align*}}
where one has, for $L\geq 2$,
\begin{equation}\label{eq:vlasov-L0-junk}
\|(v^0)^L\mathfrak{F}_{L,0}\|_{L^2_{1,\G}(\change{T^\ast M)}}\lesssim a^{-1-c\sigma}\left(\sqrt{\E^{(\leq L-1)}_{1,\leq 1}(f,\cdot)}+\|\Gamma-\Gamhat\|_{L^2_G(M)}+\|\change{G^{-1}-\gamma^{-1}}\|_{H^1_G(M)}\right)
\end{equation}
\end{subequations}
and said error term vanishes for $L=1$, as well as
\begin{align*}
\numberthis\label{eq:vlasov-L0-inhom}\|(v^0)^L\nabsak_{vert}^L&\X f_{FLRW}\|_{L^2_{1,\G}(\change{T^\ast M})}\lesssim a^{-3-c\sqrt{\epsilon}}\change{\left(\sqrt{a^4\E^{(1)}(N,\cdot)}+\sqrt{\E^{(0)}(N,\cdot)}\right)}\\
&\,+a^{-1-c\sqrt{\epsilon}}\left(\|\Gamma-\Gamhat\|_{L^2_G(M)}+\|\change{G^{-1}-\gamma^{-1}}\|_{L^2_G(M)}\right)\,.
\end{align*}
The latter bound also extends to $L=0$. \\
Furthermore, for $0<K<L$, we have
\begin{subequations}
\change{\begin{align*}
\numberthis\label{eq:vlasov-LM}\nabsak_{vert}^{L-K}\nabsak_{hor}^K\del_tf=&\,\X\nabsak_{vert}^{L-K}\nabsak_{hor}^Kf+a^{-1}(N+1)\nabla^{K-1}\Ric[G]\ast\nabsak_{vert}^{L-K+1}\left(\frac{v\ast v}{v^0}f\right)\\
&\,+a^{-1}(N+1)\left(\frac{1}{v^0}+\frac{v\ast v}{(v^0)^3}\right)\ast\nabsak_{vert}^{L-K-1}\nabsak_{hor}^{K+1}f\\
&+a^{-1}\frac{v}{v^0}\ast\nabla N\ast\B\nabsak_{vert}^{L-K-1}\nabsak_{hor}^Kf
\phantom{\nabsak_{vert}^{L-K}\nabsak_{hor}^K\del_tf}\\
&\,+\mathfrak{F}_{L,K}
\end{align*}}
with
\begin{align*}
\numberthis\label{eq:vlasov-LM-junk}\|(v^0)^{L-K}&\mathfrak{F}_{L,K}\|_{L^2_{1,\G}(\change{T^\ast M})}\lesssim\change{a^{-1-c\sigma}}\sqrt{\E^{(\leq L-1)}_{1,\leq K}(f,\cdot)}\change{+a^{-1-c\sqrt{\epsilon}}\sqrt{\E^{(\leq K+1)}(N,\cdot)}}\\
&\,+\underbrace{\sqrt{\epsilon}a^{-1-c\sqrt{\epsilon}}\left(\sqrt{\E^{(\leq K-1)}(\Ric,\cdot)}+\|\Gamma-\Gamhat\|_{H^1_G(M)}+\|\change{G^{-1}-\gamma^{-1}}\|_{H^1_G(M)}\right)}_{\text{if }K\geq 2}
\end{align*}
\end{subequations}
as well as
\begin{subequations}
\change{\begin{align*}
\numberthis\label{eq:vlasov-LL}\nabsak_{hor}^L\del_tf=&\,\bm{X}\nabsak_{hor}^Lf+a^{-1}(N+1)\nabla^{L-1}\Ric[G]\ast\nabsak_{vert}\left(\frac{v\ast v}{v^0}f\right)+a^{-1}\frac{v}{v^0}\nabla N\ast\nabsak_{hor}^Lf+\mathfrak{F}_{L,L}
\end{align*}}
with
\begin{align*}
\numberthis\label{eq:vlasov-LL-junk}\|&\mathfrak{F}_{L,L}\|_{L^2_{1,\G}(\change{T^\ast M})}\lesssim\epsilon a^{-1-c\sigma}\sqrt{\E^{(L)}_{1,\leq L-1}(f,\cdot)}+\change{a^{-1-c\sigma}}\sqrt{\E^{(\leq L-1)}_{1,\leq L-1}(f,\cdot)}\\
&\,\change{+a^{-3-c\sqrt{\epsilon}}\sqrt{a^4\E^{(L+1)}(N,\cdot)}+a^{-1-c\sqrt{\epsilon}}\sqrt{\E^{(\leq L)}(N,\cdot)}}\\
&\,+\underbrace{\change{\sqrt{\epsilon}a^{-1-c\sqrt{\epsilon}}}\left(\sqrt{\E^{(\leq L-2)}(\Ric,\cdot)}+\|\Gamma-\Gamhat\|_{H^1_G(M)}+\|\change{G^{-1}-\gamma^{-1}}\|_{H^1_G(M)}\right)}_{\text{ if }L\geq 2}.
\end{align*}
\end{subequations}

\end{lemma}

\begin{proof} We first consider the purely vertical case \eqref{eq:vlasov-L0}, then the purely horizontal case \eqref{eq:vlasov-LL}, and then briefly \eqref{eq:vlasov-LM} as a combination of both cases before discussing the respective error bounds.\\

For \eqref{eq:vlasov-L0}, we first replace $f$ with $(f-f_{FLRW})+f_{FLRW}$, which allows us to collect all reference error terms in $\nabsak_{vert}^L\X f_{FLRW}$. Since $\B_i$ and $\nabsak_{i+3}$ coincide when acting on purely vertical tensors (see \eqref{eq:conn-coeff}) and vertical frame derivatives commute, the only error terms generated \change{from the vertical term }in computing \eqref{eq:vlasov-L0} are down to applying the product rule to momentum prefactors. \change{Regarding the horizontal term in $\X$, }we additionally need to apply \eqref{eq:com-sasaki-gen-hor} for $M_2=0$, which yields the term in the second line of \eqref{eq:vlasov-L0} as the leading error term, while all others are absorbed into the error terms. In particular, we note that this commutation only generates terms of the type in the second line of \eqref{eq:vlasov-L0}, where we can always pull the horizontal derivative back to the front up to error terms as in Lemma \ref{lem:vlasov-rearrange} that are absorbed into $\mathfrak{F}_{L,0}$. We also note that, whenever vertical derivatives hit $v$-prefactors, both a power of $v$ and a derivative are lost, and thus the momentum scaling of lower order error terms always precisely matches that of the respective $\G_0$-norm, up to estimating $v$ by $v^0$ in $\lvert \cdot \rvert_G$.\\

For \eqref{eq:vlasov-LL}, we need to ensure that, after taking horizontal derivatives of $\X f$ and applying the product rule, we can still express all terms covariantly despite Christoffel symbols occuring explicitly in \change{the commutator $[\A_i,\B^j]f=-\Gamma[G]^j_{il}\B^lf$, see \eqref{eq:AB-comm}}. To this end, we \change{commute horizontal derivatives past the vertical part of the operator $\X$ as follows.
\begin{align*}
\nabsak_i\left[v^0\nabla_jN\B^jf\right]=&\,v^0\A_j\left[\nabla_jN\B^jf\right]=\,v^0\left[\del_i\nabla_jN\,\B^jf+\nabla_jN\,\A_i\B^jf\right]\\
=&\,v^0\left[\del_i\nabla_jN\,\B^jf+\nabla_jN\,\left(\B^j\A_if-\Gamma[G]^j_{il}\B^lf\right)\right]\\
=&\,v^0\left[\nabla_i\nabla_jN\,\B^jf+\nabla_jN\B^j\nabsak_if\right]
\end{align*}
At higher orders, one additionally incurs terms of the form $v\ast\Riem[G]$ from connection coefficients with two horizontal indices. Thus, repeating this argument, we obtain the following schematic representation:
\begin{align*}
&\,[\nabsak_{hor}^L,v^0\nabla_jN\,\B^jf]f\\
=&\,\sum_{K_N+K_f=L,\,K_f\neq L}v\ast\nabla^{K_N+1}N\ast\B\nabsak_{hor}^{K_f}f\\
&\,+\sum_{K_\Ric+K_N+K_f=L-2}v\ast\nabla^{K_\Ric}\Ric[G]\ast\nabla^{K_N+1}N\ast\B\nabsak_{hor}^{K_f}f\\
&\,+\langle \text{low order nonlinear error terms with multiple curvature terms}\rangle
\end{align*}}
\change{Regarding }the horizontal term in $\bm{X}f$, we apply \eqref{eq:com-sasaki-gen-hor} with $M_1=0, M_2=L$, again noting that we can rearrange the covariant derivatives up to error terms into the order presented, and can deal with the error terms as in Lemma \ref{lem:vlasov-rearrange}. Finally, \eqref{eq:vlasov-LM} follows from combining the ideas from both of these proofs, commuting with horizontal derivative terms as in the purely vertical case and noting that $\B$ and $\nabsak_{vert}^{L-K}$ commute.\\

The error bounds \eqref{eq:vlasov-LL-junk}, \eqref{eq:vlasov-LM-junk} and \eqref{eq:vlasov-L0-junk} simply follow $L^2-L^\infty-$estimating the remaining terms and using \change{\eqref{eq:APVlasov} }for low order Vlasov terms, \eqref{eq:APmidRic} for the curvature terms and \eqref{eq:APMom} to control momentum terms as well as the lapse bootstrap assumption from \eqref{eq:BsC}. In particular, we note that $\B f$ and higher order vertical derivatives of it are only bounded by $a^{-c\sqrt{\epsilon}}$, leading to the \change{lapse }energies in the second line of \eqref{eq:vlasov-LL-junk}. Furthermore, we apply \eqref{eq:vlasov-rearrange-no-ref-K} to order the derivatives in the Vlasov error terms as needed, and note that the leading curvature terms arise from this procedure.\\

Finally, for \eqref{eq:vlasov-L0-inhom}, we recall that $v^0$, $\langle v\rangle_G$ and $\B f_{FLRW}$ are bounded by $a^{-c\sqrt{\epsilon}}$ on the support of $f_{FLRW}$ since the latter has bounded support with respect to $\gamma$ and one has \eqref{eq:APmidG}. Using \eqref{eq:ref-hor-formula} for the horizontal term in $\bm{X}$ then directly leads to the stated bound.
\end{proof}

\subsection{Vlasov energy estimates}

Recalling our conventions on $t$ and $\omega$ from the start of this section, we now start by proving the Vlasov integral energy estimate at order $0$:

\begin{lemma}[Zero order Vlasov energy estimates]\label{lem:en-est-Vlasov} The following estimate holds:
\begin{subequations}
\begin{align*}
\numberthis\label{eq:en-est-Vlasov-0-scaled}a(t)^\omega\E^{(0)}_{1,0}(f,t)\lesssim&\,\epsilon^4+\int_t^{t_0}\left(a(s)^{-3+\change{\frac{\omega}4}}+\epsilon\,a^{-3}\right)\cdot a(s)^\omega\E^{(0)}_{1,0}(f,s)\,ds\\
&\,+\int_t^{t_0}a(s)^{-3-c\sqrt{\epsilon}+\change{\frac{\omega}4}}\left(\change{\epsilon^2\E^{(0)}(\Sigma,s)}+\E^{(0)}(\phi,s)+\change{a(s)^{4+\frac{\omega}2}}\|\Gamma-\Gamhat\|_{L^2_G(M_s)}^2\right)\,ds\\
&\,+\int_t^{t_0}\change{a(s)^{-1-c\sqrt{\epsilon}+\change{\frac{\omega}4}}\cdot a(s)^{\frac{\omega}{2}}}\|\change{G^{-1}-\gamma^{-1}}\|_{L_G^2(M_s)}^2\,ds\\
\numberthis\label{eq:en-est-Vlasov-0}\E^{(0)}_{1,0}(f,t)\lesssim&\,\epsilon^4+\int_t^{t_0}\left(a(s)^{-1-c\sigma}+\epsilon^\frac18a(s)^{-3}\right)\E^{(0)}_{1,0}(f,s)\,ds\\
&\,+\int_t^{t_0}\epsilon^{-\frac18}a(s)^{-3-c\sqrt{\epsilon}}\left(\change{\epsilon^2\E^{(0)}(\Sigma,s)}+\E^{(0)}(\phi,s)\right)\,ds\\
&\,+\int_t^{t_0}a(s)^{-1-c\sqrt{\epsilon}}\left(\|\Gamma-\Gamhat\|_{L_G^2(M_s)}^2+\|\change{G^{-1}-\gamma^{-1}}\|_{L_G^2(M_s)}^2\right)\,ds
\end{align*}
\end{subequations}
\end{lemma}
\begin{proof} We first prove the unscaled estimate \eqref{eq:en-est-Vlasov-0} before sketching the necessary adaptations to arrive at \eqref{eq:en-est-Vlasov-0-scaled}:
\begin{align*}\numberthis\label{eq:VlasovL2-step1}
-\del_t\E^{(0)}_{1,0}(f,\cdot)=&\,\change{-\int_{\change{T^\ast M}} 2\frac{\del_t\langle v\rangle_G}{\langle v\rangle_G}\,\langle v\rangle^{2}_G\lvert f-f_{FLRW}\rvert^2\,\vol{\G}}\\
&\,-2\int_{\change{T^\ast M}}\langle v\rangle_G^{2}(f-f_{FLRW})\cdot\del_tf\,\vol{\G}
\end{align*}
The first line can be bounded by $\lesssim\epsilon a^{-3}\E^{(0)}_{1,0}(f,\cdot)$ with \eqref{eq:deltv0} and \eqref{eq:BsC}. We insert the rescaled Vlasov equation \eqref{eq:vlasov-resc} in the second line and expand as follows:
\begin{subequations}
\begin{align*}
&\,-2\int_{\change{T^\ast M}}\langle v\rangle_G^{2}(f-f_{FLRW})\cdot\del_tf\,\vol{\G}\\
\numberthis\label{eq:Vlasov-en-est-0-1}=&\,-2\int_{\change{T^\ast M}}\langle v\rangle_G^{2}(f-f_{FLRW})\cdot\X(f-f_{FLRW})\,\vol{\G}\\
\numberthis\label{eq:Vlasov-en-est-0-2}&\,-2\int_{\change{T^\ast M}}\langle v\rangle_G^{2}(f-f_{FLRW})\X f_{FLRW}
\end{align*}
\end{subequations}
\eqref{eq:Vlasov-en-est-0-1} is bounded by $\epsilon a^{-3}\E^{(0)}_{1,0}(f,\cdot)$ due to core energy mechanism from Lemma \ref{lem:vlasov-energy-mech}. Regarding \eqref{eq:Vlasov-en-est-0-2}, we apply Lemma \ref{lem:lapse-en-est} for \change{$L=0$ to \eqref{eq:vlasov-L0-inhom} and collect}:
\begin{align*}
\|\X f_{FLRW}\|_{L^2_{1,\G}(\change{T^\ast M})}\lesssim&\,a^{-1-c\sqrt{\epsilon}}\left[\|\change{G^{-1}-\gamma^{-1}}\|_{L^2_G(M)}+\sqrt{\E^{(0)}_{1,0}(f,\cdot)}\right]\\
&\,+a^{-3-c\sqrt{\epsilon}}\left[\sqrt{\change{\epsilon^2\E^{(0)}(\Sigma,\cdot)}}+\sqrt{\E^{(0)}(\phi,\cdot)}+a^4\|\Gamma-\Gamhat\|_{L^2_G(M)}\right]
\end{align*}
Putting all of the above together and applying the Young inequality on each summand, we obtain
\begin{align*}
-\del_t\E^{(0)}_{1,0}(f,\cdot)\lesssim&\,\left(a^{-1-c\sigma}+\epsilon^\frac18a^{-3}\right)\E^{(0)}_{1,0}(f,\cdot)+\epsilon^{-\frac18}a^{-3-c\sqrt{\epsilon}}\left(\change{\epsilon^2\E^{(0)}(\Sigma,\cdot)}+\E^{(0)}(\phi,\cdot)+a^4\|\Gamma-\Gamhat\|_{L^2_G}^2\right)\\
&\,+a^{-1-c\sqrt{\epsilon}}\|\change{G^{-1}-\gamma^{-1}}\|_{L^2_G}^2
\end{align*}
and thus the result after applying the Gronwall lemma. \\

For \eqref{eq:en-est-Vlasov-0-scaled}, we argue along the same lines, with the estimates for \eqref{eq:VlasovL2-step1} and \eqref{eq:Vlasov-en-est-0-1} treated identically and where $(-\del_ta^\omega)\E^{(0)}(f,\cdot)$ can be dropped since it is nonpositive. Regarding the terms in \eqref{eq:Vlasov-en-est-0-2}, we apply the same estimates, but now distribute the additional Vlasov energy scaling to mitigate the divergence of the other energies:
\begin{align*}
&\,a^{-3}\cdot\left(a^\frac{\omega}4\|\bm{X}f_{FLRW}\|_{L^2_{\G}}\right)\cdot \left(a^{\frac{3\omega}4}\sqrt{\E_{1,0}^{(0)}(f,\cdot)}\right)\,\\
\lesssim&\,\change{a^{-3+\frac{\omega}4}}\left(a^{\omega}\E^{(0)}_{1,0}(f,\cdot)\right)+\change{a^{-3-c\sqrt{\epsilon}+\frac{\omega}4}\left(\epsilon^2\E^{(0)}(\Sigma,\cdot)+\E^{(0)}(\phi,\cdot)+a^{4+\frac{\omega}2}\|\Gamma-\Gamhat\|^2_{L^2_G}\right)}\\
&\,+\change{a^{-1-c\sqrt{\epsilon}+\frac{\omega}4}\cdot a^\frac{\omega}2}\|\change{G^{-1}-\gamma^{-1}}\|_{L^2_G}^2\,.
\end{align*}
\end{proof}

\begin{lemma}[Vertical Vlasov energy estimates]\label{lem:vlasov-vertical}
\change{Let $\beta\in\{0,4\}$. Then:
\begin{align*}
\numberthis\label{eq:vlasov-vertical-en}
a(t)^{\beta+\omega}&\,\E^{(L)}_{1,0}(f,t)\lesssim\epsilon^4+\int_t^{t_0}\left(a(s)^{-1-c\sigma}+a(s)^{-3+\frac{\omega}{4}}+\epsilon a(s)^{-3}\right)a(s)^{\beta+{\omega}}\E^{(L)}_{1,0}(f,s)\,ds\\
&\,+\int_t^{t_0}a(s)^{-1-c\sigma-{\omega}}\cdot a(s)^{\beta+2{\omega}}\E_{1,\leq 1}^{(\leq L)}(f,s)\,ds\\
&\,+\int_t^{t_0}a(s)^{-3-c\sqrt{\epsilon}+\beta+\change{\frac{\omega}2}}\left(\E^{(0)}(\phi,s)+\change{\epsilon^2a(s)^\frac{\omega}4\E^{(0)}(\Sigma,s)}\right)\,ds\\
&\,+\int_t^{t_0}a(s)^{-1-c\sqrt{\epsilon}+\beta+\frac{\omega}4}\left(a(s)^\frac{\omega}{2}\|\Gamma-\Gamhat\|_{L^2_G(M_s)}^2+a(s)^\frac{\omega}{2}\|G^{-1}-\gamma^{-1}\|_{L^2_G(M_s)}^2\right)\,ds\\
\numberthis\label{eq:vlasov-vertical-en-unscaled}\E^{(L)}_{1,0}(f,t)\lesssim&\,\epsilon^4+\int_t^{t_0}\left(a(s)^{-1-c\sigma}+\epsilon^\frac18 a(s)^{-3}\right)\E^{(L)}_{1,0}(f,s)\,ds\\
&\,+\int_t^{t_0}a(s)^{-1-c\sigma-2\omega}\cdot a(s)^{2\omega}\E_{1,\leq 1}^{(\leq L)}(f,s)\,ds\\
&\,+\int_t^{t_0}\epsilon^{-\frac18}a(s)^{-3-c\sqrt{\epsilon}+\beta}\left(\E^{(0)}(\phi,s)+\epsilon^2\E^{(0)}(\Sigma,s)\right)\,ds\\
&\,+\int_t^{t_0}a(s)^{-1-c\sqrt{\epsilon}}\left(\|\Gamma-\Gamhat\|_{L^2_G(M_s)}^2+\|\change{G^{-1}-\gamma^{-1}}\|_{L^2_G(M_s)}^2\right)\,ds
\end{align*}}
\end{lemma}
\begin{proof}
We only prove \change{\eqref{eq:vlasov-vertical-en-unscaled} }in full since the other cases are proven similarly, keeping in mind that $-\del_t(a^{\beta+\change{\omega}})\E^{(L)}_{1,0}\leq 0$.\\
Since $\del_t$ and $\B_i$ commute and since $\nabsak_{vert}^L\xi=\B^L\xi$ holds, we get the following:
\begin{align*}
-\del_t\E^{(L)}_{1,0}(f,\cdot)=&\,-\int_{\change{T^\ast M}} \langle v\rangle_G^{2}\change{\left(2\frac{\del_t\langle v\rangle_G}{\langle v\rangle_G}+2L\frac{\del_tv^0}{v^0}\right)}\lvert\nabsak_{vert}^L(f-f_{FLRW})\rvert^2_{\G_0}\vol{\G}\\
&\,-\int_{\change{T^\ast M}}\langle v\rangle_G^{2}(v^0)^{2L}\del_tG\ast_{\G}\nabsak_{vert}^L(f-f_{FLRW})\ast_{\G}^L\nabsak_{vert}(f-f_{FLRW})\,\vol{\G}\\
&\,-2\int_{\change{T^\ast M}}\langle v\rangle_G^2\langle\nabsak_{vert}(f-f_{FLRW}),\nabsak_{vert}\del_tf\rangle_{\G_0}\,\vol{\G}
\end{align*}
The first two lines are bounded by $\epsilon a^{-3}\E^{(1)}_{1,0}(f,\cdot)$ up to constant, with the first line due to \eqref{eq:deltv0} and \eqref{eq:BsN}, and with the second since $\lvert \del_tG\rvert_G\lesssim \epsilon a^{-3}$ holds by \eqref{eq:APSigma} and \eqref{eq:BsN} applied to \eqref{eq:REEqG}. Regarding the final line, we insert \eqref{eq:vlasov-L0}: 
\begin{subequations}
\begin{align*}
&\,-2\int_{\change{T^\ast M}}\langle v\rangle_G^2\langle\nabsak_{vert}^L(f-f_{FLRW}),\nabsak_{vert}^L\del_tf\rangle_{\G_0}\,\vol{\G}\\
\lesssim&\,-2\int_{\change{T^\ast M}}\langle v\rangle_G^2(v^0)^{2L}\langle\nabsak_{vert}^L(f-f_{FLRW}),\nabsak_{vert}^L\X f_{FLRW})\rangle_{\G}\vol{\G}\numberthis\label{eq:vlasov-vert-en1}\\
&+\int_{\change{T^\ast M}}\langle v\rangle_G^2\change{\left[a^{-1}\lvert N+1\rvert\left(1+\lvert v\rvert_G^2\lvert\Ric[G]\rvert_G\right)\lvert \nabsak_{vert}^{L-1}\nabsak_{hor}(f-f_{FLRW})\rvert_{\G_0}\right.}\numberthis\label{eq:vlasov-vert-en2}\\
&\quad\left.+\change{a^{-1}\lvert \nabla N\rvert_G}\lvert \nabsak_{vert}^L(f-f_{FLRW})\rvert_{\G_0}\right]\lvert \nabsak_{vert}^L(f-f_{FLRW})\rvert_{\G_0}\,\vol{\G}\numberthis\label{eq:vlasov-vert-en3}\\
&\,+\left(\|\mathfrak{F}_{L,0}\|_{L^2_{1,\G}}+\|\nabsak_{vert}^L\bm{X}f_{FLRW}\|_{L^2_{1,\G}}\right)\sqrt{\E^{(L)}_{1,0}(f,\cdot)}\numberthis\label{eq:vlasov-vert-en4}
\end{align*}
\end{subequations}

Going through the right hand side line by line, \eqref{eq:vlasov-vert-en1} can be bounded by $\epsilon a^{-3}\E^{(L)}_{1,0}(f,\cdot)$ by Lemma \ref{lem:vlasov-energy-mech}. 
For \eqref{eq:vlasov-vert-en2}, we apply \eqref{eq:APmidRic} and \eqref{eq:APMom} on the prefactor to estimate it as follows:
\[\lvert \eqref{eq:vlasov-vert-en2}\rvert\lesssim a^{-1-c\sqrt{\epsilon}}\left(\sqrt{\E^{(L)}_{1,1}(f,\cdot)}+\|\nabsak_{vert}^{L-1}\nabsak_{hor}f_{FLRW}\|_{\change{L^2_{1,\G_0}}}\right)\sqrt{\E^{(L)}_{1,0}(f,\cdot)}\]
For the reference term, we apply \eqref{eq:hor-deriv-ref-mixed} for $K=1$, and obtain the bound
\[\lvert \eqref{eq:vlasov-vert-en2}\rvert\lesssim a^{-1-c\sqrt{\epsilon}}\E^{(L)}_{1,1}(f,\cdot)+a^{-1-c\sqrt{\epsilon}}\left(\|\Gamma-\Gamhat\|_{L^{2}_G(M_t)}^2+\|\change{G^{-1}-\gamma^{-1}}\|_{L^2_G(M_t)}^2\right)+a^{-1-c\sqrt{\epsilon}}\E^{(L)}_{1,0}(f,\cdot)\,.\]
For \eqref{eq:vlasov-vert-en3}, we apply \eqref{eq:APMom} \change{and \eqref{eq:BsN} }to bound the terms by $\epsilon a^{1-c\sigma}\E^{(L)}_{1,0}$. Finally, we apply \eqref{eq:vlasov-L0-junk} and \eqref{eq:vlasov-L0-inhom} to bound \eqref{eq:vlasov-vert-en4}, where we insert Lemma \ref{lem:lapse-en-est} at order $l=0$ for the occuring lapse energy terms and then apply the Young inequality. The integral estimate \eqref{eq:vlasov-vertical-en} then follows by integrating.\\
For $\omega^\prime=\omega$, we again use the additional scaling to prevent otherwise strongly divergent energy terms to diverge more than $a^{-3}$ as in the zero order case. For the Vlasov energy gaining a horizontal derivative, we distribute the scaling as follows:
\[a^{\omega}\sqrt{\E^{(L)}_{1,0}(f,\cdot)}\sqrt{\E^{(L)}_{1,1}(f,\cdot)}=\sqrt{a^{\omega}\E^{(L)}_{1,0}(f,\cdot)}\cdot\sqrt{a^{-\omega}\cdot a^{2\omega}\E^{(L)}_{1,1}(f,\cdot)}\]
The respective proofs with $\beta=4$ instead of $\beta=0$ are identical.
\end{proof}

\begin{lemma}[Mixed Vlasov energy estimates]\label{lem:vlasov-interim} Let $0<K<L$ and $\beta\in\{0,4\}$. Then, the following estimates hold:
\begin{align*}
\numberthis\label{eq:vlasov-interim-scaled}a(t)^{\beta+(K+1){\omega}}&\,\E^{(L)}_{1,K}(f,t)\lesssim\,\epsilon^4+\int_t^{t_0}\left(a(s)^{-3+\change{\frac{\omega}{4}}}+\epsilon^\frac18 a(s)^{-3}\right)a(s)^{\beta+(K+1){\omega}}\E^{(L)}_{1,K}(f,s)\,ds\\
&\,+\int_t^{t_0}a(s)^{-1-{\omega}}\cdot a(s)^{\beta+(K+2){\omega}}\E^{(L)}_{1,K+1}(f,s)\,ds\\
&\,+\int_t^{t_0}\left(a(s)^{-1-c\sqrt{\epsilon}+\frac{\omega}{2}}+\epsilon^2 a(s)^{-3-c\sqrt{\epsilon}+\frac{\omega}{2}}\right)\cdot a(s)^{\beta+((K-1)+1){\omega}}\E^{(\leq L)}_{1,\leq K-1}(f,s)\\
&\,+\int_t^{t_0}a(s)^{-3-c\sqrt{\epsilon}+\frac{\omega}{2}}a(s)^\beta\left(\change{a(s)^\frac{\omega}2}\E^{(\leq K)}(\Sigma,s)+\E^{(\leq K)}(\phi,s)\right)\,ds\\
&\,+\int_t^{t_0}\left(a(s)^{-1-c\sqrt{\epsilon}+\change{K{\omega}}}+\epsilon a(s)^{-3-c\sqrt{\epsilon}+\change{K{\omega}}}\right)\cdot a(s)^{\beta+\frac{\omega}2}\E^{(\leq K-1)}(\Ric,s)\,ds\\
&\,\change{+\int_t^{t_0}\left(a(s)^{-1-c\sqrt{\epsilon}+\frac{(2K+1){\omega}}2}+\epsilon a(s)^{-3-c\sqrt{\epsilon}+\change{K\omega}}\right)\cdot}\\
&\,\change{\phantom{\int_t^{t_0}}\quad\cdot\left(a(s)^{\beta+\frac{\omega}2}\|\Gamma-\Gamhat\|_{H^1_G(M_s)}^2+a(s)^{\beta+\frac{\omega}{2}}\|G^{-1}-\gamma^{-1}\|_{L^2_G(M_s)}^2\right)\,ds}
\end{align*}
\begin{align*}
\numberthis\label{eq:vlasov-interim}\E^{(L)}_{1,K}(f,t)\lesssim&\,\epsilon^4+\int_t^{t_0}\left(a(s)^{-1-c\sigma}+\epsilon^\frac18 a(s)^{-3}\right)\E^{(L)}_{1,K}(f,s)+a(s)^{-1}\E^{(L)}_{1,K+1}(f,s)\,ds\\
&\,+\int_t^{t_0}\left(\epsilon^\frac{7}8a(s)^{-3-c\sqrt{\epsilon}}+a(s)^{-1-c\sqrt{\epsilon}}\right)\left(\E^{(L)}_{1,K-1}(f,s)+\E^{(\leq L-1)}_{1,\leq K-1}(f,s)\right)\,ds\\
&\,+\int_t^{t_0}\epsilon^{-\frac18}a(s)^{-3-c\sqrt{\epsilon}}\left(\E^{(\leq K)}(\Sigma,s)+\E^{(\leq K)}(\phi,s)\right)\,ds\\
&\,+\int_t^{t_0}\left(a(s)^{-1-c\sqrt{\epsilon}}+\epsilon^\frac{15}8a(s)^{-3-c\sqrt{\epsilon}}\right)\E^{(\leq K-1)}(\Ric,s)\,ds\\
&\,\change{+\int_t^{t_0}\left(a(s)^{-1-c\sqrt{\epsilon}}+\epsilon^\frac{15}8 a(s)^{-3-c\sqrt{\epsilon}}\right)\left(\|\Gamma-\Gamhat\|_{H^1_G(M_s)}^2+\|G^{-1}-\gamma^{-1}\|_{L^2_G(M_s)}^2\right)\,ds}
\end{align*}
For $K<10$, the curvature energies can be dropped.
\end{lemma}
\begin{proof} We again focus on proving the unscaled energy inequality \eqref{eq:vlasov-interim}, and touch on how the scale factors are distributed for \eqref{eq:vlasov-interim-scaled} at the end. One computes:
\begin{subequations}
\begin{align*}
&\,-\del_t\int_{\change{T^\ast M}}\lvert v\rvert_G^2\lvert\nabsak^{L-K}_{vert}\nabsak_{hor}^Kf\rvert_{\G_0}^2\,\vol{\G}\\
\numberthis\label{eq:vlasov-interim-1}=&\,-\int_{\change{T^\ast M}} \change{\langle v\rangle_G^{2}\left(2\frac{\del_t\langle v\rangle_G}{\langle v\rangle_G}+2(L-K)\frac{\del_tv^0}{v^0}\right)}\lvert\nabsak_{vert}^{L-K}\nabsak_{hor}^Kf\rvert^2_{\G_0}\vol{\G}\\
\numberthis\label{eq:vlasov-interim-2}&\,-\int_{\change{T^\ast M}}\langle v\rangle_G^{2}(v^0)^{2(L-K)}(\del_t\G^{-1})\ast_{\G}\nabsak_{vert}^{L-K}\nabsak_{hor}^Kf\ast_{\G}\nabsak_{vert}^{L-K}\nabsak_{hor}^Kf\,\vol{\G}\\
\numberthis\label{eq:vlasov-interim-3}&\,-2\int_{\change{T^\ast M}}\langle v\rangle_G^{2}\langle[\del_t,\nabsak_{vert}^{L-K}\nabsak_{hor}^K]f,\nabsak_{vert}^{L-K}\nabsak_{hor}^Kf\rangle_{\G_0}\,\vol{\G}\\
\numberthis\label{eq:vlasov-interim-4}&\,-2\int_{\change{T^\ast M}}\langle v\rangle_G^{2}\langle\nabsak_{vert}^{L-K}\nabsak_{hor}^K\del_tf,\nabsak_{vert}^{L-K}\nabsak_{hor}^Kf\rangle_{\G_0}\,\vol{\G}
\end{align*}
\end{subequations}
\eqref{eq:vlasov-interim-1}-\eqref{eq:vlasov-interim-2} can be bounded by $\epsilon a^{-3}\E^{(L)}_{1,K}(f,\cdot)$ up to constant as in the purely vertical case. Regarding \eqref{eq:vlasov-interim-3}, \eqref{eq:com-sasaki-gen-t} implies the following:
\begin{align*}
&\,\left\lvert\langle[\del_t,\nabsak_{vert}^{L-K}\nabsak_{hor}^K]f,\nabsak_{vert}^{L-K}\nabsak_{hor}^Kf\rangle_{\G_0}\right\rvert\\
\lesssim&\,\lvert \nabsak_{vert}^{L-K}\nabsak_{hor}^Kf\rvert_{\G_0}\cdot\left[\frac{\lvert v\rvert_G}{v^0}\lvert \nabsak_{vert}^{L-K+1}f\rvert_{\G_0}\lvert\nabla^{K-2}\del_t\Ric[G]\rvert_{G}\right.\\
&\,\change{\left.+\frac{\lvert v\rvert_G}{v^0}\lvert\del_t\Gamma[G]\rvert_G\lvert \nabsak_{vert}^{L-K+1}\nabsak_{hor}^{K-1}f\rvert_{\G_0}+(v^0)^{L-K}\lvert\Xi_{L-K,K,t}f\rvert_{\G}\right]}
\end{align*}
Recall that, by integration by parts and the Young inequality, one has
\[a^{-3-c\sqrt{\epsilon}}\E^{(l-1)}(\Sigma,\cdot)\lesssim a^{-3}\E^{(l)}(\Sigma,\cdot)+a^{-3-2c\sqrt{\epsilon}}\E^{(l-2)}(\Sigma,\cdot)\]
and the analogous estimates for $\E^{(l-1)}(N,\cdot)$ and $\E^{(l-1)}(\Ric,\cdot)$. Applying \eqref{eq:APmidSigma} and \eqref{eq:BsN} to \eqref{eq:REEqRic} and \eqref{eq:REEqChr}, as well as updating $c>0$, one then gets
\begin{align*}
&\,\|\del_t\Ric[G]\|_{H^{l-2}_G}+\|\del_t\Gamma[G]\|_{H^{l-1}_G}\\
\lesssim&\,a^{-3}\left(\sqrt{\E^{(l)}(\Sigma,\cdot)}+\sqrt{\E^{(l)}(N,\cdot)}\right)+a^{-3-c\sqrt{\epsilon}}\left(\sqrt{\E^{(\leq l-2)}(\Sigma,\cdot)}+\sqrt{\E^{(\leq l-2)}(N,\cdot)}\right)\\
&\,+\epsilon a^{-3-c\sqrt{\epsilon}}\sqrt{\E^{(\leq l-2)}(\Ric,\cdot)}\\
\end{align*} 
as well as 
\[\|\del_t\Ric[G]\|_{C^{10}_G(M)}+\|\del_t\change{\Gamma[G]}\|_{C^{11}_G(M)}\lesssim \epsilon a^{-3-c\sqrt{\epsilon}}\,.\]
Further, the supremum norm over $\xi=f$ in \eqref{eq:com-sasaki-gen-t-err} can be bounded by $a^{-c\sqrt{\epsilon}}$ due to \eqref{eq:APVlasov} whenever it occurs. Altogether, this implies
\begin{align*}
\lvert \eqref{eq:vlasov-interim-3}\rvert\lesssim&\,\epsilon a^{-3-c\sqrt{\epsilon}}\left(\sqrt{\E^{(L)}_{1,K-1}(f,\cdot)}+\sqrt{\E^{(\leq L-1)}_{1,\leq K-1}(f,\cdot)}\right)\\
&\,+a^{-3-c\sqrt{\epsilon}}\left(\sqrt{\E^{(\leq K)}(\Sigma,\cdot)}+\sqrt{\E^{(\leq K)}(N,\cdot)}+\epsilon \sqrt{\E^{(\leq K-2)}(\Ric,\cdot)}\right)\,.
\end{align*}
Considering \eqref{eq:vlasov-interim-4}, we can apply \eqref{eq:vlasov-LM} and using $\lvert v\rvert_G\leq v^0$ as well as $\lvert N+1\rvert\lesssim 1$ to obtain the following:
\begin{subequations}
\begin{align*}
\eqref{eq:vlasov-interim-4}+&\,2\int_{\change{T^\ast M}}\langle v\rangle_G^{2}\langle \bm{X}\nabsak^{L-K}_{vert}\nabsak^K_{hor}f,\nabsak^{L-K}_{vert}\nabsak^K_{hor}f\rangle\vol{\G}\\
\numberthis\label{eq:vlasov-interim-step1}\lesssim&\,\int_{\change{T^\ast M}}\langle v\rangle_G^{2}\left[a^{-1}\change{\lvert\nabla^{K-1}\Ric[G]\rvert_G\cdot (v^0)^{2(L-K)}\left\lvert \nabsak_{vert}^{L-K+1}\left(\frac{v\ast v}{v^0}f\right)\right\rvert_{\G}}+a^{-1}\lvert\nabsak_{vert}^{L-K-1}\nabsak_{hor}^{K+1}f\rvert_{\G_0}\right.\\
\numberthis\label{eq:vlasov-interim-step2}&\,\change{+a^{-1}(v^0)^{L-K}\lvert\nabla N\rvert_G\lvert \B\nabsak_{vert}^{L-K-1}\nabsak_{hor}^Kf\rvert_{\G}+\left.(v^0)^{L-K}\lvert\mathfrak{F}_{L,K}\rvert_{\G}\right]}\cdot \lvert\nabsak^{L-K}_{vert}\nabsak^K_{hor}f\rvert_{\G_0}\,\vol{\G}\\
\end{align*}
\end{subequations}
The extra term on the left hand side is dealt with using Lemma \ref{lem:vlasov-energy-mech}. \change{Using \eqref{eq:APMom} and \eqref{eq:APVlasov}, the vertical derivative factor in the first term of \eqref{eq:vlasov-interim-step1} can be estimated by $a^{-c\sqrt{\epsilon}}$. }Regarding \change{the first term in \eqref{eq:vlasov-interim-step2}}, note that one schematically has
\[\B\nabsak_{vert}^{L-K-1}\nabsak_{hor}^Kf=\nabsak_{vert}^{L-K}\nabsak_{hor}^Kf+\Gamma[G]\ast\nabsak^{L-1}f,\]
where the second term always contains horizontal derivatives except for the case $K=1$. Then, one expands the term as follows:
\[\lvert (v^0)^{L-1}\Gamma[G]\ast\nabsak_{vert}^{L-1}f\rvert_G\lesssim a^{-c\sqrt{\epsilon}}\left(\lvert \nabsak^{L-1}_{vert}(f-f_{FLRW})\rvert_{\G_0}+\lvert\nabsak_{vert}^{L-1}f_{FLRW}\rvert_{\G_0}\right)\,.\]
The former can be related back to $\E^{(L-1)}_{1,0}(f,\cdot)$ in estimation and the latter is simply bounded uniformly by $a^{-c\sqrt{\epsilon}}$. Else, this lower order term can be bounded by $\sqrt{\E^{(L-1)}_{1,\leq K-1}(f,\cdot)}$ in estimation. 
We can treat $\B\nabsak_{vert}^{L-K}\nabsak_{hor}^{K-1}f$ identically, except that the case of only vertical derivatives appearing in error terms now happens for $K\leq 2$. 
Thus and applying \eqref{eq:APmidRic} \change{and }\eqref{eq:APVlasov} as well as \eqref{eq:BsC} for the lapse, we obtain
\begin{align*}
\eqref{eq:vlasov-interim-4}\lesssim&\,\left(a^{-1-c\sigma}+\epsilon^\frac18 a^{-3}\right)\E^{(L)}_{1,K}(f,\cdot)+a^{-1}\E^{(L)}_{1,K+1}(f,\cdot)+\epsilon^\frac{15}8 a^{-3-c\sqrt{\epsilon}}\E^{(L)}_{1,K-1}(f,\cdot)\\
&\change{\,+\epsilon a^{-3-c\sqrt{\epsilon}}\E^{(\leq L-1)}_{1,\leq K-1}(f,\cdot)+ a^{-1-c\sqrt{\epsilon}}\E^{(K-1)}(\Ric,\cdot)+a^{-3-c\sqrt{\epsilon}}{a^4\E^{(1)}(N,\cdot)}}\\
&\,+\|(v^0)^{(L-K)}\mathfrak{F}_{L,K}\|_{L^2_{1,\G}}\sqrt{\E^{(L)}_{1,K}(f,\cdot)}
\end{align*}
Note that \change{the last term in the second line }comes from the discussed error terms arising from replacing $\B$ with $\nabsak_{vert}$ at low horizontal orders.\\
\change{The statement now follows by inserting \eqref{eq:vlasov-LM-junk}, combining the bounds for \eqref{eq:vlasov-interim-1}-\eqref{eq:vlasov-interim-4}, inserting Lemma \ref{lem:lapse-en-est} and integrating.}\\

For the scaled version, the main point is that one needs to be careful to ensure that $\E^{(L)}_{1,K-1}(f,\cdot)$ terms do not enter at a power worse than $a^{-3}$, which is the case in the unscaled version. To this end, we distribute the scaling as follows:
\[a^{-3-c\sqrt{\epsilon}+(K+1){\omega}}\sqrt{\E^{(L)}_{1,K}(f,\cdot)}\sqrt{\E^{(L)}_{1,K-1}(f,\cdot)}=a^{-3+\frac{\omega}{2}}\sqrt{a^{(K+1)\omega}\E^{(L)}_{1,K}(f,\cdot)}\sqrt{a^{K\omega-2c\sqrt{\epsilon}}\E^{(L)}_{1,K}(f,\cdot)}\]
Then, we apply the Young inequality and update $c$. We note that one never encounters more than $K$ horizontal derivatives in the error terms and thus can never lose scaling. For the non-Vlasov energies, we distribute scaling as in the previous arguments. Again, replacing $\beta=0$ with $\beta=4$ goes through identically.
\end{proof}

\begin{lemma}[Horizontal Vlasov energy estimates]\label{lem:vlasov-hor} Let $\beta\in\{0,4\}$ \change{and $L\geq 2$}. Then:
\begin{align*}
\numberthis\label{eq:vlasov-hor-scaled}a(t)^{\beta+(L+1)\omega}&\E^{(L)}_{1,L}(f,t)\lesssim\epsilon^4+\int_t^{t_0}\left(a(s)^{-3+\frac{\omega}{2}}+\epsilon^\frac18 a(s)^{-3}\right)a(s)^{\beta+(L+1)\omega}\E^{(L)}_{1,L}(f,s)\,ds\\
&\,+\int_t^{t_0}\epsilon a(s)^{-3-c\sqrt{\epsilon}+\frac{\omega}{2}}\cdot a(s)^{\beta+L{\omega}}\E^{(\leq L)}_{1,\leq L-1}(f,s)\,ds\\
&\,+\int_t^{t_0}a(s)^{-3-c\sqrt{\epsilon}+\omega}\cdot a(s)^\beta\left(\change{a(s)^\frac{\omega}4}\E^{(\leq L)}(\Sigma,s)+\E^{(\leq L)}(\phi,s)
\right)\,ds\\
&\,\change{+\int_t^{t_0}a(s)^{-3-c\sqrt{\epsilon}+\frac{(2L-1){\omega}}2}\cdot a(s)^{4+\beta+\omega}\E^{(L-1)}(\Ric,s)\,ds}\\
&\,\change{+\int_t^{t_0}a(s)^{-3-c\sqrt{\epsilon}+L\omega}\,a(s)^{\beta+\frac{\omega}2}\E^{(L-2)}(\Ric,s)\,ds}\\
&\,\change{+\int_t^{t_0}\left(\epsilon a(s)^{-3-c\sqrt{\epsilon}+\omega}+a(s)^{-1-c\sqrt{\epsilon}}\right)\,a(s)^{\beta+\frac{\omega}2}\E^{(\leq L-4)}(\Ric,s)\,ds}\\
&\,\change{+\int_t^{t_0}\left(a(s)^{-1-c\sqrt{\epsilon}+\frac{(2L+1){\omega}}2}+\epsilon a(s)^{-3-c\sqrt{\epsilon}+\change{L\omega}}\right)\cdot}\\
&\,\change{\phantom{\int_t^{t_0}}\quad\cdot\left(a(s)^{\beta+\frac{\omega}2}\|\Gamma-\Gamhat\|_{H^1_G(M_s)}^2+a(s)^{\beta+\frac{\omega}{2}}\|G^{-1}-\gamma^{-1}\|_{H^2_G(M_s)}^2\right)\,ds}
\end{align*}

\begin{align*}
\numberthis\label{eq:vlasov-hor}\E^{(L)}_{1,L}(f,t)\lesssim&\,\epsilon^4+\int_t^{t_0}\left(a(s)^{-1-c\sigma}+\epsilon^\frac18 a(s)^{-3}\right)\E^{(L)}_{1,L}(f,s)+\epsilon^\frac{15}8 a(s)^{-3-c\sqrt{\epsilon}}\E^{(L)}_{1,L-1}(f,s)\,ds\\
&\,+\int_t^{t_0}\epsilon^{-\frac18}a(s)^{-3-c\change{\sqrt{\epsilon}}}\left(\E^{(\leq L)}(\Sigma,s)+\E^{(\leq L)}(\phi,s)
\right)\,ds\\
&\,+\int_t^{t_0}\epsilon^\frac78a(s)^{-3-c\sqrt{\epsilon}}\cdot \left(a(s)^4\E^{(L-1)}(\Ric,s)+\E^{(\leq L-2)}(\Ric,s)\right)\,ds\\
\end{align*}
For $L\leq 10$, the curvature energies can be dropped. For $L=1$, one \change{has}
\begin{align*}\numberthis\label{eq:vlasov-hor-scaled-1}
a(t)^{4+2\omega}\E^{(1)}_{1,1}(f,t)\lesssim&\,\epsilon^4+\int_t^{t_0}\left(\change{a(s)^{-1-c\sqrt{\epsilon}}+}a(s)^{-3+\frac{\omega}2}+\epsilon^\frac18 a(s)^{-3}\right)\cdot a(s)^{4+2\omega}\E^{(1)}_{1,1}(f,s)\,ds\\
&\,+\int_t^{t_0}\epsilon a(s)^{-3-c\sqrt{\epsilon}+\frac{\omega}2}\cdot a(s)^{4+\omega}\E^{(1)}_{1,0}(f,s)\,ds\\
&\,+\int_t^{t_0}a(s)^{-3-c\sqrt{\epsilon}+\frac54\omega}\cdot a(s)^4\left(a(s)^{\frac{\omega}4}\E^{(1)}(\Sigma,s)+\E^{(1)}(\phi,s)\right)\,ds\\
&\,\change{+\int_t^{t_0}a(s)^{-1-c\sqrt{\epsilon}+\frac32\omega}\,\cdot a(s)^{4+\frac{\omega}2}\|\Gamma-\Gamhat\|_{L^2_G(M_s)}^2\,ds}\\
&\,\change{+\int_t^{t_0}a(s)^{-1-c\sqrt{\epsilon}+\frac{3}2\omega}\,\cdot a(s)^{\frac{\omega}2}\|G^{-1}-\gamma^{-1}\|_{L^2_G(M_s)}^2\,ds}\\
\numberthis\label{eq:vlasov-hor-1} \E^{(1)}_{1,1}(f,t)\lesssim&\,\epsilon^4+\int_t^{t_0}\left(a(s)^{-1-c\sigma}+\epsilon^\frac18 a(s)^{-3}\right)\E^{(1)}_{1,1}(f,s)+\epsilon^\frac{15}8 a(s)^{-3-c\sqrt{\epsilon}\change{+\omega}}\E^{(1)}_{1,0}(f,s)\,ds\\
&\,+\int_t^{t_0}\epsilon^{-\frac18}a(s)^{-3-c\change{\sqrt{\epsilon}}}\left(\E^{(\leq 1)}(\Sigma,s)+\E^{(\leq 1)}(\phi,s)\right)\,ds\\
&\,\change{+\int_t^{t_0}a(s)^{-1-c\sqrt{\epsilon}}\,\left(\|\Gamma-\Gamhat\|_{L^2_G(M_s)}^2+\|G^{-1}-\gamma^{-1}\|_{L^2_G(M_s)}^2\right)\,ds}
\end{align*}
\end{lemma}

\begin{proof}
The proof of these estimates is analogous to that of Lemma \ref{lem:vlasov-interim} since \eqref{eq:vlasov-LL} is just a simpler version of \eqref{eq:vlasov-LM}, and the commutator formula \eqref{eq:com-sasaki-gen-t} is also used for $M_1=0$ in a similar fashion to the previous proof. As a difference in estimation, we note that the high order curvature term present in \eqref{eq:vlasov-LL} is now of order $L-1$, and thus needs to be thought of as being of order $L+1$ in the metric, \change{i.e., }of higher order than the Vlasov energy itself. Consequently, the resulting energy controlling this term needs to be additionally scaled to allow the argument to close, which makes the resulting energy term borderline. 
\end{proof}
\section{Improving the Bootstrap assumptions}\label{sec:improve}

In this section, we combine the energy estimates from Sections \ref{sec:energy-new}, and \ref{sec:vlasov} to improve the bootstrap assumption \eqref{eq:BsC} -- recall that \eqref{eq:BsVlasovhor} is already improved due to \eqref{eq:APVlasov}. To this end, we first derive improved energy bounds for the spacetime and scalar field variables and then use these to derive bounds for Vlasov energies. We then translate these bounds to Sobolev norm bounds, and finally to bounds on $\mathcal{C}$.\\
In the following, we again assume $t\in(t_{Boot},t_0]$ and that $\omega\gg\sigma$ is a suitably small constant\change{, see also the start of Section 5}.\\

To deal with linear borderline terms in the energy estimates for the shear (see Lemma \ref{lem:en-est-Sigma}) and Bel-Robinson variables (see Lemma \ref{lem:en-est-BR}), we need to extend \cite[Lemma 7.1]{FU23} to ensure that the error terms essentially have favourable sign even for flat and spherical spatial geometry:

\begin{lemma}\label{lem:en-error-cancellation} Let $L\in 2\N$. Then, the following estimate is satisfied:
\begin{align*}
\numberthis\label{eq:en-error-cancellation}
&\,\int_M\left[16\pi C^2a^{-3}(N+1)\langle\Lap^\frac{L}2\RE,\Lap^{\frac{L}2}\Sigma\rangle_G+8\pi C^2\frac{\dot{a}}a(N+1)\lvert\Lap^\frac{L}2\Sigma\rvert_G^2+6\frac{\dot{a}}a(N+1)\lvert\Lap^\frac{L}2\RE\rvert_G^2\right]\vol{G}\\
&\begin{cases}
\geq 0 & \kappa\leq 0\\
\gtrsim -a^{-1}\left[4\pi C^2\E^{(L)}(\Sigma,\cdot)+\E^{(L)}(W,\cdot)\right] & \kappa>0
\end{cases}
\end{align*}
\end{lemma}
\begin{proof}
Firstly, the bootstrap assumption \eqref{eq:BsN} implies $N+1>0$ for $\epsilon>0$ small enough. For $\kappa\leq 0$, we have $a^{-3}\leq \frac{3}{4\pi C^2}\frac{\dot{a}}a$ by \eqref{eq:Friedman} since the Vlasov density is nonnegative. Then, the statement follows directly from this and the Young inequality as in \cite[Lemma 7.1]{FU23}. Else, one has 
\[a^{-3}\leq \frac{3}{4\pi C^2}\frac{\sqrt{\dot{a}^2+\kappa}}{a}\leq \frac{3}{4\pi C^2}\left(\frac{\dot{a}}a+\sqrt{\kappa}a^{-1}\right)\]
and thus the integrand on the left hand side in \eqref{eq:en-error-cancellation} is bounded from below by
\[\gtrsim-a^{-1}\lvert\Lap^{\frac{L}2}\Sigma\rvert_G\lvert\Lap^\frac{L}2\RE\rvert_G\,,\]
implying the statement.
\end{proof}

To obtain bounds on the spacetime and scalar field energies, we need to combine Vlasov energies at order $L$ with appropriate scaling as follows:

\begin{definition}[Total scaled Vlasov energies]
For $0\leq L\leq 19$, we define
\begin{equation}\label{eq:def-total-vlasov}
\E^{(L)}_{total,Vl}=\sum_{K=0}^La^{(K+1)\omega}\E^{(L)}_{1,K}(f,\cdot)\,.
\end{equation}
\end{definition}

With this in hand, we can combine the estimates from Section \ref{sec:vlasov} as follows:

\begin{lemma}[Total scaled Vlasov energy estimate]\label{lem:vlasov-total-scaled}
For $L=1$, one has
\begin{align*}\numberthis\label{eq:vlasov-total-scaled-1}
a(t)^4\E^{(1)}_{total,Vl}&(t)\lesssim\epsilon^4+\int_t^{t_0}\left(\change{a(s)^{-1-{\omega}-c\sigma}}+a(s)^{-3\change{+\frac{\omega}{4}-c\sqrt{\epsilon}}}+\epsilon^\frac18a(s)^{-3}\right)a(s)^4\E^{(1)}_{total,Vl}(s)\,ds\\
&\,+\int_t^{t_0}a(s)^{-3-c\sqrt{\epsilon}+\change{\frac{\omega}2}}\cdot a(s)^4\left(\E^{(\leq 1)}(\phi,s)+\change{a(s)^{\frac{\omega}4}\E^{(\leq 1)}(\Sigma,s)+\epsilon^2\E^{(0)}(\Sigma,s)}\right)\,ds\\
&\,+\int_t^{t_0}a(s)^{-1-c\sqrt{\epsilon}}\left(a(s)^{4+\change{\frac{\omega}2}}\|\Gamma-\Gamhat\|_{L^2_G(M_s)}^2+a(s)^{4+\change{\frac{\omega}2}}\|\change{G^{-1}-\gamma^{-1}}\|_{L^2_G(M_s)}^2\right)\,ds\\
\end{align*}
Let $L\geq 2$. Then, the following holds:
\begin{align*}\numberthis\label{eq:vlasov-total-scaled}
a(t)^\beta\E^{(L)}_{total,Vl}(t)\lesssim&\,\epsilon^4+\int_t^{t_0}\left(a(s)^{\change{-1-{\omega}-c\sigma}}+a(s)^{-3+\frac{\omega}{4}}+\epsilon^\frac18a(s)^{-3}\right)a(s)^\beta\E^{(L)}_{total,Vl}(s)\,ds\\
&\,+\int_t^{t_0}a(s)^{-3-c\sqrt{\epsilon}+\change{\frac{\omega}4}}\cdot a(s)^\beta\left(\E^{(L)}(\phi,s)+\change{a(s)^\frac{\omega}4}\E^{(L)}(\Sigma,s)\right)\,ds\\
&\,+\int_t^{t_0}\epsilon a(s)^{-3-c\sqrt{\epsilon}+\change{{(L-1)}\omega}}\cdot a(s)^{4+\beta}\E^{(L-1)}(\Ric,s)\,ds\\
&\,+\int_t^{t_0}\change{a(s)^{-3-c\sqrt{\epsilon}+L\omega}}\cdot a(s)^\beta\E^{(L-2)}(\Ric,s)\,ds\\
&+\int_t^{t_0}a(s)^{-3-c\sqrt{\epsilon}+\change{\frac{\omega}4}}\cdot a(s)^\beta\left(\E^{(\leq L-2)}(\phi,s)+\change{a(s)^\frac{\omega}4}\E^{(\leq L-2)}(\Sigma,s)\right)\,ds\\
&\,+\int_t^{t_0}\left(a(s)^{-1-\omega\change{-c\sigma}}+a(s)^{-3-c\sqrt{\epsilon}+\frac{\omega}2}\right)a(s)^{\beta+(L-1)\omega}\E^{(\leq L-2)}_{1,\leq L-2}(f,s)\,ds\\
&\,+\int_t^{t_0}\change{\left(a(s)^{-1-c\sqrt{\epsilon}+\omega}+\epsilon a(s)^{-3-c\sqrt{\epsilon}+\omega}\right)}\E^{(\leq L-4)}(\Ric,s)\,ds\\
&\,+\int_t^{t_0}\Big(a(s)^{-1-c\sqrt{\epsilon}+\beta}+a(s)^{-3-c\sqrt{\epsilon}+\beta+\change{\omega}}\Big)\cdot\\
&\,\qquad \cdot \left(\change{a(s)^\frac{\omega}2}\|\Gamma-\Gamhat\|_{H^1_G(M_s)}^2+\change{a(s)^\frac{\omega}2}\|\change{G^{-1}-\gamma^{-1}}\|_{H^2_G(M_s)}^2\right)\,ds
\end{align*}
\end{lemma}
\begin{proof}
This is a combination of Lemma \ref{lem:vlasov-vertical}, \ref{lem:vlasov-interim} and \ref{lem:vlasov-hor}. Observe that the final line arises from Lemma \ref{lem:vlasov-vertical}, while the remaining lines arise from Lemma \ref{lem:vlasov-interim} for $K=L-1$ and Lemma \ref{lem:vlasov-hor}. Furthermore, we immediately rewrote $\E^{(L-1)}_{1,\leq K}(f,\cdot)$ as follows: First, let $\mathcal{K}_{L,K}>0$ be a suitable constant depending on $L$ and $K$ and \change{$\nabsak_v=v_l\nabsak_{\B^l}$. }For $L-1>K>0$, we integrate by parts:
\begin{align*}
a^{-c\sqrt{\epsilon}}\E^{(L-1)}_{1,K}(f,\cdot)=&\,-a^{-c\sqrt{\epsilon}}\int_{\change{T^\ast M}} \langle v\rangle_G^{2}(v^0)^{2(L-K-1)}\langle\nabsak_{vert}^{L-K-2}\nabsak_{hor}^{K}f,\nabsak_{vert}^{L-K}\nabsak_{hor}^Kf\rangle_{\G}\,\vol{\G}\\
&\,-2a^{-c\sqrt{\epsilon}}\int_{\change{T^\ast M}}\left((L-K-1)(v^0)^{2(L-K-2)}\langle v\rangle_G^2+(v^0)^{2(L-K-1)}\right)\cdot \\
&\,\qquad\qquad\qquad \cdot \langle\nabsak_{vert}^{L-K-2}\nabsak_{hor}^Kf,\nabsak_v\nabsak_{vert}^{L-K-2}\nabsak_{hor}^Kf\rangle\,\vol{\G}
\end{align*}
Integrating by parts again for the second term, one obtains
\begin{align*}
a^{-c\sqrt{\epsilon}}\E^{(L-1)}_{1,K}(f,\cdot)\leq&\,\frac12\left(\E^{(L)}_{1,K}(f,\cdot)+a^{-2c\sqrt{\epsilon}}\E^{(L-2)}_{1,K}(f,\cdot)\right)\\
&\,+\mathcal{K}_{L,K}a^{-c\sqrt{\epsilon}}\int_{\change{T^\ast M}}(v^0)^{2(L-K-2)}\langle v\rangle_G^2\lvert \nabsak_{vert}^{L-K-2}\nabsak_{hor}^Kf\rvert_{\G}^2\,\vol{\G}\\
\lesssim&\,\E^{(L)}_{1,K}(f,\cdot)+a^{-c\sqrt{\epsilon}}\E^{(L-2)}_{1,K}(f,\cdot)\,.
\end{align*}
For $K=L-1>1$, one can argue identically without an error term after the first integration by parts since $\nabsak_{hor}v^0=0$, obtaining
\[a^{-c\sqrt{\epsilon}}\E^{(L-1)}_{1,L-1}(f,\cdot)\lesssim \E^{(L)}_{1,L}(f,\cdot)+a^{-c\sqrt{\epsilon}}\E^{(L-2)}_{1,L-2}(f,\cdot).\]
For $K=0$, one can argue as above replacing $f$ with $f-f_{FLRW}$. Finally, for $L-1=1=K$, we need to perform the following replacement after integrating by parts:
\begin{align*}
&\,a^{-3-c\sqrt{\epsilon}+\frac52\omega+\beta}\E^{(1)}_{1,1}(f,\cdot)\\
=&\,a^{-3-c\sqrt{\epsilon}+\frac{3}2\omega+\beta}\left[-\int_{\change{T^\ast M}} \langle v\rangle_G^2(f-f_{FLRW})\cdot (\change{{G}^{-1})^{ij}\nabsak_{i}\nabsak_{j}f}\,\vol{\G}\right.\\
&\,-\left.\int_{\change{T^\ast M}}\langle v\rangle_G^2(f-f_{FLRW})({\G}^{-1})^{ij}\nabsak_i\nabsak_jf_{FLRW}+\langle v\rangle_G^2\lvert \nabsak_{hor} f_{FLRW}\rvert_{\G_0}^2\,\vol{G}\right]\\
\lesssim&\,a^{-3+\frac{\omega}2}\left(a^{\beta+3\omega}\E^{(2)}_{1,2}(f,\cdot)+a^{-c\sqrt{\epsilon}}\cdot a^{\beta+\omega}\E^{(0)}_{1,0}(f,\cdot)\right)\\
&\,+a^{-3-c\sqrt{\epsilon}+\beta+\frac{5}2\omega}\|\nabsak_{hor} f_{FLRW}\|_{L^2_{1,\G}(\change{T^\ast M})}^2+a^{3-c\sqrt{\epsilon}+\beta+\frac{7}2\omega}\|\nabsak_{hor}^2f_{FLRW}\|_{L^2_{1,\G}(\change{T^\ast M})}^2\\
\lesssim&\,a^{-3+\frac{\omega}2}\left(a^{\beta+3\omega}\E^{(2)}_{1,2}(f,\cdot)+a^{-c\sqrt{\epsilon}}\cdot a^{\beta+\omega}\E^{(0)}_{1,0}(f,\cdot)\right)\\
&\,\change{+a^{-3-c\sqrt{\epsilon}+\beta+2\omega}\left(a^\frac{\omega}2\|\Gamma-\Gamhat\|_{H^1_G(M)}^2+a^\frac{\omega}2\|\change{G^{-1}-\gamma^{-1}}\|^2_{H^2_G(M)}\right)}\,
\end{align*}
The final step follows by applying \eqref{eq:hor-deriv-ref}.
\end{proof}

\change{Recalling the discussion in Section \ref{subsubsec:integrate-Vlasov-3+1}, we would like to construct a total energy as in \cite{FU23} that admits a strong bound via the Gronwall lemma. To this end, we essentially take the energy $\E^{(L)}_{quiesc}$ that describes the quiescent formalism as in \cite{FU23} and add the total Vlasov energy to it, which forces us to add further scaled terms to the total energy $\E^{(L)}_{total}$. Before explaining the precise rationale behind these terms, we introduce these energies. 
\begin{definition}[Total energies]\label{def:total}
\begin{subequations}\label{eq:def-total-en-both}
We define:
\begin{align*}\numberthis\label{eq:def-total-en0}
\E^{(0)}_{total}=&\,\E^{(0)}(\phi,\cdot)+\left(\epsilon^\frac14+a^\frac{\omega}4\right)\left(\E^{(0)}(W,\cdot)+4\pi C^2\E^{(0)}(\Sigma,\cdot)\right)\\
&\,+a^4\E^{(1)}(\phi,\cdot)+\left(\epsilon^\frac12+a^\frac{\omega}2\right)\,a^4\E^{(1)}(\Sigma,\cdot)\\
&\,+\E^{(0)}_{total,Vl}+a^4\E^{(1)}_{total,Vl}+a^\frac{\omega}2\|G^{\pm 1}-\gamma^{\pm 1}\|_{L^2_G(M)}^2+a^{4+\frac{\omega}2}\|\Gamma-\Gamhat\|^2_{L^2_G(M)}\\
\numberthis\label{eq:def-quiesc-en0}\E^{(0)}_{quiesc}=&\,\E^{(0)}(\phi,\cdot)+\epsilon^\frac14\left(\E^{(0)}(W,\cdot)+4\pi C^2\E^{(0)}(\Sigma,\cdot)\right)+a^4\E^{(1)}(\Sigma,\cdot)\\
\end{align*}
and, for $L\in 2\N, 2\leq L\leq 18$
\begin{align*}\numberthis\label{eq:def-total-en}
\E^{(L)}_{total}=&\,\E^{(L)}(\phi,\cdot)+\left(\epsilon^\frac14+a^{\frac{\omega}4}\right)\left(\E^{(L)}(W,\cdot)+4\pi C^2\E^{(L)}(\Sigma,\cdot)\right)+\left(\epsilon^\frac12+a^{\frac{\omega}2}\right)\E^{(L-2)}(\Ric,\cdot)\\
&\,+a^4\E^{(L+1)}(\phi,\cdot)+\left(\epsilon^\frac12+a^{\frac{\omega}2}\right)a^4\E^{(L+1)}(\Sigma,\cdot)+\left(\epsilon^\frac34+a^{\omega}\right) a^4\E^{(L-1)}(\Ric,\cdot)\\
&\,+\left(\epsilon^\frac34+a^\omega\right)a^{8}\E^{(L)}(\Ric,\cdot)+\E^{(L)}_{total,Vl}+a^4\E^{(L+1)}_{total,Vl}\\
&\,+\underbrace{a^\frac{\omega}2\left(\|\Gamma-\Gamhat\|_{H^1_G(M)}^2+\|G^{\pm 1}-\gamma^{\pm 1}\|_{H^2_G(M)}^2\right)}_{\text{if }L=2}\\
\numberthis\label{eq:def-quiesc-en}\E^{(L)}_{quiesc}=&\,\E^{(L)}(\phi,\cdot)+\epsilon^\frac14\left(\E^{(L)}(W,\cdot)+4\pi C^2\E^{(L)}(\Sigma,\cdot)\right)+a^4\E^{(L+1)}(\phi,\cdot)+\epsilon^\frac12 a^4\E^{(L+1)}(\Sigma,\cdot)\\
&\,+\epsilon^\frac12\E^{(L-2)}(\Ric,\cdot)+\epsilon^\frac34 a^4\E^{(L-1)}(\Ric,\cdot)
\end{align*}
\end{subequations}
\end{definition}

\begin{remark}[On the construction of $\E^{(L)}_{total}$] Essentially, the total energy can be split into three components -- $\E^{(L)}_{quiesc}$ which describes the core quiescent mechanism as in \cite{FU23}, $\E^{(L)}_{total,Vl}$ which is a natural energy for Vlasov matter in near-FLRW regimes, and various scaled energies controlling components of the the spacetime metric. We briefly explain why the latter needs to be introduced when one attempts to combine $\E^{(L)}_{quiesc}$ and $\E^{(L)}_{total,Vl}$ into a total energy that admits a useful Gronwall estimate:\\

First, we notice that Vlasov energies occur linearly in energy estimates for the scalar field, see Lemma \ref{lem:en-est-SF}, and also vice versa, see \eqref{eq:vlasov-total-scaled}. Thus, we do not gain smallness in terms of $\epsilon$, and the Vlasov energy must enter the total energy without $\epsilon$-weights. Second, Vlasov matter also linearly couples linearly in the geometric evolution equations for the shear and for Bel-Robinson variables, see Lemmas \ref{eq:en-est-Sigma-all} and \ref{lem:en-est-BR}, and vice versa, see \eqref{eq:vlasov-total-scaled}. However, to obtain a closed energy estimates for the quiescent components, one needs to introduce an $\epsilon$-weight in $\E^{(L)}_{quiesc}$ as in \cite{FU23}. On the other hand, geometric energies occur linearly in \eqref{eq:vlasov-total-scaled}. Thus, to obtain closed estimates, we must introduce geometric energies that are not weighted by $\epsilon$. To nevertheless surpress linear scalar field terms of order $t^{-1}$ in the geometric estimates, we must then introduce scale factor weights on these terms, which must be weaker than those on Vlasov matter. This then extends to Ricci curvature terms as well as low order geometric error terms in the Vlasov equation, which must both exhibit stronger scaling than the second fundamental form, which occurs linearly at order $t^{-1}$ in the evolution equations of either. For the second fundamental form at order $L+1$, the additional scale factor weight must be stronger than the one at order $L$ so that, after applying Lemma \ref{lem:en-est-Sigma-top}, all resulting terms can be absorbed into the left hand side if $a(t_0)$ is sufficiently small. The curvature term at order $L_1$ must, in turn, be scaled more strongly than the second fundamental form at order $L+1$ as before, and the top order curvature term must be similarly strongly weighted so that one can absorb terms resulting from \eqref{eq:en-est-Ric-toptop} into the left hand side.
\end{remark}}

\begin{prop}[Total energy bound] \label{prop:en-imp}
For $L\in 2\N, L\leq 18$, the following holds:
\begin{equation}\label{eq:total-en-imp}
\E^{(L)}_{total}\lesssim\epsilon^4a^{-c\epsilon^\frac18}
\end{equation}
\end{prop}
\begin{proof}\change{We start proving the case of $L=0$ by combining the following energy estimates:
\begin{itemize}
\item For the first line, we use \eqref{eq:en-est-SF} for $L=0$ for the scalar field energy. For the Bel-Robinson and shear energies, we treat their sum scaled by $\epsilon^\frac14$ and their sum scaled by $a^\frac{\omega}4$ separately. For the former, we apply \eqref{eq:en-est-BR} and \eqref{eq:en-est-Sigma} for $L=0$, scaled by $\epsilon^\frac14$. For the latter, we apply \eqref{eq:en-est-BR-scaled} and \eqref{eq:en-est-Sigma-scaled} for $L=0$.
\item For the second line, we apply \eqref{eq:en-est-SF-top} for $L=1$ and the elliptic shear estimate \eqref{eq:en-est-Sigma-1} scaled by $\left(\epsilon^\frac12+a^{\frac{\omega}2}\right)$
\item For the third line, we apply \eqref{eq:en-est-Vlasov-0-scaled}, \eqref{eq:vlasov-total-scaled-1} and the metric estimate \eqref{eq:metric-int-est-scaled}.
\end{itemize}
Notice that, by \eqref{eq:en-error-cancellation}, the indefinite terms on the left hand sides of the analogues of \eqref{eq:en-est-BR} and \eqref{eq:en-est-Sigma}, respectively \eqref{eq:en-est-BR-scaled} and \eqref{eq:en-est-Sigma-scaled}, are nonnegative for $\kappa\leq 0$, and else are so up to a term that can be pulled to right and estimated from above by $\int_t^{t_0}a(s)^{-1}\E^{(0)}_{total}(s)\,ds$. This leads to the following estimate:
\begin{align*}
\E^{(0)}_{total}(t)\lesssim&\,\epsilon^4+\int_t^{t_0}\left(\epsilon^\frac18a(s)^{-3}+a(s)^{-3+\change{\frac{\omega}8}}+a(s)^{-2-c\sigma-2\omega}\right)\E^{(0)}_{total}(s)\,ds+\left(\epsilon^\frac14+a(t)^{\frac{\omega}4}\right)\E^{(0)}_{total}(t)
\end{align*}
Recall that we can choose both $\epsilon>0$ and $a(t_0)$ to be sufficiently small without loss of generality, see Remark \ref{rem:close-to-bb} for the latter. Thus, the final term can be absorbed into the left hand side after adapting implicit constants, and the Gronwall lemma then yields
\[\E^{(0)}_{total}(t)\lesssim \epsilon^4a(t)^{-c\epsilon^\frac18}\,.\]
For $L\geq 2$, we proceed iteratively, assuming the bound to have been shown up to ${\E}_{total}^{(\leq L-2)}$. One now similarly combines the following estimates, going through \eqref{eq:def-total-en}.
\begin{itemize}
\item  For the first line: \eqref{eq:en-est-SF} for the scalar field energy, \eqref{eq:en-est-BR}, \eqref{eq:en-est-Sigma}  and \eqref{eq:en-est-Ric} for the Bel-Robinson, shear and Ricci energies scaled by $\epsilon^\frac14$, and \eqref{eq:en-est-BR-scaled}, \eqref{eq:en-est-Sigma-scaled} and \eqref{eq:en-est-Ric-scaled} for the energies scaled by powers of $a^\frac{\omega}4$
\item For the first two terms in the second line: \eqref{eq:en-est-SF-top} and \eqref{eq:en-est-Sigma-top}
\item For the curvature energies second line: \eqref{eq:en-est-Ric} and \eqref{eq:en-est-Ric-scaled}, respectively \eqref{eq:en-est-Ric0} and \eqref{eq:en-est-Ric0-scaled} for $L=2$, as well as \eqref{eq:en-est-Ric-top} and \eqref{eq:en-est-Ric-top-scaled} for the energies at order $L-1$
\item For the curvature energy in the third line: \eqref{eq:en-est-Ric-toptop} with $l=L$, scaled appropriately
\item For the Vlasov energies in the second line: \eqref{eq:vlasov-total-scaled} with $\beta=0$ at order $L$ and with $\beta=4$ at order $L+1$
\item For the metric terms if $L=2$: \eqref{eq:chr-int-est-scaled} with $l=2,\beta=0$ and \eqref{eq:metric-int-est-scaled} with $l=2$
\end{itemize}
Note that, for $L=2$, the metric and Christoffel symbol norms become a part of the total energy, and else they are absorbed into the lower order terms since they have already been improved.\\
Altogether, this yields the following:
\begin{align*}
\E^{(L)}_{total}(t)\lesssim&\,\epsilon^4+\int_t^{t_0}\left[\left(\epsilon^\frac18a(s)^{-3}+a(s)^{-3-c\sqrt{\epsilon}+\change{\frac{\omega}8}}+a(s)^{-2-c{\sigma}-(L+2)\omega}\right)\E^{(L)}_{total}(s)\right.\\
&\,\qquad\qquad\left.+\left(\epsilon^\frac18a(s)^{-3-c\sqrt{\epsilon}}+a(s)^{-3-c\sqrt{\epsilon}+\frac{\omega}8}+a(s)^{-2-c{\sigma}-(L+1)\omega}\right)\E^{(\leq L-2)}_{total}(s)\right]\,ds\\
&\,+\left(\epsilon^\frac14+a(t)^{\frac{\omega}4}\right)\E^{(L)}_{total}(t)+\left(\epsilon^\frac14+a(t)^{\frac{\omega}4}\right)a(t)^{-c\sqrt{\epsilon}}\E^{(\leq L-2)}_{total}(t)
\end{align*}}
Since the statement has already been shown for $\E^{(\leq L-2)}_{total}$, this implies after rearranging \change{the first term in the final line as before }that
\[\E^{(L)}_{total}(t)\lesssim \epsilon^4a(t)^{-c\epsilon^\frac18}+\int_t^{t_0}\left(\epsilon^\frac18a(s)^{-3}+a(s)^{-3-c\sqrt{\epsilon}+\change{\frac{\omega}8}}+a(s)^{-2-c\sigma-(L+2)\omega}\right)\E^{(L)}_{total}(s)\,ds\]
and thus
\[\E^{(L)}_{total}(t)\lesssim \epsilon^4a(t)^{-c\epsilon^\frac18}\,.\]
\end{proof}

\change{\begin{prop}[Quiescent energy bound]\label{prop:en-imp-quiesc}
For $L\in 2\N, L\leq 18$, the following holds:
\begin{equation}\label{eq:quiesc-en-imp}
\E^{(L)}_{quiesc}\lesssim\epsilon^4a^{-c\epsilon^\frac18}
\end{equation}
\end{prop}
\begin{proof}
In short, the strategy in this step of the argument is to rerun the argument that yielded \eqref{eq:total-en-imp}, but using this bound to directly estimate all Vlasov terms and metric errors, since it directly implies, for any $l\in\N, l\leq 18$ and any $m\in\N, m\leq l$,
\begin{equation}\label{eq:vlasov-imp-interim}
a^{(m+1)\omega}\E^{(l)}_{1,m}(f,\cdot)+a^{\frac{\omega}2}\left(\|G^{\pm 1}-\gamma^{\pm 1}\|^2_{H^2_G}+\|\Gamma-\Gamhat\|_{H^1_G}^2\right)\lesssim \epsilon^4 a^{-c\epsilon^\frac18}
\end{equation}
For $L=0$, we use the analogues of \eqref{eq:en-est-SF}, \eqref{eq:en-est-BR}, \eqref{eq:en-est-Sigma} and \eqref{eq:en-est-SF-top} for $L=0$ and the elliptic shear estimate \eqref{eq:en-est-Sigma-1}. This yields the following:
\begin{align*}
\E^{(0)}_{quiesc}(t)\lesssim&\,\epsilon^4+\int_t^{t_0}\left(\epsilon^\frac18a(s)^{-3}+a(s)^{-3+\change{\frac{\omega}4}}+a(s)^{-2-c\sqrt{\epsilon}}\right)\,\E^{(0)}_{quiesc}(s)\\
&\,+\int_t^{t_0}a(s)^{-2-c\sqrt{\epsilon}-2\omega}\left(\E^{(0)}_{total,Vl}(s)+a(s)^4\E^{(1)}_{total,Vl}+a(s)^\frac{\omega}2\|G^{-1}-\gamma^{-1}\|_{L^2_G(M_s)}^2\right)\,ds\\
&\,+\epsilon^\frac14\E^{(0)}_{quiesc}(t)+a(t)^{2-c\sqrt{\epsilon}-\omega}\E^{(0)}_{total,Vl}(t)
\end{align*}
After pulling the first term in the last line to the left and applying \eqref{eq:total-en-imp}, this implies
\begin{align*}
\E^{(0)}_{quiesc}(t)\lesssim&\,\epsilon^4+\int_t^{t_0}\left(\epsilon^\frac18a(s)^{-3}+a(s)^{-3+\change{\frac{\omega}4}}+a(s)^{-2-c\sqrt{\epsilon}}\right)\,\E^{(0)}_{quiesc}(s)\\
&\,+\int_t^{t_0}\epsilon^4a(s)^{-2-2\omega-c\epsilon^\frac18}\,ds+\epsilon^4a(t)^{4-\omega-c\epsilon^\frac18}
\end{align*}
Since Lemma \ref{lem:scale-factor} implies that the second line can be bounded by $\epsilon^4$, the Gronwall lemma leads to $\E^{(0)}_{quiesc}(t)\lesssim \epsilon^4a(t)^{-c\epsilon^\frac18}$.\\
We now proceed iteratively for $L\in 2\N, L\leq 18$, assuming that one has shown
\[\E^{(\leq L-2)}_{quiesc}\lesssim \epsilon^4a^{-c\epsilon^\frac18}\]
We combine the following estimates to control the individual terms in \eqref{eq:def-quiesc-en}.
\begin{itemize}
\item  For the first line: \eqref{eq:en-est-SF}, \eqref{eq:en-est-BR}, \eqref{eq:en-est-Sigma},  \eqref{eq:en-est-SF-top}, \eqref{eq:en-est-Sigma-top} 
\item For the curvature energies second line: \eqref{eq:en-est-Ric} (respectively \eqref{eq:en-est-Ric0} for $L=2$) and \eqref{eq:en-est-Ric-top}
\end{itemize}
This yields the following as before.
\begin{align*}
\E^{(L)}_{quiesc}(t)\lesssim&\,\epsilon^4+\int_t^{t_0}\left(\epsilon^\frac18a(s)^{-3}+a(s)^{-3+\frac{\omega}4}+a(s)^{-2-c\sqrt{\epsilon}}\right)\E^{(L)}_{quiesc}(s)\,ds\\
&\,+\int_t^{t_0}a(s)^{-2-c\sqrt{\epsilon}-(L+2)\omega}\left(\E^{(\leq L)}_{total,Vl}(s)+a(s)^4\E^{(L+1)}_{total,Vl}(s)+\|G^{-1}-\gamma^{-1}\|_{L^2_G(M_s)}^2\right)\,ds\\
&\,+\int_t^{t_0}\left(\epsilon^{\frac18}a(s)^{-3-c\sqrt{\epsilon}}+a(s)^{-3+\frac{\omega}4}\right)\E^{(\leq L-2)}_{quiesc}(s)\,ds\\
&\,+\epsilon^\frac14\E^{(L)}_{quiesc}(t)+a(t)^{2-c\sqrt{\epsilon}-(L+1)\omega}\E^{(L)}_{total,Vl}(t)\\
\lesssim&\,\epsilon^4a(t)^{-c\epsilon^\frac18}+\int_t^{t_0}\left(\epsilon^\frac18a(s)^{-3}+a(s)^{-3+\frac{\omega}4}+a(s)^{-2-c\sqrt{\epsilon}}\right)\E^{(L)}_{quiesc}(s)\,ds+\epsilon^\frac14 \E^{(L)}_{quiesc}(t)
\end{align*}
The bound \eqref{eq:quiesc-en-imp} now follows at order $L$ after rearranging and applying the Gronwall lemma, proving the statement.
\end{proof}}

\begin{corollary}[Sobolev norm estimates for the solution variables]\label{cor:en-imp-interim}
\begin{subequations}\label{eq:sob-norm-imp}
\begin{align}
\|\Psi\|_{H^{18}_G(M)}\lesssim&\,\epsilon^2a^{-c\epsilon^\frac18}\label{eq:norm-imp-psi}\\
\|\Sigma\|_{H^{18}_G(M)}\lesssim&\,\epsilon^\frac{15}8a^{-c\epsilon^\frac18}\label{eq:norm-imp-sigma}\\
\|\RE\|_{H^{18}_G(M)}+\|\RB\|_{H^{18}_G(M)}\lesssim&\,\epsilon^\frac{15}8a^{-c\epsilon^\frac18}\\
\|\Ric[G]\change{-2\kappa G}\|_{H^{16}_G(M)}+\|\Gamma-\Gamhat\|_{H^{17}_G(M)}+\|\change{G^{\pm 1}-\gamma^{\pm 1}}\|_{H^{18}_G(M)}\lesssim&\,\epsilon^\frac74 a^{-c\epsilon^\frac18}\label{eq:metric-norm-imp}\\
\|N\|_{H^{16}_G(M)}\lesssim&\,\epsilon^\frac74a^{4-c\epsilon^\frac18}
\end{align}
\end{subequations}
Further, \change{one has }the following high order bounds:
\begin{subequations}\label{eq:metric-top-imp}
\begin{align}
a^4\E^{(17)}(\Ric,\cdot)\lesssim \epsilon^\frac72 a^{-c\epsilon^\frac18}\\
\|\Gamma-\Gamhat\|_{H^{17}_G}+a^4\|\Gamma-\Gamhat\|_{\dot{H}^{18}_G}\lesssim&\,\epsilon^\frac74 a^{-c\epsilon^\frac18}
\end{align}
\end{subequations}
\end{corollary}
\begin{proof}
The first three bounds as well as the Ricci bounds follow directly from Proposition \ref{prop:en-imp} using \cite[Corollary 4.5]{FU23} to switch bound the respective Sobolev norms up to curvature energies. To obtain the lapse bounds, we apply \eqref{eq:lapse-en-est-tilde0} and \eqref{eq:lapse-en-est}, where we observe that, for suitable $\omega$, we can estimate all Vlasov energies using \eqref{eq:vlasov-imp-interim}. The metric and Christoffel bounds then follow from \eqref{eq:norm-est-G} and \eqref{eq:chr-norm-est}.
\end{proof}

With these improvements in hand, we are now able to derive energy bounds for Vlasov matter that are not scaled:

\begin{prop}[Vlasov energy bound]\label{prop:vlasov-improvement} The following bound holds:
\begin{equation}\label{eq:vlasov-en-imp}
\E^{(\leq 18)}_{1,\leq 18}(f,\cdot)\lesssim \epsilon^\frac72a^{-c\epsilon^\frac18}
\end{equation}
\end{prop}
\begin{proof}
As in the proof of Proposition \ref{prop:en-imp}, we again iterate over the number of derivatives, this time in increments of $1$ starting from $L=1$ after treating $L=0$ as a special case first:\\
For $L=0$, note applying the bounds from Corollary \ref{cor:en-imp-interim} to the energy estimate \eqref{eq:en-est-Vlasov-0} implies
\[\E^{(0)}_{1,0}(f,t)\lesssim\epsilon^4+\int_t^{t_0}\left(a(s)^{-1-c\sigma}+\epsilon^\frac18 a(s)^{-3}\right)\E^{(0)}_{1,0}(f,s)+\epsilon^\frac{29}8a(s)^{-3-c\epsilon^\frac18}+\epsilon^\frac72a(s)^{-1-c\epsilon^\frac18}\,ds\,.\]
The stated bound then follows with the Gronwall lemma.\\
For $L=1$, we start with \change{\eqref{eq:vlasov-vertical-en-unscaled}}. Notice that, using \eqref{eq:vlasov-imp-interim}, the term containing the Vlasov energy of higher horizontal order $\E^{(1)}_{1,1}(f,\cdot)$ can be bounded by
\begin{align*}
\int_t^{t_0}\change{a(s)^{-1-2\omega}\cdot a(s)^{2\omega}}\E^{(1)}_{1,1}(f,s)\,ds\lesssim
\int_t^{t_0}a(s)^{-1-\change{2\omega}}\cdot \change{\epsilon^4}a(s)^{-c\epsilon^\frac18}\,ds\,.
\end{align*}
Using the bounds from Corollary \ref{cor:en-imp-interim} on the remaining energies similarly, we obtain
\begin{align*}
\E^{(1)}_{1,0}(f,t)\lesssim&\,\epsilon^4+\int_t^{t_0}\left(a(s)^{-1-c\sigma}+\change{\epsilon^\frac18}a(s)^{-3}\right)\E^{(\leq 1)}_{1,0}(f,s)+\epsilon^\frac72a(s)^{-1-c\epsilon^\frac18}+\change{\epsilon^\frac{29}8}a(s)^{-3-c\epsilon^\frac18}\,ds\,
\end{align*}
and thus, again applying the Gronwall lemma,
\begin{align*}
\E^{(1)}_{1,0}(f,t)\lesssim&\,\epsilon^\frac72a(t)^{-c\epsilon^\frac18}\,.
\end{align*}
To deduce the bound on $\E^{(1)}_{1,1}(f,\cdot)$, we consider \eqref{eq:vlasov-hor-1} and \change{dedude }that the respective integral containing $\E^{(1)}_{1,0}(f,t)$ therein can be bounded as follows:
\begin{align*}
\int_t^{t_0}\epsilon^\frac{15}8 a(s)^{-3-c\epsilon^\frac18}\E^{(1)}_{1,0}(f,s)\,ds\lesssim \epsilon^\frac{21}4 a(t)^{-c\epsilon^\frac18}
\end{align*}
Applying the bounds from Corollary \ref{cor:en-imp-interim} for the remaining energies proves the bound for $\E^{(1)}_{1,1}(f,\cdot)$ with the Gronwall lemma as before.\\
For $1<L\leq 18$, assume the statement has been shown up to $L-1$. Starting again with the vertical energy estimate \eqref{eq:vlasov-vertical-en} for $\beta=0$, the bound 
\[\E^{(L)}_{1,0}(f,\cdot)\lesssim\epsilon^\frac72 a^{-c\epsilon^\frac18}\]
follows by the same argument as for $L=1$, in particular using \eqref{eq:vlasov-imp-interim} to bound $\E^{(L)}_{1,1}(f,\cdot)$. Next, let $K\in\N, L>K>0$ and assume the bound to hold for $\E^{(L)}_{1,\leq K-1}(f,\cdot)$, and consider the intermediate energy estimate \eqref{eq:vlasov-interim}. Again, notice that
\[\int_t^{t_0}a(s)^{-1}\E^{(L)}_{1,K+1}(f,s)\,ds\lesssim \epsilon^\frac72 a^{-c\epsilon^\frac18}\]
holds by \eqref{eq:vlasov-imp-interim}, and on the other hand,
\begin{align*}
\int_t^{t_0}\left(\epsilon^\frac78a(s)^{-3-c\sqrt{\epsilon}}+a(s)^{-1-c\sqrt{\epsilon}}\right)\left(\E^{(L)}_{1,K-1}(f,s)+\E^{(\leq L-1)}_{1,\leq K}(f,s)\right)\,ds\lesssim \epsilon^\frac72 a^{-c\epsilon^\frac18}
\end{align*}
is satisfied since improved bounds have been established for all energies in the integrand. The shear and scalar field energies can be treated as before, and \eqref{eq:sob-norm-imp} leads to similar terms for curvature and metric error terms that are not at top order. At highest order (i.e., $L=18$), we note that \change{high order curvature terms }that occur in \eqref{eq:vlasov-imp-interim} are appropriately scaled by \eqref{eq:metric-top-imp}. Altogether, \eqref{eq:vlasov-interim} then leads to the following estimate:
\[\E^{(L)}_{1,K}(f,\cdot)\lesssim\epsilon^4+\int_t^{t_0}\left(a(s)^{-1-c\sigma}+\epsilon^\frac18a(s)^{-3}\right)\E^{(L)}_{1,K}(f,s)+\epsilon^\frac72a(s)^{-1-\omega-c\sqrt{\epsilon}}+\epsilon^\frac{29}8a(s)^{-3-c\epsilon^\frac18}\,ds,\]
and thus $\E^{(L)}_{1,\leq K}(f,\cdot)\lesssim\epsilon^\frac72 a^{-c\epsilon^\frac18}$. Having now proven the bound up to $\E^{(L)}_{1,\leq L-1}(f,\cdot)$ by iterating over $K$, we can apply the same argument to \eqref{eq:vlasov-hor}. Again, note that even at highest order, the metric and curvature bounds \eqref{eq:sob-norm-imp} and \eqref{eq:metric-top-imp} are sufficient. This implies the stated bound for $\E^{(L)}_{1,\leq L}(f,\cdot)$, completing the iteration step and thus proving \eqref{eq:vlasov-en-imp}.
\end{proof}

Now, we are in the position to improve the remaining bootstrap assumption \eqref{eq:BsC}:

\begin{corollary}[Bootstrap improvement]\label{cor:bootstrap-imp} For $\mathcal{C}$ as in \eqref{eq:def-C} and $t\in(t_{Boot},t_0]$, the following holds:
\begin{equation}\label{eq:BsCimp}
\mathcal{C}(t)\lesssim\epsilon^\frac74 a(t)^{-c\epsilon^\frac18}
\end{equation}
\end{corollary}
\begin{proof}
Regarding the Vlasov matter quantities, note that Lemma \ref{lem:density-control} only requires bounds on $\E^{(L)}_{1,L}(f,\cdot)$ for $L\in\N$. Inserting \eqref{eq:vlasov-en-imp} and the metric bound from \eqref{eq:sob-norm-imp} then implies
\begin{align*}
\|\rho^{Vl}-{\rho}_{FLRW}^{Vl}\|_{H^{18}_G(M)}+\|\mathfrak{p}^{Vl}-{\mathfrak{p}}_{FLRW}^{Vl}\|_{H^{18}_G(M)}+\|\change{\j^{Vl}}\|_{H^{18}_G(M)}+\|(S^{Vl})^\parallel\|_{H^{18}_G(M)}\lesssim \epsilon^\frac74 a^{-c\epsilon^\frac18}
\end{align*}
Furthermore, applying the above along with \eqref{eq:norm-imp-psi} and \eqref{eq:norm-imp-sigma} to Lemma \ref{lem:norm-est-nablaphi} for $0\leq l\leq 18$ yields
\[\|\nabla\phi\|_{H^{18}_G(M)}\lesssim \epsilon^\frac74 a^{-c\epsilon^\frac18}\]
after updating $c>0$. This, along with Corollary \ref{cor:en-imp-interim}, proves
From here, the statement follows for all variables measured by $\mathcal{C}$ from the Sobolev embedding $H^{l+2}_\gamma(M)\hookrightarrow C^{l}_\gamma(M)$ for $l\in\N$, as well as the estimates in \cite[Lemma 7.4]{FU23} to relate norms with respect to $\gamma$ and $G$ to one another.
\end{proof}

Additionally, Proposition \ref{prop:vlasov-improvement} implies the following bound on the distribution function compared to the reference data:

\begin{corollary}[Sobolev norm bound for the Vlasov distribution function]\label{cor:norm-imp-vlasov} On the bootstrap interval, one has the following:
\begin{align}
\|f-f_{FLRW}\|_{H^{18}_{1,{\G_0}}(\change{T^\ast M_t})}\lesssim&\,\epsilon^\frac74 a(t)^{-c\epsilon^\frac18}
\end{align}
\end{corollary}
\begin{proof}
For $L\leq 18$, we sum \eqref{eq:vlasov-rearrange-no-ref-K} over $K=1,\dots,L-1$ as well as all possible derivative combinations represented by $\mathcal{D}_{L,K}$.  
Note that the curvature error terms in \eqref{eq:vlasov-rearrange-no-ref-K} only necessitate control of the curvature energy $\E^{\change{(\leq 16)}}(\Ric,\cdot)$ as well as low order norms for the Christoffel symbols and the metric, all of which are covered in Corollary \ref{cor:en-imp-interim}. Thus, inserting the Vlasov energy bound \eqref{eq:vlasov-en-imp} and Corollary \ref{cor:en-imp-interim} then implies
\begin{equation*}
\|f\|^2_{\dot{H}^L_{1,\G_0}(\change{T^\ast M})}-\|\nabsak_{vert}^Lf\|^2_{L^2_{1,\G_0}(\change{T^\ast M})}\lesssim \epsilon^\frac72 a^{-c\epsilon^\frac18}
\end{equation*}
Furthermore, by Lemma \ref{lem:hor-deriv-ref} and Corollary \ref{cor:en-imp-interim}, one also has
\begin{equation*}
\|f_{FLRW}\|^2_{\dot{H}^L_{1,\G_0}(\change{T^\ast M})}-\|\nabsak_{vert}^Lf_{FLRW}\|^2_{L^2_{1,\G_0}(\change{T^\ast M})}\lesssim \epsilon^\frac72 a^{-c\epsilon^\frac18}
\end{equation*}
The statement then follows by expanding $\|f-f_{FLRW}\|^2_{\dot{H}^L_{1,\G_0}(\change{T^\ast M_t})}$ and using \eqref{eq:vlasov-imp-interim} to bound $\E^{(L)}_{1,0}(f,\cdot)$.
\end{proof}

\section{The main theorem}\label{sec:main-thm}

\noindent We are now in the position to prove the main theorem of this paper:

\begin{theorem}[Past stability of FLRW solutions to the ESFV system]\label{thm:main} Consider FLRW solutions of the form \eqref{eq:FLRW} to the Einstein scalar-field Vlasov equations \eqref{eq:EVSF} for which, in the massless case, $\change{f}_{FLRW}$ vanishes on an open neighbourhood of the zero section of $P$. Let $(M,\mathring{g},\mathring{k},\mathring{\pi},\mathring{\psi},\mathring{f})$ be CMC initial data to the Einstein scalar-field Vlasov system (see Subsection \ref{subsubsec:initial-data}) that is $\epsilon^2$-close to the FLRW data on $M_{t_0}$ in the sense of Assumption \ref{ass:init}. \\

Then, the past maximal globally hyperbolic development $(\M,\g,\phi,\overline{f})$ to the Einstein scalar-field Vlasov system \eqref{eq:EVSF} admits a foliation by CMC Cauchy hypersurfaces $M_s=t^{-1}(\{s\}),\ s\in(0,t_0]$ such that the solution norms $\mathcal{H}$ and $\mathcal{C}$ (see Definition \ref{def:sol-norm}) satisfy
\begin{subequations}
\begin{equation}\label{eq:norm-main-thm}
\mathcal{H}(t)+\mathcal{C}(t)\lesssim \epsilon^\frac74a(t)^{-c\epsilon^\frac18}
\end{equation}
and one has
\begin{equation}\label{eq:low-Vlasov-main-thm}
\|(f-f_{FLRW})\|_{C^{11}_{1,\G_0}(\change{T^\ast M_t})}\lesssim \sqrt{\epsilon} a(t)^{-c\sqrt{\epsilon}}\,.
\end{equation}
\end{subequations}
In particular, this implies the following asymptotic behaviours for the solution variables: There exist a scalar function $\Psi_{Bang}\in C_\gamma^{14}(M)$, a $(1,1)$-tensor field $K_{Bang}\in C^{14}_{\gamma}(M)$, a volume form $\vol{Bang}\in C^{14}_\gamma(M)$ and a $(0,2)$-tensor field $M_{Bang}\in C^{14}_\gamma(M)$ such that the following bounds hold for any $t\in(0,t_0]$:
\begin{subequations}\label{eq:asymp}
\begin{align}
\|n-1\|_{C^{14}_\gamma(M_t)}\lesssim&\,\epsilon a(t)^{2-c\epsilon^\frac18}\label{eq:asymp-lapse}\\
\|a^{-3}\vol{g}-\vol{Bang}\|_{C^{14}_\gamma(M_t)}\lesssim&\,\epsilon a(t)^{2-c\epsilon^\frac18}\label{eq:asymp-vol}\\
\|a^3\del_t\phi-(\Psi_{Bang}+C)\|_{C^{14}_\gamma(M_t)}\lesssim&\,\epsilon a(t)^{2-c\epsilon^\frac18}\label{eq:asymp-psi}\\
\change{\left\|\nabhat\phi(t,\cdot)-\nabhat\phi(t_0,\cdot)+\int_t^{t_0}a(s)^{-3}\,ds\cdot \nabhat\Psi_{Bang}\right\|_{C^{14}_\gamma(M_t)}}\lesssim&\,\epsilon a(t)^{2-c\epsilon^\frac18}\\
\|a^3k-K_{Bang}\|_{C_\gamma^{14}(M_t)}\lesssim&\,\epsilon a(t)^{2-c\epsilon^\frac18}\label{eq:asymp-K}\\
\left\|g\odot\exp\left[-2\int_t^{t_0}a(s)^{-3}\,ds\cdot K_{Bang}\right]-M_{Bang}\right\|_{C^{12}_\gamma(M_t)}\lesssim&\,\epsilon a(t)^{2-c\epsilon^\frac18}\label{eq:asymp-g}
\end{align}
\end{subequations}
Furthermore, one has
\begin{subequations}
\begin{equation}\label{eq:asymp-constr}
8\pi(\Psi_{Bang}+C)^2+\lvert K_{Bang}\rvert^2=12\pi C^2,\quad {(K_{Bang})^a}_a=-\sqrt{12\pi}C
\end{equation}
and, denoting ${\left(K_{Bang}^\parallel\right)^i}_j={(K_{Bang})^i}_j+\sqrt{\frac{4\pi}3}C\I^i_j$,
\begin{equation}\label{eq:footprint-control}
\|\vol{\gamma}-\vol{Bang}\|_{C^{15}_\gamma(M)}+\|M_{Bang}-\gamma\|_{C^{15}_\gamma(M)}+\|K_{Bang}^\parallel\|_{C^{15}_\gamma(M)}+\|\Psi\|_{C^{15}_\gamma(M)}\lesssim\epsilon\,.
\end{equation}
\end{subequations}
\change{Additionally, there exists $f_{\mathrm{Bang}}\in C^{10}_{1,\underline{\gamma}_0}(T^\ast M)\,$ such that, for any $t\in(t,t_0]$, one has
\begin{subequations}
\begin{equation}\label{eq:vtd-f}
\|f-f_{\mathrm{Bang}}\|_{C^{10}_{1,\underline{\gamma}_0}(T^\ast M_t)}\lesssim\sqrt{\epsilon}a(t)^{2-c\epsilon^\frac18}
\end{equation}
where the footprint state $f_{\mathrm{Bang}}$ can be controlled as follows.
\begin{equation}\label{eq:footprint-control-Vlasov}
\|f_{\mathrm{Bang}}-f_{FLRW}\|_{C^{10}_{1,\underline{\gamma}_0}(T^\ast M)}\lesssim\sqrt{\epsilon}
\end{equation}
The momentum support of $f_{\mathrm{Bang}}$ is bounded and, in the massless case, bounded away from the origin with respect to $\underline{\gamma}\vert_{vert}$. The solution norm bound \eqref{eq:norm-main-thm} also implies the following Vlasov matter bounds:
\begin{align}
\|f-f_{FLRW}\|_{H^{18}_{1,\G_0}(T^\ast M_t)}\lesssim&\,\epsilon^\frac74a(t)^{-c\epsilon^\frac18}\label{eq:sob-norm-f-main}\\
\left\|T_{00}^{Vl}[\g,{f}]-T_{00}^{Vl}[\g_{FLRW},f_{FLRW}]\right\|_{C^{16}_\gamma(M_t)}\lesssim&\,\epsilon^\frac74 a(t)^{-4-c\epsilon^\frac18}\label{eq:sob-norm-f-rho}\\
\left\|T_{0(\cdot)}^{Vl}[\g,{f}]-T_{0(\cdot)}^{Vl}[\g_{FLRW},{f}_{FLRW}]\right\|_{C^{16}_\gamma(M_t)}\lesssim&\,\epsilon^\frac74 a(t)^{-3-c\epsilon^\frac18}\label{eq:sob-norm-f-j}\\
\left\|T_{(\cdot)(\cdot)}^{Vl}[\g,{f}]-T_{(\cdot)(\cdot)}^{Vl}[\g_{FLRW},{f}_{FLRW}]\right\|_{C^{16}_\gamma(M_t)}\lesssim&\,\epsilon^\frac74 a(t)^{-4-c\epsilon^\frac18}\label{eq:sob-norm-f-S}
\end{align}
\end{subequations}}
The FLRW solution exhibits stable mean, scalar, Ricci and Kretschmann curvature scalar blow-up toward the Big Bang in the following sense:
\begin{subequations}\label{eq:blow-up-rates}
\begin{align}
\left\|a^6\lvert k\rvert_g^2-\left\lvert K_{Bang}\right\rvert^2\right\|_{C^0(M_t)}\lesssim&\,\epsilon a(t)^{2-c\epsilon^\frac18}\\
\left\|a^6R[\g]-8\pi(\Psi_{Bang}+C)^2\right\|_{C^0(M_t)}\lesssim&\,\epsilon a(t)^{2-c\epsilon^\frac18}\\
\left\|a^{12}\Ric[\g]_{\alpha\beta}\Ric[\g]^{\alpha\beta}-(8\pi)^2(\Psi_{Bang}+C)^4\right\|_{C^0(M_t)}\lesssim&\,\epsilon a(t)^{2-c\epsilon^\frac18}\\
\left\|a^{12}\mathcal{K}-\frac52\cdot (8\pi)^2(\Psi+C)^2\right\|_{C^0(M_t)}\lesssim&\, \epsilon a(t)^{2-c\epsilon^\frac18}
\end{align}
\end{subequations}
Finally, $\left(\M,\g\right)$ is causally disconnected and geodesically incomplete toward the past in the same sense as in \cite[Theorem 8.2]{FU23}. 
\end{theorem}

\begin{remark}[On AVTD behaviour]\label{rem:avtd}
\eqref{eq:asymp} demonstrates that the evolution of metric and scalar field are asymptotically velocity term dominated (AVTD) in the sense that $(g,k,n,\del_0\phi)$ exhibit asymptotic behaviour which, agrees at first order with the solution to the truncated Einstein Vlasov scalar-field system in CMC gauge in which we drop all spatial derivatives from \eqref{eq:REEqG}, \eqref{eq:REEqSigmaSharp}, \eqref{eq:REEqLapse1} and \eqref{eq:REEqWave}. Note that, when considering the evolution of ${k^i}_j$ (or, equivalently, its tracefree part ${\hat{k}^i}_j$) as well as the lapse equation, Vlasov matter terms are not truncated, leading to lower order terms in the VTD solution -- as opposed to the pure scalar-field case, where the VTD solution is simply described by $a^{-3}K_{Bang}$. \change{Similarly truncating the rescaled Vlasov equation \eqref{eq:vlasov-resc} simply leads to the equation $\del_tf_{VTD}=0$, if one views the full Sasaki derivative $\nabsak_i$ as the horizontal term to be truncated as the form of the isotropic solution suggests. Thus, the Vlasov distribution on the co-mass shell is also AVTD, but one loses control of the derived energy-momentum quantities since the leading order term is influenced by the level of inhomogeneity encoded in even the formal VTD solution via the asymptotics of the spatial metric. Thus, these asymptotics are in line with the critical behaviour of radiation fluids discussed in Section \ref{subsubsec:euler-sf}, since the corresponding macroscopic quantities degenerate towards the Big Bang, but at controllable rates.}
\end{remark}

\begin{proof}[Proof of Theorem \ref{thm:main}]
Firstly, we can assume without loss of generality that the initial data is CMC by Remark \ref{rem:cmc}. If $\epsilon$ is chosen sufficiently small, Lemma \ref{lem:local-wp-CMC} implies that there is a past maximal solution $(\M,\g,\phi,\overline{f})$ to the ESFV-system with $\M\simeq(\mathfrak{t},t_0]\times M$ for $\mathfrak{t}\in[0,t_0]$. Furthermore, we can assume sufficient regularity for the initial data such that all norms and energies are continuous in time (see Remark \ref{rem:cmc})\footnote{In particular, for any $l\in\N$, we again note that the support assumptions on $\mathring{f}$ ensure that \change{$\mathring{f}\in H^l_{1,{\underline{\gamma}}_0,m}(T^\ast M_{t_0})$ } is equivalent to  \change{$\mathring{f}\in H^l_{\mu,{\underline{\gamma}}_0,m}(T^\ast M_{t_0})$ } for any $\mu\geq 0$.}.\\

\noindent Thus, for some suitable $K_0>0$, the bootstrap assumptions hold on $(t_{Boot},t_0]\times M$ for $t_{Boot}\in[\mathfrak{t},t_0)$ (see Assumption \ref{ass:bootstrap}). If $\epsilon>0$ is, again, chosen to be small enough, the assumption \eqref{eq:BsC} is improved by Corollary \ref{cor:bootstrap-imp}, and \eqref{eq:BsVlasovhor} is improved by \eqref{eq:APVlasov}. Thus, \eqref{eq:norm-main-thm} and \eqref{eq:low-Vlasov-main-thm} hold on $(\mathfrak{t},t_0]$, and consequently \eqref{eq:asymp} as in the proof of \cite[Theorem 8.2]{FU23}. In particular, if $\mathfrak{t}>0$ were to hold, none of the blow-up criteria in Lemma \ref{lem:local-wp-CMC} would be satisfied approaching $\mathfrak{t}$:
\begin{enumerate}
\item By \eqref{eq:BsCimp}, the eigenvalues of $(G^{-1})^{ij}$ differ from those of $(\gamma^{-1})^{ij}$ by at most $K\epsilon^\frac74 a^{-c\epsilon^\frac18}$ for some $K>0$, and are thus bounded from above by $1+\epsilon^\frac74a^{-c\epsilon^\frac18}$ up to multiplicative constant. Hence, the eigenvalues of $G_{ij}$ are bounded away from $0$.
\item \eqref{eq:asymp-lapse} prevents (2).
\item \eqref{eq:asymp-psi} prevents (3).
\item \eqref{eq:blowup-crit-geom} and \eqref{eq:blowup-crit-Vlasov-quant} are directly prevented by \eqref{eq:norm-main-thm}. Regarding \eqref{eq:blowup-crit-Vlasov-harm}, we need to verify after scaling transformation that
\begin{equation}\label{eq:vlasov-blowup-check}
\int_{\R^3}\langle v\rangle_\gamma^{2\cdot 3}\cdot \left(\frac{\langle v\rangle_\gamma}{v^0_\gamma}\right)^{2l}\left(\lvert \nabsak_{vert}^{\leq 1}f\rvert_{\underline{\gamma}_0}^2\right)\change{\mu_\gamma^{-1}} dv<\infty
\end{equation}
holds, where $l\geq 18$ is a sufficiently large natural number. For the massive case, the factor containing $l$ is simply $1$. Thus, expanding $f$ into $f_{FLRW}+(f-f_{FLRW})$, and using \eqref{eq:asymp-g}, \eqref{eq:norm-main-thm} and the momentum support bound \eqref{eq:APMom}, \eqref{eq:vlasov-blowup-check} is bounded by
\begin{equation*}
\lesssim \|f_{FLRW}\|^2_{H^1_{3,\underline{\gamma}_0}(\change{T^\ast M})}+\left(1+\|\change{G^{-1}-\gamma^{-1}}\|_{C^0_G(M)}\right)\cdot (\P^0)^4\cdot \mathcal{H}^2\lesssim 1+\epsilon^\frac72 a^{-c\epsilon^\frac18}\,.
\end{equation*}
In the massless case, \change{\eqref{eq:APMomMasslessgamma} }implies
\[\left(\frac{\langle v\rangle_\gamma}{v^0_\gamma}\right)^{2l}=(1+\lvert v\rvert_\gamma^{-2})^l\lesssim \change{1}\,,\]
and thus the same upper bound holds after updating constants. 
\end{enumerate}
Hence, one must have $\mathfrak{t}=0$. \eqref{eq:asymp-constr}-\eqref{eq:footprint-control} follow directly from \eqref{eq:asymp}, the rescaled constraint equations \eqref{eq:REEqHam} and \eqref{eq:REEqMom} and the initial data assumption, as in the proof of \cite[(8.4)-(8.6)]{FU23}.\\

\change{\noindent Regarding \eqref{eq:vtd-f}, we use \eqref{eq:norm-main-thm} and \eqref{eq:low-Vlasov-main-thm} to obtain
\begin{align*}
\left\lvert\frac{d}{dt}f(t,x,v_i)\right\rvert\lesssim&\,a^{-1-c\epsilon^\frac18}\left(\|N\|_{C^1_G(M_t)}+\|f-f_{FLRW}\|_{\dot{C}^1_{1,\G_0}(T^\ast M_t)}+\|G^{-1}-\gamma^{-1}\|_{C^1(M_t)}\right)\\
\lesssim&\,\epsilon a^{-1-c{\epsilon}^\frac18}
\end{align*}
Further, recall that $\langle v\rangle_\gamma=\sqrt{1+\lvert v\rvert_\gamma^2}$ remains bounded on the momentum support of $f$, see \eqref{eq:APMomgamma}. Thus, one has
\begin{equation*}
\left\lvert \langle v\rangle_\gamma \frac{d}{dt}f(t,x,v)\right\rvert\lesssim\sqrt{\epsilon}a^{-1-c\epsilon^\frac18}\,.
\end{equation*}
By integrating in time, this implies for any $s,t\in(0,t_0)$ that
\[\sup_{{x,v}\in \change{T^\ast M}}\langle v\rangle_\gamma\left\lvert f(t,x,v)-f(s,x,v)\right\rvert\lesssim\sqrt{\epsilon}\left\lvert a(t)^{2-c\epsilon^\frac18}-a(s)^{2-c\epsilon^\frac18}\right\rvert\]
holds and thus that $f$ converges to a continuous function $f_{Bang}$ in $C^0_{1,\underline{\gamma}_0}(T^\ast M)$ uniformly when approaching the Big Bang singularity. The bound \eqref{eq:footprint-control-Vlasov} in $C^0_{1,\underline{\gamma}_0}(T^\ast M)$ then follows directly from above for $t=t_0$ and $s\downarrow 0$. For higher orders, the argument is analogous, instead considering 
\[\frac{d}{dt}\left(\langle v\rangle_\gamma^{L-K}\,\hat{\mathcal{D}}_{L,K}f\right)=\langle v\rangle_\gamma^{L-K}\,\hat{\mathcal{D}}_{L,K}\left(\X f\right)\,,\]
where $\hat{\mathcal{D}}_{L,K}$ denotes as combination of $L-K$ vertical and $K$ horizontal covariant derivatives with respect to $(T^\ast M,\underline{\gamma})$. More precisely, it goes through as long as \eqref{eq:low-Vlasov-main-thm} can be used to control derivatives of one order higher, i.e., the argument extends to $C^{10}_{1,{\underline{\gamma}}_0}(\change{T^\ast M})$. In particular, note that $\del_t$ and covariant derivatives with respect to $\hat{\nabsak}$ commute, and commutation errors between Sasaki metrics with respect to $\gamma$ and $G$ can be controlled sufficiently by \eqref{eq:norm-main-thm}. The momentum support bounds for $f_{Bang}$ are obtained by considering the limit of the support bounds with respect to $\gamma$ in Lemma \ref{lem:APMom}.\\}
The remaining Vlasov matter bounds follow directly from \eqref{eq:norm-main-thm}. Everything else is proven identically to \cite[Theorem 8.2]{FU23} and \cite[Theorem 15.1]{Rodnianski2014} -- we note the convergence rates in \eqref{eq:blow-up-rates} are weaker than in these results since the presence of Vlasov matter leads to a weaker convergence rate for the lapse (see \eqref{eq:asymp-lapse}).
\end{proof}

\change{It was not necessary in the analysis above to assume that the momentum support of the initial perturbation was close to that of the FLRW solution, as long as their joint support was uniformly bounded and, in the massless case, bounded away from zero. However, if one additionally requires this, this property is preserved throughout the evolution.

\begin{corollary}\label{cor:main-thm-support}
Consider the setting of Theorem \ref{thm:main} and let $\mathrm{dist}_{\underline{\gamma}}$ denote the distance function on $T^\ast M$ induced by the Sasaki metric $\underline{\gamma}$. If one additionally assumes
\begin{equation}\label{eqref:footprint-control-Vlasov-support}
\mathrm{dist}_{\underline{\gamma}}(\supp \mathring{f},\supp f_{FLRW})\leq \delta
\end{equation}
for some $\delta>0$, then the following holds for the limiting Vlasov distribution $f_{Bang}$ on the Big Bang hypersurface.
\begin{equation}\label{eq:footprint-control-Vlasov-support}
\mathrm{dist}_{\underline{\gamma}}(\supp f_{\mathrm{Bang}},\supp f_{FLRW})\lesssim \delta+\epsilon^\frac74\,.
\end{equation}
\end{corollary}

\begin{proof}
Let $(\hat{X}_{s}(t;x,v),\hat{V}_{t}(s;x,v))$ denote the characteristics of the FLRW solution, emanating from $(x,v)\in T^\ast M_{t_0}$, i.e., the solution to the following ODE.
\begin{subequations}\label{eq:char-eq-gamma}
\begin{align}
\frac{d}{dt}{\hat{X}_s}^i=&\,a^{-1}\,(\gamma^{-1})^{ij}\,\frac{\hat{V}_i}{\sqrt{m^2a^2+\lvert \hat{V}\rvert_\gamma^2}} &\hat{X}_{s}(s;x,v)=&\,x\\
\frac{d}{dt}(\hat{V}_s)_i=&\,a^{-1}\,(\gamma^{-1})^{jk}\,\Gamhat^l_{ik}\,\frac{\hat{V}_j\hat{V}_l}{\sqrt{m^2a^2+\lvert \hat{V}\rvert_\gamma^2}}& \hat{V}_{s}(s;x,v)=&\,v
\end{align}
\end{subequations}
Our first goal is to prove that these characteristics are close to the characteristics $(X,V)$ of f given by \eqref{eq:char}, i.e., that
\begin{equation}\label{eq:char-bound-goal}
\lvert X_{t_0}-\hat{X}_{t_0}\rvert_\gamma+\lvert V_{t_0}-\hat{V}_{t_0}\rvert\lesssim\epsilon^\frac74\,.
\end{equation}
holds on $(t,x,v)\in (0,t_0]\times T^\ast M$. As in the proof of Lemma \ref{lem:APVlasov}, the characteristic equations for $X$ and $\hat{X}$ ensure that, for a suitable finite partition $\{(t_j,t_j-1]\}_{j=1,\dots,J}$ of $(0,t_0]$, we can choose a finite covering of $M$ such that any characteristic remains contained in a coordinate neighbourhood for each step of the partition. We can also choose this neighbourhood such that Christoffel symbols $\Gamma$ and $\Gamhat$ are Lipschitz continuous. it is sufficient to prove
\begin{equation}\label{eq:char-bound-goal-step}
\lvert (X_{t_{j-1}}-\hat{X}_{t_{j-1}})(t,x,v)\rvert_\gamma+\lvert (V_{t_{j-1}}-\hat{V}_{t_{j-1}})(t,x,v)\rvert\lesssim\epsilon^\frac74
\end{equation}
for any $t\in(t_j,t_{j-1}]$ and $(x,v)$ in within a coordinate neighbourhood of $(t_j,t_{j-1}]\times T^\ast M$. We will drop the subscript $t_{j-1}$ for the proof \eqref{eq:char-bound-goal-step} for notational convenience.\\
Using \eqref{eq:norm-main-thm} as well as Lemma \ref{lem:APMom}, \eqref{eq:char-eq-gamma} and \eqref{eq:char} imply the following bound on such a neighbourhood:
\begin{align*}
\left\lvert\frac{d}{dt}\left(V-\hat{V}\right)\right\rvert\lesssim&\, a^{-1-c\epsilon^\frac18}\,\left(\lvert N\rvert+\lvert G^{-1}-\gamma^{-1}\rvert_\gamma+\lvert \Gamma-\Gamhat\rvert_\gamma+\frac{\lvert V\rvert_\gamma+\lvert \hat{V}\rvert_\gamma}{V^0}\,\lvert V-\hat{V}\rvert_\gamma\right)\\
\lesssim&\,a^{-1-c\epsilon^\frac18} (\epsilon^\frac74 +\lvert V-\hat{V}\rvert_\gamma)
\end{align*}
Since one has $(V-\hat{V})(t_{j-1},\cdot,\cdot)=0$, the Gronwall lemma implies 
\[\lvert (V-\hat{V})_{t_0}(t,x,v)\rvert\lesssim\epsilon^\frac74\] using \eqref{eq:a-integrals}. Consequently, one has
\begin{align*}
\left\lvert\frac{d}{dt}\left(X-\hat{X}\right)\right\rvert_\gamma\lesssim&\, a^{-1-c\epsilon^\frac18}\left(\lvert N\rvert+\lvert\Gamma-\Gamhat\rvert_\gamma+\lvert V-\hat{V}\rvert_\gamma\right)\\
\lesssim&\,\epsilon^\frac74 a^{-1-c\epsilon^\frac18}
\end{align*}
and we obtain \eqref{eq:char-bound-goal-step} after integrating. By patching coordinate neighbourhoods and iterating over $j=1,\dots,J$, this implies \eqref{eq:char-bound-goal}.\\

Now, note that solutions to \eqref{eq:char-eq-gamma} depend on their initial data Lipschitz continuously with respect to $\gamma$ due to our coordinate choice. Thus, if one has
\[\mathrm{dist}_{\underline{\gamma}}((y,w),\supp f_{FLRW}(t_0,\cdot,\cdot))\leq\delta\]
for some $(y,w)\in T^\ast M_{t_0}$, then one has
\begin{equation}\label{eq:char-cont-init-data}
\mathrm{dist}_{\underline{\gamma}}((\hat{X},\hat{V})_{t_0}(0;y,w),\supp f_{FLRW}(0,\cdot,\cdot))\lesssim\delta
\end{equation}
since $f_{FLRW}$ is time-independent. Finally, for any $(x,v)\in \supp f_{\mathrm{Bang}}$, one has $(X_0,V_0)(t_0;x,v)\in\supp \mathring{f}$. By assumption, we can apply \eqref{eq:char-cont-init-data} to the characteristic flow of this point with respect to the $(\hat{X},\hat{V})$-characteristics and the distance of the end point of this characterstic to $(x,v)$ is bounded up to a constant by $\epsilon^\frac74$ due to \eqref{eq:char-bound-goal}. This implies \eqref{eq:footprint-control-Vlasov-support}.
\end{proof}

Additionally, the following analogue of \eqref{eq:asymp-g} holds and admits a limiting renormalised metric on the Big Bang hypersurface in the style used in \cite{GPR23}.
\begin{corollary}\label{cor:renorm-3+1}
Let $(\M,\g,\phi,f)$ be as in Theorem \ref{thm:main} and consider the renormalised metrics
\[\mathcal{G}_{t}(X,Y)=a(t)^{-2}\,\left(g_{t}(\mathcal{W}_{t}(X),\mathcal{W}_{t}(Y))\right)\,,\]
where
\begin{equation*}
\mathcal{W}_{t,x}(w)=\exp\left[-\int_t^{t_0}a(s)^{-3}\,ds\cdot (\Sigma^\sharp)_{t,x}\right]\odot w\,.
\end{equation*}
Then, there exists a Riemannian metric $\mathcal{G}_{\normalfont Bang}$ on $M$ such that the following holds.
\begin{align}
\label{eq:control-renorm-3+1}\|\mathcal{G}_{\normalfont Bang}-\gamma\|_{C^{15}_\gamma(M)}\lesssim&\,\epsilon\\
\label{eq:asymp-renorm-3+1}\|\mathcal{G}-\mathcal{G}_{\normalfont Bang}\|_{C^{14}_\gamma(M_t)}\lesssim&\, \epsilon\,a(t)^{2-c\,\epsilon^\frac18}
\end{align}
\end{corollary}
\begin{proof}
The evolution equation \eqref{eq:REEqG} implies
\begin{align*}
\del_t\mathcal{G}_{ij}=&\,-2\,a^{-3}\,(N+1)\,\Sigma_{rs}\,{\mathcal{W}^r}_i\,{\mathcal{W}^s}_j+2\,N\,\frac{\dot{a}}a\,\mathcal{G}_{ij}\\
&\,+2\,a^{-2}\,G_{rs}\,(\Sigma^\sharp)^l_{\ i}{\mathcal{W}^r}_l\,{\mathcal{W}^s}_j\\
&\,-2\,\int_{(\cdot)}^{t_0}a(s)^{-3}\,ds\cdot G_{rs}\,(\del_t\Sigma^\sharp)^l_{\ i}\,{\mathcal{W}^r}_l\,{\mathcal{W}^s}_j\,.
\end{align*}
Using that $\Sigma^\sharp$ and, consequently, $\mathcal{W}$ are self-adjoint with respect to $G$, the second line cancels with terms in the first, leading to
\[\del_t\mathcal{G}_{ij}=2\,N\,\frac{\dot{a}}a\,\mathcal{G}_{ij}-2\,a^{-3}\,N\,\Sigma_{rs}\,{\mathcal{W}^r}_i\,{\mathcal{W}^s}_j-2\,G_{rs}\,\int_{(\cdot)}^{t_0}a(s)^{-3}\,ds\,(\del_t\Sigma^\sharp)^l_{\ i}\,{\mathcal{W}^r}_l\,{\mathcal{W}^s}_j\,.\]
Note that Theorem \ref{thm:main} implies the following bounds for some $c>0$, using Lemma \ref{lem:scale-factor} and \eqref{eq:REEqSigmaSharp}.
\begin{align*}
\lvert\mathcal{W}\rvert_\gamma\lesssim&\,a^{-c\,\epsilon}\\
\lvert \del_t\Sigma^\sharp\rvert_\gamma\lesssim&\,{\epsilon}\,a^{-1-c\,\epsilon^\frac18}
\end{align*}
Additionally controlling the lapse with \eqref{eq:asymp-lapse}, one has
\[\lvert \del_t\mathcal{G}\rvert_\gamma\lesssim\epsilon\,a^{-1-c\,\epsilon^\frac18}(\lvert\mathcal{G}\rvert_\gamma+1)\,.\]
The existence of $\mathcal{G}_{\normalfont Bang}$ as well as the bounds \eqref{eq:control-renorm-3+1}-\eqref{eq:asymp-renorm-3+1} at order $0$ now follow as in the proof of Theorem \ref{thm:main}. At higher orders, the respective estimates are obtained similarly.
\end{proof}}

\begin{remark}[Interpreting the asymptotics on the mass shell]\label{rem:save-my-skin}
\change{One might hope to obtain similar asymptotics for the Vlasov distribution on the mass shell. In particular, when interpreting the Vlasov distribution as a kinetic description of particle trajectories, the velocity of a particle is more naturally understood as a section of the mass shell, while the co-mass shell is the more natural choice from the Hamiltonian perspective. However, even on the level of the VTD solution, understanding the picture on the mass shell becomes difficult: Viewing the Vlasov equation \eqref{eq:Vlasov-sharp} on the mass shell, it is convenient to view this equation in CMCTC gauge in expansion normalised momentum variables $w^i=a(t)^2q^i$, the characteristics of which are clearly preserved in the FLRW case, and to study 
\[f^{\sharp}_{\mathrm{resc}}(t,x,w):=f^\sharp(t,x,a(t)^{-2}w).\]
Truncating the resulting evolution equation in terms of these variables and denoting the lapse and shear of the VTD solution by $n_{VTD}$ and $k_{VTD}^\parallel$, one obtains the following:
\[\del_tf^\sharp_{VTD}=-2(n_{VTD}-1)\frac{\dot{a}}aw^j-2n_{VTD}{(k_{VTD}^\parallel)^j}_{i}w^i\del_{w^j}f_{VTD}\]
For initial data $\mathring{f}^\sharp$, this transport equation with initial data at $M_{t_0}$ is solved by $f^\sharp_{VTD}(t,w)=\mathring{f}(x,W_{t_0}(t,x,w))$, where the characteristic $W_{t_0}$ obeys
\begin{equation}\label{eq:char-vtd}
\dot{W}^j=-2(n_{VTD}-1)\frac{\dot{a}}aW^j+2{(k_{VTD}^\parallel)^j}_{i}W^i\,.
\end{equation}
as well as $W_{t_0}(t,x,w)=q$. Since $n_{VTD}$ converges to $1$, the leading term in \eqref{eq:char-vtd} is given by the leading term of the shear, i.e., one roughly has
\begin{equation}\label{eq:char-vtd-v}
W_{t_0}(t,x,w)=\exp\left[2\int_t^{t_0}a(s)^{-3}\,ds\cdot (K_{VTD}^{\parallel})_x\right]\odot w+\langle\text{lower order terms}\rangle\,.
\end{equation}
When performing the analysis above for $f^\sharp$, one must thus renormalise the particle velocities as follows to account for the VTD behaviour
\begin{align*}
f^\sharp_{\text{asymp}}(t,x,w)=&\,f^\sharp_{\mathrm{resc}}\left(t,x,\exp\left[-2\int_t^{t_0}a(s)^{-3}\,ds\cdot (K_{Bang}^{\parallel})_x\right]\odot w\right)\,.\\
=&\,f^\sharp\left(t,x,\exp\left[-2\int_t^{t_0}a(s)^{-3}\,ds\cdot (K_{Bang})_x\right]\odot q\right)\,.
\end{align*}
Indeed, $f^\sharp_{\text{asymp}}$ converges toward the Big Bang in a similar manner to $f$, as can be proven by a similar argument to that on the co-mass shell. However, these results can not be directly translated into one another on the Big Bang hypersurface itself since the spatial metric degenerates.\\


To interpret this asymptotic profile, we analyze $(K_{Bang}^{\parallel})_x$ more closely for an arbitrary point $x\in M$: First, we note that $\Sigma^\sharp$ is self-adjoint with respect to $\mathcal{G}$ on $M_t$ for any $t\in(0,t_0]$ and, consequently, $K_{Bang}^{\parallel}$ is self-adjoint with respect to $\mathcal{G}_{Bang}$ by taking the limit. In particular, $K_{Bang}^{\parallel}$  is diagonalisable. Let $E_{x,+}, E_{x,-}$  and $E_{x,0}$ denote the respective span of eigenspaces associated to positive, zero and negative eigenvalues of $(K_{Bang}^{\parallel})_x$. Note that, if $q\in T^\ast_xM$ is not orthogonal to $E_{x,+}$, $W_{t_0}(t,x,w(q))$ diverges as $t\downarrow 0$, and thus leaves the support of $f$ if $t$ is sufficiently close to $0$. Thus, approaching the Big Bang, the majority of particles have momenta that are orthogonal to $E_{x,+}$. Elements of $E_{x,0}$ are left unchanged by the evolution, while momenta in $E_{x,-}$ approach the origin as $t\downarrow 0$. Thus,  Vlasov particle velocities concentrate in the following manner, depending on the signs of eigenvalues of $(K_{Bang}^{\parallel})_x$:
\begin{enumerate}
\item $(+,+,-)$ or $(+,-,-)$: Particle velocities are attracted by $E_{x,-}$.
\item $(+,0,-)$: Particle velocities concentrate close to the Minkowski sum\linebreak $E_{x,-}+\left[E_{x,0}\cap \supp\mathring{f}(x,\cdot)\right]$.
\item $(0,0,0)$: This implies $f^\sharp_{\text{asymp}}(t,x,\cdot)=\mathring{f}(x,\cdot)$ for any $t\in(0,t_0)$, and particles are asymptotically distributed close to the FLRW distribution function.
\end{enumerate}
If the momentum support of $\mathring{f}$ is bounded away from the origin, the same interpretations hold, where $E_{x,-}$ can be replaced with $\lim_{r\to\infty}E_{x,-}\cap \{w\in T_xM \vert \lvert q\vert_\gamma=r\}$. Recall that we assume this property for $f_{FLRW}$ and $\mathring{f}$ in the massless case.\\}
\end{remark}

%
\section*{Declarations}

\noindent \textbf{Competing Interests.} David Fajman acknowledges support through the project \textit{Relativistic Fluids in cosmology} (project number P34313) of the Austrian Science Fund. Liam Urban acknowledges the support by the START-Project \textit{Isoperimetric study of initial data for the Einstein equations} (project number Y963), also by the Austrian Science Fund. Liam Urban is a recipient of a DOC Fellowship of the Austrian Academy of Sciences at the Institute of Mathematics at the University of Vienna, and \change{received a scholarship from the German Academic Scholarship Foundation during a majority of the work on this paper.}\\

\noindent \textbf{Data availability.} Data sharing is not applicable to this article as no datasets were generated or analysed during this work.

\bibliographystyle{alpha}
\bibliography{bibliography}

\newcommand{\etalchar}[1]{$^{#1}$}
\begin{thebibliography}{AnCGS22}

\bibitem[AF20]{AndFaj20}
Lars Andersson and David Fajman.
\newblock Nonlinear stability of the {M}ilne model with matter.
\newblock {\em Comm. Math. Phys.}, 378(1):261--298, 2020.

\bibitem[AnCGS22]{A-CGaSar22}
Rub\'{e}n~O. Acu\~{n}a C\'{a}rdenas, Carlos Gabarrete, and Olivier Sarbach.
\newblock An introduction to the relativistic kinetic theory on curved
  spacetimes.
\newblock {\em Gen. Relativity Gravitation}, 54(3):Paper No. 23, 120, 2022.

\bibitem[And11]{Andr11}
H{\aa}kan Andr{\'e}asson.
\newblock The {E}instein-{V}lasov system/{K}inetic {T}heory.
\newblock {\em Living Reviews in Relativity}, 14:1--55, 2011.

\bibitem[BF22]{BaFaj20}
Hamed Barzegar and David Fajman.
\newblock Stable cosmologies with collisionless charged matter.
\newblock {\em J. Hyperbolic Differ. Equ.}, 19(4):587--634, 2022.

\bibitem[BFJ{\etalchar{+}}21]{BFJSTh21}
L\'{e}o Bigorgne, David Fajman, J\'{e}r\'{e}mie Joudioux, Jacques Smulevici,
  and Maximilian Thaller.
\newblock Asymptotic stability of {M}inkowski space-time with non-compactly
  supported massless {V}lasov matter.
\newblock {\em Arch. Ration. Mech. Anal.}, 242(1):1--147, 2021.

\bibitem[BK73]{BK73}
Vladimir~A. Belinskiǐ and Isaak~M. Khalatnikov.
\newblock Effect of scalar and vector fields on the nature of the cosmological
  singularity.
\newblock {\em Soviet Journal of Experimental and Theoretical Physics}, 36:591,
  1973.

\bibitem[BKL70]{BKL70}
Vladimir~A. Belinskiǐ, Isaak~M. Khalatnikov, and Evgeny~M. Lifshitz.
\newblock Oscillatory approach to a singular point in the relativistic
  cosmology.
\newblock {\em Advances in Physics}, 19(80):525--573, 1970.

\bibitem[BL17]{BeyLeF17}
Florian Beyer and Philippe~G. LeFloch.
\newblock Self-gravitating fluid flows with {G}owdy symmetry near cosmological
  singularities.
\newblock {\em Comm. Partial Differential Equations}, 42(8):1199--1248, 2017.

\bibitem[BO24a]{BeyOl21}
Florian Beyer and Todd~A. Oliynyk.
\newblock Localized {B}ig {B}ang {S}tability for the {E}instein-{S}calar
  {F}ield {E}quations.
\newblock {\em Arch. Ration. Mech. Anal.}, 248(1):3, 2024.

\bibitem[BO24b]{BO24EESF}
Florian Beyer and Todd~A. Oliynyk.
\newblock Past stability of {FLRW} solutions to the {E}instein-{E}uler-scalar
  field equations and their {B}ig {B}ang singularites.
\newblock {\em Beijing Journal of Pure and Applied Mathematics}, 1(2):515--637,
  2024.

\bibitem[BT08]{BT08}
James Binney and Scott Tremaine.
\newblock {\em Galactic Dynamics, Second Edition}.
\newblock Princeton University Press, Princeton, 2008.

\bibitem[CBG69]{CBGer69}
Yvonne Choquet-Bruhat and Robert~P. Geroch.
\newblock Global aspects of the {C}auchy problem in general relativity.
\newblock {\em Communications in Mathematical Physics}, 14:329--335, 1969.

\bibitem[CR95]{CR95}
Piotr~T. Chrusciel and Alan~D. Rendall.
\newblock Strong {C}osmic {C}ensorship in {V}acuum {S}pace-{T}imes with
  {C}ompact, {L}ocally {H}omogeneous {C}auchy {S}urfaces.
\newblock {\em Annals of Physics}, 242(2):349–385, Sep 1995.

\bibitem[DR16]{DaRe16}
Mihalis Dafermos and Alan~D. Rendall.
\newblock Strong cosmic censorship for surface-symmetric cosmological
  spacetimes with collisionless matter.
\newblock {\em Comm. Pure Appl. Math.}, 69(5):815--908, 2016.

\bibitem[Eva98]{Evans98}
Lawrence~C. Evans.
\newblock {\em Partial Differential Equations}.
\newblock American Mathematical Society, 1st edition, 1998.

\bibitem[Faj16]{Faj16}
David Fajman.
\newblock Future asymptotic behavior of three-dimensional spacetimes with
  massive particles.
\newblock {\em Classical Quantum Gravity}, 33(11):11LT01, 9, 2016.

\bibitem[Faj17a]{Faj17}
David Fajman.
\newblock The nonvacuum {E}instein flow on surfaces of negative curvature and
  nonlinear stability.
\newblock {\em Comm. Math. Phys.}, 353(2):905--961, 2017.

\bibitem[Faj17b]{Faj17m0}
David Fajman.
\newblock Topology and incompleteness for {$2+1$}-dimensional cosmological
  spacetimes.
\newblock {\em Lett. Math. Phys.}, 107(6):1157--1176, 2017.

\bibitem[Faj18]{Faj18}
David Fajman.
\newblock The nonvacuum {E}instein flow on surfaces of nonnegative curvature.
\newblock {\em Comm. Partial Differential Equations}, 43(3):364--402, 2018.

\bibitem[FB52]{FB52}
Yvonne Four{\`e}s-Bruhat.
\newblock Th{\'e}or{\`e}me d'existence pour certains syst{\`e}mes
  d'{\'e}quations aux d{\'e}riv{\'e}es partielles non lin{\'e}aires.
\newblock {\em Acta Mathematica}, 88:141--225, 1952.

\bibitem[FJS21]{FJS21}
David Fajman, J{\'e}r{\'e}mie Joudioux, and Jacques Smulevici.
\newblock {The stability of the Minkowski space for the Einstein–Vlasov
  system}.
\newblock {\em Analysis \& PDE}, 14(2):425 -- 531, 2021.

\bibitem[FK20]{FajKr20}
David Fajman and Klaus Kr\"{o}ncke.
\newblock Stable fixed points of the {E}instein flow with positive cosmological
  constant.
\newblock {\em Comm. Anal. Geom.}, 28(7):1533--1576, 2020.

\bibitem[FRS23]{RodSpFou20}
Grigorios Fournodavlos, Igor Rodnianski, and Jared Speck.
\newblock Stable big bang formation for {E}instein's equations: the complete
  sub-critical regime.
\newblock {\em Journal of the American Mathematical Society}, 36(3):827--916,
  2023.

\bibitem[FU22]{FU22}
David Fajman and Liam Urban.
\newblock Blow-up of waves on singular spacetimes with generic spatial metrics.
\newblock {\em Lett. Math. Phys.}, 112(2):Paper No. 42, 31, 2022.

\bibitem[FU24]{FU23}
David Fajman and Liam Urban.
\newblock {C}osmic {C}ensorship near {FLRW} spacetimes with negative spatial
  curvature, 2024.
\newblock arXiv:2211.08052v3. To appear in \textit{Analysis \& PDE}.

\bibitem[GPR23]{GPR23}
Hans~Oude Groeniger, Oliver Petersen, and Hans Ringstr{\"o}m.
\newblock Formation of quiescent big bang singularities, 2023.
\newblock arXiv:2309.11370.

\bibitem[Haw67]{Hawk67}
Stephen~W. Hawking.
\newblock The occurrence of singularities in cosmology. iii. {C}ausality and
  {S}ingularities.
\newblock {\em Proceedings of the Royal Society of London. Series A,
  Mathematical and Physical Sciences}, 300(1461):187--201, 1967.

\bibitem[Li24]{Li24}
Warren Li.
\newblock Scattering towards the singularity for the wave equation and the
  linearized {E}instein-scalar field system in {K}asner spacetimes, 2024.
\newblock arXiv:2401.08437.

\bibitem[LT20]{LiTay20}
Hans Lindblad and Martin Taylor.
\newblock Global stability of {M}inkowski space for the {E}instein-{V}lasov
  system in the harmonic gauge.
\newblock {\em Arch. Ration. Mech. Anal.}, 235(1):517--633, 2020.

\bibitem[Mis69]{Misner69}
Charles~W Misner.
\newblock Mixmaster universe.
\newblock {\em Physical Review Letters}, 22(20):1071, 1969.

\bibitem[O'N83]{ONeill83}
Barrett O'Neill.
\newblock {\em Semi-Riemannian geometry with applications to relativity}.
\newblock Academic Press, 1983.

\bibitem[Pen65]{Pen65}
Roger Penrose.
\newblock Gravitational collapse and space-time singularities.
\newblock {\em Phys. Rev. Lett.}, 14:57--59, Jan 1965.

\bibitem[Rei96]{Rein96}
Gerhard Rein.
\newblock Cosmological solutions of the {V}lasov-{E}instein system with
  spherical, plane, and hyperbolic symmetry.
\newblock {\em Math. Proc. Cambridge Philos. Soc.}, 119(4):739--762, 1996.

\bibitem[Rin01]{Rin01}
Hans Ringstr\"{o}m.
\newblock The {B}ianchi {IX} attractor.
\newblock {\em Ann. Henri Poincar\'{e}}, 2(3):405--500, 2001.

\bibitem[Rin09a]{Ring09}
Hans Ringstr\"{o}m.
\newblock {\em The {C}auchy problem in general relativity}.
\newblock ESI Lectures in Mathematics and Physics. European Mathematical
  Society (EMS), Z\"{u}rich, 2009.

\bibitem[Rin09b]{Ring09Annals}
Hans Ringstr\"{o}m.
\newblock Strong cosmic censorship in {$T^3$}-{G}owdy spacetimes.
\newblock {\em Ann. of Math. (2)}, 170(3):1181--1240, 2009.

\bibitem[Rin13]{Rin13}
Hans Ringstr\"{o}m.
\newblock {\em On the topology and future stability of the universe}.
\newblock Oxford Mathematical Monographs. Oxford University Press, Oxford,
  2013.

\bibitem[RS18a]{Rodnianski2018}
Igor Rodnianski and Jared Speck.
\newblock A regime of linear stability for the {E}instein-scalar field system
  with applications to nonlinear {B}ig {B}ang formation.
\newblock {\em Ann. of Math. (2)}, 187(1):65--156, 2018.

\bibitem[RS18b]{Rodnianski2014}
Igor Rodnianski and Jared Speck.
\newblock Stable big bang formation in near-{FLRW} solutions to the
  {E}instein-scalar field and {E}instein-stiff fluid systems.
\newblock {\em Selecta Math. (N.S.)}, 24(5):4293--4459, 2018.

\bibitem[RS22]{RodSp22}
Igor Rodnianski and Jared Speck.
\newblock On the nature of {H}awking's incompleteness for the {E}instein-vacuum
  equations: the regime of moderately spatially anisotropic initial data.
\newblock {\em J. Eur. Math. Soc. (JEMS)}, 24(1):167--263, 2022.

\bibitem[Ryd16]{Ryd16}
Barbara Ryden.
\newblock {\em Introduction to Cosmology}.
\newblock Cambridge University Press, 2 edition, 2016.

\bibitem[Sha11]{Shao11}
Arick Shao.
\newblock On breakdown criteria for nonvacuum {E}instein equations.
\newblock {\em Ann. Henri Poincar\'{e}}, 12(2):205--277, 2011.

\bibitem[Spe18]{Speck2018}
Jared Speck.
\newblock The maximal development of near-{FLRW} data for the
  {E}instein-{S}calar {F}ield system with spatial topology $\mathbb{S}^3$.
\newblock {\em Communications in Mathematical Physics}, 364(3):879–979, Oct
  2018.

\bibitem[Sve12]{Sve12}
Christopher Svedberg.
\newblock {\em Future stability of the Einstein-Maxwell-Scalar field system and
  non-linear wave equations coupled to generalized massive-massless Vlasov
  equations}.
\newblock PhD thesis, KTH, Mathematics (Div.), 2012.
\newblock QC 20120503.

\bibitem[Tay17]{Tay17}
Martin Taylor.
\newblock The global nonlinear stability of {M}inkowski space for the massless
  {E}instein-{V}lasov system.
\newblock {\em Ann. PDE}, 3(1):Paper No. 9, 177, 2017.

\bibitem[TNR04]{TNR04}
David Tegankong, Norbert Noutchegueme, and Alan~D. Rendall.
\newblock Local existence and continuation criteria for solutions of the
  {E}instein-{V}lasov-scalar field system with surface symmetry.
\newblock {\em J. Hyperbolic Differ. Equ.}, 1(4):691--724, 2004.

\bibitem[TR06]{TR06}
David Tegankong and Alan~D. Rendall.
\newblock On the nature of initial singularities for solutions of the
  {E}instein-{V}lasov-scalar field system with surface symmetry.
\newblock {\em Math. Proc. Cambridge Philos. Soc.}, 141(3):547--562, 2006.

\bibitem[Urb24]{U24}
Liam Urban.
\newblock Quiescent {B}ig {B}ang formation in $2+1$ dimensions, 2024.
\newblock arXiv:2412.03396. Submitted to \textit{Annales Henri Poincar{\'e}}.

\bibitem[WE97]{WE97}
John Wainwright and George Ellis.
\newblock {\em Dynamical Systems in Cosmology}.
\newblock Cambridge University Press, 1997.

\end{thebibliography}

\end{document}